\documentclass[12pt]{article} 

\usepackage{amssymb,amsfonts,amsmath,latexsym,amsthm}
\usepackage[usenames]{color}
\usepackage[]{graphicx}
\usepackage[space]{grffile}
\usepackage{mathrsfs}   % fancy math font
\usepackage{caption}
\usepackage{subcaption}
\usepackage{setspace}
\usepackage{verbatim}
\usepackage{url}
\usepackage{bm}
\usepackage{extarrows}
\usepackage{enumitem}
\usepackage{dsfont}
\usepackage{epstopdf}
\usepackage[colorlinks,citecolor=black,linkcolor=black,urlcolor=black]{hyperref}
\usepackage[authoryear,round]{natbib}
\usepackage{tikz}
\usetikzlibrary{fit}					% fitting shapes to coordinates
\usetikzlibrary{backgrounds}	% drawing the background after the foreground

\renewcommand*{\arraystretch}{0.85}  % reduce space between rows of matrices
\theoremstyle{plain}
\newtheorem{theorem}{Theorem}[section]

\newtheorem{proposition}[theorem]{Proposition}
\newtheorem{condition}[theorem]{Condition}

\theoremstyle{definition}

\numberwithin{equation}{section}
\renewcommand{\L}{\mathcal{L}}
\newcommand{\ps}{\texttt{\small\textup{+}}}
\newcommand{\bc}{\bm{:}}
\newcommand{\Id}{\mathrm{I}}
\renewcommand{\tt}{\texttt}
\newcommand{\Diag}{\mathrm{Diag}}

\DeclareMathOperator*{\Bernoulli}{Bernoulli}

\DeclareMathOperator*{\Ga}{Gamma}
\DeclareMathOperator*{\Poisson}{Poisson}
\DeclareMathOperator*{\NegBin}{NegBin}
\DeclareMathOperator*{\LogNormal}{LogNormal}
\DeclareMathOperator*{\LNP}{LNP}
\DeclareMathOperator*{\Geometric}{Geometric}

\DeclareMathOperator*{\Var}{Var}

\DeclareMathOperator*{\diag}{diag}

\newcommand{\T}{\mathtt{T}}
\newcommand{\R}{\mathbb{R}}
\newcommand{\Z}{\mathbb{Z}}
\newcommand{\E}{\mathbb{E}}
\renewcommand{\Pr}{\mathbb{P}}
\newcommand{\I}{\mathds{1}}

\newcommand{\N}{\mathcal{N}}

\newcommand{\branch}[4]{
\left\{
	\begin{array}{ll}
		#1  & \mbox{if } #2 \\
		#3 & \mbox{if } #4
	\end{array}
\right.
}

\def\app#1#2{%
  \mathrel{%
    \setbox0=\hbox{$#1\sim$}%
    \setbox2=\hbox{%
      \rlap{\hbox{$#1\propto$}}%
      \lower1.3\ht0\box0%
    }%
    \raise0.25\ht2\box2%
  }%
}

\graphicspath{{figures/}}

\usepackage[authoryear,round]{natbib}

\allowdisplaybreaks

\newcommand{\blind}{0}

\begin{document}

\def\spacingset#1{\renewcommand{\baselinestretch}%
{#1}\small\normalsize} \spacingset{1}

%%%%%%%%%%%%%%%%%%%%%%%%%%%%%%%%%%%%%%%%%%%%%%%%%%%%%%%%%%%%%%%%%%%%%%%%%%%%%%

\if0\blind
{
  \title{\bf Inference in generalized bilinear models}
  \author{Jeffrey W. Miller\thanks{
	J.W.M.\ was supported by the Department of Defense grant W81XWH-18-1-0357,
    the National Institutes of Health grant R01GM083084, and
	the Zhu Family Center for Global Cancer Prevention.
	S.L.C.\ was supported by National Cancer Institute grant 1R01CA227156-01, the Department of Defense grant W81XWH-18-1-0357, the American Brain Tumor Association, the Wong Family Awards in Translational Oncology, and the LUNGstrong Fund.
	S.L.C.\ is also affiliated with the Department of Biostatistics at Harvard University and the Broad Institute of MIT and Harvard.
	}\hspace{.2cm}\\
    Department of Biostatistics, Harvard University \\
    and \\
    Scott L. Carter \\
    Division of Computational Biology, Dana--Farber Cancer Institute}
  \maketitle
} \fi
% Department of Biostatistics, Harvard T.H.\ Chan School of Public Health \\

\if1\blind
{
  \bigskip
  \bigskip
  \bigskip
  \begin{center}
    {\LARGE\bf Inference in generalized bilinear models}
\end{center}
  \medskip
} \fi

\bigskip
\begin{abstract}
Latent factor models are widely used to discover and adjust for hidden variation in modern applications.  However, most methods do not fully account for uncertainty in the latent factors, which can lead to miscalibrated inferences such as overconfident p-values.  In this article, we develop a fast and accurate method of uncertainty quantification in generalized bilinear models, which are a flexible extension of generalized linear models to include latent factors as well as row covariates, column covariates, and interactions.  In particular, we introduce delta propagation, a general technique for propagating uncertainty among model components using the delta method.  Further, we provide a rapidly converging algorithm for \textit{maximum a posteriori} GBM estimation that extends earlier methods by estimating row and column dispersions.  In simulation studies, we find that our method provides approximately correct frequentist coverage of most parameters of interest.  We demonstrate on RNA-seq gene expression analysis and copy ratio estimation in cancer genomics.
\end{abstract}

\noindent%
\vfill
% generalized linear model,
% genomics, 
% latent factor, 

\spacingset{1.6}
\newpage

\section{Introduction}
\label{section:introduction}

% Outline (overall: ~1.5 pages)
% - Background/motivation (~1 paragraphs)
% - Limitations of current methods
% - In this paper, we ... (~1 paragraph)
% - Novel contributions (~1 paragraph)
% - Organization of the paper (1 paragraph)

% (What is the problem?)

Latent factor models have become an essential tool for analyzing and adjusting for hidden sources of variation in complex data.
Building on generalized linear model theory, % and latent factor models with normal outcomes,
generalized bilinear models (GBMs) provide a flexible framework incorporating 
latent factors along with row covariates, column covariates, and interactions to analyze matrix data
\citep{choulakian1996generalized,gabriel1998generalised,de2000gbms,perry2013degrees,hoff2015multilinear,buettner2017f}.
% matrices of non-normal data
However, uncertainty quantification is a persistent problem in large complex models, and 
GBMs are particularly challenging due to the
non-linearity of multiplicative terms,
constraints and strong dependencies among parameters, 
inapplicability of normal model theory,
and the fact that the number of parameters grows with the data.

% (Limitations of current approaches)

% For example, in genomics and genetics,
Most latent factor methods do not fully account for uncertainty in the latent factors.
% which may lead to spurious findings due to overconfident inferences.
For example, to remove batch effects in gene expression analysis, several methods first estimate a factorization $U V^\T$ and then 
treat $V$ as a known matrix of covariates, accounting for uncertainty only in $U$  % using standard regression %techniques
\citep{leek2007capturing,leek2008general,sun2012multiple,risso2014normalization}.
In copy number variation detection, it is common to treat the estimated $U V^\T$ as known and subtract it off \citep{fromer2012discovery,krumm2012copy,jiang2015codex}.
In principle, Bayesian inference provides full uncertainty quantification \citep{carvalho2008high},
however, Markov chain Monte Carlo tends to be slow in large parameter spaces with strong dependencies, as in the case of GBMs.
Variational Bayes approaches are faster \citep{stegle2010bayesian,buettner2017f,babadi18gatk},
but rely on factorized approximations that tend to underestimate uncertainty.
Meanwhile, the classical method of inverting the Fisher information matrix is computationally prohibitive in large GBMs.

% Some approaches instead use variational Bayes \citep{stegle2010bayesian,buettner2017f,babadi18gatk},
% however, these methods employ fully factorized approximations that cannot represent dependencies among parameters and tend to underestimate uncertainty.
% Finally, classical frequentist theory is not directly applicable, since in large GBMs, 
% it is computionally prohibitive to invert the full constraint-augmented Fisher information matrix \citep{silvey1975statistical}
% and methods based on normal model theory \citep{perry2013degrees} do not apply.

% (In this paper, we ...)

In this article, we introduce a novel method for uncertainty quantification in GBMs, focusing on the case of count data with negative binomial outcomes.
The basic idea is to propagate uncertainty between model components using the delta method,
which can be done analytically using closed-form expressions involving the gradient and the Fisher information; we refer to this as \textit{delta propagation}.
The method facilitates computation of p-values and confidence intervals with approximately correct frequentist properties in GBMs.
Further, we provide an algorithm for \textit{maximum a posteriori} GBM estimation
that extends previous work by 
estimating row- and column-specific dispersion parameters,
improving numerical stability,
and explicitly handling identifiability constraints.

In a suite of simulation studies, we find that our methods perform favorably in terms of
consistency, frequentist coverage, computation time, algorithm convergence, and robustness to the outcome distribution.
We then apply our methods to gene expression analysis,
(a) comparing performance with DESeq2 \citep{love2014moderated} on a benchmark dataset of RNA-seq samples from lymphoblastoid cell lines,
and (b) testing for age-related genes using RNA-seq data from the Genotype-Tissue Expression (GTEx) project \citep{mele2015human}.
Finally, we apply our methods to copy ratio estimation in cancer genomics,
comparing performance with the Genome Analysis Toolkit (GATK) \citep{gatk}
on whole-exome sequencing data from the Cancer Cell Line Encyclopedia \citep{ghandi2019next}.

The article is organized as follows.
In Section~\ref{section:model}, we define the GBM model and we address identifiability, interpretability, and residuals.
In Sections~\ref{section:estimation} and \ref{section:inference}, we describe our estimation and inference methods, respectively.
In Section~\ref{section:theory}, we establish theoretical results,
% on identifiability, interpretability, likelihood-preserving projections, and computational complexity.
and Section~\ref{section:simulations} contains simulation studies.
In Sections~\ref{section:gene-expression} and \ref{section:cancer}, we apply our methods to gene expression analysis
and copy ratio estimation in cancer genomics.
The supplementary material contains a discussion of previous work and challenges,
additional empirical results, mathematical derivations and proofs,
and step-by-step algorithms.

\iffalse

% residuals and their estimated standard deviations

% Starting from a baseline covariance equal to the inverse Fisher information matrix of each parameter block,
% our proposed inference algorithm requires a single pass through the model.

Methodological contributions
- fast and accurate UQ for nearly all parameters, scalable to large matrices
    --- novel general method: delta propagation
- including dispersions in the model, and estimating them --- this is actually nontrivial
- rapidly converging, linear time estimation algorithm that is numerically stable

GBMs are a generalization of both bilinear models (cite Takane and Shibayama (1991)) and generalized linear models (McCullagh and Nelder (1989))

% Define acronym GLM since we use it later.

% Cite gnomAD as example of giant collection of WES/WGS data: https://gnomad.broadinstitute.org/blog/2018-10-gnomad-v2-1/

Estimation in large GBMs is computationally and statistically challenging due to
(a) the high-dimensionality of the parameter space,
(b) the need to enforce identifiability constraints, and
(c) the nonlinearities and dependencies induced by the latent factor term $U D V^\T$.
We address these issues with a fast constrained optimization algorithm for \textit{maximum a posteriori} (MAP) estimation in EDF-GBMs.

% While previous authors have considered various special cases, 
% to our knowledge, this paper is the first to \todo{briefly make claims of novelty here};
% see the discussion of previous work in Section~\ref{section:previous-work}.

\fi

\section{Model}
\label{section:model}

In this section, we define the class of models considered in this paper
and we provide conditions guaranteeing identifiability and interpretability of the model parameters.
% We consider the following class of regression models with feature covariates, sample covariates, interactions, and latent factors.
For $i \in \{1,\ldots,I\}$ and $j \in \{1,\ldots,J\}$, suppose $Y_{i j}\in\R$ is a random variable such that 
\begin{equation}
\label{equation:model-univariate}
g(\E(Y_{i j})) = \sum_{k = 1}^K x_{i k} a_{j k} + \sum_{\ell = 1}^L b_{i \ell} z_{j \ell} + \sum_{k = 1}^K \sum_{\ell = 1}^L x_{i k} c_{k \ell} z_{j \ell}
+ \sum_{m = 1}^M u_{i m} d_{m m} v_{j m}
\end{equation}
where $x_{i k}$ and $z_{j\ell}$ are observed covariates, $a_{j k}$, $b_{i\ell}$, $c_{k \ell}$, $u_{i m}$, $d_{m m}$, and $v_{j m}$ are parameters to be estimated,
and $g(\cdot)$ is a smooth function such that $g'$ is positive, referred to as the link function.
In matrix form, denoting $\bm{Y} = (Y_{i j})\in\R^{I\times J}$, Equation~\ref{equation:model-univariate} is equivalent to
\begin{equation}
\label{equation:model}
g(\E(\bm{Y})) = X A^\T + B Z^\T + X C Z^\T + U D V^\T
\end{equation}
where $g$ is applied element-wise to the matrix $\E(\bm{Y})$.
To be able to use capital $Y$ to denote scalar random variables, we use bold $\bm{Y}$ to denote the data matrix. 
% and unbold $Y$ for its scalar entries.
% The matrices $X$ and $Z$ contain observed covariates, and the matrices $A$, $B$, $C$, $D$, $U$, and $V$ are parameters to be estimated.
We refer to this as a \textit{generalized bilinear model} (GBM), following the terminology of \citet{choulakian1996generalized}.

In the genomics applications in Sections~\ref{section:gene-expression} and \ref{section:cancer},
we use negative binomial outcomes with log link $g(\mu) = \log(\mu)$, and
the role of each piece is as follows (see Figure~\ref{figure:model}):
$Y_{i j}$ is the read count for feature $i$ in sample $j$,
$X\in\R^{I\times K}$ contains feature covariates and $A\in\R^{J\times K}$ contains the corresponding coefficients,
$Z\in\R^{J\times L}$ contains sample covariates and $B\in\R^{I\times L}$ contains the corresponding coefficients,
$C\in\R^{K\times L}$ contains intercepts and coefficients for interactions between the $x$'s and $z$'s, and
$U D V^\T$ is a low-rank matrix that captures latent effects due, for example, to unobserved covariates such as batch.
% For gene expression RNA-seq data, the features are genes or transcripts.    
% For cancer coverage data, the features are targets or bins in the case of whole-exome or whole-genome sequencing, respectively.

\begin{figure}
  \centering
  \includegraphics[width=1\textwidth]{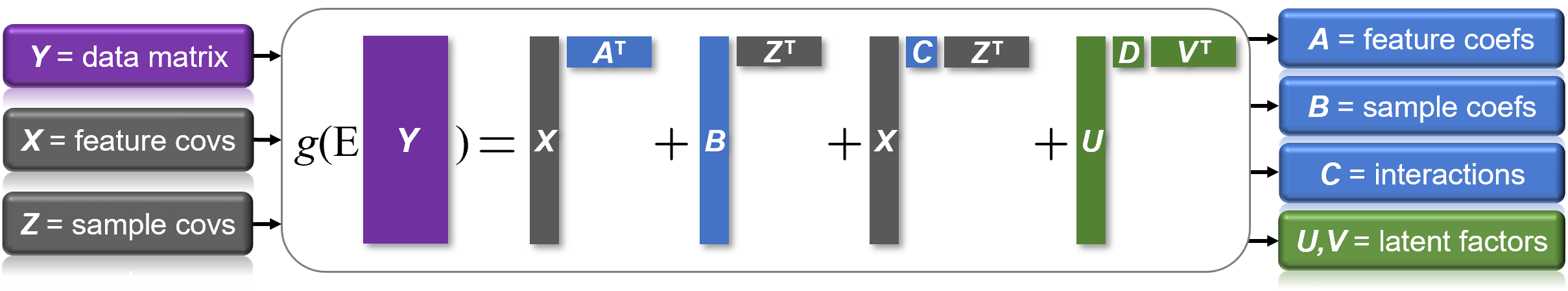}
  \caption{Schematic of the generalized bilinear model (GBM) structure.}
  \label{figure:model}
\end{figure}

\subsection{Identifiability and interpretation}
\label{section:identifiability}

For identifiability and interpretability, we impose certain constraints (Conditions~\ref{condition:identifiability} and \ref{condition:interpretation}).
The only constraints on the covariates are that $X$ and $Z$ are full-rank, centered, and include a column of ones for intercepts; see below.
% --- or more precisely, that (i) $X^\T X$ and $Z^\T Z$ are invertible, (ii) $\sum_{i=1}^I x_{i k} = 0$ and $\sum_{j=1}^J z_{j \ell} = 0$ for all $k\geq 2$, $\ell \geq 2$, and (iii) $x_{i 1} = 1$ and $z_{j 1} = 1$ for all $i,j$.
% These are standard conditions on design matrices for regression.
The rest of the constraints are enforced on the parameters during estimation.
We write $\Id$ to denote the identity matrix to distinguish it from the number of features, $I$.

\begin{condition}[Identifiability constraints]
\label{condition:identifiability}
Assume the following constraints:
\begin{enumerate}[label=\textup{(\alph*)}, ref=\ref{condition:identifiability}(\alph*)]
\item\label{condition:invertible} $X^\T X$ and $Z^\T Z$ are invertible,
\item\label{condition:orthogonal} $X^\T B = 0$, $Z^\T A = 0$, $X^\T U = 0$, and $Z^\T V = 0$,
\item\label{condition:orthonormal} $U^\T U = \Id$ and $V^\T V = \Id$,
\item\label{condition:diagonal} $D$ is a diagonal matrix such that $d_{1 1} > d_{2 2} > \cdots > d_{M M} > 0$, and
\item\label{condition:sign} the first nonzero entry of each column of $U$ is positive,
\end{enumerate}
where $A\in \R^{J\times K}$, $B\in\R^{I\times L}$, $C\in\R^{K\times L}$, $D\in\R^{M\times M}$, $U\in\R^{I\times M}$, $V\in\R^{J\times M}$, 
$X\in\R^{I\times K}$, $Z\in\R^{J\times L}$, and $M < \min\{I,J\}$.
\end{condition}

In Theorem~\ref{theorem:identifiability}, we show that Condition~\ref{condition:identifiability} is sufficient to guarantee identifiability 
of $A$, $B$, $C$, $D$, $U$, and $V$ in any GBM satisfying Equation~\ref{equation:model}. More precisely, letting
\begin{align} \label{equation:eta}
\eta(A,B,C,D,U,V) := X A^\T + B Z^\T + X C Z^\T + U D V^\T
\end{align}
for some fixed full-rank $X$ and $Z$, Theorem~\ref{theorem:identifiability} shows that $\eta(\cdot)$ is a one-to-one function on the set of parameters satisfying Condition~\ref{condition:identifiability}.
%, provided that $X^\T X$ and $Z^\T Z$ are invertible.
% \todo{Maybe mention that we are (essentially) using Agresti's definition of identifiability for GLMs (unique distn for each X and Z)? See 1.4.1 of Agresti's GLM book.}
% We also impose the following interpretability constraints.
\begin{condition}[Interpretability constraints]
\label{condition:interpretation}
Assume that \textup{(a)} $x_{i 1} = 1$ and $z_{j 1} = 1$ for all $i,j$, and \textup{(b)} $\sum_{i=1}^I x_{i k} = 0$ and $\sum_{j=1}^J z_{j \ell} = 0$ for all $k\geq 2$, $\ell \geq 2$.
% \begin{enumerate}[label=\textup{(\alph*)}, ref=\ref{condition:interpretation}(\alph*)]
% \item\label{condition:intercept} $x_{i 1} = 1$ and $z_{j 1} = 1$ for all $i,j$, and
% \item\label{condition:centered} $\sum_{i=1}^I x_{i k} = 0$ and $\sum_{j=1}^J z_{j \ell} = 0$ for all $k\geq 2$, $\ell \geq 2$.
% \end{enumerate}
\end{condition}
\noindent When Condition~\ref{condition:interpretation}(a) holds, we can rearrange the right-hand side of Equation~\ref{equation:model-univariate} as
\begin{align*}
% \label{equation:interpretation}
% g(\E(Y_{i j})) = 
c_{1 1} + a_{j 1} + b_{i 1} + \sum_{k = 2}^K (c_{k 1} + a_{j k}) x_{i k} + \sum_{\ell = 2}^L (c_{1 \ell} + b_{i \ell}) z_{j \ell}  
+ \sum_{k = 2}^K \sum_{\ell = 2}^L c_{k \ell} x_{i k} z_{j \ell} + \sum_{m = 1}^M u_{i m} d_{m m} v_{j m}. 
\end{align*}
Using this and assuming Conditions~\ref{condition:identifiability} and \ref{condition:interpretation}, 
we show in Theorem~\ref{theorem:interpretation} that the interpretation of each parameter is:
(1) $c_{1 1}$ is the overall intercept, $c_{1 1} = \frac{1}{I J}\sum_{i = 1}^I \sum_{j = 1}^J g(\E(Y_{i j}))$,
(2) $a_{j 1}$ is a sample-specific offset and $b_{i 1}$ is a feature-specific offset,
(3) $c_{k 1}$ is the mean effect of the $k$th feature covariate and $a_{j k}$ is the sample-specific offset of this effect,
(4) $c_{1 \ell}$ is the mean effect of the $\ell$th sample covariate and $b_{i \ell}$ is the feature-specific offset of this effect, and
(5) $c_{k \ell}$ is the effect of the interaction $x_{i k} z_{j \ell}$ for $k\geq 2$ and $\ell \geq 2$.

% \begin{enumerate}
% \item $c_{1 1}$ is the overall intercept, $\displaystyle c_{1 1} = \frac{1}{I J}\sum_{i = 1}^I \sum_{j = 1}^J g(\E(Y_{i j}))$,
% \item $a_{j 1}$ is a sample-specific offset and $b_{i 1}$ is a feature-specific offset,
% \item $c_{k 1}$ is the mean effect of the $k$th feature covariate ($x_{i k}$), and $a_{j k}$ is the sample-specific offset of this effect for sample $j$,
% \item $c_{1 \ell}$ is the mean effect of the $\ell$th sample covariate ($z_{j \ell}$), and $b_{i \ell}$ is the feature-specific offset of this effect for feature $i$, and
% \item $c_{k \ell}$ is the effect of the interaction $x_{i k} z_{j \ell}$ for $k\geq 2$ and $\ell \geq 2$.
% \end{enumerate}

GBMs can be decomposed in terms of the sum-of-squares of each component's contribution to the overall model fit,
enabling one to interpret the proportion of variation explained by each component.
Specifically, in Theorem~\ref{theorem:sum-of-squares}, we show that 
$$ \mathrm{SS}(X A^\T + B Z^\T + X C Z^\T + U D V^\T) = \mathrm{SS}(X A^\T) + \mathrm{SS}(B Z^\T) + \mathrm{SS}(X C Z^\T) + \mathrm{SS}(U D V^\T) $$
whenever Condition~\ref{condition:orthogonal} holds,
where $\mathrm{SS}(Q) := \sum_{i,j} q_{i j}^2$ denotes the sum of squares of the entries of a matrix.
This extends a similar result by \citet{takane1991principal}.
% in the case of Normal outcomes and $X A^\T + B Z^\T + X C Z^\T$ only.

\subsection{Outcome distributions}
\label{section:outcome}

% \textbf{Exponential dispersion families (EDFs).}
For the distribution of $Y_{i j}$, we focus on discrete exponential dispersion families \citep{jorgensen1987exponential,agresti2015foundations}.
Specifically, suppose $Y_{i j}\sim f(y \mid \theta_{i j},r_{i j})$ where for $y\in\Z$,
\begin{equation}
\label{equation:edf}
f(y\mid\theta,r) = \exp(\theta y - r\kappa(\theta)) h(y,r)
\end{equation}
is a probability mass function for all $\theta\in\Theta$ and $r\in\mathcal{R}$;
this is referred to as a \textit{discrete EDF}.
Here, $\Theta\subseteq\R$, $\mathcal{R}\subseteq (0,\infty)$, and $\Z$ denotes the integers.  % = \{\ldots,-2,-1,0,1,2,\ldots\}$.
For any discrete EDF, the mean and variance are
$\E(Y) = r\kappa'(\theta)$ and $\Var(Y) = r\kappa''(\theta)$ \citep{jorgensen1987exponential};
also see Section~\ref{section:edfs}.
We refer to $r$ as the inverse dispersion parameter, and $1/r$ is the dispersion.
% Results for discrete EDFs can be translated into results for standard EDFs of the form $\exp(r[\theta y - \kappa(\theta)]) h(y,r)$ via the transformation $y \mapsto r y$.
Discrete EDFs can be translated into standard EDFs of the form $\exp(r[\theta y - \kappa(\theta)]) h(y,r)$ via the transformation $y \mapsto r y$.
We refer to a GBM with EDF outcomes as an EDF-GBM.

\textbf{Negative binomial (NB) outcomes.}
In the applications in Sections~\ref{section:gene-expression} and \ref{section:cancer},
we use negative binomial outcomes: $Y_{i j}\sim \NegBin(\mu_{i j},r_{i j})$
where $\mu_{i j}$ is the mean and $1/r_{i j}$ is the dispersion.
This is a discrete EDF as in Equation~\ref{equation:edf} with 
$\theta = \log(\mu/(\mu + r))$ and $\kappa(\theta) = -\log(1 - \exp(\theta))$;
% and thus, it is straightforward to verify that 
thus, $\E(Y) = \mu$ and $\mathrm{Var}(Y) = \mu + \mu^2/r$.
The NB distribution is an overdispersed Poisson
since if $Y|\lambda \sim \Poisson(\lambda)$ and $\lambda\sim \Ga(r,\,r/\mu)$, then
integrating out $\lambda$, we have $Y\sim\NegBin(\mu,r)$.
We refer to a GBM with NB outcomes as an NB-GBM.

We parametrize the dispersions as $1/r_{i j} = \exp(s_i + t_j + \omega)$
and work in terms of $S = (s_1,\ldots,s_I)^\T\in\R^I$, $T = (t_1,\ldots,t_J)^\T\in\R^J$,
and $\omega\in\R$,
subject to the identifiability constraints $\frac{1}{I}\sum_i e^{s_i} = 1$ and $\frac{1}{J}\sum_j e^{t_j} = 1$.
Note that this makes $\frac{1}{I J}\sum_{i,j} 1/r_{i j} = \exp(\omega)$.
% The matrix of inverse dispersions is denoted $R = (r_{i j})\in\R^{I\times J}$.

% and we work in terms of the feature-specific log-dispersion offsets $S = (s_1,\ldots,s_I)^\T\in\R^I$,
% the sample-specific log-dispersion offsets $T = (t_1,\ldots,t_J)^\T\in\R^J$,

% \begin{align}
% \label{equation:negative-binomial}
% \NegBin(y \mid \mu,r) = \frac{\Gamma(r + y)}{\Gamma(y+1)\Gamma(r)} \Big(\frac{\mu}{\mu + r}\Big)^y \Big(\frac{r}{\mu + r}\Big)^r
% \end{align}
% for $\mu>0$ and $r>0$.
% Note that in this parametrization $\mu$ is the mean and $1/r$ is the dispersion.

\subsection{Residuals and adjusting out selected effects}
\label{section:residuals}

Residuals are useful for many purposes, such as visualization, model criticism, and downstream analyses.
We define GBM residuals as $\varepsilon_{i j} := g(Y_{i j} + \epsilon) - \eta_{i j}$
where $\eta = \eta(A,B,C,D,U,V) \in\R^{I \times J}$ as in Equation~\ref{equation:eta}
and $\epsilon$ is a small constant to make $\varepsilon_{i j}$ well-defined; for NB-GBMs, we use $\epsilon = 1/8$ as a default.
% \todo{Should we use $\hat{\eta}_{i j}$ here and below? Need to be consistent with notation.}
% \todo{Do we need to reintroduce the $\mu_{i j}$ notation in the model definition?}
A model-based estimate of the variance of $\varepsilon_{i j}$ is given by $\sigma^2_{i j} g'(\mu_{i j})^2$
where $\mu_{i j}$ and $\sigma_{i j}^2$ are the mean and variance of $Y_{i j}$ under the model.
% where $\sigma_{i j}^2 = \Var(Y_{i j}\mid\theta_{i j},r_{i j})$ in the notation of Section~\ref{section:outcome}.
This formula can be derived either from the Fisher information or from a first-order Taylor approximation to $g$.
It turns out that the corresponding precisions $w_{i j} := 1/(\sigma^2_{i j} g'(\mu_{i j})^2)$
play a key role in our GBM estimation and inference algorithms.
% so they are readily available as an output of the algorithms.
In the NB case with $g(\mu) = \log(\mu)$, these residual precisions are $w_{i j} = r_{i j} \mu_{i j} / (r_{i j} + \mu_{i j})$.

Often, it is useful to adjust out some effects but not others.
Let $\mathcal{R}_x$, $\mathcal{R}_z$, and $\mathcal{R}_u$ be the indices of the columns of $X$, $Z$, and $U$ (or $V$) that one does not wish to adjust out.
We define the partial residuals $\varepsilon^\mathcal{R}_{i j} := \eta^\mathcal{R}_{i j} + \varepsilon_{i j}$
where $\eta^\mathcal{R}\in\R^{I \times J}$ is defined as in Equation~\ref{equation:eta} but with 
$x_{i k}$, $z_{j \ell}$, and $u_{i m}$ replaced by 
$x_{i k} \I(k\in \mathcal{R}_x)$, $z_{j \ell} \I(\ell\in \mathcal{R}_z)$, and $u_{i m} \I(m\in \mathcal{R}_u)$.

\section{Estimation}
\label{section:estimation}

We provide a general algorithm for estimating the parameters of a discrete EDF-GBM (Equations~\ref{equation:model} and \ref{equation:edf}),
and we augment the algorithm to estimate the NB-GBM dispersion parameters as well.
% , that is, a generalized bilinear model (Equation~\ref{equation:model}) with discrete exponential dispersion family outcomes (Equation~\ref{equation:edf}).
% Given the dispersions, NB-GBMs are covered as a special case,
Here, we only give an outline; see Section~\ref{section:estimation-details} for step-by-step details. 
% \todo{Check whether this is true: It turns out that the algorithm is identical for standard EDF and discrete EDF outcomes.}

% Thus, the basic idea of the algorithm is straightforward, but there are a number of nontrivial issues, % that must be overcome,
% as discussed in Section~\ref{section:challenges}.
% The algorithm below has been carefully constructed to overcome these challenges.

% Estimation in large GBMs is computationally and statistically challenging due to the issues discussed in Section~\ref{section:challenges}.
% We address these issues with a fast optimization algorithm for \textit{maximum a posteriori} (MAP) estimation in EDF-GBMs.
% an unconstrained optimization step followed by a likelihood-preserving projection onto the constrained space.

\noindent\textbf{Inputs.}
The required inputs are $\bm{Y}\in\Z^{I\times J}$, $X\in\R^{I\times K}$, $Z\in\R^{J\times L}$, and $M\in\{0,1,2,\ldots\}$.
% the data matrix $\bm{Y}\in\Z^{I\times J}$, the
% feature covariate matrix $X\in\R^{I\times K}$, the sample covariate matrix $Z\in\R^{J\times L}$, and the number of latent factors $M\in\{0,1,2,\ldots\}$.
Optional inputs are the maximum step size $\rho > 0$,
prior precisions $\lambda_a,\lambda_b,\lambda_c,\lambda_d,\lambda_u,\lambda_v > 0$,
prior means and precisions $m_s,m_t,\lambda_s,\lambda_t$ for the log-dispersions in the NB-GBM case,
and the convergence criterion.
As defaults, we use $\rho = 5$, $\lambda_a = \lambda_b = \lambda_c = \lambda_d = \lambda_u = \lambda_v = 1$,
$m_s = m_t = 0$, $\lambda_s = \lambda_t = 1$, convergence tolerance $\tau = 10^{-6}$ for the relative change in log-likelihood+log-prior,
and a maximum of $50$ iterations.

\noindent\textbf{Preprocessing.}
    \begin{enumerate}[label=\textup{(\arabic*)},noitemsep,topsep=0pt]
    \item Ensure that $X$ and $Z$ satisfy Condition~\ref{condition:invertible} and Condition~\ref{condition:interpretation}.
    \item Unless the covariates are already on a common scale in terms of units, standardize $X$ and $Z$
    such that $\frac{1}{I}\sum_{i=1}^I x_{i k}^2 = 1$ and $\frac{1}{J}\sum_{j=1}^J z_{j \ell}^2 = 1$ for all $k\geq 2$ and $\ell\geq 2$.
    % (This is done because a common prior precision is used for all entries of each component.)
    \item Precompute the pseudoinverses $X^\ps = (X^\T X)^{-1} X^\T$ and $Z^\ps = (Z^\T Z)^{-1} Z^\T$.
    \end{enumerate}
\bigskip
\noindent\textbf{Initialization.} 
    \begin{enumerate}[label=\textup{(\arabic*)},noitemsep,topsep=0pt]
    \item Solve for $A$, $B$, and $C$ to minimize the sum-of-squares of the GBM residuals $\varepsilon_{i j}$.
    \item Randomly initialize $D$, $U$, and $V$ by computing the truncated singular value decomposition (of rank $M$) of a random matrix with i.i.d.\ $\N(0,10^{-16})$ entries.
    \item In the case of NB-GBMs, iteratively update $S$, $T$, and $\omega$ for a few iterations.
    \end{enumerate}
\bigskip
\noindent\textbf{Iteration.}
In each iteration, we cycle through the components of the model, updating each in turn using an optimization-projection step,
consisting of an unconstrained optimization step and a likelihood-preserving projection onto the constrained parameter space.
% this is repeated until convergence.
We use a bounded, regularized version of Fisher scoring to perform the unconstrained optimization step 
for each of $A$, $B$, $C$, $D$, $U D$, and $V D$, separately, holding all the other parameters fixed.
For a generic parameter vector $\beta$, the (unbounded) regularized Fisher scoring step is
$\beta \gets \beta + (\E(-\nabla_{\!\beta}^2 \L) + \lambda \Id)^{-1} (\nabla_{\!\beta}\L - \lambda \beta)$
where $\L$ is the log-likelihood and $\lambda>0$ is a regularization parameter. 
This arises from optimizing the log-likelihood plus the log-prior,
where the prior on $\beta$ is $\pi(\beta) = \N(\beta \mid 0,\lambda^{-1}\Id)$, since then the gradient and Fisher information are
$\nabla_{\!\beta} (\L + \log\pi) = \nabla_{\!\beta} \L - \lambda \beta$ and
$\E(-\nabla_{\!\beta}^2 (\L + \log\pi)) = \E(-\nabla_{\!\beta}^2 \L) + \lambda \Id$.
Since these Fisher scoring steps occasionally diverge,
for numerical stability we bound them using 
\begin{align}
\label{equation:generic-fisher-update}
\begin{split}
\xi &\gets (\E(-\nabla_{\!\beta}^2 \L) + \lambda \Id)^{-1} (\nabla_{\!\beta}\L - \lambda \beta) \\
\beta &\gets \beta + \xi \min\{1,\, \rho \sqrt{\mathrm{dim}(\xi)}/\|\xi\|\}
\end{split}
\end{align}
where $\|\xi\| = (\sum_i |\xi_i|^2)^{1/2}$ is the Euclidean norm.
The idea is that $\xi \min\{1,\, \rho \sqrt{\mathrm{dim}(\xi)}/\|\xi\|\}$ points in the same direction as $\xi$, 
but its root-mean-square is capped at $\rho$.
Similarly, for $S$ and $T$ in the NB-GBM, we use bounded regularized Newton steps
for the unconstrained optimizations, since the (expected) Fisher information for $S$ and $T$ is not closed form.
% cannot be computed in closed form.
% Likelihood-preserving projections are performed after the unconstrained optimization step to each parameter matrix/vector;
% see Theorem~\ref{theorem:projections} for the justification.

See Section~\ref{section:estimation-details} for full step-by-step details.
See Section~\ref{section:theory} for time complexity analysis.

% We analyze the theoretical computation time complexity in Section~\ref{section:theory}.

% Full step-by-step details are in Section~\ref{section:detailed-updates} to avoid overburdening the reader with minutiae.

\section{Inference}
\label{section:inference}

In this section, we introduce our methodology for computing approximate standard errors for the parameters of a GBM.
% in order to enable uncertainty quantification.  
Since the standard technique of inverting the Fisher information matrix does not work well on GBMs,
we develop a novel technique for propagating uncertainty from one part of the model to another; see Section~\ref{section:delta-propagation}.
% We also use constraint-augmented Fisher information matrices to handle some of the identifiability constraints; see Section~\ref{section:uv-inference}.
We provide an outline of our inference algorithm in Section~\ref{section:inference-outline}, and
step-by-step details are in Section~\ref{section:inference-details}.

\subsection{Delta propagation method}
\label{section:delta-propagation}

In fixed-dimension parametric models, the asymptotic covariance of the maximum likelihood estimator is equal to the inverse of the Fisher information matrix.
Thus, classically, approximate standard errors are given by the square roots of the diagonal entries of the inverse Fisher information.
% obtained by taking the square root of each entry of the diagonal of this matrix.
However, in GBMs, inverting the full (constraint-augmented) Fisher information does not work well for two reasons:
(1) it is computationally intractable for large data matrices, and
(2) it does not yield well-calibrated standard errors in terms of coverage, presumably because the number of parameters grows with the amount of data.
Meanwhile, the inverse Fisher information for each component individually (for instance, $F_a^{-1}$ for $A$) 
is computationally efficient, but severely underestimates uncertainty since it treats the other components as known,
and thus, can be thought of as representing the conditional uncertainty in each component given the other components.
% Instead, what is required is the ``marginal'' uncertainty in each component, not conditional.

We propose a general technique for approximating, for each model component, the additional variance due to uncertainty in the other components.
By adding this to the conditional variance (that is, $\diag(F_a^{-1})$ in the case of $A$),
we obtain approximate variances that are better calibrated, empirically.
The basic idea is to write the estimator for each component as a function of the other components,
and propagate the variance of the other components through this function using the same idea as the delta method.

In general, suppose we have a model with parameters $\theta\in\R^d$ and $\nu\in\R^k$,
and we wish to quantify the uncertainty in $\theta$ due to uncertainty in $\nu$.
Suppose the true values are $\theta_0$ and $\nu_0$, and let $\hat\theta$ and $\hat\nu$ be the maximum likelihood estimators.
Define
$$ h(\nu) := \theta_0 + F(\theta_0,\nu)^{-1} g(\theta_0,\nu) $$
where $g(\theta,\nu) := \nabla_\theta \mathcal{L}$ is the gradient of the log-likelihood and
$F(\theta,\nu) := \E(-\nabla_\theta^2 \mathcal{L})$ is the Fisher information matrix for $\theta$, evaluated at $(\theta,\nu)$.
The interpretation is that $\hat\theta \approx h(\hat\nu)$, since 
$h(\hat\nu)$ is a Fisher scoring step on $\theta$ starting at $\theta_0$, when the current estimate of $\nu$ is $\hat\nu$.

By a Taylor approximation to $h$ at $\nu_0$, we have
$h(\hat\nu) \approx h(\nu_0) + h'(\nu_0)^\T (\hat\nu - \nu_0)$
where $h'(\nu)\in\R^{k\times d}$ such that $h'(\nu)_{i j} = \partial h_j / \partial \nu_i$.
Define the random element $\varphi = (h(\nu_0),h'(\nu_0))$, where the randomness comes from the data.
Assuming $\hat\nu|\varphi\approx\N(\nu_0,\Sigma_\nu)$ for some $\Sigma_\nu$, 
we have $h(\hat\nu)|\varphi \approx \N(h(\nu_0),\, h'(\nu_0)^\T \Sigma_\nu h'(\nu_0))$.
Meanwhile, under standard regularity conditions, $\E(h(\nu_0)) = \theta_0$ and $\mathrm{Cov}(h(\nu_0)) = F(\theta_0,\nu_0)^{-1}$.
Then, by the law of total covariance, 
$$ \mathrm{Cov}(h(\hat\nu)) = \mathrm{Cov}(\E(h(\hat\nu)|\varphi)) + \E(\mathrm{Cov}(h(\hat\nu)|\varphi))
\approx F(\theta_0,\nu_0)^{-1} + \E(h'(\nu_0)^\T \Sigma_\nu h'(\nu_0)). $$
The interpretation of this decomposition is that $F(\theta_0,\nu_0)^{-1}$ represents the uncertainty in $\theta$ given $\nu$,
and $\E(h'(\nu_0)^\T \Sigma_\nu h'(\nu_0))$ represents the uncertainty in $\theta$ due to uncertainty in $\nu$.
Since $\hat\theta \approx h(\hat\nu)$, plugging in empirical estimates leads to the approximation
\begin{align}
\label{equation:covariance-decomposition}
\mathrm{Cov}(\hat\theta) \approx F(\hat\theta,\hat\nu)^{-1} + \hat{h}'(\hat\nu)^\T \hat{\Sigma}_\nu \hat{h}'(\hat\nu),
\end{align}
where $\hat{h}(\nu) = \hat\theta + F(\hat\theta,\nu)^{-1} g(\hat\theta,\nu)$.

To compute the Jacobian matrix $h'(\nu)^\T$, observe that the $i$th column of $h'(\nu)^\T$ is
\begin{align}
\label{equation:delta-prop-jacobian}
\frac{\partial h}{\partial\nu_i} %= \begin{bmatrix} \partial h_1/\partial\nu_i \\ \vdots \\ \partial h_d/\partial\nu_i \end{bmatrix}
= \frac{\partial F^{-1}}{\partial\nu_i} g + F^{-1} \frac{\partial g}{\partial\nu_i} 
= -F^{-1} \frac{\partial F}{\partial\nu_i} F^{-1} g + F^{-1} \frac{\partial g}{\partial\nu_i}
\end{align}
by \citet[Eqn C.21]{bishop2006pattern},
where $F = F(\theta_0,\nu)$, $g = g(\theta_0,\nu)$,
and $\partial/\partial\nu_i$ is applied element-wise.
% that is, $\partial g/\partial\nu_i = (\partial^2 \mathcal{L}/\partial\nu_i\partial\theta_1,\ldots, \partial^2 \mathcal{L}/\partial\nu_i\partial\theta_d)^\T$
% and $(\partial F / \partial\nu_i)_{j j'} = (\partial / \partial\nu_i) \E(- \partial^2 \mathcal{L}/\partial\theta_j\partial\theta_{j'})$
% evaluated at $(\theta_0,\nu)$.
Equation~\ref{equation:delta-prop-jacobian} also holds for $\partial\hat{h}/\partial\nu_i$, but with $\hat\theta$ in place of $\theta_0$;
this facilitates computing $\hat{h}'(\hat\nu)$ in Equation~\ref{equation:covariance-decomposition}.
To obtain standard errors for each element of $\theta$, computation is simplified by the fact that we only need the diagonal of $\mathrm{Cov}(\hat\theta)$,
and to simplify computation even further we use a diagonal matrix for $\hat\Sigma_\nu$ when we apply this technique to GBMs.
Additionally, since we use \textit{maximum a posteriori} estimates, we use the regularized Fisher information and the gradient of the log-posterior in 
the formulas above.

% since $\E(g(\theta_0,\nu_0)) = 0$ and $\mathrm{Cov}(g(\theta_0,\nu_0)) = F(\theta_0,\nu_0)$.
% $\E(g(\theta_0,\nu_0)) = 0$ and 
% $$ \hat\theta \approx h(\hat\nu) \approx \N(\theta_0, \, F(\theta_0,\nu_0)^{-1} + h'(\nu_0)^\T \Sigma_\nu h'(\nu_0))$.

% Assuming $\E(\hat\nu) \approx \nu_0$, with $h$ held fixed we have 

% Thus, we use delta propagation to account for the additional uncertainty in each component due to the uncertainty in the other components.

\subsection{Outline of inference algorithm}
\label{section:inference-outline}

Here we outline our procedure for computing GBM standard errors;
see Section~\ref{section:inference-details} for step-by-step details.
The strategy is to handle $U$ and $V$ jointly by inverting the regularized, constraint-augmented Fisher information matrix for $(U,V)$
and then employ delta propagation for $A$, $B$, $C$, $S$, and $T$; see Figure~\ref{figure:delta-prop}.
Empirically, we find that some of the delta propagation terms are negligible, thus, these have been excluded.
% since they have no noticeable effect.

\begin{figure}
\centering
\tikzstyle{var}=[circle,
                thick,
                draw=black,
                fill=white]

\begin{tikzpicture}[>=latex] %,text height=1.5ex,text depth=0.25ex]
  \matrix[row sep=2.5ex,column sep=2.5ex] {
        & & \node (C) [var] {$C$}; & & \\
        \node (A) [var] {$A$}; & & \node (UV) [var] {$U,V$}; & & \node (B) [var] {$B$}; \\
        & \node (S) [var] {$S$}; & & \node (T) [var] {$T$}; & \\
    };
    
    % The diagram elements are now connected through arrows:
    \path[->]
        (A) edge[thick] (C)
        (B) edge[thick] (C)
        (UV) edge[thick] (A)
        (UV) edge[thick] (B)
        (A) edge[thick] (S)
        (A) edge[thick] (T)
        (B) edge[thick] (S)
        (B) edge[thick] (T)
        (UV) edge[thick] (S)
        (UV) edge[thick] (T)
;
        
\end{tikzpicture}
\caption{Diagram of uncertainty propagation scheme for GBM inference.} 
\label{figure:delta-prop}
\end{figure}
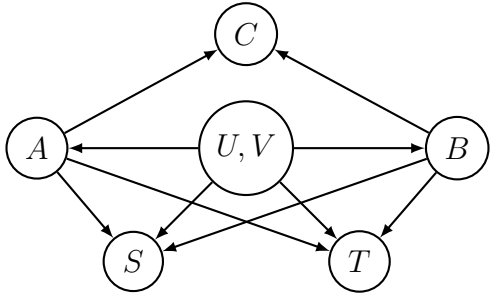

For notational convenience, we vectorize the parameter matrices as follows.
For $Q\in\R^{m\times n}$, define $\mathrm{vec}(Q) := (q_{1 1},q_{2 1},\ldots,q_{m 1},q_{1 2},q_{2 2},\ldots,q_{m 2},\ldots\ldots,q_{m n})\in\R^{m n}$.
% Let $\mathrm{vec}(Q)$ denote the column-wise vectorization of a matrix $Q\in\R^{m\times n}$, that is, $\mathrm{vec}(Q)$ is the concatenation of the columns:
Denote $\vec{a} := \mathrm{vec}(A^\T)$, $\vec{b} := \mathrm{vec}(B^\T)$, $\vec{c} := \mathrm{vec}(C)$, $\vec{d} := \diag(D)$, $\vec{u} := \mathrm{vec}(U^\T)$, and $\vec{v} := \mathrm{vec}(V^\T)$.
It is easier to work with the vectorized transposes of $A$, $B$, $U$, and $V$ (that is, $\mathrm{vec}(A^\T)$ rather than $\mathrm{vec}(A)$),
since then the Fisher information matrices have a block diagonal structure.
% \todo{Maybe use $\mathrm{vec}(C^\T)$ as well, for consistency???}
We write $F_a$, $F_b$, $F_c$, $F_d$, $F_u$, and $F_v$ to denote the regularized Fisher information
for $\vec{a}$, $\vec{b}$, $\vec{c}$, $\vec{d}$, $\vec{u}$, and $\vec{v}$, respectively,
for instance, $F_a = \E(-\nabla_{\!\vec{a}}^2\,(\L+\log \pi))$.
We write $F_{s,\mathrm{obs}}$ and $F_{t,\mathrm{obs}}$ for the regularized, observed Fisher information for $S$ and $T$,
that is, $F_{s,\mathrm{obs}} = -\nabla_{\! S}^2\,(\L+\log \pi)$.

\bigskip
\noindent\textbf{Inputs.}
The required inputs are $\bm{Y}$, $X$, $Z$, 
and the estimates of $A$, $B$, $C$, $D$, $U$, $V$, $S$, $T$, and $\omega$.
Optional inputs are the prior parameters ($\lambda_a,\lambda_b,\lambda_c,\lambda_d,\lambda_u,\lambda_v,\lambda_s,\lambda_t,m_s,m_t$).
% the count data matrix $\bm{Y}$, the feature covariate matrix $X$, the sample covariate matrix $Z$, 

\bigskip
% \noindent\textbf{Preprocessing.}
    % \begin{enumerate}[label=\textup{(\arabic*)},noitemsep,topsep=0pt]
    % \item Compute inverse dispersions $r_{i j} \gets \exp(-s_i-t_j-\omega)$ for all $i,j$.
    % \item Compute $\mu$, $W$, and $E$.
    % \end{enumerate}
% \bigskip
\noindent\textbf{Compute conditional uncertainty for each component.}
    \begin{enumerate}[label=\textup{(\arabic*)},noitemsep,topsep=0pt]
    \item Compute $F_{c}^{-1}$ and the diagonal blocks of $F_{a}^{-1}$, $F_{b}^{-1}$, $F_{u}^{-1}$, and $F_{v}^{-1}$.
    \item Compute the diagonals of $F_{s,\mathrm{obs}}^{-1}$ and $F_{t,\mathrm{obs}}^{-1}$.
    \end{enumerate}
\bigskip
\noindent\textbf{Compute joint uncertainty in $(U,V)$ accounting for constraints.}
    \begin{enumerate}[label=\textup{(\arabic*)},noitemsep,topsep=0pt]
    % \item Compute $F_{u\times v}$, the Fisher information between $U$ and $V$.
    % \item Compute $Q_{u}$ and $Q_{v}$, the constraint Jacobian matrices for $U$ and $V$.
    \item Compute $\tilde{F}_{(u,v)}$, the regularized constraint-augmented Fisher information for $\big[\begin{smallmatrix}\vec{u}\\\vec{v}\end{smallmatrix}\big]$.
    % $\tilde{F}_{(u,v)} = \begin{bmatrix} F & Q \\ Q^\T & 0 \end{bmatrix}$
    % where $F = \begin{bmatrix} F_u & F_{u\times v} \\ F_{u\times v}^\T & F_{v} \end{bmatrix}$
    % and $Q = \begin{bmatrix} Q_{u} & 0 \\ 0 & Q_{v} \end{bmatrix}$.
    \item Compute $\diag(\tilde{F}_{(u,v)}^{-1})$. (It is key to do this in a computationally efficient way.)
    \item Define $\widehat{\mathrm{var}}_{u}$ and $\widehat{\mathrm{var}}_{v}$ to be the entries of $\diag(\tilde{F}_{(u,v)}^{-1})$ corresponding to $U$ and $V$.
    \end{enumerate}
\bigskip
\noindent\textbf{Propagate uncertainty between components using delta propagation.}
    \begin{enumerate}[label=\textup{(\arabic*)},noitemsep,topsep=0pt]
    \item Propagate uncertainty in $(U,V)$ to $A$ and $B$, to obtain $\widehat{\mathrm{var}}_{(u,v)\to a}$ and $\widehat{\mathrm{var}}_{(u,v)\to b}$, the additional variance of the estimators of $A$ and $B$ due to uncertainty in $(U,V)$.
    % \item Propagate uncertainty in $U$ and $V$ through to $B$, to obtain $\widehat{\mathrm{var}}_{(u,v)\to b}$.
    \item Propagate uncertainty in $A$ and $B$ through to $C$, to obtain $\widehat{\mathrm{var}}_{(a,b)\to c}$.
    \item Propagate uncertainty in $A$, $B$, $U$, $V$ to $S$ and $T$, to get $\widehat{\mathrm{var}}_{(a,b,u,v)\to s}$ and $\widehat{\mathrm{var}}_{(a,b,u,v)\to t}$.
    \end{enumerate}
\bigskip
\noindent\textbf{Compute approximate standard errors.}
    \begin{enumerate}[label=\textup{(\arabic*)},noitemsep,topsep=0pt]
    \item $\hat{\mathrm{se}}_a \gets \mathrm{sqrt}(\diag(F_{a}^{-1}) + \widehat{\mathrm{var}}_{(u,v)\to a})$ ~~and~~
          $\hat{\mathrm{se}}_b \gets \mathrm{sqrt}(\diag(F_{b}^{-1}) + \widehat{\mathrm{var}}_{(u,v)\to b})$
    \item $\hat{\mathrm{se}}_c \gets \mathrm{sqrt}(\diag(F_{c}^{-1}) + \widehat{\mathrm{var}}_{(a,b)\to c})$
    \item $\hat{\mathrm{se}}_u \gets \mathrm{sqrt}(\widehat{\mathrm{var}}_{u})$ ~~and~~
          $\hat{\mathrm{se}}_v \gets \mathrm{sqrt}(\widehat{\mathrm{var}}_{v})$
    \item $\hat{\mathrm{se}}_s \gets \mathrm{sqrt}(\diag(F_{s,\mathrm{obs}}^{-1}) + \widehat{\mathrm{var}}_{(a,b,u,v)\to s})$ ~~and~~
          $\hat{\mathrm{se}}_t \gets \mathrm{sqrt}(\diag(F_{t,\mathrm{obs}}^{-1}) + \widehat{\mathrm{var}}_{(a,b,u,v)\to t})$
    \end{enumerate}
Here, $\mathrm{sqrt}(\cdot)$ is the element-wise square root.
We do not provide standard errors for $D$ and $\omega$, since it seems difficult to estimate them without non-negligible bias.
See Section~\ref{section:inference-details} for the complete step-by-step algorithm.
See Section~\ref{section:theory} for computation time complexity.
% We analyze the theoretical computation time complexity in Section~\ref{section:theory}.

% We reshape the vectorized standard errors to matrices/vectors of the same dimension as the corresponding components,
% for instance, $\hat{\mathrm{se}}_a \in \R^{J K}$ is reshaped to a matrix $\hat{\mathrm{se}}_A \in \R^{J \times K}$.

% Uncertainty quantification
    % Fast frequentist UQ based on generalization of standard GLM theory
    % Difficulty of performing fully Bayesian UQ

% In this model, the full Hessian is large and very expensive to invert.
% Also, full Hessian is singular unless constraints are accounted for using the restricted MLE approach of Silvey.

\section{Theory}
\label{section:theory}

In this section, we provide theoretical results on GBMs. % on identifiability, interpretation of parameters, and likelihood-preserving projections in GBMs.
The proofs are in Section~\ref{section:proofs}.

% \textbf{Identifiability and interpretation.}

\begin{theorem}[Identifiability]
\label{theorem:identifiability}
If $(A_1,B_1,C_1,D_1,U_1,V_1,X,Z)$ and $(A_2,B_2,C_2,D_2,U_2,V_2,\allowbreak X,Z)$ satisfy Condition~\ref{condition:identifiability} and
\begin{equation}
\label{equation:identifiability}
X A_1^\T + B_1 Z^\T + X C_1 Z^\T + U_1 D_1 V_1^\T = X A_2^\T + B_2 Z^\T + X C_2 Z^\T + U_2 D_2 V_2^\T,
\end{equation}
then $A_1 = A_2$, $B_1 = B_2$, $C_1 = C_2$, $D_1 = D_2$, $U_1 = U_2$, and $V_1 = V_2$.
In particular, for any GBM satisfying Equation~\ref{equation:model} for some $X$, $Z$, and $M$,
if Condition~\ref{condition:identifiability} holds, then $A$, $B$, $C$, $D$, $U$, and $V$ are identifiable
in the sense that they are uniquely determined by the distribution of $\bm{Y}$;
in fact, they are uniquely determined by $\E(\bm{Y})$.
\end{theorem}

% See Section~\ref{section:proofs} for the proof.

\begin{theorem}[Interpretation of parameters]
\label{theorem:interpretation}
If Conditions~\ref{condition:identifiability} and \ref{condition:interpretation} hold
and $\mu_{i j} := \E(Y_{i j})$ satisfies Equation~\ref{equation:model-univariate}, then:
\begin{enumerate}[label=\textup{(\alph*)}, ref=\ref{theorem:interpretation}(\alph*)]
\item $\sum_{j = 1}^J a_{j k} = 0$ and $\sum_{i = 1}^I b_{i \ell} = 0$ for all $k\in\{1,\ldots,K\}, \ell\in\{1,\ldots,L\}$,
\item $\sum_{i = 1}^I u_{i m} = 0$ and $\sum_{j = 1}^J v_{j m} = 0$ for all $m\in\{1,\ldots,M\}$,
\item\label{theorem:interpretation:item:c} $\frac{1}{I}\sum_{i = 1}^I g(\mu_{i j}) = c_{1 1} + a_{j 1} + \sum_{\ell=2}^L c_{1 \ell} z_{j\ell}$ for all $j\in\{1,\ldots,J\}$,
\item $\frac{1}{J}\sum_{j = 1}^J g(\mu_{i j}) = c_{1 1} + b_{i 1} + \sum_{k=2}^K c_{k 1} x_{i k}$ for all $i\in\{1,\ldots,I\}$, and
\item $\frac{1}{I J}\sum_{i = 1}^I \sum_{j = 1}^J g(\mu_{i j}) = c_{1 1}$.
\end{enumerate}
\end{theorem}

% See Section~\ref{section:proofs} for the proof.

For a matrix $Q\in\R^{m\times n}$, we write $\mathrm{SS}(Q) := \sum_{i=1}^m \sum_{j=1}^n q_{i j}^2$ for the sum of squares.

\begin{theorem}[Sum-of-squares decomposition]
\label{theorem:sum-of-squares}
% If $A$, $B$, $C$, $D$, $U$, $V$, $X$, and $Z$ are matrices with dimensions as in Condition~\ref{condition:identifiability},
If Condition~\ref{condition:orthogonal} holds, then
$$ \mathrm{SS}(X A^\T + B Z^\T + X C Z^\T + U D V^\T) = \mathrm{SS}(X A^\T) + \mathrm{SS}(B Z^\T) + \mathrm{SS}(X C Z^\T) + \mathrm{SS}(U D V^\T). $$
\end{theorem}
% Assume $X^\T B = 0$, $Z^\T A = 0$, $X^\T U = 0$, and $Z^\T V = 0$, that is, assume Condition~\ref{condition:orthogonal} holds.
% and $X^\T B = 0$, $Z^\T A = 0$, $X^\T U = 0$, and $Z^\T V = 0$, then

% See Section~\ref{section:proofs} for the proof.

% \textbf{Likelihood-preserving projections.}
Theorem~\ref{theorem:projections} shows that the projections we use in the estimation algorithm are likelihood-preserving.
The idea is that, for example, $\tilde{A}$ is the result of an unconstrained optimization step on $A$,
% while holding the other parameters fixed, 
and we project $(\tilde{A},C)$ to $(A_1,C_1)$ to enforce the constraints in Condition~\ref{condition:identifiability} without affecting the likelihood.
To interpret items \ref{item:UD-proj} and \ref{item:VD-proj} of Theorem~\ref{theorem:projections}, note that in the algorithm, we optimize with respect to $G := U D$ and $H := V D$, rather than $U$ and $V$.
We write $Q^\ps$ to denote the pseudoinverse.
When $(Q^\T Q)^{-1}$ exists, $Q^\ps = (Q^\T Q)^{-1} Q^\T$.
% of a matrix $Q\in\R^{m \times n}$.

\begin{theorem}[Likelihood-preserving projections]
\label{theorem:projections}
Suppose $(A,B,C,D,U,V,X,Z)$ satisfies Condition~\ref{condition:identifiability}.
Fix $X$ and $Z$ and define $\eta(\cdot)$ as in Equation~\ref{equation:eta}.
\begin{enumerate}
\item\label{item:A-proj} Let $\tilde{A}\in\R^{J\times K}$. Define $A_1 = \tilde{A} - Z (Z^\ps \tilde{A})$ and $C_1 = C + (Z^\ps \tilde{A})^\T$.
Then $\eta(A_1,B,\allowbreak C_1,D,U,V) = \eta(\tilde{A},B,C,D,U,V)$ and $(A_1,B,C_1,D,U,V)$ satisfies Condition~\ref{condition:identifiability}.
\item\label{item:B-proj} Let $\tilde{B}\in\R^{I\times L}$. Define $B_1 = \tilde{B} - X (X^\ps \tilde{B})$ and $C_1 = C + (X^\ps \tilde{B})$.
Then $\eta(A,B_1,\allowbreak C_1,D,U,V) = \eta(A,\tilde{B},C,D,U,V)$ and $(A,B_1,C_1,D,U,V)$ satisfies Condition~\ref{condition:identifiability}.
\item\label{item:UD-proj} Let $\tilde{G} \in\R^{I\times M}$.  Define $G_0 = \tilde{G} - X (X^\ps \tilde{G})$ and let $U_1 D_1 V_1^\T$ be the compact SVD (of rank $M$) of $G_0 V^\T$. Assume the singular values are distinct and positive, and choose the SVD in such a way that Conditions~\ref{condition:diagonal} and \ref{condition:sign} are satisfied. 
Define $A_0 = A + V (X^\ps \tilde{G})^\T$,  $A_1 = A_0 - Z (Z^\ps A_0)$, and $C_1 = C + (Z^\ps A_0)^\T$.
Then $\eta(A_1,B,C_1,\allowbreak D_1,U_1,V_1) = \eta(A,B,C,\Id,\tilde{G},V)$ and $(A_1,B,C_1,D_1,U_1,V_1)$ satisfies Condition~\ref{condition:identifiability}.
\item\label{item:VD-proj} Let $\tilde{H} \in\R^{J\times M}$.  Define $H_0 = \tilde{H} - Z (Z^\ps \tilde{H})$ and let $U_1 D_1 V_1^\T$ be the compact SVD (of rank $M$) of $U H_0^\T$. Assume the singular values are distinct and positive, and choose the SVD in such a way that Conditions~\ref{condition:diagonal} and \ref{condition:sign} are satisfied. 
Define $B_0 = B + U (Z^\ps \tilde{H})^\T$,  $B_1 = B_0 - X (X^\ps B_0)$, and $C_1 = C + X^\ps B_0$.
Then $\eta(A,B_1,C_1,D_1,U_1,V_1) = \eta(A,B,C,\Id,U,\tilde{H})$ and $(A,B_1,C_1,D_1,U_1,V_1)$ satisfies Condition~\ref{condition:identifiability}.
\end{enumerate}
\end{theorem}

% (Projecting $\tilde{A}$.)
% (Projecting $\tilde{B}$.)
% (Projecting $\tilde{G}$.)
% (Projecting $\tilde{H}$.)

% Fix $X\in\R^{I\times K}$ and $Z\in\R^{J\times L}$ such that $X^\T X$ and $Z^\T Z$ are invertible.
% For $A\in \R^{J\times K}$, $B\in\R^{I\times L}$, $C\in\R^{K\times L}$, $D\in\R^{M\times M}$, $U\in\R^{I\times M}$, $V\in\R^{J\times M}$, define
% $$ \eta(A,B,C,D,U,V) := X A^\T + B Z^\T + X C Z^\T + U D V^\T. $$
% Suppose $(A,B,C,D,U,V)$ satisfies Condition~\ref{condition:identifiability} with respect to $X$ and $Z$.

% See Section~\ref{section:proofs} for the proof.

% \textbf{Computation time complexity.}
Next, we provide the computation time complexity of our estimation and inference algorithms in Sections~\ref{section:estimation} and \ref{section:inference}.
To simplify the expressions, here we assume 
\begin{align}
\label{equation:cost-dimension-bound}
\max\{K^2,L^2,M\} \leq \min\{I,J\}.
\end{align}
For the estimation algorithm,
preprocessing and initialization take $O(I J \max\{K,L,M\})$ time, and 
Table~\ref{table:estimation-time-complexity} summarizes the computation time for updating each model component.
The table first breaks out the time required to compute $\eta$ (Equation~\ref{equation:eta}), which is a prerequisite within each other update, 
and then lists the time required for each update given $\eta$.
% These expressions show how the computation time scales as a function of $I$, $J$, $K$, $L$, and $M$, assuming Equation~\ref{equation:cost-dimension-bound}.
In total, it takes $O(I J \max\{K^2,L^2,M^2\})$ time to perform each overall iteration.

For the inference algorithm, Table~\ref{table:inference-time-complexity} shows the computation time 
% as a function of $I$, $J$, $K$, $L$, and $M$, 
assuming Equation~\ref{equation:cost-dimension-bound} and also assuming $I \geq J$.  These are one-time costs since there are not repeated iterations.
When $M > 0$, the most expensive operation tends to be computing the joint uncertainty in $(U,V)$,
and as $J$ grows this dominates the cost.
We have experimented extensively but have not found a faster alternative that provides well-calibrated standard errors.
% See Section~\ref{section:proofs} for the derivation of Tables~\ref{table:estimation-time-complexity} and \ref{table:inference-time-complexity}.

% For the iterative updates, Theorem~\ref{theorem:computation-time} summarizes the computation time.
% \begin{theorem}
% \label{theorem:computation-time}
% Assume $\max\{K^2,L^2,M\} \leq \min\{I,J\}$.
% As a function of $I$, $J$, $K$, $L$, and $M$, it takes $O(I J \max\{K^2,L^2,M^2\})$ time to perform each iteration of the algorithm.
% Further, Table~\ref{table:estimation-time-complexity} shows the time required for 
% the updates to each of $A$, $B$, $C$, $D$, $U$, and $V$; we also include $S$ in the NB-GBM case.
% \end{theorem}

\begingroup
\renewcommand*{\arraystretch}{1.2}
\begin{table}  %[b]
\centering
\small
\caption{Computation time complexity of each update in the estimation algorithm.}%
\begin{tabular}{|l|c|}%
\hline
Operation & Time complexity \\
\hline
Computing $\eta$  &  $O(I J \max\{K,L,M\})$ \\ 
Updating $A$  &  $O(I J K^2)$ \\
Updating $B$  &  $O(I J L^2)$ \\
Updating $C$  &  $O(I J \max\{K^2,L^2\})$ \\
Updating $D$, $U$, and $V$  &  $O(I J M^2)$ \\
Updating $S$ and $T$  &  $O(I J)$ \\
\hline
Total per iteration  &  $O(I J \max\{K^2,L^2,M^2\})$ \\
\hline
\end{tabular}
\label{table:estimation-time-complexity}
\end{table}
\endgroup

% Computing $\mu$, $W$, and $E$  &  $O(I J \max\{K,L,M\})$ \\ 
% Updating $D$  &  $O(I J M^2)$ \\
% Updating $G = U D$  &  $O(I J M^2)$ \\
% Updating $H = V D$  &  $O(I J M^2)$ \\

\begingroup
\renewcommand*{\arraystretch}{1.2}
\begin{table}  %[b]
\centering
\small
\caption{Computation time complexity of the inference algorithm.}%
\begin{tabular}{|l|c|}%
\hline
Operation & Time complexity \\
\hline
Preprocessing  &  $O(I J \max\{K,L,M\})$ \\ 
Conditional uncertainty for each component  &  $O(I J \max\{K^2,L^2,M^2\})$  \\
Joint uncertainty in $(U,V)$ accounting for constraints  &  $O(I J^2 M^3)$  \\
Propagate uncertainty between components  &  $O(I J \max\{K^3, L^3, M^3\})$  \\
Compute approximate standard errors  &  $O(I J)$  \\
\hline
Total  &  $O(I J \max\{K^3, L^3, J M^3\})$ \\
\hline
\end{tabular}
\label{table:inference-time-complexity}
\end{table}
\endgroup

\section{Simulations}
\label{section:simulations}

In this section, we present simulation studies assessing
% assessing performance in terms of
(a) consistency and statistical efficiency, 
(b) accuracy of standard errors,
(c) computation time and algorithm convergence, and
(d) robustness to the outcome distribution.
See Section~\ref{section:simulation-details} for more simulation results.

In each simulation run, the data are generated as follows;
see Section~\ref{section:simulation-details} for full details.
We generate the covariates using one of three schemes, \texttt{Normal}, \texttt{Gamma}, or \texttt{Binary},
then we generate the true parameters using either a \texttt{Normal} or \texttt{Gamma} scheme,
and finally we generate the outcome data using the log link and a
\texttt{NB} (negative binomial), \texttt{LNP} (log-normal Poisson), \texttt{Poisson}, or \texttt{Geometric} distribution.
For brevity, we refer to each combination of choices by the triplet of outcomes/covariates/parameters, 
for instance, \texttt{NB/Binary/Normal}.
% The true parameter values are denoted using subscript $0$, for instance, $A_0$.
% refers to using the \texttt{NB}, \texttt{Binary}, and \texttt{Normal} simulation schemes for 
% the generating the outcomes, covariates, and true parameters, respectively.

\begin{figure}
  \centering
  \includegraphics[trim=0.6cm 0 1.5cm 0, clip, height=0.2\textheight]{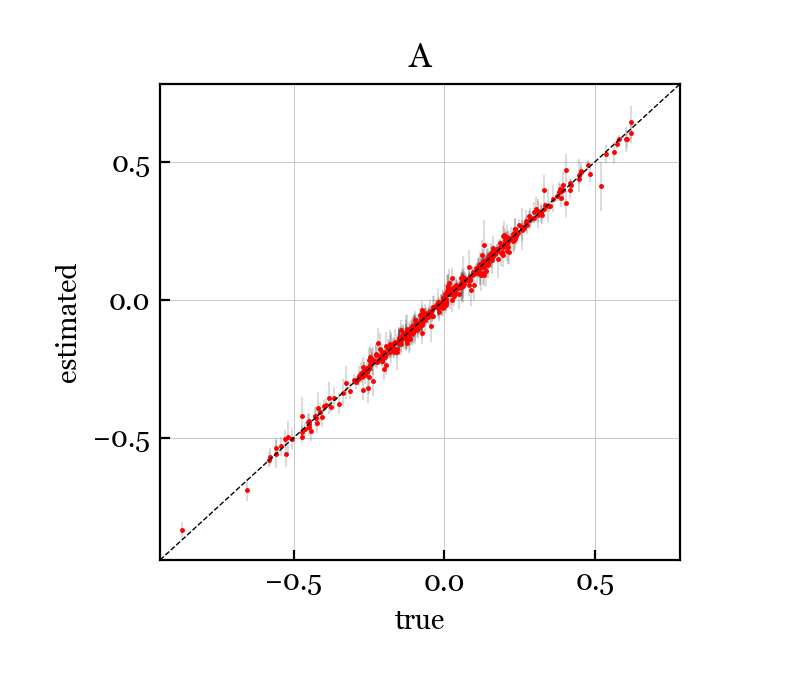}
  \includegraphics[trim=1.2cm 0 1.5cm 0, clip, height=0.2\textheight]{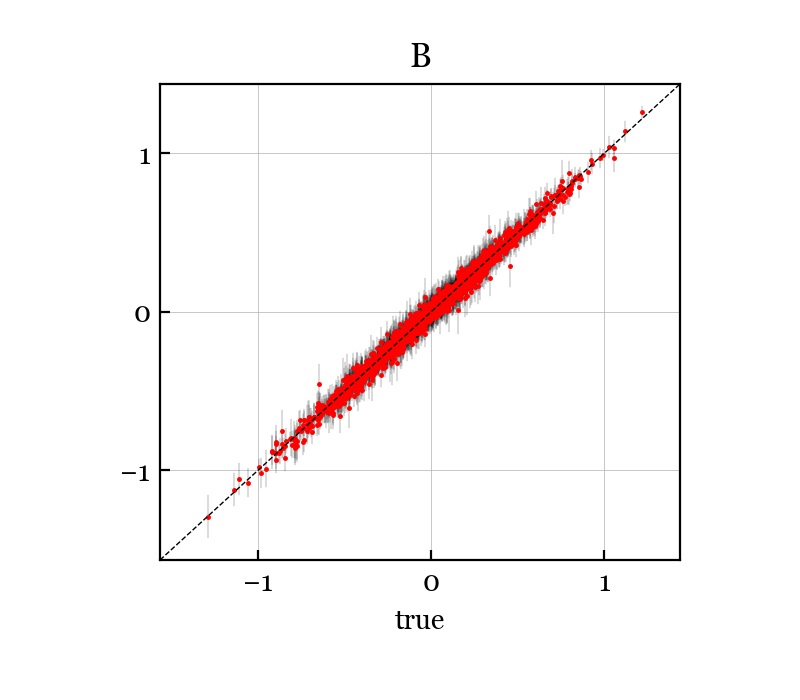}
  \includegraphics[trim=1.0cm 0 1.5cm 0, clip, height=0.2\textheight]{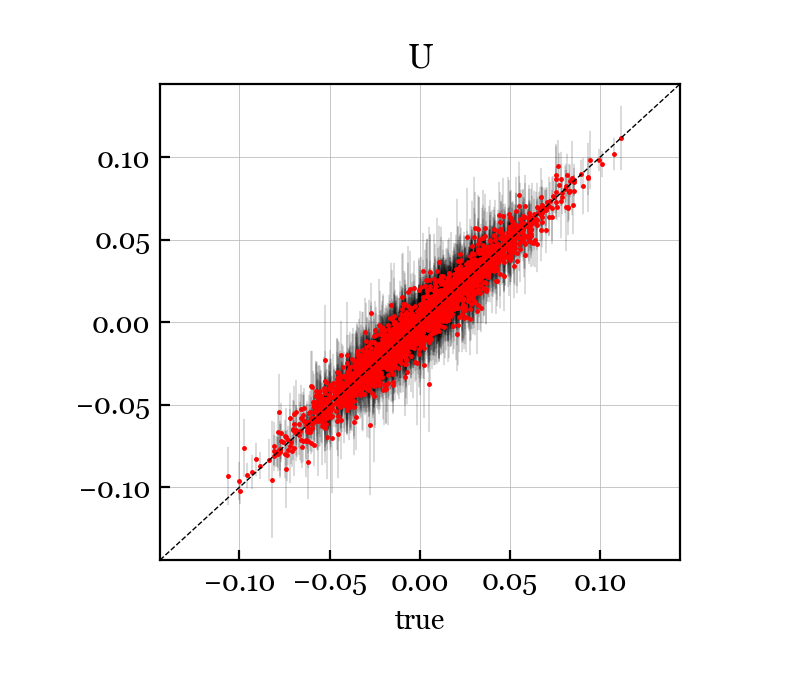}
  \includegraphics[trim=1.2cm 0 1.5cm 0, clip, height=0.2\textheight]{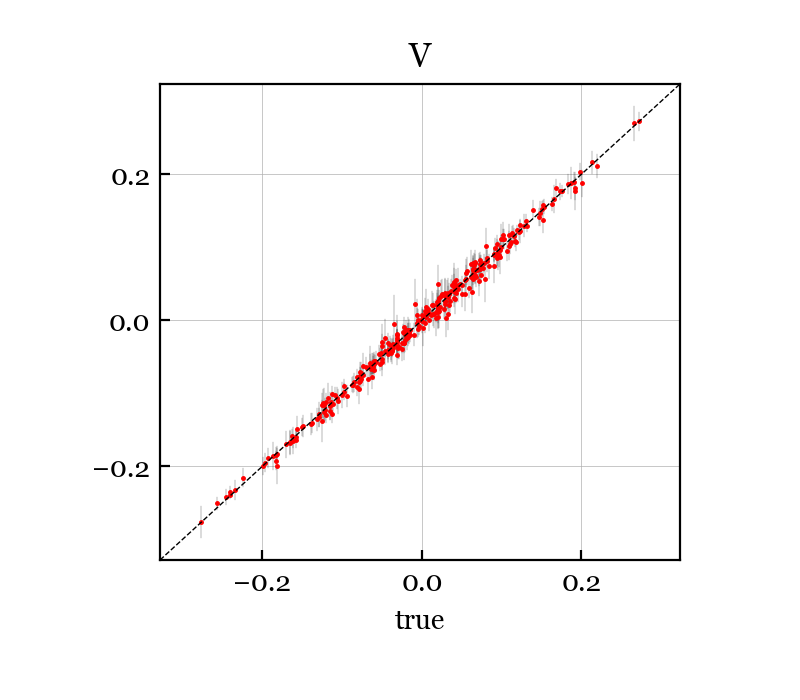}
  \includegraphics[trim=0.6cm 0 1.5cm 0, clip, height=0.2\textheight]{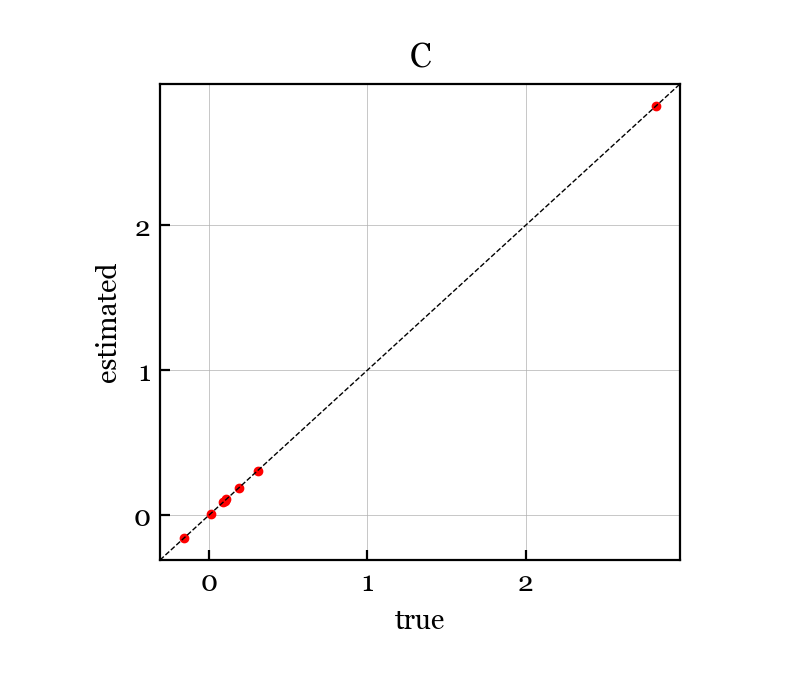}
  \includegraphics[trim=1.2cm 0 1.5cm 0, clip, height=0.2\textheight]{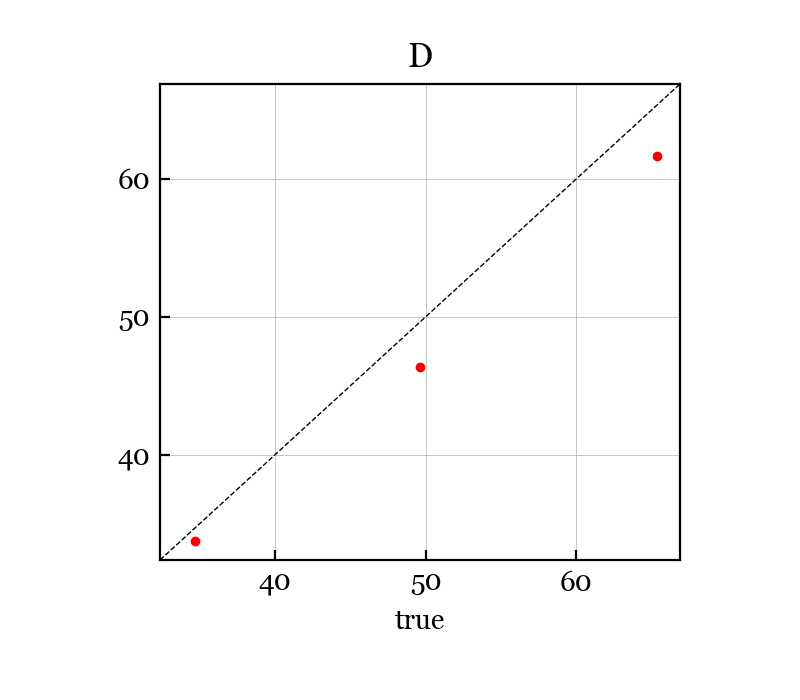}
  \includegraphics[trim=1.0cm 0 1.5cm 0, clip, height=0.2\textheight]{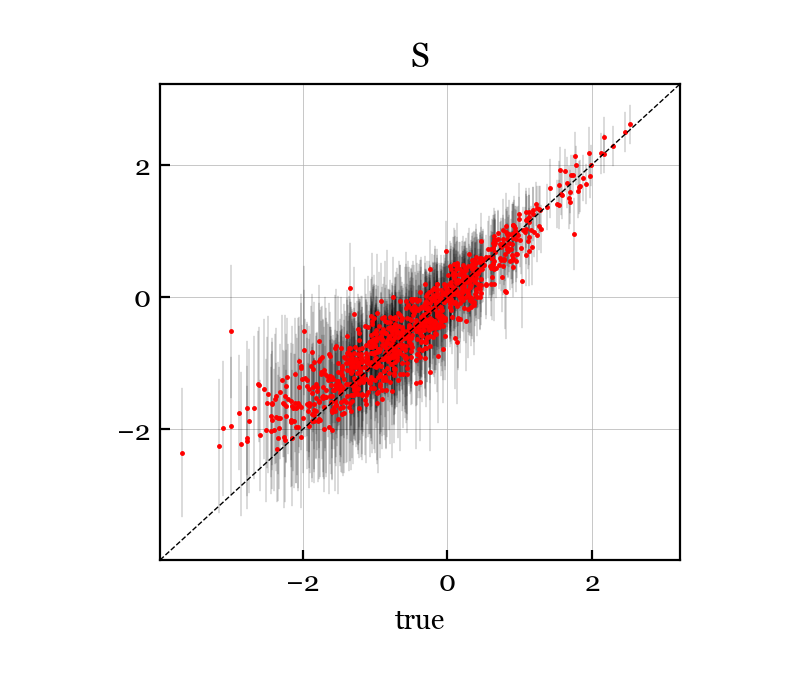}
  \includegraphics[trim=1.2cm 0 1.5cm 0, clip, height=0.2\textheight]{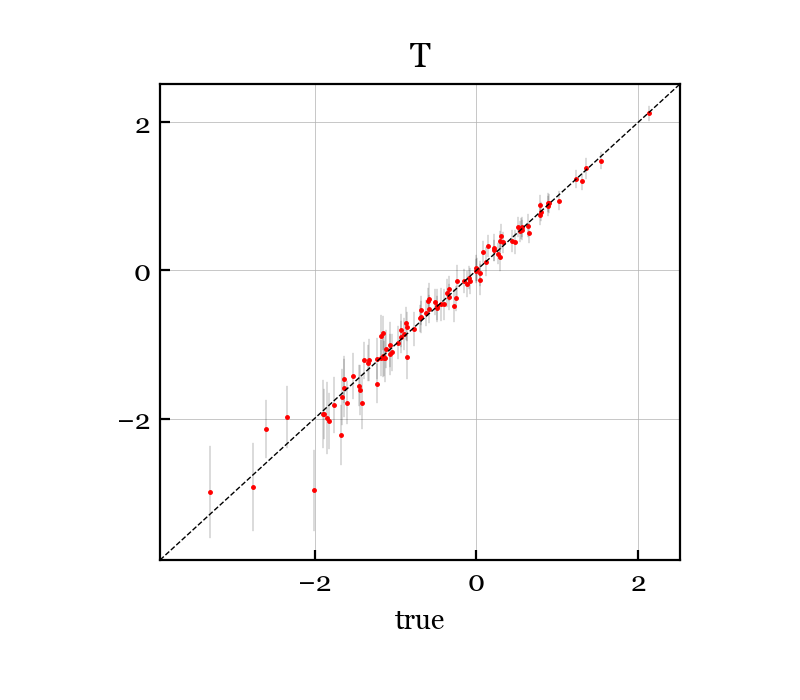}
  \caption{Scatterplots of estimated versus true parameters for a typical simulated data matrix.}
  \label{figure:scatterplot}
\end{figure}

\textbf{Typical example.} 
Figure~\ref{figure:scatterplot} shows scatterplots of the estimated versus true parameters for an \texttt{NB/Normal/Normal} simulation with $I = 1000$, $J = 100$, $K = 4$, $L = 2$, and $M = 3$. Each dot represents a single univariate parameter, for example, the plot for $A$ contains $J K$ dots, one for each entry $a_{j k}$.
The error bars are $\pm 2 \,\hat{\mathrm{se}}$.
Visually, the estimates are close to true values, and the standard errors look appropriate.
% for each model component.
% for each of $A$, $B$, $C$, $D$, $U$, $V$, $S$, and $T$; 
% Also, the estimate of $\omega$ is $-2.34$, which is reasonably close to the true value $\omega_0 = -2.3$.
Since the likelihood and prior are invariant to permutations and sign changes of the latent factors, in this section 
we permute and flip signs to find the correct assignment to the true latent factors.
Note that the $s_i$ estimates are biased upward when the true value of $s_i$ is very low; 
this is because very low values of $s_i$ make row $i$ roughly Poisson, and in this case 
% there is little information in the data to inform the estimate since
any value of $s_i$ from $-\infty$ to $\approx -2$ could yield a reasonable fit.
The prior on $s_i$ prevents the estimate from diverging, but also leads to an upward bias when the true value is very low.

\subsection{Consistency and statistical efficiency}
\label{section:consistency}

In many applications, $I$ is much larger than $J$.
% we are particularly interested in the asymptotics as $I\to\infty$ with $J$ fixed.
One would hope that for any $J$, the estimates of $A$, $C$, $V$, and $T$ would be consistent as $I\to\infty$
since then these parameters have fixed dimension.
% since the dimension of these parameters is fixed.
% and the amount of data informing each univariate parameter coordinate is growing.
% Fortunately, this hope appears to be fulfilled based on our empirical experiments.
Meanwhile, one cannot hope for consistency in $B$, $U$, and $S$ as $I\to\infty$. % with $J$ fixed.
% since the dimension of these parameter matrices/vectors is growing and the amount of data informing each univariate parameter coordinate is fixed.

\begin{figure}
  \centering
  \includegraphics[trim=0.8cm 0 1.3cm 0, clip, height=0.2\textheight]{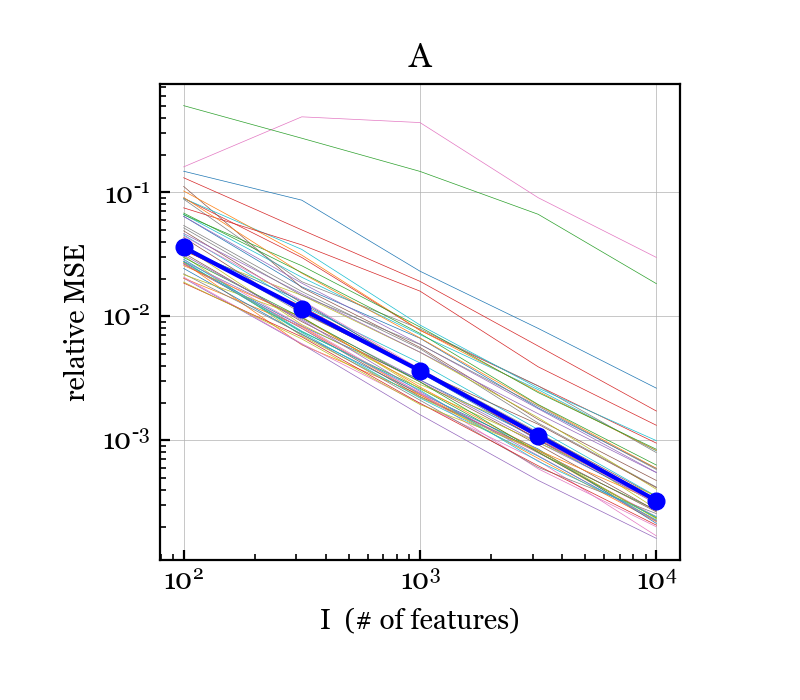}
  \includegraphics[trim=1.3cm 0 1.3cm 0, clip, height=0.2\textheight]{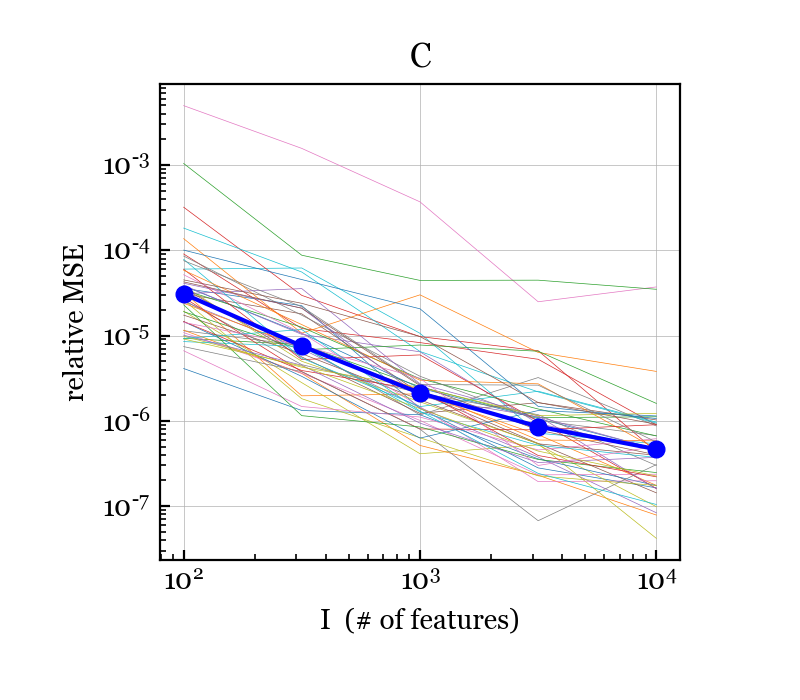}
  \includegraphics[trim=1.3cm 0 1.3cm 0, clip, height=0.2\textheight]{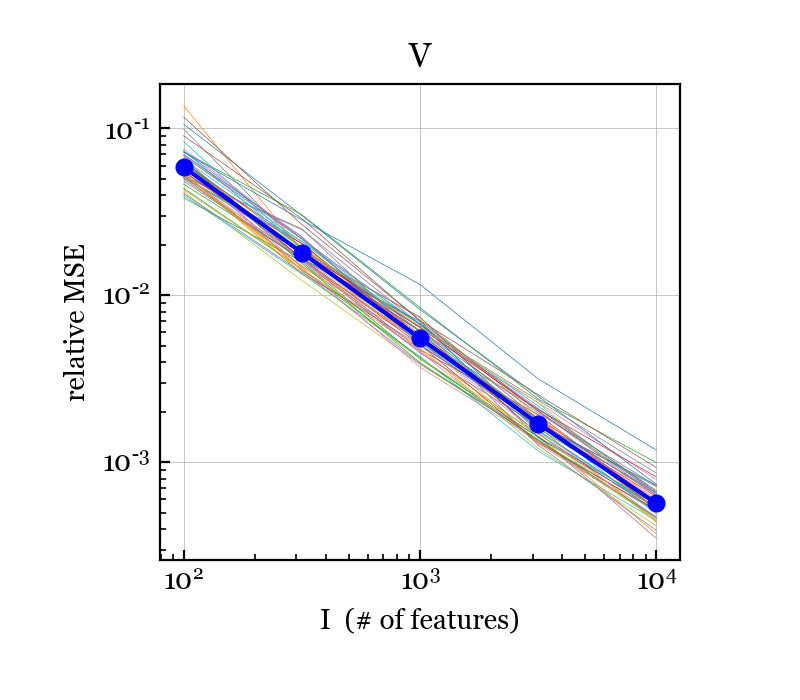}
  \includegraphics[trim=1.3cm 0 1.3cm 0, clip, height=0.2\textheight]{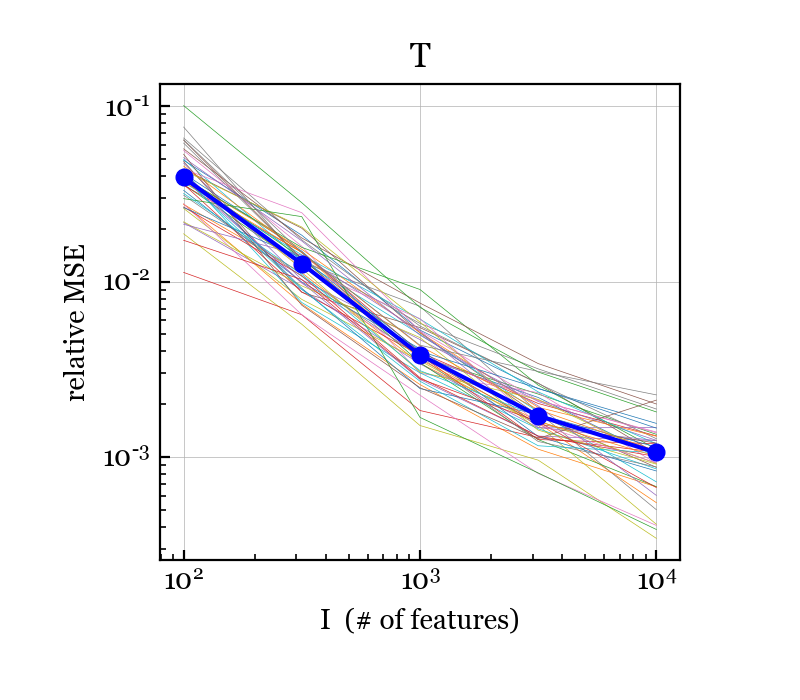}
  \caption{Relative mean-squared error between estimated and true parameter values. Each plot contains $50$ thin lines, one for each run, along with the median over the $50$ runs (thick blue line).}
  \label{figure:consistency}
\end{figure}

To assess consistency and efficiency, 
for each $I\in\{100,316,1000,3162,10000\}$
we use the \texttt{NB/Normal/Normal} simulation scheme to generate $50$ data matrices with $J = 100$, $K = 4$, $L = 2$, and $M = 3$,
each with a different set of covariates and parameters.
For each data matrix, we run our NB-GBM estimation algorithm with convergence tolerance $\tau = 10^{-8}$.
Figure~\ref{figure:consistency} shows the relative mean-squared error (MSE) between the estimates and the true values for $A$, $C$, $V$, and $T$.
For $T$, we measure the relative MSE in the dispersion parametrization (rather than log-dispersion) since there is little difference between, say, $t_j = -5$ and $t_j = -100$; both make column $j$ approximately Poisson distributed.
% Thus, in Figure~\ref{figure:consistency}, the relative MSEs for $S$ and $T$ are computed using $e^{s_i}$ and $e^{t_j}$ rather than $s_i$ and $t_j$.

% In Figure~\ref{figure:consistency} 
We see that for $A$, $C$, $V$, and $T$, the relative MSE is decreasing to zero, suggesting that the estimates of these parameters are consistent as $I\to\infty$. Further, for $A$ and $V$, the relative MSE appears to be $O(1/I)$, which is the optimal rate of convergence
even for fixed-dimension parametric models.
For $B$, $U$, and $S$ (Figure~\ref{figure:consistency-all}), the relative MSE hovers around a small nonzero value, but does not appear to be trending to zero, as expected.
For $D$ and $\omega$  (Figure~\ref{figure:consistency-all}), the relative MSE is small and the trend is suggestive but less clear.
% See Section~\ref{section:consistency-details} for further details on these results.

\subsection{Accuracy of standard errors}
\label{section:coverage}

Next, we assess the accuracy of the standard errors produced by our inference algorithm,
in terms of the coverage.
% To do so, we construct confidence intervals for each univariate parameter and compare the actual coverage to the target coverage.  
Ideally, a 95\% confidence interval would contain the true parameter 95\% of the time, 
but even when the model is correct, this is not guaranteed since intervals are usually based on an approximation to the distribution of an estimator.
% the underlying assumptions usually do not hold exactly.
To assess coverage, for each $I\in\{100,1000,10000\}$,
we use the \texttt{NB/Normal/Normal} scheme to generate $50$ data matrices with $J = 100$, $K = 4$, $L = 2$, and $M = 3$,
each with a different set of covariates and parameters.
For each data matrix, we run our NB-GBM estimation algorithm (with tolerance $\tau = 10^{-8}$)
and then we run our NB-GBM inference algorithm to obtain approximate standard errors.
We construct Wald-type confidence intervals for each univariate parameter, for example,
the 95\% confidence interval for $a_{j k}$ is $\hat{a}_{j k} \pm 1.96\, \hat{\mathrm{se}}$ where
$\hat{\mathrm{se}}$ is the approximate standard error for $\hat{a}_{j k}$.

% \todo{Actually, I'd have the check whether the following paragraph is right since I haven't tried computing p-values.}
% In addition to quantifying uncertainty in the parameter estimates, accurate standard errors are also essential for hypothesis testing.
% For instance, one can compute p-values for testing whether any given parameters equal zero,
% using the estimates and standard errors.
% Ideally, a p-value will be uniformly distributed under the null hypothesis, but as in the case of confidence intervals, this is not guaranteed due to inexact assumptions.
% Here, we report only the coverage since the p-value CDF is simply a transformation of the coverage curve \todo{check this}.

Figure~\ref{figure:coverage} shows actual coverage versus target coverage,
% for each parameter matrix/vector ($A$, $B$, $C$, $U$, $V$, $S$, and $T$),
estimated by combining across all $50$ runs and across all entries of each parameter matrix/vector.
Perfect coverage would be a straight line on the diagonal. % where actual coverage is exactly equal to target coverage.
We exclude $c_{1 1}$, $D$, and $\omega$ in these coverage results since it seems challenging to estimate them without non-negligible bias, skewing the results.
% We exclude the overall intercept $c_{1 1}$ in the calculation of coverage for $C$ since the estimates of $c_{1 1}$ appear to be significantly biased, skewing the results.
% Similarly, we do not attempt to provide standard errors for $D$ and $\omega$ since, empirically, it seems challenging to estimate these parameters without significant bias.
% thus, we do not calculate the coverage for $D$ and $\omega$.
We see in Figure~\ref{figure:coverage} that the actual coverage for $A$, $B$, $C$, $U$, and $S$ is excellent at every target coverage level from 0\% to 100\%.
For $V$ and $T$, the coverage is reasonably good for the smaller values of $I$ but appears to degrade when $I$ increases.
% Overall, for the parameters that are usually of primary interest, namely $A$, $B$, and $C$, the coverage appears to be excellent.

% although it is slightly lower for $I \in \{100, 10000\}$.
% The coverage for $C$ also appears to be very good; the curves for $C$ are noisier due to the smaller number of parameters in $C$.
% For $S$, the coverage is nearly perfect for $I \in \{1000, 10000\}$, and slightly conservative for $I = 100$.

\begin{figure}
  \centering
  \includegraphics[trim=0.8cm 0 1.3cm 0, clip, height=0.2\textheight]{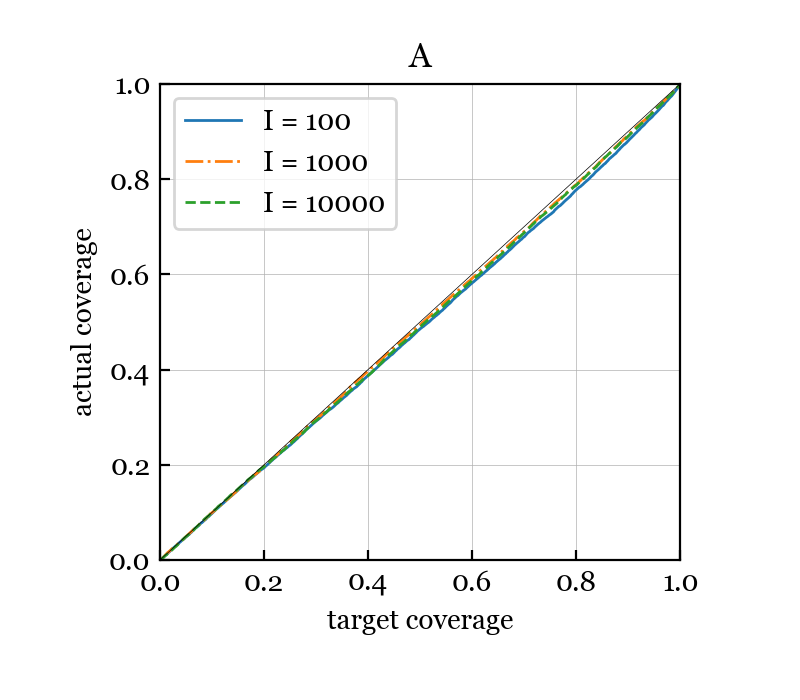}
  \includegraphics[trim=1.3cm 0 1.3cm 0, clip, height=0.2\textheight]{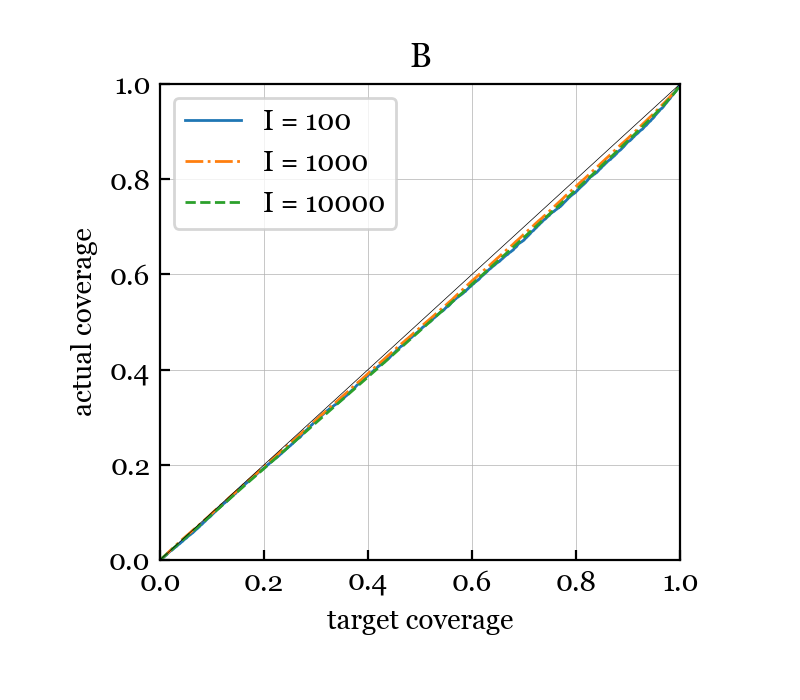}
  \includegraphics[trim=1.3cm 0 1.3cm 0, clip, height=0.2\textheight]{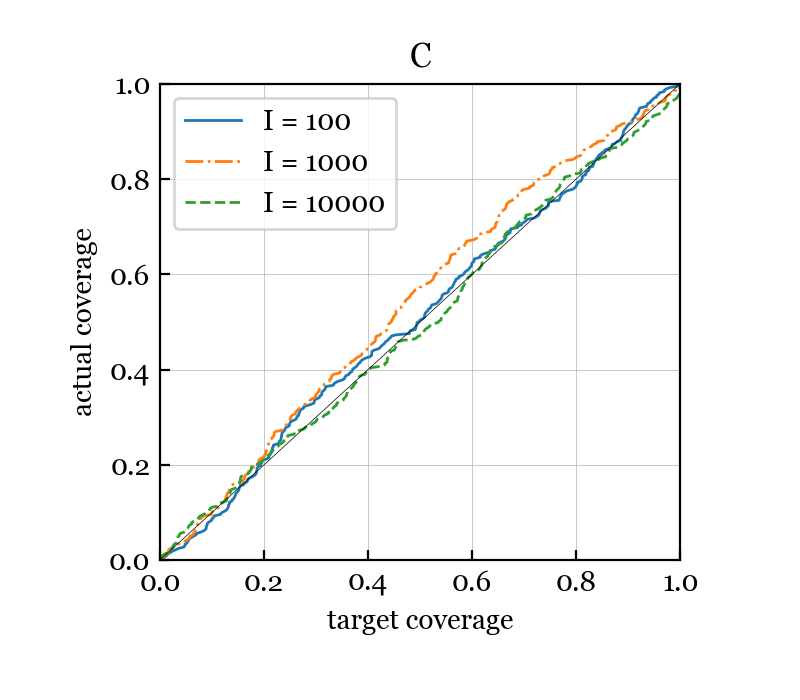}\\
  \includegraphics[trim=0.8cm 0 1.3cm 0, clip, height=0.2\textheight]{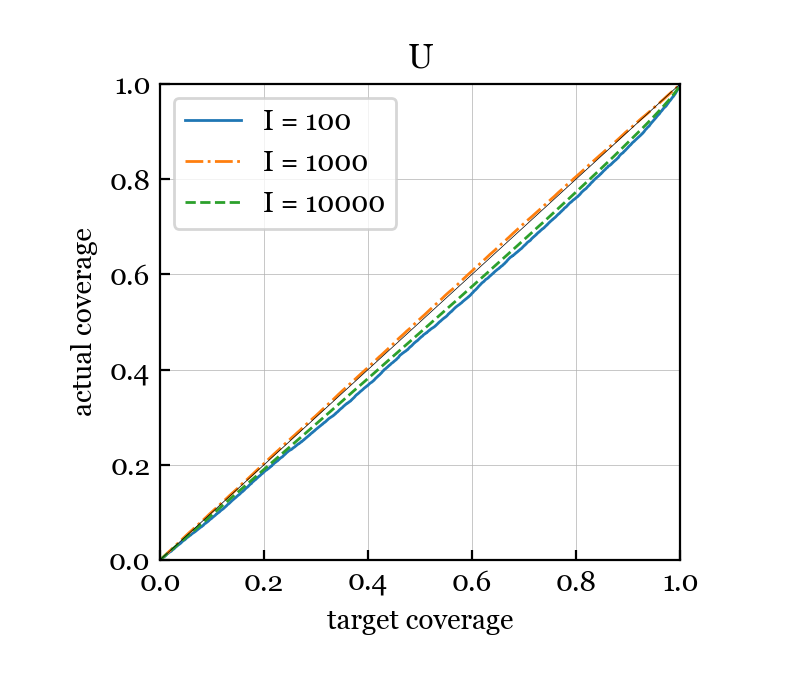}
  \includegraphics[trim=1.3cm 0 1.3cm 0, clip, height=0.2\textheight]{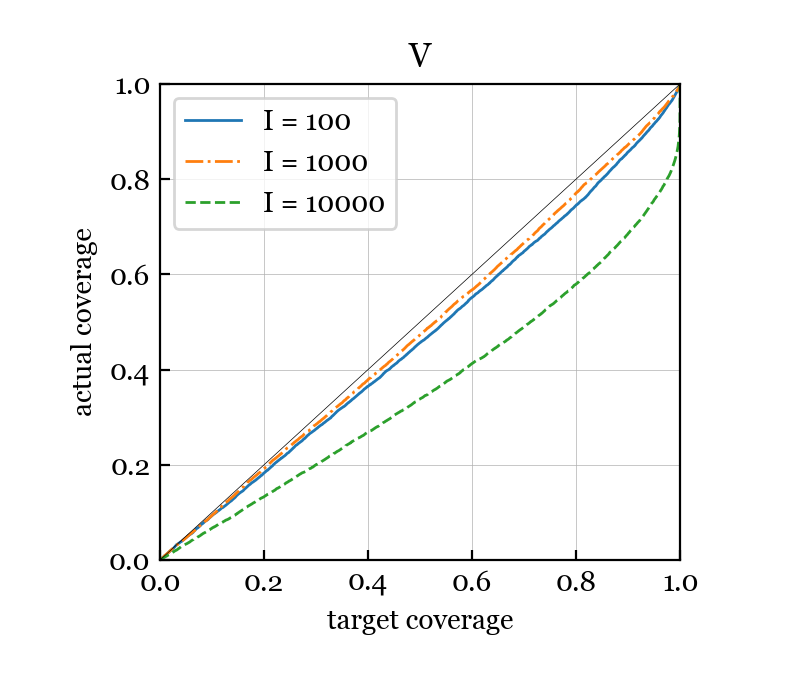}
  \includegraphics[trim=1.3cm 0 1.3cm 0, clip, height=0.2\textheight]{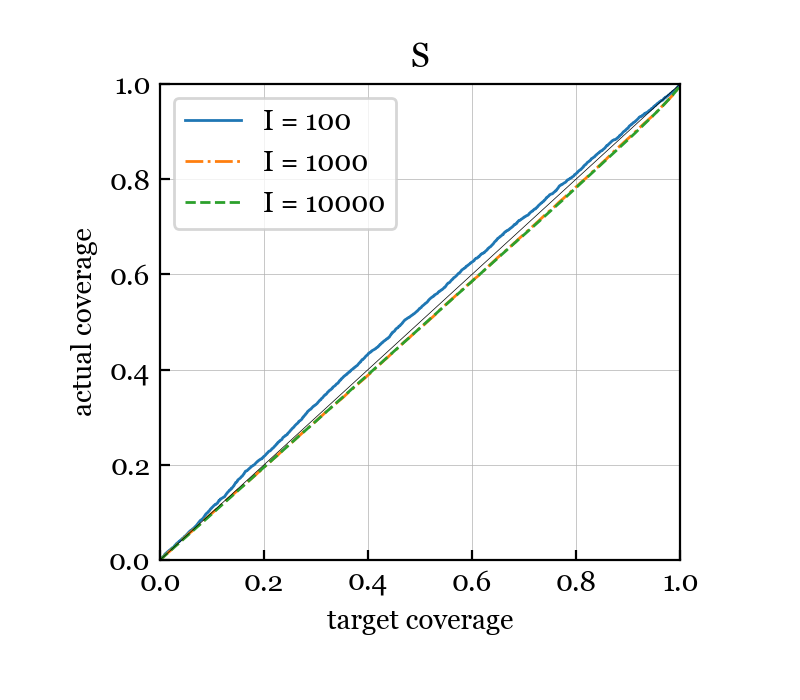}
  \includegraphics[trim=1.3cm 0 1.3cm 0, clip, height=0.2\textheight]{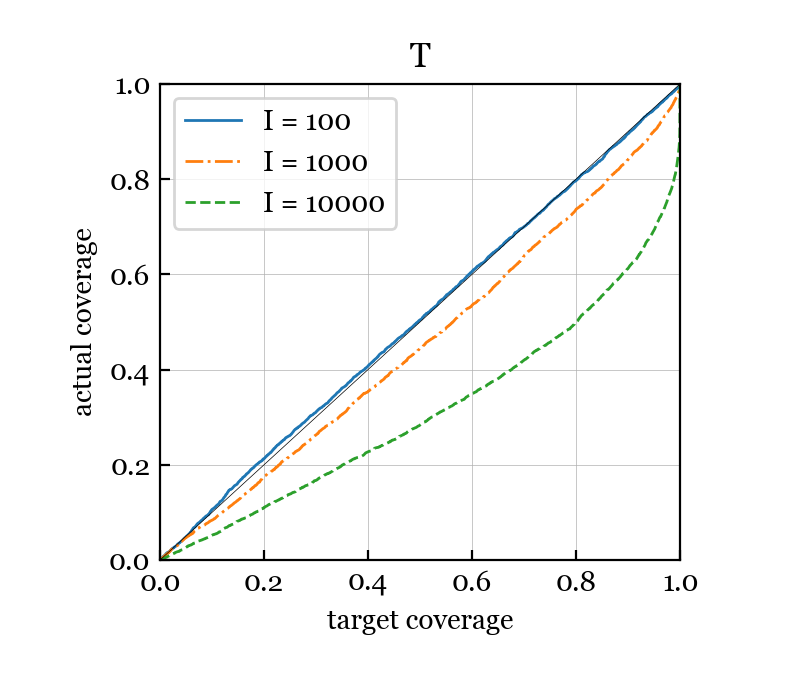}
  \caption{Coverage of confidence intervals for the entries of each parameter matrix/vector.}
  \label{figure:coverage}
\end{figure}

\subsection{Computation time and algorithm convergence}
\label{section:convergence}

\textbf{Computation time.}
% To assess the computation time of our proposed algorithm, 
For each combination of $I\in\{100,316,1000,3162,10000\}$ and $J\in\{15,30,60,120,240\}$, 
we generate $10$ data matrices using the \texttt{NB/Normal/Normal} simulation scheme with $K=4$, $L=2$, and $M=3$.
% For each data matrix, we run our estimation algorithm for $50$ iterations, followed by our inference algorithm.
For each $I$ and $J$, Figure~\ref{figure:computation-time} shows the average computation time per iteration of the estimation algorithm,
along with the average computation time for the inference algorithm.
These empirical results agree with our theory in Section~\ref{section:theory} showing that the time per iteration scales like $I J$ (that is, linearly with the size of the data matrix) and the time for inference scales like $I J^2$.

\begin{figure}
  \centering
  \includegraphics[trim=0.5cm 0 2.65cm 0, clip, height=1.6in]{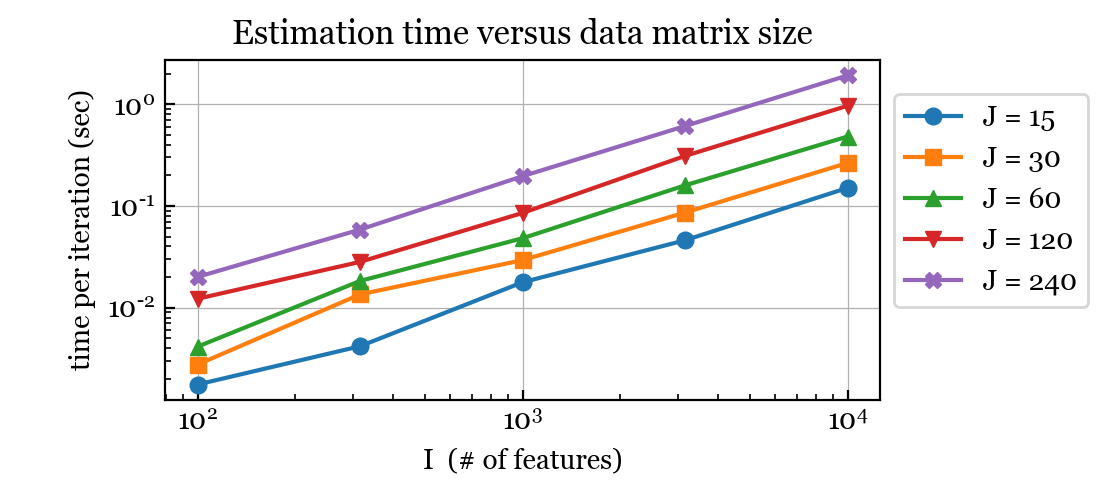}
  \includegraphics[trim=0.5cm 0 0      0, clip, height=1.6in]{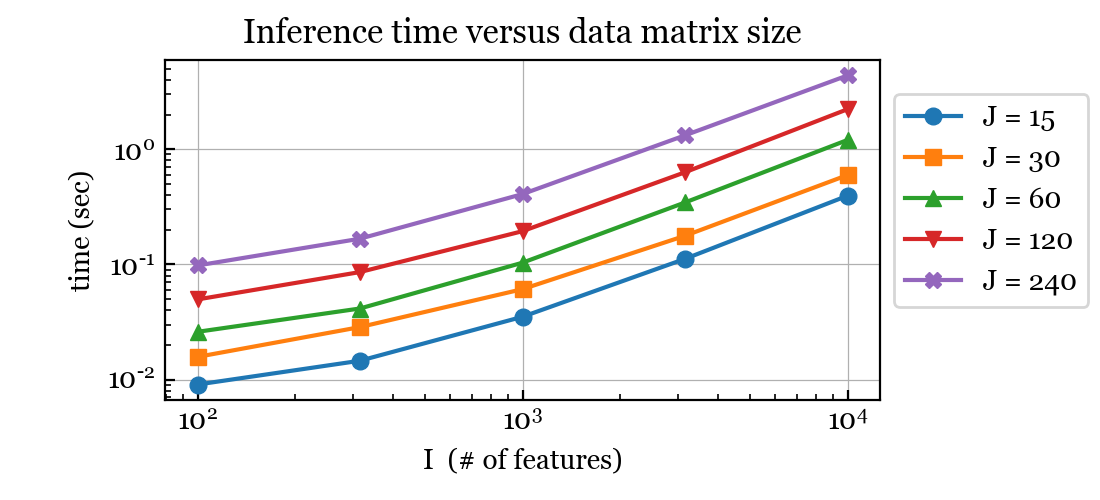}
  \caption{Computation time of our GBM algorithms as a function of $I$ and $J$.}
  \label{figure:computation-time}
\end{figure}

\textbf{Algorithm convergence.}
Next, we evaluate the number of iterations required for the estimation algorithm to converge.
Similarly to before, for $I\in\{100,1000,10000\}$, we run the \texttt{NB/Normal/Normal} scheme $25$ times with $J=100$, $K=4$, $L=2$, and $M=3$.
% we generate $25$ data matrices
% using the \texttt{NB/Normal/Normal} scheme with $J=100$, $K=4$, $L=2$, and $M=3$, and we run our estimation algorithm for $50$ iterations.
Figure~\ref{figure:algorithm-convergence} shows the log-likelihood+log-prior (plus a constant) versus iteration number for each simulation run.
In these simulations, the log-likelihood+log-prior levels off after around 5 or fewer iterations,
indicating that the algorithm converges rapidly.

% To summarize the state of the algorithm in a single number,
% we compute the log posterior density at the current estimate after each iteration.

\subsection{Robustness to the outcome distribution}
\label{section:robustness}

To assess the robustness of the NB-GBM to the assumption that the outcome distribution is negative binomial,
we rerun the experiments in Sections~\ref{section:consistency} and \ref{section:coverage}
using the following data simulation schemes:
(a) \texttt{LNP/Normal/Normal}
(b) \texttt{Poisson/Normal/Normal}, and
(c) \texttt{Geometric/Normal/Normal}.
The results in Figures~\ref{figure:LogNormalPoisson-outcomes}, \ref{figure:Poisson-outcomes}, and \ref{figure:Geometric-outcomes}
show that the algorithms are quite robust to misspecification of the outcome distribution.

\section{Application to gene expression analysis}
\label{section:gene-expression}

In this section, we evaluate our GBM algorithms on RNA-seq gene expression data.
An RNA-seq dataset consists of a matrix of counts in which entry $(i,j)$ is the number of high-throughput sequencing reads that were mapped to gene $i$ for sample $j$.
These read counts are related to gene expression level, but there are many biases -- both sample-specific and gene-specific. 
More generally, there are often significant sources of unwanted variation---both biological and technical---that obscure the signal of interest.
% that may be biological or technical and can be both sample-specific and gene-specific.
Most methods use pipelines that adjust for each bias sequentially, rather than in an integrated way.
GBMs enable one to use a single coherent model that adjusts for gene covariates and sample covariates as well as unobserved factors such as batch effects.
% Further, since any available covariates can be used, GBMs provide a flexible framework that can be adapted to the application at hand.

\subsection{Comparing to DESeq2 on lymphoblastoid cell lines}
\label{section:pickrell}

As a test of our GBM methods, we compare with DESeq2 \citep{love2014moderated}, a leading method for RNA-seq differential expression analysis.  
We first consider a benchmark dataset used by \citet{love2014moderated} consisting of 161 samples from lymphoblastoid cell lines \citep{pickrell2010understanding}.
% derived from 69 subjects from the Yoruba ethnic group in Nigeria
% we refer to this as the Pickrell data since it was provided by \citet{pickrell2010understanding}.
% Since many genes with essentially no counts are included in the original data, 
We use the subset of 20,815 genes with nonzero median count across samples.

In both DESeq2 and the GBM, we adjust for two sample covariates: sequencing center (Argonne or Yale) and cDNA concentration.
To adjust for GC content, which is often the most important gene covariate, 
in the GBM we construct the $X$ matrix using a natural cubic spline basis with knots at the 2.5\%, 25\%, 50\%, 75\%, and 97.5\% quantiles of GC content.
% quantiles of the distribution of GC content over the genes used.
DESeq2 does not have a built-in capacity to adjust for gene covariates, so for DESeq2, 
we adjust for GC using their recommended approach of pre-computing normalization factors using CQN \citep{hansen2012removing},
which uses the same spline basis.
Since DESeq2 does not adjust for latent factors, we first set $M = 0$ for direct comparison; later, we set $M = 2$. % for PCA-like dimension reduction.
% to construct a PCA-like low-dimensional representation.
% in Section~\ref{section:pickrell-visualization}, we set $M = 2$ to construct a PCA-like low-dimensional representation.

It is natural to use negative binomial (NB) outcomes for sequencing data since the technical variability is close to Poisson \citep{marioni2008rna}, and biological variability introduces overdispersion \citep{robinson2010edger}.
These modeling choices yield an NB-GBM with $I = 20{,}815$, $J = 161$, $K = 7$, and $L = 3$.
DESeq2 also uses an NB model, so the main difference between DESeq2 and this particular GBM is the 
way that the parameters and standard errors are estimated.
Using a 1.8GHz processor, 
GBM estimation and inference took 42 seconds, whereas DESeq2+CQN took 105 seconds.

\begin{figure}
  \centering
  \includegraphics[trim=0.6cm 0 1.1cm 0, clip, width=0.325\textwidth]{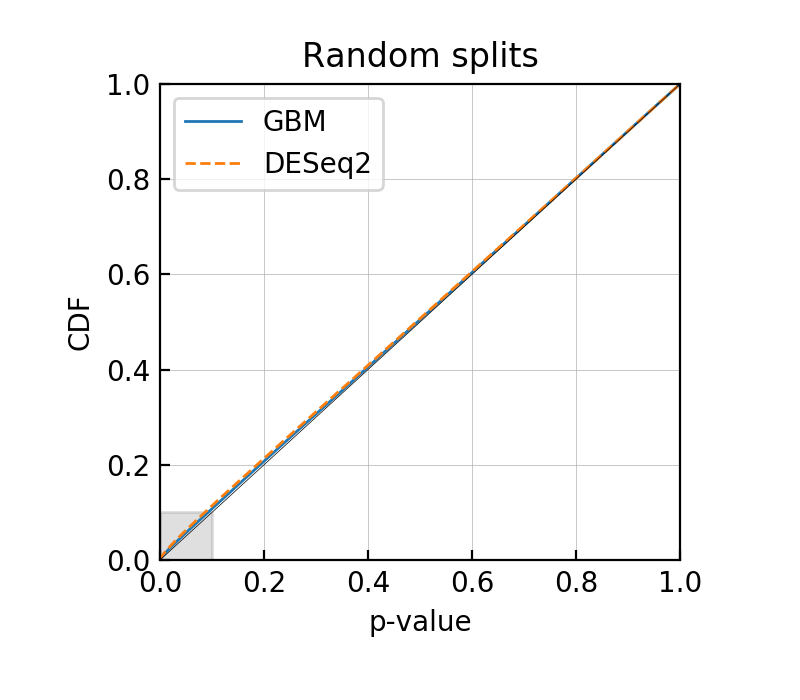}
  \includegraphics[trim=0.6cm 0 1.1cm 0, clip, width=0.325\textwidth]{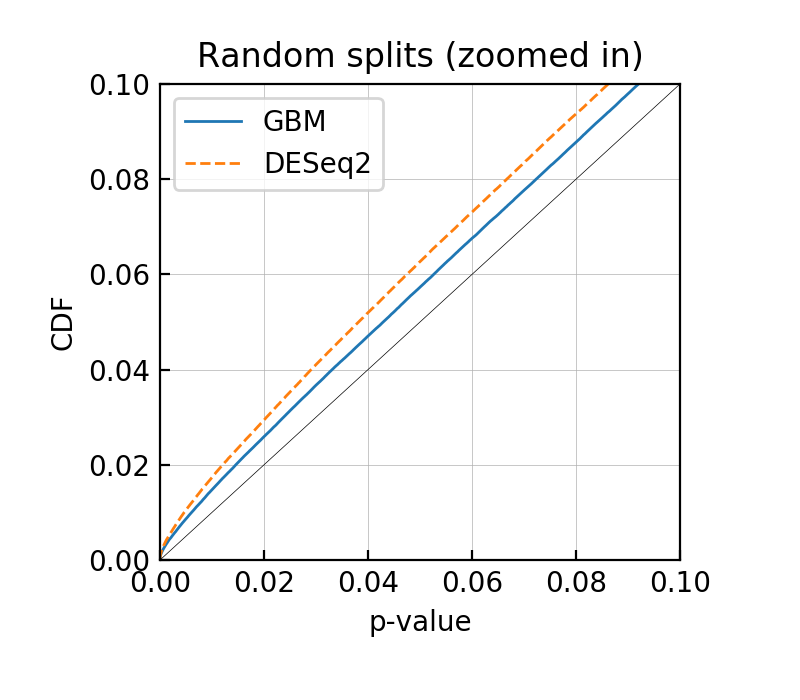}
  \includegraphics[trim=0.6cm 0 1.1cm 0, clip, width=0.325\textwidth]{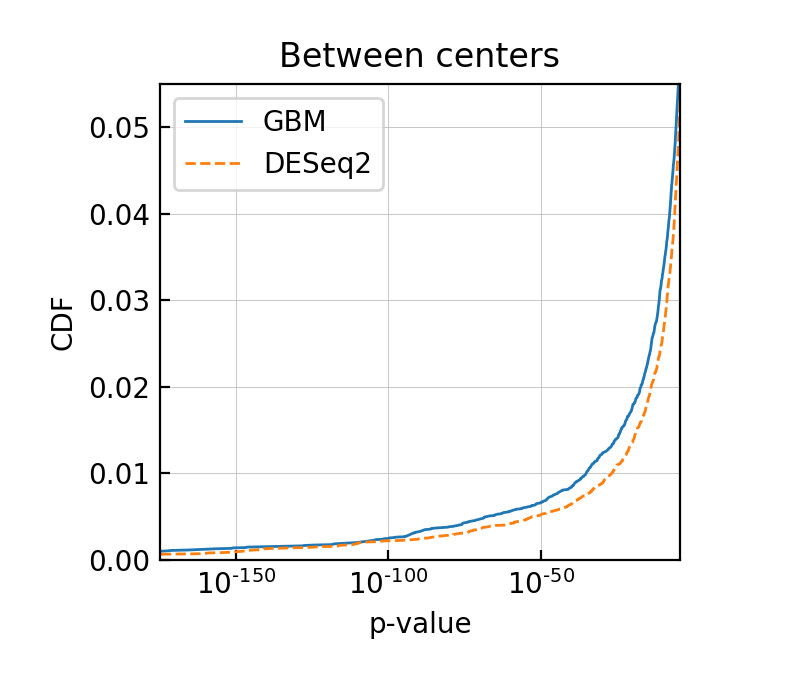}
  \caption{Comparison of p-value distributions on Pickrell data using DESeq2 and the GBM.
  (Left) CDFs of p-values for mock null comparisons over 50 random splits.
  (Middle) Same, but zoomed in to region shaded in left-hand figure.
  (Right) CDFs of p-values for testing for a difference between sequencing centers.
  The x-axis ends at the significance threshold for controlling FWER at 0.05.}
  % to facilitate visualization, the rightmost end of the x-axis is the significance threshold using the Bonferroni adjustment to control FWER at 0.05.}
  \label{figure:pickrell-pvalues}
\end{figure}

\textbf{Correctness of p-values under mock null comparisons.}
First, we assess the calibration of p-values for testing for differential expression between two conditions.
Under the null hypothesis of no difference, the p-values would ideally be uniformly distributed on $[0,1]$.
Since the Pickrell samples appear to be relatively homogeneous (when adjusting for sequencing center and cDNA concentration),
we can assess how well this ideal is attained 
by randomly splitting the samples into two groups and testing for differential expression.

To this end, we add a sample covariate $z_{j 4}$ consisting of a dummy variable for the assignment of samples to the two random groups.
Thus, the null hypothesis of no difference for gene $i$ is $b_{i 4} = 0$ and the alternative is $b_{i 4} \neq 0$.
The (two-sided) p-value for gene $i$ is 
$p_i := 2 \big(1 - \Phi(|\hat{b}_{i 4} / \hat{\mathrm{se}}(b_{i 4})|)\big)$
% where $\ell = 4$ is the index of the group assignment covariate and 
where $\Phi(x)$ is the standard normal CDF (cumulative distribution function).
Figure~\ref{figure:pickrell-pvalues} shows the p-value CDF over all genes,
aggregating over 50 random splits into two groups containing 80 and 81 samples, respectively.
Both DESeq2 and the GBM yield p-values that are very close to the ideal uniform distribution.
% with the GBM being slightly closer to uniform on these data.
This indicates that both methods are accurately controlling the false positive rate.

\textbf{Sensitivity to detect actual differences.}
% Next, we compare sensitivity.
% The Pickrell data consist of 80 samples from the Argonne sequencing center and 81 samples from the Yale sequencing center.
% Two sequencing centers, Argonne and Yale, were used to generate the samples, with 80 samples from Argonne and 81 samples from Yale.
To compare sensitivity,
we test for differential expression between sequencing centers by computing p-values 
$p_i = 2 \big(1 - \Phi(|\hat{b}_{i \ell} / \hat{\mathrm{se}}(b_{i \ell})|)\big)$
where $\ell$ is the index of the sequencing center covariate.
% Adjusting for multiple testing 
Using Bonferroni to control the family-wise error rate (FWER) at 0.05,
the number of genes detected as differentially expressed by the GBM and DESeq2 are 1038 and 892, respectively.
Figure~\ref{figure:pickrell-pvalues} shows the lower tail of the p-value CDFs.
In these results, the GBM yields equal or greater sensitivity.
% On these data, the GBM yields better performance.

\begin{figure}
  \centering
  \includegraphics[trim=1cm 0.5cm 1.8cm 2cm, clip, height=0.4\textheight]{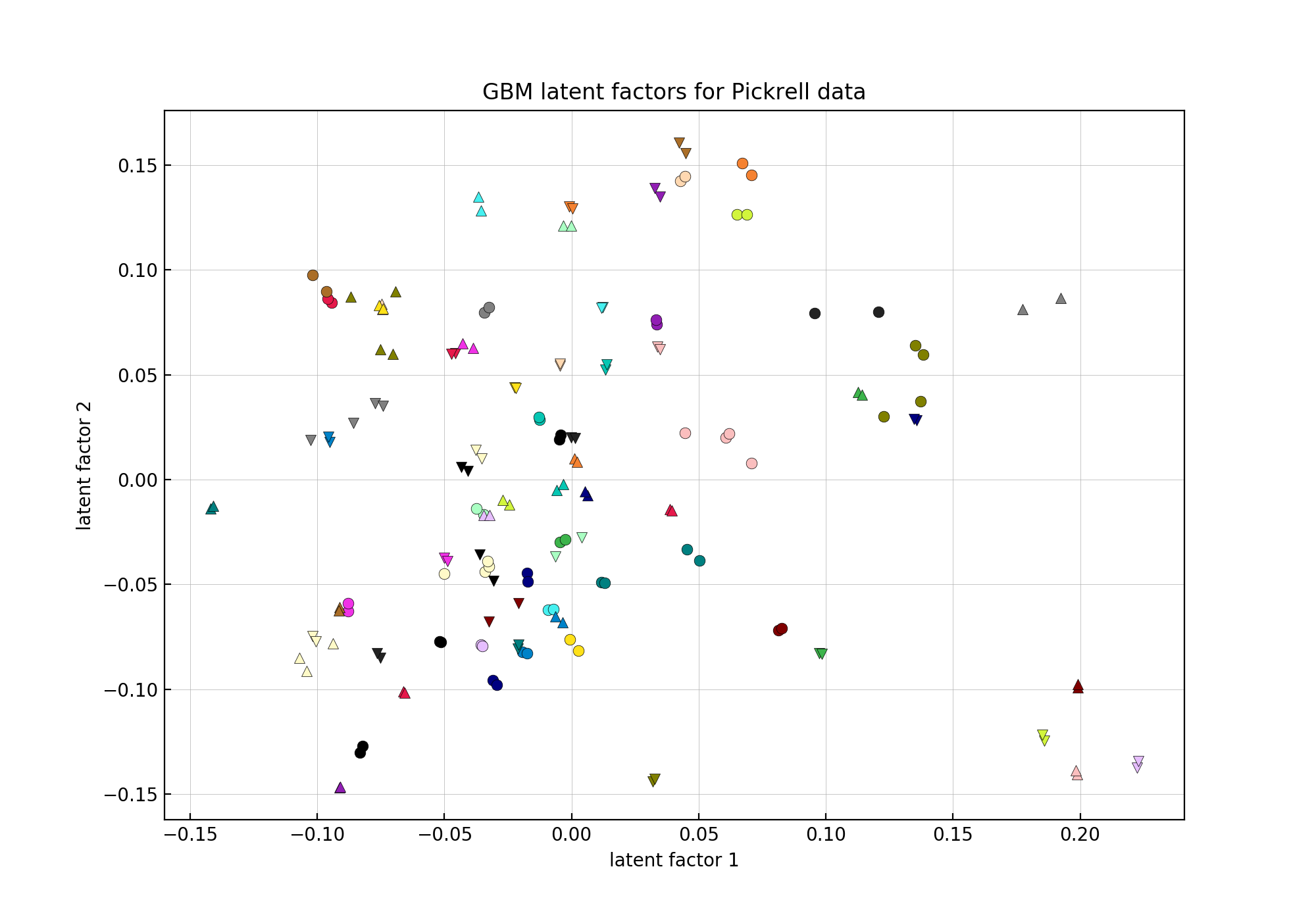}
  \caption{Visualization of Pickrell data using GBM latent factors, adjusting for covariates. Each dot represents one of the 161 samples, and 
  each subject ID is indicated by a different combination of color and shape (by cycling through 32 colors and 3 shapes).}
  % the color+shape indicates the subject ID.
  \label{figure:pickrell-latent-factors}
\end{figure}

\textbf{Visualization using GBM latent factors.}
% \label{section:pickrell-visualization}
The latent factors of the GBM provide a model-based approach to visualizing high-dimensional count data, while adjusting for covariates.
To illustrate, we modify the model to use $M = 2$. % and run our GBM estimation algorithm.
Figure~\ref{figure:pickrell-latent-factors} shows a scatterplot of $v_{j 2}$ versus $v_{j 1}$ for the estimated $V$ matrix.
Observe that this yields very tightly grouped clusters of samples from the same subject.
This is analogous to plotting the first two scores in principal components analysis (PCA).
Thus, for comparison, Figure~\ref{figure:pickrell-pca} shows the PCA plots based on (a) log-transformed TPMs (Transcripts per Million), 
specifically, $\log(\texttt{TPM}_{i j}+1)$,
and (b) the  variance stabilizing transform (VST) in the DESeq2 package, using the GC adjustment from CQN.
The DESeq2 model does not estimate latent factors, which is why PCA is used in DESeq2.
The TPM plot is very noisy in terms of subject ID clusters. The VST plot is better than TPMs, but still not quite as clean as the GBM plot.

% The TPM-based plot is far noisier than the GBM plot in terms of subject ID clusters.
% The VST-based plot is much better than using TPMs, but is still not quite as clean as the GBM plot.
% This illustrates that the model-based approach is more effectively extracting the signal from the noise.

% Although the DESeq2 method does not directly provide residuals suitable for PCA,
% the DESeq2 software package provides a separate function for this purpose; 
% in Figure~\ref{figure:pickrell-pca-deseq}, we see that this is comparable to the GBM PCA results.

Overall, in terms of sensitivity, controlling false positives, computation time, and visualization,
these results suggest that the GBM performs very well.
When $J$ is very small, it may be beneficial to augment the GBM to use DESeq2-like shrinkage estimates for $s_i$.
% compared to 
% we find that our GBM approach yields better performance than
% a state-of-the-art competitor that is specifically designed for RNA-seq differential expression analysis.

% (GBM)    Number of p-values significant at p < 0.05 after Bonferroni correction: 1038
% (DESeq2) Number of p-values significant at p < 0.05 after Bonferroni correction: 892

% This is highly elevated over the number of genes that would be significant under the null, 
% which would be typically be zero since the expected number is $I (0.05/I) = 0.05 \ll 1$.

\subsection{Analyzing GTEx data for aging-related genes}
\label{section:gtex}

Next, we test our methods on an application of scientific interest, 
using RNA-seq data from the Genotype-Tissue Expression (GTEx) project \citep{mele2015human}, 
consisting of 8{,}551 samples from 30 tissues in the human body, obtained from 544 subjects.
We apply the GBM to find genes whose expression changes with age, adjusting for technical biases.
See \citet{jia2018analysis} and \citet{zeng2020transcriptome} for studies of age-related genes using GTEx.

We use the GTEx RNA-seq data from recount2 \citep{collado2017reproducible},\footnote{Downloaded from \url{https://jhubiostatistics.shinyapps.io/recount} on 8/7/2020.} normalized using the \texttt{scale\_counts} function in the \texttt{recount} R library.
We use the subset of 8{,}551 samples that passed GTEx quality control,
% (\texttt{smafrze} = ``\texttt{USE ME}'')
and the subset of genes in chromosomes 1--22 that have an HGNC-approved gene symbol and have nonzero median across all samples.

\begin{figure}
  \centering
  \includegraphics[trim=0.5cm 0cm 0.5cm 0cm, clip, height=0.5\textheight]{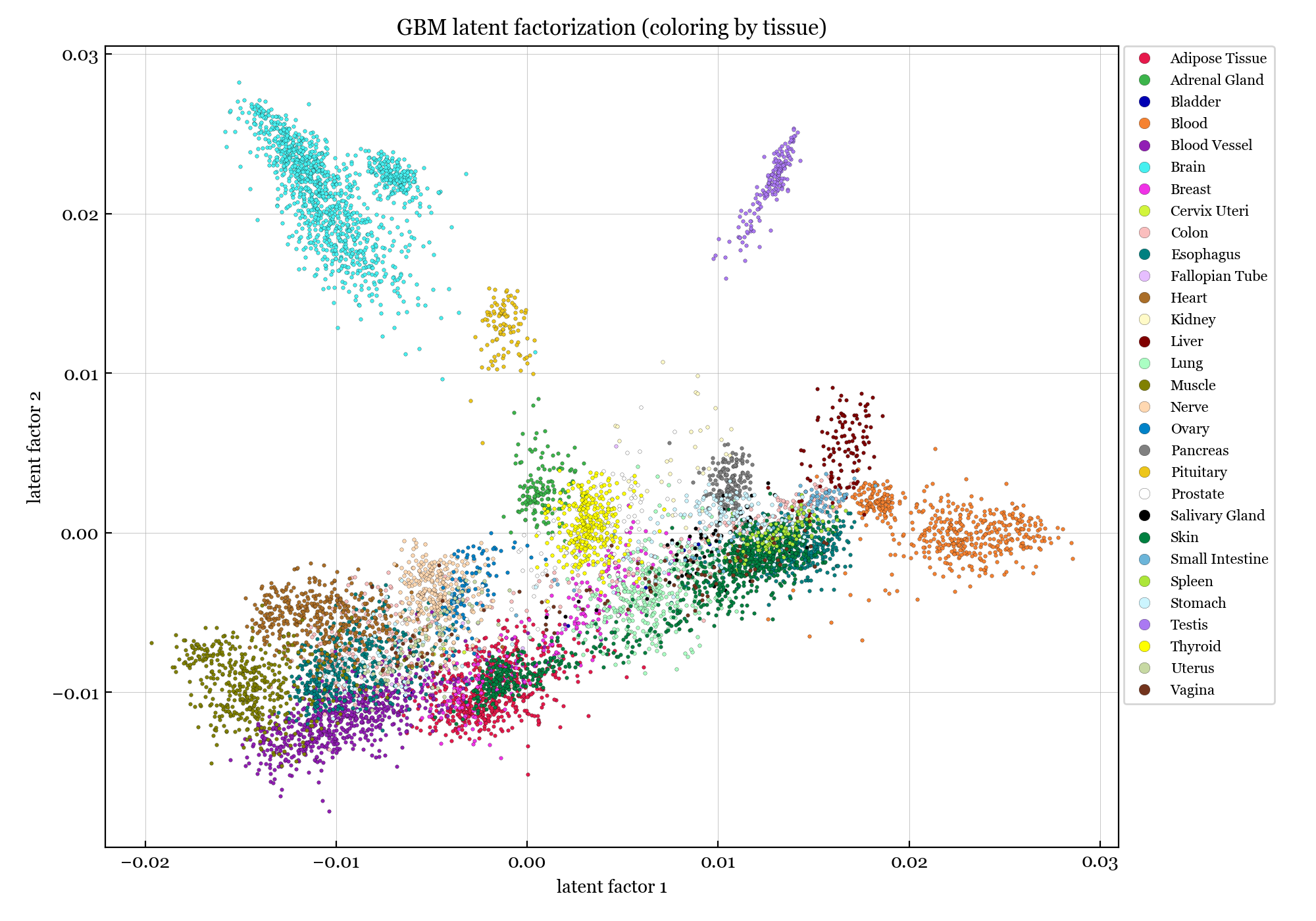}
  \caption{Visualization of GTEx data using NB-GBM latent factors, adjusting for covariates. Each dot represents one of the 8{,}551 samples, and the color indicates the tissue type.}
  \label{figure:gtex-smts}
\end{figure}

To visualize the samples, we take a random subset of 5{,}000 genes and estimate an NB-GBM with two latent factors and no sample covariates.
For the gene covariates, we use $\log(\texttt{length}_i)$, $\texttt{gc}_i$, and $(\texttt{gc}_i - \overline{\texttt{gc}})^2$ where $\texttt{length}_i$ is the sum of the exon lengths and $\texttt{gc}_i$ is the GC content of gene $i$.  Thus, in this initial model for visualization, $I = 5{,}000$, $J = 8{,}551$, $K = 4$, $L = 1$, and $M = 2$.
Figure~\ref{figure:gtex-smts} shows the latent factors ($v_{j 2}$ versus $v_{j 1}$), similarly to Figure~\ref{figure:pickrell-latent-factors}.
The samples tend to fall into clusters according to the tissue from which they were taken.
Some tissues, such as brain and blood, clearly contain two or more subclusters which turn out to correspond to subtissue types (Figure~\ref{figure:gtex-smtsd}).
Meanwhile, when more latent factors are used (that is, $M > 2$), some clusters that overlap in Figure~\ref{figure:gtex-smts} become well-separated in higher latent dimensions.
For comparison, running PCA on the log TPMs is not nearly as clear in terms of tissue/subtissue clusters (Figure~\ref{figure:gtex-pca-tpms}).

\textbf{Testing for age-related genes.}
To find genes that are related to aging, we add subject age as a sample covariate.  
Each gene then has a coefficient describing how its expression changes with age, and we compute a p-value for each gene to test whether its coefficient is nonzero.
Due to the heterogeneity of tissue/subtissue types, we analyze each subtissue type separately.
To perform both exploratory analysis and valid hypothesis testing, we used a random subset of 108 subjects during an exploratory model-building phase and then used the remaining 436 subjects during a testing phase with the selected model.

In the exploratory phase, we considered adjusting for various technical sample covariates and gene covariates, and varied $M$ from 0 to 10.  For each model and each subtissue type, we used the GBM to find the set of genes exhibiting a significant association with age, controlling FWER at 0.05 using the Bonferroni correction.  To score the relevance of each of these gene sets in terms of aging biology,
we computed its F1 score for overlap with the set of aging-related genes identified by 
\citet{de2009meta}.\footnote{From \url{https://genomics.senescence.info/gene_expression/signatures.html} on 8/11/2020.}
Based on this exploratory analysis, we chose to keep $\log(\texttt{length}_i)$, $\texttt{gc}_i$, and $(\texttt{gc}_i - \overline{\texttt{gc}})^2$ as gene covariates, and use \texttt{smexncrt} (exonic rate, the fraction of reads that map within exons)
as well as \texttt{age} (subject age, coded as a numerical value in $\{25, 35, \ldots, 75\}$) as sample covariates.
For each subtissue, we chose the $M$ that yielded the highest F1 score on the exploratory data.
% F1 score of 0.0767 -- see f1scores_meta.txt in this folder: exploratory_results_testing_age_by_tissue-adjust_for_smexncrt-passed_qc

In the testing phase, we apply the selected model for each subtissue to test for age-associated genes.
For illustration, we present results for the ``Heart - Left Ventricle'' subtissue (Heart-LV), which had the highest F1 score across all subtissues
on the exploratory data.
We ran the GBM on the 176 Heart-LV samples in the test set,
using the 19{,}853 genes with nonzero median across these samples, with $M = 3$ based on the exploratory phase.
Thus, in this model, $I = 19{,}853$, $J = 176$, $K = 4$, $L = 3$, and $M = 3$.

\begin{figure}
  \centering
  \includegraphics[trim=0.6cm 0 1.1cm 0, clip, width=0.325\textwidth]{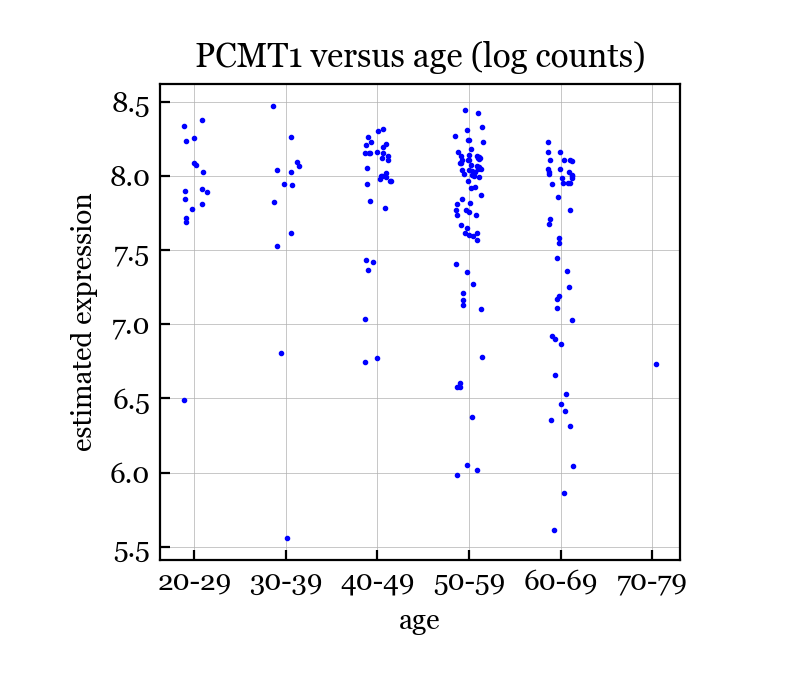}
  \includegraphics[trim=0.6cm 0 1.1cm 0, clip, width=0.325\textwidth]{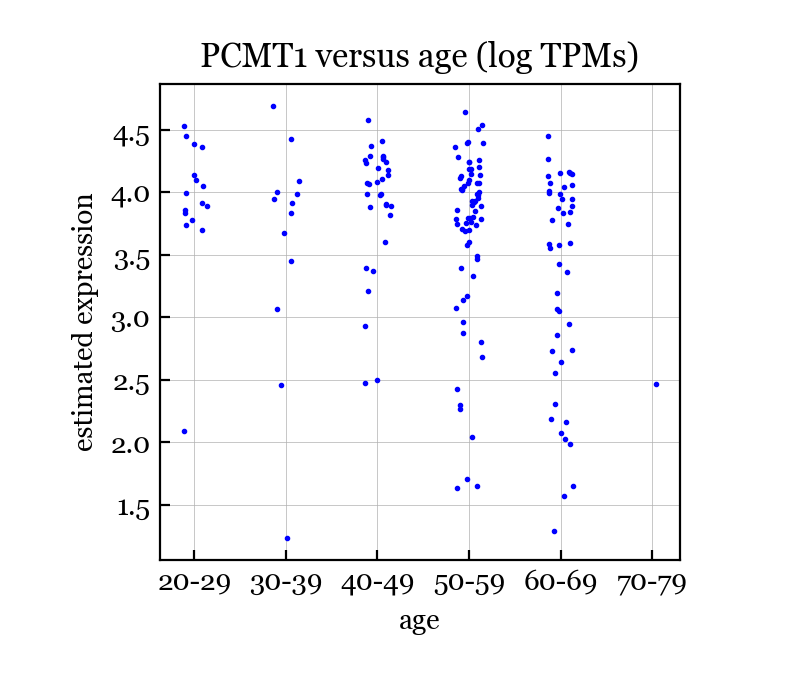}
  \includegraphics[trim=0.6cm 0 1.1cm 0, clip, width=0.325\textwidth]{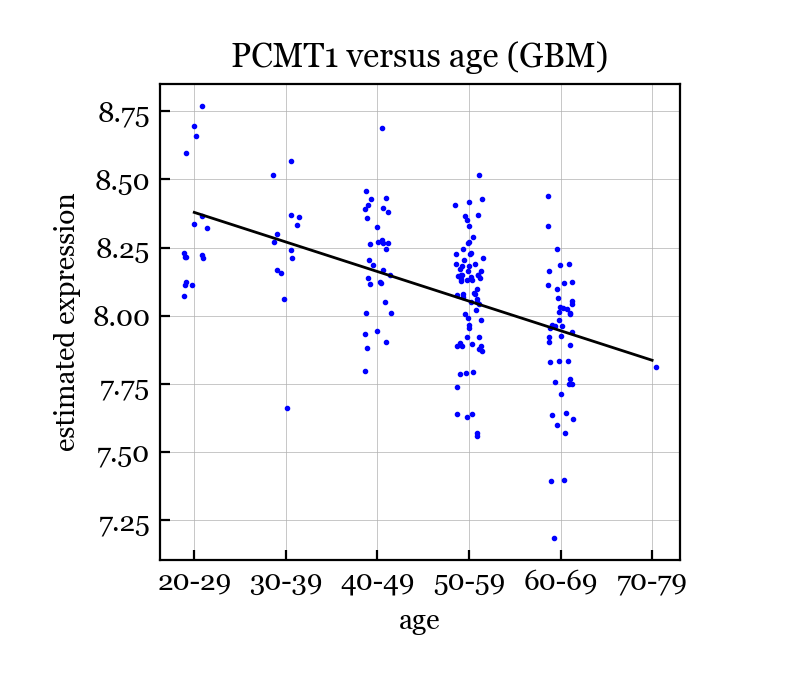}
  \caption{Estimated PCMT1 expression based on log counts, log TPMs, and GBM residuals.}
  %Comparing estimates of PCMT1 gene expression versus age.
  \label{figure:pcmt1}
\end{figure}

\textbf{Results.}
We found 2{,}444 genes to be significantly associated with age in Heart-LV, controlling FWER at 0.05.
For comparison, simple linear regression on the log TPMs yields only 1 significant gene;
thus, the GBM has much greater power than a simple standard approach.
To validate the biological relevance of the GBM hits, % for age-associated genes, 
we compare with what is known from the aging literature.
First, the top GBM hit for Heart-LV is PCMT1 (Entrez gene ID 5110) with a p-value of $1.1\times 10^{-47}$.
PCMT1 is involved in the repair and degradation of damaged proteins, and is a well-known aging gene,
being one of 307 human genes in the GenAge database (build version 20) from \citet{tacutu2018human}.
Figure~\ref{figure:pcmt1} shows the estimated expression of PCMT1 versus age for the Heart-LV samples.
The GBM-estimated expression exhibits a clear downward linear trend with age.
For comparison, Figure~\ref{figure:pcmt1} shows that the log TPMs are considerably noiser and the trend is much less clear.
Simple linear regression on the log TPMs yields a p-value of $1.1\times 10^{-3}$ for PCMT1, which does not reach the Bonferroni significance threshold of $0.05/I = 2.5\times 10^{-6}$.
Here, we define the GBM-estimated expression as the partial residual
$\hat{c}_{1 1} + \hat{a}_{j 1} + \hat{b}_{i 1} + (\hat{c}_{1 \ell} + \hat{b}_{i \ell}) z_{j \ell} + \varepsilon_{i j}$
where $\ell$ is the index of the $\texttt{age}$ column in $Z$, and
$\varepsilon_{i j}$ is the GBM residual (Section~\ref{section:residuals}).

To evaluate the GBM hits altogether for biological relevance to aging, we test for enrichment of Gene Ontology (GO) term gene sets using DAVID v6.8 \citep{huang2009bioinformatics,huang2009systematic}.
We run DAVID on the top 1000 GBM hits for Heart-LV, using all 19{,}853 tested genes as the background list. (DAVID allows at most 1000 genes.)
Tables~\ref{table:go-bp-terms} and \ref{table:go-cc-terms} show the top 20 enriched GO terms in the Biological Process and Cellular Component categories.
These results are highly consistent with known aging biology \citep{lopez2013hallmarks}.

% \todo{Possibly cite the following papers that studied age-related changes in gene expression (see comment).}
% in this excerpt from Song et al (2018) "The Systems Biology of Single-Cell Aging":
%     Moreover, both protein coding (Harris et al., 2017; Janssens et al., 2015) and non-coding
%     (Kim et al., 2016; Thum, 2014) RNA transcript levels change with age, and a meta-analysis has found 
%     aging-related changes in the transcription from genes associated with certain pathways (inflammation,
%     mitochondria, and lysosome) across several species (de Magalhaes et al., 2009), although the extent to
%     which these changes are causative for or responsive to aging is not clear. Although regulation of global
%     transcriptional activity is a key mediator within longevity pathways (Filer et al., 2017), transcriptome
%     modifications may be a second-order consequence of other effectors in the complex aging network.
%     However, as potential contributors to emergent network properties, transcriptome modifications may
%     indeed be important effectors of the aging phenotype.

% Just to highlight a few connections, \todo{finish -- NAD and ``TCA cycle'' + ``pyruvate metabolic process'', mitochondrial regulation, protein catabolism, cadherins, chaperonins and telomeres, DNA damage reponse, proteasome, NIK/NF-kappaB signaling, mTOR and ``formation of translation preinitiation complex'', Wnt signaling, autophagy  \dots  Can we group some of these perhaps?}.

\section{Application to cancer genomics}
\label{section:cancer}

Next, we apply the GBM to estimate copy ratios for sequencing data from cancer cell lines.
Copy ratio estimation is an essential step in detecting somatic copy number alterations (SCNAs), 
that is, duplications or deletions of segments of the genome.
The input data is a matrix of counts where entry $(i,j)$ is the number of reads from sample $j$ that map to target region $i$ of the genome.
The goal is to estimate the copy ratio of each region, that is,
the relative concentration of copies of that region in the original DNA sample.

Simple estimates based on row and column normalization are very noisy and are contaminated by significant technical biases.
State-of-the-art methods employ a panel of normals (that is, sequencing samples from non-cancer tissues)
to estimate technical biases using principal components analysis (PCA), and then
use linear regression to remove these biases from the cancer samples of interest.
We take an analogous approach, first running a GBM on a panel of normals, and then running a GBM on the cancer samples 
using a feature covariate matrix $X$ that includes the $U$ matrix estimated from the panel of normals.

To assess performance, we compare with the state-of-the-art method provided by the Broad Institute's Genome Analysis Toolkit (GATK)
\citep{gatk} on the 326 whole-exome sequencing samples from the Cancer Cell Line Encyclopedia (CCLE) \citep{ghandi2019next}.
These samples are from a wide range of cancer types, including lung, breast, colon, prostate, brain, and many others. 
% Each sample consists of a vector of counts, where count $i$ is the number of sequencing reads mapping to target region $i$ in the genome.
We use the subset of 180{,}495 target regions that are in chromosomes 1--22 and have nonzero median count across the 326 samples.

Since there are essentially no normal samples in the CCLE dataset, we create a panel of pseudo-normals by
taking a random subset of 163 samples as a training set and de-segmenting them to adjust for copy number alterations;
see Section~\ref{section:cancer-details} for details.  The remaining 163 samples are used as a test set.
For the GBM, we use $\log(\texttt{length}_i)$, $\texttt{gc}_i$, and $(\texttt{gc}_i - \overline{\texttt{gc}})^2$ as region covariates,
no sample covariates, and 5 latent factors.
Thus, $I = 180{,}495$, $J = 163$, $K = 4$, $L = 1$, and $M = 5$ on the training set, while
$I = 180{,}495$, $J = 163$, $K = 9$, $L = 1$, and $M = 0$ on the test set.
We define the GBM copy ratio estimates as the exponentiated residuals $\tilde{Y}_{i j}/\hat{\mu}_{i j}$ where $\tilde{Y}_{i j} := Y_{i j} + 1/8$;
see Section~\ref{section:residuals}.
The GBM took 10 minutes and 4.3 minutes to run on the training and test sets, respectively, 
while GATK took 3.3 minutes and 28 minutes on training and test, respectively.
The slowness of GATK on the test set is likely due to having to run it separately on every test sample.

% The GBM estimation algorithm took 10 minutes to run on the training set, and 4.3 minutes on the test set.
% GATK took 3.3 minutes on the training set and 28 minutes on the test set, however, 
% the slowness of GATK on the test set is likely due to having to perform a separate call for each test sample.

\begin{figure}
  \centering
  \includegraphics[trim=2.0cm 0cm 0.5cm 0cm, clip, width=0.95\textwidth]{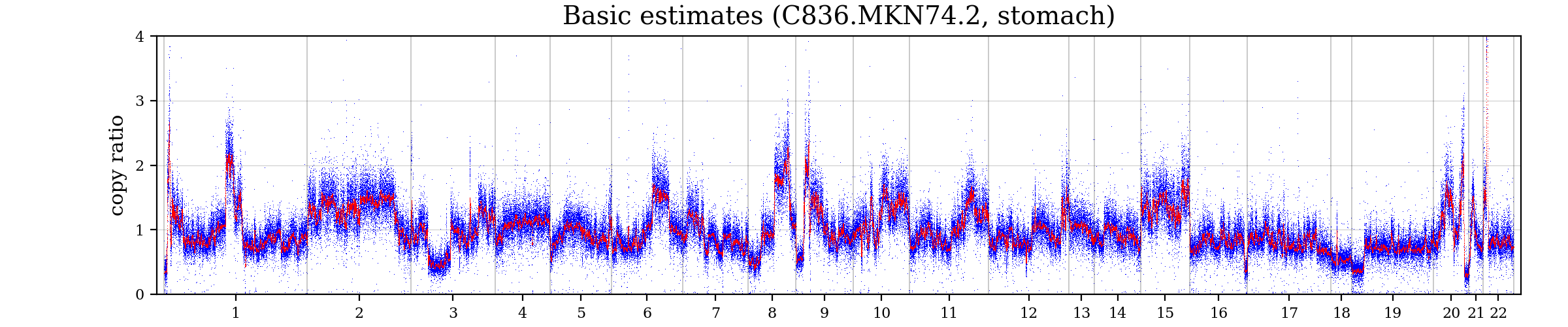} \\
  \includegraphics[trim=2.0cm 0cm 0.5cm 0cm, clip, width=0.95\textwidth]{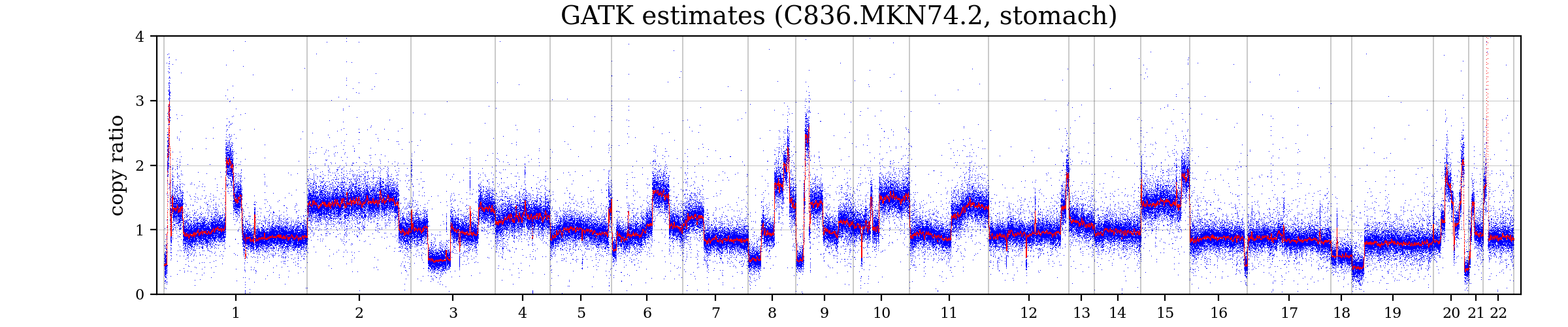} \\
  \includegraphics[trim=2.0cm 0cm 0.5cm 0cm, clip, width=0.95\textwidth]{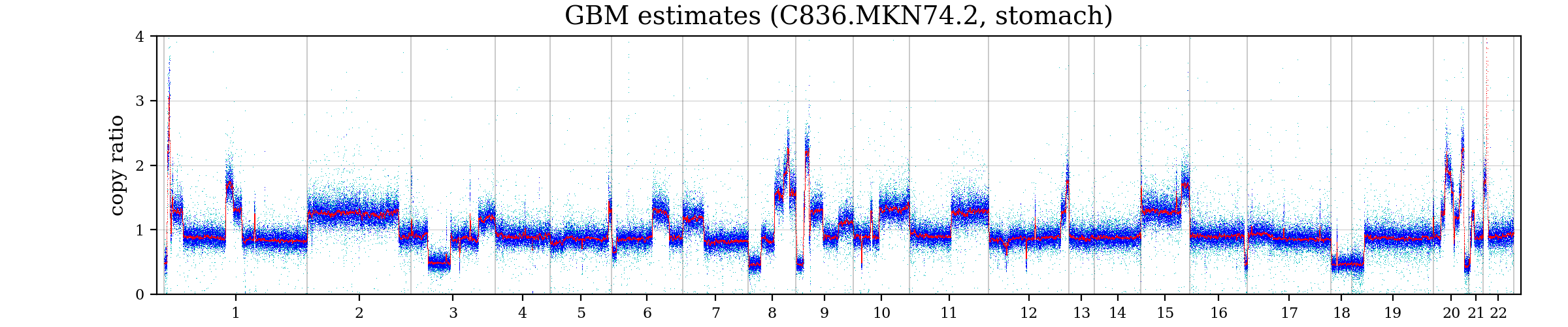} \\ \vspace{0.5em}
  \caption{Copy ratio estimates for an illustrative sample from the CCLE data. % whole-exome sequencing dataset. 
  The x-axis is genomic position, and each blue dot is the estimate for one region; moving averages are in red.
  For the GBM, regions with high and low precision estimates are plotted in blue and cyan, respectively.}
  \label{figure:ccle-example}
\end{figure}

Figure~\ref{figure:ccle-example} shows the GBM and GATK copy ratio estimates for an illustrative sample from the test set.
As a baseline, we also show the simple normalization-based estimates defined as $\tilde{Y}_{i j} / (\alpha_i \beta_j)$ 
where $\alpha_i = \frac{1}{J}\sum_{j=1}^J \tilde{Y}_{i j}$ and $\beta_j = \frac{1}{I}\sum_{i=1}^I \tilde{Y}_{i j} / \alpha_i$.
%; here, $1/8$ is a pseudocount that avoids issues when taking logs.
A major advantage of the GBM is that it provides uncertainty quantification.
Here, the estimated precision (that is, the inverse variance) of each log copy ratio estimate is $w_{i j}$ (Section~\ref{section:residuals}).
In the GBM plot in Figure~\ref{figure:ccle-example}, this is illustrated by using cyan for the regions with low relative precision;
see Section~\ref{section:cancer-details}.
% This improves the accuracy of downstream analyses such as SCNA detection by allowing one to downweight regions with low estimated precision.
By downweighting regions with low estimated precision, downstream analyses such as SCNA detection can be made more accurate.

\begin{figure}
  \centering
  \includegraphics[trim=0.6cm 0 1.1cm 0, clip, width=0.4\textwidth]{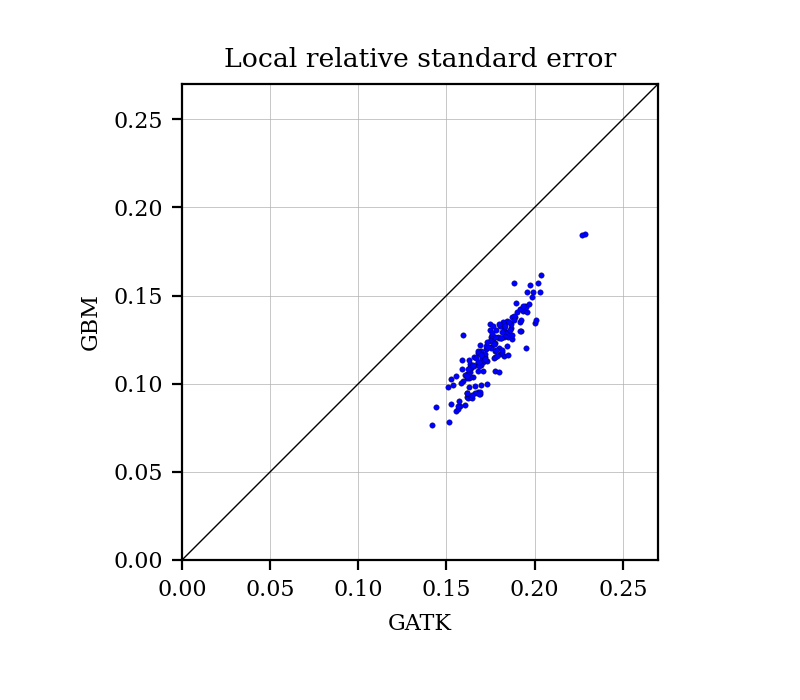} 
  \includegraphics[trim=0.6cm 0 1.1cm 0, clip, width=0.4\textwidth]{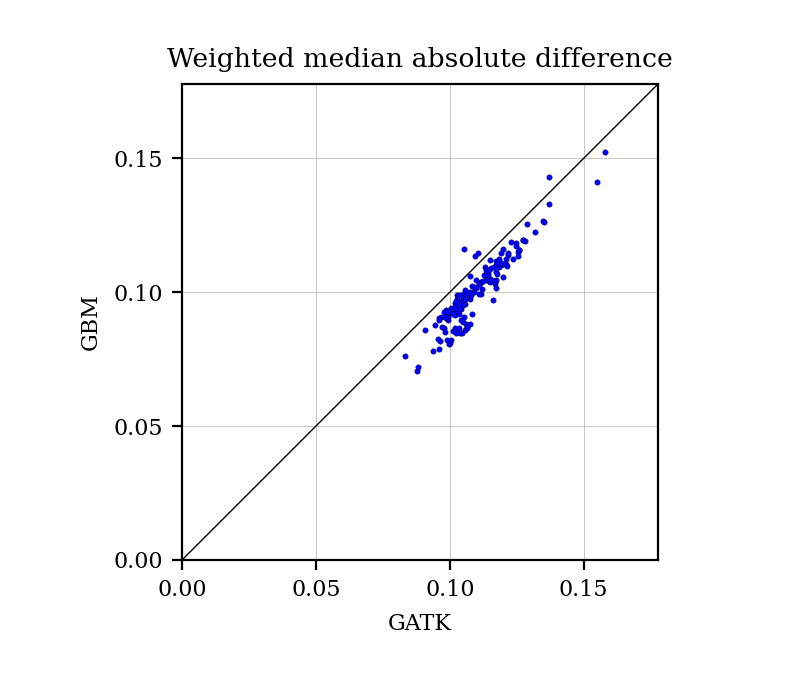}
  \caption{Performance of the GBM versus GATK on the 163 test samples from the CCLE whole-exome sequencing dataset.
  The GBM exhibits better performance in terms of both metrics.}
  \label{figure:ccle-metrics}
\end{figure}

To quantify performance, Figure~\ref{figure:ccle-metrics} compares the GBM and GATK in terms of two performance metrics.
The local relative standard error quantifies the variability of the log copy ratio estimates around a weighted moving average,
accounting for the precision of each estimate.
Meanwhile, the weighted median absolute difference quantifies the typical magnitude of the slope of a weighted moving average.
On these data, the GBM exhibits better performance in terms of both metrics;
see Section~\ref{section:cancer-details} for details.
Figures~\ref{figure:ccle-heatmap-gatk} and \ref{figure:ccle-heatmap-gbm} show the GBM and GATK copy ratio estimates for all 163 test samples.
% For visualization purposes, Figure~\ref{figure:ccle-heatmaps} is smoothed using a moving weighted average.
% The GBM estimates are visibly less noisy and appear to infer copy neutral regions more accurately.
The GBM estimates are visibly less noisy than the GATK estimates.

% Overall, we find that the GBM outperforms GATK in terms of removing technical biases and denoising,
Overall, the GBM appears to perform very well in terms of removing technical biases and denoising,
particularly when using uncertainty quantification to downweight low precision regions.
The improved performance appears to be due to (a) model-based uncertainty quantification and (b) using a robust probabilistic model for count data.

% # TIMING RESULTS:
% #   GATK 
% #       Creating panel of normals from "training set"
% #           CreateReadCountPanelOfNormals:   200.613734 seconds (14.79 k allocations: 760.185 KiB)
% #           DenoiseReadCounts:              1765.105546 seconds (82.88 k allocations: 3.818 MiB)
% #
% #       Denoising samples in "test set"
% #           DenoiseReadCounts:              1660.271496 seconds (82.90 k allocations: 3.819 MiB)

% #   GBM
% #       Running GBM to estimate factors and denoise samples in "training set":  594.429744 seconds (124.94 M allocations: 229.333 GiB, 1.92% gc time)
% #           (Finished in 17 iterations, with logp = -1.389259814088012e8)
% #
% #       Running GBM to denoise samples in "test set":  255.905412 seconds (53.56 M allocations: 65.118 GiB, 1.30% gc time)
% #           (Finished in 6 iterations with logp = -1.555575450125233e8)

\section{Conclusion}
Generalized bilinear models provide a flexible framework for the analysis of matrix data,
and the delta propagation method enables accurate GBM uncertainty quantification in modern applications.
In future work, it would be interesting to
extend to the more general model of \citet{gabriel1998generalised},
provide theoretical guarantees for delta propagation, and
try applying delta propagation to other models.

\if0\blind
{
\section*{Acknowledgments}
We would like to thank Jonathan Huggins, Will Townes, Mehrtash Babadi, Samuel Lee, David Benjamin, Robert Klein, Samuel Markson, Philipp H\"{a}hnel, and Rafael Irizarry for many helpful conversations.
} \fi

\spacingset{1}
\bibliographystyle{abbrvnatcaplf}
\small
\bibliography{refs}
\normalsize
\spacingset{1.6}

%%%%%%%%%%%%%%%%%%%%%%%%%%%%%%%%%%%%%%%%%%%%%%%%%%%%%%%%%%%%%%%%%%%%%%%%%%%%%%%%
% Supplementary material

\newpage
\setcounter{page}{1}
\setcounter{section}{0}
\setcounter{table}{0}
\setcounter{figure}{0}
\renewcommand{\theHsection}{SIsection.\arabic{section}}
\renewcommand{\theHtable}{SItable.\arabic{table}}
\renewcommand{\theHfigure}{SIfigure.\arabic{figure}}
\renewcommand{\thepage}{S\arabic{page}}  
\renewcommand{\thesection}{S\arabic{section}}   
\renewcommand{\thetable}{S\arabic{table}}   
\renewcommand{\thefigure}{S\arabic{figure}}

\bigskip
\bigskip
\bigskip
\begin{center}
{\LARGE\bf Supplementary material for ``Inference in generalized bilinear models''}
\end{center}
\medskip

\section{Discussion}
\label{section:discussion}

\subsection{Previous work}
\label{section:previous-work}

There is an extensive literature on models involving an unknown low-rank matrix, going by a variety of names including latent factor models, factor analysis models, multiplicative models, bi-additive models, and bilinear models.
In particular, a large number of models can be viewed as special cases of generalized bilinear models (GBMs).
Since a full review is beyond the scope of this article, we settle for covering the main threads in the literature.

% However, the vast majority of work focuses only on estimation, rather than inference and uncertainty quantification.
% We first discuss work on estimation, and then inference.

\subsubsection{Normal bilinear models without covariates.}
Principal components analysis (PCA) is equivalent to maximum likelihood estimation in a GBM with only the $U D V^\T$ term ($K=0$, $L=0$, $M>0$), 
assuming normally distributed outcomes with common variance $\sigma^2$.
PCA (or equivalently, the SVD) is often performed after centering the rows and columns of the data matrix, and from a model-based perspective,
this is equivalent to including intercepts ($K=1$, $L=1$, $M>0$):
\begin{align}
\label{equation:ammi}
Y_{i j} = c + a_i + b_j + \sum_{m=1}^M u_{i m} d_m v_{j m} + \varepsilon_{i j}
\end{align}
where $\varepsilon_{i j}$ is a normal residual. 
Similarly, scaling the rows and columns is analogous to using a rank-one factorization of the variance,
that is, $\varepsilon_{i j} \sim \mathcal{N}(0,\sigma_i^2\sigma_j^2)$.

\textbf{Estimation.} 
Equation~\ref{equation:ammi} is often called the AMMI (additive main effects and multiplicative interaction) model,
and a range of techniques for using it have been developed \citep{gauch1988model}.
The least squares fit of an AMMI model can be obtained by first fitting the linear terms $c + a_i + b_j$ ignoring the non-linear term,
and then estimating the non-linear term $\sum_{m=1}^M u_{i m} d_m v_{j m}$ using PCA on the residuals \citep{gilbert1963non,gollob1968statistical,mandel1969partitioning,gabriel1978least}.
Estimation is more difficult when each entry is allowed to have a different variance, that is, when $\varepsilon_{i j} \sim \mathcal{N}(0,\sigma_{i j}^2)$ with $\sigma_{i j}^2$ known;
this is sometimes called a weighted AMMI model \citep{van1995multiplicative}.
To handle this, \citet{gabriel1979lower} develop the criss-cross method of estimation, 
which successively fits the non-linear terms $m = 1,\ldots,M$, one by one, using weighted least squares.

\textbf{Hypothesis testing for model selection.}
While most applications of PCA only use the estimates, without any uncertainty quantification,
statistical research on the AMMI model has largely focused on hypothesis testing for which factors $m$ to include in the model.
Early contributions on testing in this model were made by 
\citet{fisher1923studies}, \citet{cochran1943comparison}, \citet{tukey1949one}, \citet{williams1952interpretation}, \citet{mandel1961non}, \citet{gollob1968statistical}, and \citet{mandel1969partitioning}.
Methods of this type are very widely used, particularly in the study of genotype-environment interactions in agronomy;
see reviews by \citet{freeman1973statistical}, \citet{gauch2006statistical}, and \citet{gauch2008statistical}.

\textbf{Confidence regions for parameters.}
Uncertainty quantification for the AMMI model parameters has also been studied.
Asymptotic covariance formulas for the least squares estimates have been given by 
\citet{goodman1990analysis}, \citet{chadoeuf1991asymptotic}, \citet{dorkenoo1993etude} and \citet{denis1996asymptotic} 
for the AMMI model and various special cases.
These results are based on inverting the constraint-augmented Fisher information matrix \citep{aitchison1958maximum,silvey1959lagrangian};
we use the same technique in Section~\ref{section:constraint-augmented-fisher}
to estimate standard errors for $U$ and $V$ in our more general GBM model and we extend it using delta propagation.

% \todo{weighted AMMI? van Eeuwijk (1995): ``Denis (1991) described an alternating least squares algorithm for weighted AMMI estimation including offsets.''}

\subsubsection{Normal bilinear models with covariates.}

\textbf{Estimation.}
In a wide-ranging article, \citet{tukey1962future} discussed the possibility of combining regression with factor analysis,
by factoring the matrix of residuals after adjusting for covariates.   % Note to self: This is on page 36 of his tukey1962future.
% In a major advance toward realizing Tukey's vision 
Indeed, for the case of normal outcomes with common variance,
\citet[Cor 3.1]{gabriel1978least} showed that when using a model of the form $\bm{Y} = X A^\T + B Z^\T + U D V^\T + \bm{\varepsilon}$,
the least squares fit can be obtained simply by first fitting $A$ and $B$ using regression (ignoring $U D V^\T$), then fitting $U D V^\T$ to the residuals.
This can be viewed as a generalization of the AMMI estimation procedure.
% In the case of common variance, this vastly simplifies the fitting of bilinear models.
% \citet{gabriel1978least} also shows that the least squares error can be decomposed in various ways, 
% and demonstrates that several previous methods can be viewed as least-squares fits of certain models.
\citet{takane1991principal} extend the results of \citet{gabriel1978least} by first fitting 
$\bm{Y} = X A^\T + B Z^\T + X C Z^\T + \bm{\varepsilon}$ using least squares,
and then using PCA to analyze the residuals as well as each fitted component of the model, that is,
$X \hat{A}^\T$, $\hat{B} Z^\T$, $X \hat{C} Z^\T$, $\hat{\bm{\varepsilon}}$, and combinations thereof.

In a complementary direction, reduced-rank regression \citep{davies1982procedures} and CANDELINC \citep{carroll1980candelinc}
use least squares to fit models of the form 
$\bm{Y} = X A^\T + \bm{\varepsilon}$ and $\bm{Y} = X C Z^\T + \bm{\varepsilon}$, respectively, where $A$ and $C$ are constrained to be low-rank.

\textbf{Hypothesis testing and confidence regions.}
In the case of normal outcomes with common variance $\sigma^2$,
for the model with $\bm{Y} = X A^\T + B Z^\T + U D V^\T + \bm{\varepsilon}$,
\citet{perry2013degrees} show how to perform inference for univariate entries of $A$ and $B$ (and univariate linear projections, more generally)
accounting for the uncertainty in $U D V^\T$ via an estimate of the degrees of freedom associated with the latent factors.
Further, \citet{perry2013degrees} show that the problem can be reduced to the covariate-free case, $\bm{Y} = U D V^\T + \bm{\varepsilon}$,
enabling one to use results on hypothesis testing in the AMMI model \citep{gollob1968statistical,mandel1969partitioning}
which provide estimates of the degrees of freedom.
However, this approach appears to rely on the assumption of normal outcomes with common variance.

\subsubsection{Generalized bilinear models without covariates.}
In many applications, it is unreasonable to use a normal outcome model.
A classical approach is to transform the data and then apply a normal outcome model,
however, as discussed by \citet{van1995multiplicative}, there is unlikely to be a transformation that simultaneously achieves
(a) approximate normality, (b) common variance, and (c) additive effects.

A more principled approach is to extend the generalized linear model (GLM) framework to handle latent factors,
as suggested by \citet{gower1989discussion}.
Goodman's RC models are early contributions in this direction
\citep{goodman1979simple,goodman1981association,goodman1986some,goodman1991measures},
consisting of count models with multinomial or Poisson outcomes where 
$\log(\E(Y_{i j})) = c + a_i + b_j + \sum_{m=1}^M u_{i m} d_m v_{j m}$.
More generally, \citet{van1995multiplicative} develops the generalized AMMI (GAMMI) model,
which is a GLM version of the AMMI model in Equation~\ref{equation:ammi}, specifically,
$g(\E(Y_{i j})) = c + a_i + b_j + \sum_{m=1}^M u_{i m} d_m v_{j m}$.
\citet{van1995multiplicative} introduces a coordinate descent algorithm and discusses approaches for choosing $M$,
however, he does not consider uncertainty quantification for parameters, does not estimate dispersion parameters,
and only demonstrates the method on very small datasets ($11\times 5$ and $17\times 12$).

Correspondence analysis \citep{benzecri1973analyse,greenacre1984theory} is an SVD-based exploratory analysis method for matrices of categorical data,
and has been reinvented under many names, such as reciprocal averaging and dual scaling \citep{de2000gbms}.
Correspondence analysis bears resemblence to estimation methods for the GAMMI model, 
however, it is primarily descriptive in perspective, and thus typically does not involve quantification of uncertainty.

\subsubsection{Generalized bilinear models with covariates.}

\citet{choulakian1996generalized} defines a class of GBMs of the same form as in this article, where
$g(\E(\bm{Y})) = X A^\T + B Z^\T + X C Z^\T + U D V^\T$
and $g$ is the (a) canonical, (b) identical, or (c) logarithmic link function.
For the case of no covariates (that is, the GAMMI model), 
\citet{choulakian1996generalized} proposes an estimation algorithm that involves univariate Fisher scoring updates,
which is attractive for its simplicity, but may exhibit slow convergence or failure to converge due to strong dependencies among parameters.
While the defined model class is general, some limitations of the paper by \citet{choulakian1996generalized} are that 
uncertainty quantification is not addressed, the estimation algorithm is for the special case of GAMMI,
a single common dispersion is assumed and estimation of dispersion is not addressed,
identifiability constraints are not enforced, 
no initialization procedure is provided, 
and only very small datasets are considered ($10\times 7$ and $11\times 11\times 2$). 
% An interesting aspect is that he allows for metric matrices in the choice of identifiability constraints.

\citet{gabriel1998generalised} considers a very general class of models of the form $g(\E(\bm{Y})) = \sum_{k=1}^K X_k \Theta_k Z_k$,
where $X_k$ and $Z_k$ are observed matrices (for instance, covariates) and $\Theta_k$ is a low-rank matrix of parameters for each $k = 1,\ldots,K$.
He extends the criss-cross estimation algorithm of \citet{gabriel1979lower} to this model.
% , using updates that are analogous to the iteratively re-weighted least squares (IRLS) optimization steps for standard GLMs.
While the model of \citet{gabriel1998generalised} is very elegant, some limitations 
are that estimation is performed using a vectorization approach that is computationally prohibitive on large matrices, 
uncertainty quantification is not addressed,
a common dispersion parameter is assumed for all entries, 
and only very small datasets are considered ($10\times 9$ and $17\times 2$).
Also, it is not clear what identifiability constraints are assumed on the $\Theta_k$ matrices.

In recent work, \citet{townes2019generalized} considers a model of the form 
$g(\E(\bm{Y})) = X A^\T + B Z^\T + U D V^\T + \bm{1}\delta^\T$ 
where $\bm{1}$ is a vector of ones and $\delta\in\R^J$ is a vector of fixed column-specific offsets.
\citet{townes2019generalized} derives diagonal approximations to Fisher scoring updates for 
$\ell_2$-penalized maximum likelihood estimation, and 
in a postprocessing stage, enforces orthogonality constraints to aid interpretability.
Other differences compared to the present work are that only estimation is considered (uncertainty quantification is not addressed)
and overdispersion parameters are not estimated.
% , and the method is not demonstrated empirically.

The overview by \citet{de2000gbms} provides an interesting and insightful discussion of several threads in the literature.

\subsubsection{Recent applications of bilinear models.}
Several authors have used bilinear models or GBMs in genetics and genomics, usually to remove unwanted variation such as batch effects.
However, most of these methods do not fully account for uncertainty in the latent factors, 
which may lead to miscalibrated inferences such as overconfident p-values.
For example, to remove batch effects in gene expression analysis, several approaches involve first estimating $U D V^\T$ and then 
% either (a) treating $U D V^\T$ as fixed and known, subtracting it off as a fixed offset \citep{} or (b)
treating $V$ as a known matrix of covariates, accounting for uncertainty only in $U D$ using standard regression
\citep{leek2007capturing,leek2008general,sun2012multiple,risso2014normalization};
this is also done to adjust for population structure in genetic association studies \citep{price2006principal}.
In copy number variation detection, it is common to simply treat the estimated $U D V^\T$ as known and subtract it off 
\citep{fromer2012discovery,krumm2012copy,jiang2015codex}.

\citet{carvalho2008high} use a Bayesian sparse factor analysis model with covariates, employing evolutionary stochastic search for model selection
and Markov chain Monte Carlo (MCMC) for posterior inference within models.
% However, MCMC tends to be slow in large parameter spaces with strong dependencies, as in the case of GBMs.
\citet{stegle2010bayesian}, \citet{buettner2017f}, and \citet{babadi18gatk} use complex hierarchical models that can be viewed as Bayesian GBMs
with additional prior structure, and they employ variational methods for approximate posterior inference.
% While a Bayesian approach is potentially attractive for uncertainty quantification,
% these methods resort to a fully factorized (mean-field) variational approximation,
% which cannot capture the dependencies that exist among GBM parameters and tends to underrepresent posterior uncertainty.
% and thus is unlikely to appropriately represent uncertainty.
Another application in which bilinear models have seen recent use is longitudinal relational data such as networks, 
for which \citet{hoff2015multilinear} employs an interesting Bayesian model with $\bm{Y}_t = A X_t B^\T + \bm{\varepsilon}$,
where $X_t$ is a matrix of observed covariates that depend on time $t$.

% TODO: ON CCLE, TRY USING THE B[:,1] INFERRED FROM THE PON AS A COVARIATE IN X ON THE TEST DATA.
% TODO: ON CCLE, TRY USING THE B[:,1] INFERRED FROM THE PON AS A COVARIATE IN X ON THE TEST DATA.
% TODO: ON CCLE, TRY USING THE B[:,1] INFERRED FROM THE PON AS A COVARIATE IN X ON THE TEST DATA.

\subsection{Challenges and solutions}
\label{section:challenges}

Estimation and inference in large GBMs is complicated by a number of nontrivial challenges.
In this section, we discuss several issues and how we resolve them.
% The details of our estimation and inference algorithms are in Sections~\ref{section:estimation} and \ref{section:inference}.
% Before getting into the details of the algorithm, we give an overview of the main issues and how we resolve them.

% \todo{Someplace, we need to discuss the issue of sample covariates --- that is, how including sample covariates can lead to overfitting.}

\textbf{Estimating the dispersion parameters.} 
There are several issues with estimating the negative binomial (NB) dispersions $1/r_{i j}$.
First, since there is insufficient information to estimate all $I J$ dispersions individually, 
we use the rank-one parametrization $1/r_{i j} = \exp(s_i + t_j + \omega)$.
% Second, to avoid numerical instability due to the positivity constraint on $r_{i j}$, we work with the log-dispersion parameters $s_i$, $t_j$, and $\omega$.
Second, the choice of identifiability constraints matters --- the natural choice of contraints, $\sum_i s_i = 0$ and $\sum_j t_j = 0$, 
leads to noticeably biased estimates of $s_i$ and $t_j$, particularly for higher values; see Figure~\ref{figure:log-dispersion-constraint}.
Instead, we constrain $\frac{1}{I}\sum_i e^{s_i} = 1$ and $\frac{1}{J}\sum_j e^{t_j} = 1$, which effectively mitigates this bias, empirically.
Third, the maximum likelihood estimates sometimes exhibit a severe downward bias, particularly for low values of log-dispersion;
we use a simple heuristic bias correction to deal with this.
Fourth, to avoid arithmetic underflow/overflow in the log-dispersion update steps,
we develop carefully constructed expressions for the gradient and Hessian. % with respect to the log-dispersions.
Finally, to prevent occasional lack of convergence due to oscillating estimates, we employ an adaptive maximum step size.

% Third, maximum likelihood estimates of the log-dispersions are noticably \todo{downward? upward?} biased,
% \todo{is this a general phenomenon for NB models, not just in GBMs? Try estimating $\log(r)$ repeatedly in a simple i.i.d.\ NB model.}
% \todo{particularly for rows $i$ with very low counts?}
% which we mitigate via a carefully chosen prior on $s_i$.
% a $\N(0,1)$ prior on the log-dispersions, which, despite its simplicity, was carefully chosen to mitigate this bias.
% \todo{Use observed Fisher info since expected Fisher info is not closed form.}

\textbf{Inapplicability of standard GLM methods.}
Since $X A^\T + B Z^\T + X C Z^\T$ is linear in the parameters, one could vectorize and write it as 
$\mathrm{vec}(X A^\T + B Z^\T + X C Z^\T) = \tilde{X} \beta$
% Naively, one could write the $X A^\T + B Z^\T + X C Z^\T$ part of the model in the standard form of \tilde{X} \beta$ 
where $\beta = (\mathrm{vec}(A)^\T,\mathrm{vec}(B)^\T,\mathrm{vec}(C)^\T)^\T \in \R^{J K + I L + K L}$
and $\tilde{X}\in\R^{I J \times (J K + I L + K L)}$ is a function of $X$ and $Z$.
In principle, one could then apply standard GLM estimation methods for estimating $\beta$ to construct a joint update to $(A,B,C)$.
However, this vectorization approach is only computationally feasible for small data matrices since computing the matrix inverse
$(\tilde{X}^\T \tilde{X})^{-1}$ takes on the order of $(J K + I L + K L)^3$ time.
Further, this update would need to be done repeatedly since $D$, $U$, $V$, $S$, $T$, and $\omega$ also need to be simultaneously estimated,
and the vectorization approach does not help estimate these parameters.
% nonlinear latent factor term, $U D V^\T$.

% the induced design matrix $\tilde X$ is 
% making it necessary to invert a \todo{insert dimensions} matrix $\tilde{X}^\T \tilde{X}$ just to estimate $A$, $B$, and $C$ alone.
% just to estimate $A$, $B$, and $C$ alone.

\textbf{Inapplicability of the singular value decomposition.} 
At first glance, it might appear that the singular value decomposition (SVD) would make it straightforward to estimate $U$, $D$, and $V$ 
given the other parameters.
However, when used for estimation, the SVD implicitly assumes that every entry has the same variance.
% Note to self: I think it would be every entry of the residual matrix $\varepsilon = (\varepsilon_{i j})$.
This is far from true in GBMs, and consequently, naively using the SVD to update $U D V^\T$ leads to poor estimation accuracy.
The criss-cross algorithm of \citet{gabriel1979lower} yields a low-rank matrix factorization that accounts for entry-specific variances, 
and our algorithm provides another way of doing this while adjusting for covariates in a GBM.
In our algorithm, we only directly use the SVD for enforcing the identifiability constraints, not for estimation of $U D V^\T$.

\textbf{Computational efficiency.} 
The genomics applications in Sections~\ref{section:gene-expression} and \ref{section:cancer}
involve large count matrices $\bm{Y}\in\R^{I\times J}$ where the number of features $I$ is on the order of $10^4$ to $10^6$
and the number of samples $J$ can be as large as $10^4$ or more.
Consequently, computational efficiency is essential for practical usage of the method.
For estimation, we exploit the special structure of the GBM to derive computationally efficient Fisher scoring updates to each component of the model.
% As a second-order optimization method, our estimation algorithm converges rapidly.
For inference, we develop a novel method for efficiently propagating uncertainty between components of the model.
% approximating the diagonal of the inverse of the constraint-augmented Fisher information.
Assuming $\max\{K^2,L^2,M\} \leq J \leq I$, 
our estimation algorithm takes $O(I J \max\{K^2,L^2,M^2\})$ time per iteration, and our inference algorithm requires $O(I J \max\{K^3, L^3, J M^3\})$ time,
making them computationally feasible on large data matrices.

\textbf{Numerical stability.} Using a good choice of initialization is crucial for numerical stability. % when using second-order optimization methods.
To initialize the estimation algorithm, we analytically solve for values of $A$, $B$, and $C$ to approximate the data matrix and then, for NB-GBMs, we iteratively update $S$ and $T$ for a few iterations.
Even with a good initialization, optimization methods occasionally diverge.
In a large GBM, there are so many parameters that even occasional divergences cause the algorithm to fail with high probability.
We reduce the frequency of divergences to be negligible by enforcing a bound on the norm of the optimization steps;
see Section~\ref{section:estimation}.

\textbf{Enforcing identifiability constraints.} Rather than performing constrained optimization steps, 
we use a combination of unconstrained optimization steps and likelihood-preserving projections onto the constrained parameter space.
Although the construction of likelihood-preserving projections in a GBM is not obvious,
we show that they can be efficiently computed using simple linear algebra operations. 
This optimization-projection approach has a number of advantages; see Section~\ref{section:enforcing-constraints} for further discussion.

\textbf{Dependencies in latent factors.} 
Optimizing the latent factor term $U D V^\T$ is challenging 
due to the dependencies among $U$, $D$, and $V$ as well as the orthonormality contraints $U^\T U = \Id$ and $V^\T V = \Id$.
Consequently, updating $U$, $D$, and $V$ individually does not seem to work well.
% \todo{be more specific --- is it just slow or does it fail altogether?}.
To resolve this issue, we relax the dependencies and constraints by defining $G := U D$ and $H := V D$, and updating $G$, $H$, and $D$ separately.
% performing likelihood-preserving projections onto the contrained space after each update.
% \todo{Try updating both $G$ and $H$ before projecting?}

\textbf{Prior / regularization.} To improve estimation accuracy in a high-dimensional setting,
we place independent normal priors on the entries of $A$, $B$, $C$, $D$, $U$, $V$, $S$, and $T$,
and use \textit{maximum a posteriori} (MAP) estimation, which is equivalent to $\ell_2$-penalization/shrinkage for this choice of prior.
An additional benefit of using priors is that it improves the 
numerical stability of the estimation algorithm.
See Section~\ref{section:priors} for prior details.
% We provide default priors that we have found to work well in a wide range of settings.

\subsection{Enforcing the GBM identifiability constraints}
\label{section:enforcing-constraints}

It might seem preferable to perform unconstrained optimization throughout the estimation algorithm until convergence, and then enforce the identifiability constraints as a postprocessing step. However, in general, this would not converge to a local optimum in the constrained space because the prior does not have the same invariance properties as the likelihood.  Thus, we maintain the constraints throughout the algorithm by applying a projection at each step.

When updating each component ($A$, $B$, $C$, $D$, $U$, $V$, $S$, and $T$), 
rather than using a constrained optimization step such as equality-constrained Newton's method \citep{boyd2004convex}, we use an unconstrained optimization step followed by a likelihood-preserving projection onto the constrained space.  
% \todo{is there a standard name for this technique?}.

It is crucial to preserve the likelihood when projecting onto the constrained space, since otherwise the projection might undo all the gains obtained by the unconstrained optimization step --- in short, otherwise we might end up ``taking one step forward and two steps back.''
To this end, we employ likelihood-preserving projections for each component of the GBM.
By Theorem~\ref{theorem:projections}, the likelihood is invariant under these operations
and the projected values satisfy the identifiability constraints.
The optimization-projection approach has several major advantages.
\begin{enumerate}
\item In the likelihood surface, there can be strong dependencies among the parameters within each row of $A$, $B$, $U$, and $V$, 
whereas the between-row dependencies are much weaker (specifically, they have zero Fisher cross-information). Thus, it is desirable to optimize each row jointly, 
however, this is complicated by the fact that the constraints create dependencies between rows.
Consequently, using equality-constrained Newton appears to be computationally infeasible since it would require a joint update of each parameter matrix in entirety.
\item Since each optimization-projection step modifies multiple components of the GBM, it effectively performs a joint update on multiple components.
For instance, the likelihood-preserving projection for $A$ also modifies $C$, so the optimization-projection step on $A$ is effectively a joint update to $A$ and $C$.
This has the effect of enlarging the constrained space within which each update takes place, improving convergence.
\item For $U$ and $V$, the constrained space is particular difficult to optimize over since it involves 
not only within-column linear dependencies ($X^\T U = 0$ and $Z^\T V = 0$), but also quadratic dependencies within and between columns
($U^\T U = \Id$ and $V^\T V = \Id$). The optimization-projection approach makes it easy to handle these constraints.
\item It is straightforward to perform unconstrained optimization for each component separately,
and the projections that we derive turn out to be very easy to apply.
\end{enumerate}

% we apply an unconstrained optimization step and then perform a likelihood-preserving projection onto the constrained parameter space.
% For the updates involving the $U D V^\T$ term, the construction of likelihood-preserving projections is not obvious.
% It turns out that including the $X C Z^\T$ term is particularly helpful for enforcing the constraints on $A$, $B$, $U$, and $V$.

% To effectively perform constrained optimization in a computationally attractive way,
% at each step of the algorithm \todo{reference algorithm(s)} we use an unconstrained optimization step and then apply a likelihood-preserving projection
% onto the constrained space.

\section{Additional simulation results and details}
\label{section:simulation-details}

We present additional simulation results and details supplementing Section~\ref{section:simulations}.

\subsubsection*{Simulating covariates, parameters, and data}
\label{section:simulating-data}

Here, we provide the details of how the simulation data are generated.
First, the covariates are generated using a copula model as follows.
We describe the procedure for the feature covariate matrix $X\in\R^{I\times K}$; 
the sample covariate matrix $Z\in\R^{J\times L}$ is generated in the same way but with $J$ and $L$ in place of $I$ and $K$.
We generate a random covariance matrix $\Sigma = Q^\T Q$ where the entries of 
$Q\in\R^{K\times K}$ are $q_{k k'}\sim \N(0,1)$ i.i.d.,
and then we compute the resulting correlation matrix $\tilde\Sigma\in\R^{K\times K}$ 
by setting $\tilde\Sigma_{k k'} = \Sigma_{k k'} / \sqrt{\Sigma_{k k} \Sigma_{k' k'}}$.
We generate $(\tilde x_{i 1},\ldots,\tilde x_{i K})^\T \sim \N(0,\tilde\Sigma)$ i.i.d.\ for $i=1,\ldots,I$,
and define $X\in\R^{I\times K}$ by setting $x_{i k} = h(F^{-1}(\Phi(\tilde x_{i k})))$
where $h(x) = \mathrm{sign}(x) \min\{100,|x|\}$,
$\Phi(x)$ is the standard normal CDF, and $F^{-1}$ is the generalized inverse CDF for the desired marginal distribution,
which we take to be
$\N(0,1)$ for the \texttt{Normal} scheme,
$\Ga(2,\sqrt{2})$ for the \texttt{Gamma} scheme,
% $\Ga(2,2)$ for the \texttt{Gamma} scheme,
and $\Bernoulli(1/2)$ for the \texttt{Binary} scheme.
Finally, we standardize $X$ by setting $x_{i 1} = 1$ for all $i$
and centering/scaling so that $\sum_{i=1}^I x_{i k} = 0$ and $\frac{1}{I}\sum_{i=1}^I x_{i k}^2 = 1$ for $k\geq 2$.

The true parameters $A_0$, $B_0$, $C_0$, $D_0$, $U_0$, $V_0$, $S_0$, $T_0$, and $\omega_0$ are then generated as follows.
First, we generate matrices
$\tilde{A}$, $\tilde{B}$, and $\tilde{C}$ with i.i.d.\ entries as follows:
(\texttt{Normal} scheme)
$\tilde{a}_{j k} \sim \N(0, 1/(4 K))$,
$\tilde{b}_{i\ell} \sim \N(0, 1/(4 L))$, and
$\tilde{c}_{k\ell} \sim \N(0, 1/(K L)) + 3\,\I(k=1,\ell=1)$, or
(\texttt{Gamma} scheme)
$\tilde{a}_{j k} \sim \Ga(2,\, 2 \sqrt{2 K}))$,
$\tilde{b}_{i\ell} \sim \Ga(2,\, 2 \sqrt{2 L})$, and
$\tilde{c}_{k\ell} \sim \Ga(2,\, \sqrt{2 K L}) + 3\,\I(k=1,\ell=1)$.
These distributions are defined so that the scale of the entries of $X\tilde{A}^\T$, $\tilde{B}Z^\T$, and $X\tilde{C}Z^\T$ is not affected by $K$ and $L$.

Then we set
$A_0 = \tilde{A} - Z (Z^\ps \tilde{A})$,
$B_0 = \tilde{B} - X (X^\ps \tilde{B})$, and 
$C_0 = \tilde{C}$.
Next, we set 
$U_0 = \tilde{U} - X (X^\ps \tilde{U})$ and
$V_0 = \tilde{V} - Z (Z^\ps \tilde{V})$
where $\tilde{U}\in\R^{I\times M}$ and $\tilde{V}\in\R^{J\times M}$ are sampled uniformly from their respective Stiefel manifolds,
that is, uniformly subject to $\tilde{U}^\T \tilde{U} = \Id$ and $\tilde{V}^\T \tilde{V} = \Id$.
The diagonal entries of $D_0$ are evenly spaced from $\sqrt{I} + \sqrt{J}$ to $2(\sqrt{I} + \sqrt{J})$; 
this scaling is motivated by the Marchenko--Pastur law for the distribution of singular values of $I\times J$ random matrices \citep{marchenko1967distribution}.
For the log-dispersion parameters ($S_0$, $T_0$, and $\omega_0$), 
we generate $\tilde{s}_i,\tilde{t}_j\sim\N(0,1)$ i.i.d., and set $\omega_0 = -2.3$,
$s_{0 i} = \tilde{s}_i - \log(\frac{1}{I}\sum_{i=1}^I \exp(\tilde{s}_i))$, and
$t_{0 j} = \tilde{t}_j - \log(\frac{1}{J}\sum_{j=1}^J \exp(\tilde{t}_j))$,
% \todo{Maybe we should briefly explain the rationale of some of these choices?}

Given the true parameters and covariates, the data matrix $\bm{Y}\in\{0,1,2,\ldots\}^{I\times J}$ is generated as follows.
We compute the mean matrix $\mu_0 := g^{-1}(X A_0^\T + B_0 Z^\T + X C_0 Z^\T + U_0 D_0 V_0^\T)$, 
where the inverse link function $g^{-1}(x) = e^{x}$ is applied element-wise,
and we compute the inverse dispersions $r_{0 i j} := \exp(-s_{0 i} - t_{0 j} - \omega_0)$.
Then we sample $Y_{i j} \sim \mathcal{D}(\mu_{0 i j},r_{0 i j})$ where
$\mathcal{D}(\mu,r) = \NegBin(\mu,r)$ in the \texttt{NB} scheme, 
$\mathcal{D}(\mu,r) = \LNP(\mu,\,\log(1/r+1))$ in the \texttt{LNP} scheme,
$\mathcal{D}(\mu,r) = \Poisson(\mu)$ in the \texttt{Poisson} scheme, or
$\mathcal{D}(\mu,r) = \Geometric(1/(\mu+1))$ in the \texttt{Geometric} scheme,
so that $\E(Y_{i j}) = \mu_{0 i j}$ in each case.
% Here, $\Geometric(y\mid p) = (1-p)^y p$ for $y\in\{0,1,2,\ldots\}$.
Here, for $y\in\{0,1,2,\ldots\}$, $\Geometric(p)$ has p.m.f.\ $f(y) = (1-p)^y p$, whereas $\LNP(\mu,\sigma^2)$ has p.m.f.
\begin{align}
\label{equation:lnp}
f(y) = \int \Poisson(y | \lambda) \LogNormal(\lambda \mid \log(\mu) - \tfrac{1}{2}\sigma^2, \; \sigma^2)\, d\lambda
\end{align}
for $\mu>0$ and $\sigma^2>0$.  These outcome distributions are defined so that in each case, if $Y\sim\mathcal{D}(\mu,r)$ then $\E(Y) = \mu$.
Further, in the \texttt{LNP} case, $\mathrm{Var}(Y) = \mu + \mu^2 / r$, so the interpretation of $r$ is the same as in the \texttt{NB} case.

\subsubsection*{Consistency and statistical efficiency -- Details on Section~\ref{section:consistency}}
\label{section:consistency-details}

In these simulations, to accurately measure the trend with increasing $I$, %for each of the $50$ runs, 
we generate the covariates, true parameters, and data with $I = 10000$
and project them onto the lower-dimensional spaces for smaller $I$ values; 
for $\bm{Y}$, $X$, and $S_0$ this projection simply consists of taking the first $I$ rows/entries, 
$Z$ and $T_0$ are unaffected by the projection, 
and $A_0$, $B_0$, $C_0$, $D_0$, $U_0$, and $V_0$ are projected by matching the first $I$ rows of the mean matrix $\mu_0$.

We use the relative MSE rather than the MSE to facilitate interpretability, since this puts the errors on a common scale
that does not depend on the magnitude of the parameters.
For instance, the relative MSE for $A$ is defined as 
$$ \mathrm{MSE}_\mathrm{rel}(A,A_0) = \frac{\sum_{j=1}^J \sum_{k=1}^K |a_{j k} - a_{0 j k}|^2}{\sum_{j=1}^J \sum_{k=1}^K |a_{0 j k}|^2} $$
where $A$ is the estimate and $A_0$ is the true value.

Figure~\ref{figure:consistency-all} shows the relative MSE plots for all GBM parameter components.
For $A$, $C$, $V$, and $T$, the relative MSE appears to be decreasing to zero.
The trend for $D$ and $\omega$ is suggestive but not as clear,
making it difficult to gauge whether $D$ and $\omega$ are likely to be consistent based on these experiments.
The relative MSEs for $B$, $U$, and $S$ are small but do not appear to be going to zero as $I\to\infty$ with $J$ fixed;
this is expected since for these parameters the amount of data informing each univariate entry is fixed.

% For $C$ and $T$, the rate of convergence is less clear from these experiments. %; however, these parameters are usually of less interest than $A$ and $V$.
% Meanwhile, for $B$, $U$, and $S$, the relative MSE hovers around a small nonzero value, but does not appear to be trending to zero, as expected.

\begin{figure}
  \centering
  \includegraphics[trim=0.6cm 0 1.5cm 0, clip, width=0.325\textwidth]{figures/sim/results-J=100-K=4-L=2-M=3-NegativeBinomial-Normal-Normal/err_A.png}
  \includegraphics[trim=0.6cm 0 1.5cm 0, clip, width=0.325\textwidth]{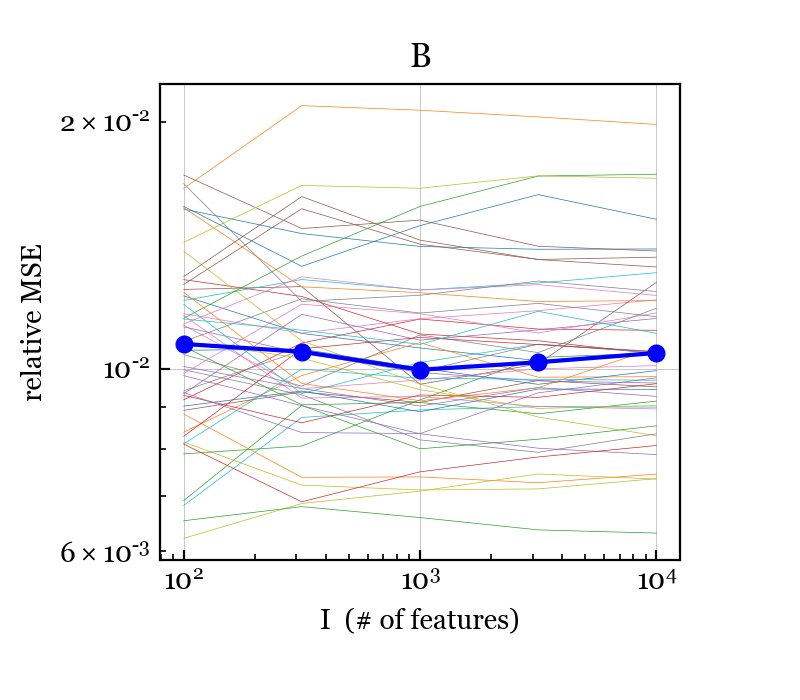}
  \includegraphics[trim=0.6cm 0 1.5cm 0, clip, width=0.325\textwidth]{figures/sim/results-J=100-K=4-L=2-M=3-NegativeBinomial-Normal-Normal/err_C.png}\\
  \includegraphics[trim=0.6cm 0 1.5cm 0, clip, width=0.325\textwidth]{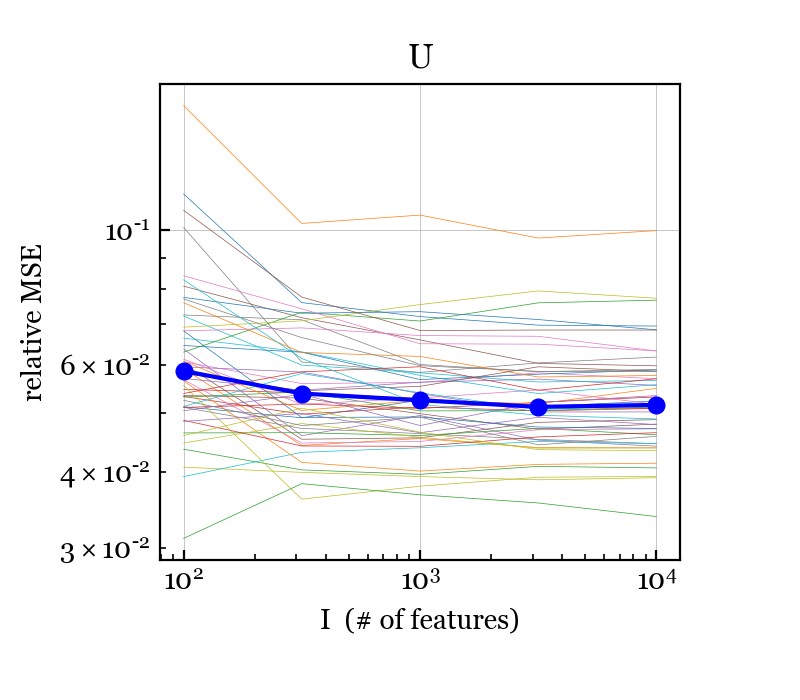}
  \includegraphics[trim=0.6cm 0 1.5cm 0, clip, width=0.325\textwidth]{figures/sim/results-J=100-K=4-L=2-M=3-NegativeBinomial-Normal-Normal/err_V.png}
  \includegraphics[trim=0.6cm 0 1.5cm 0, clip, width=0.325\textwidth]{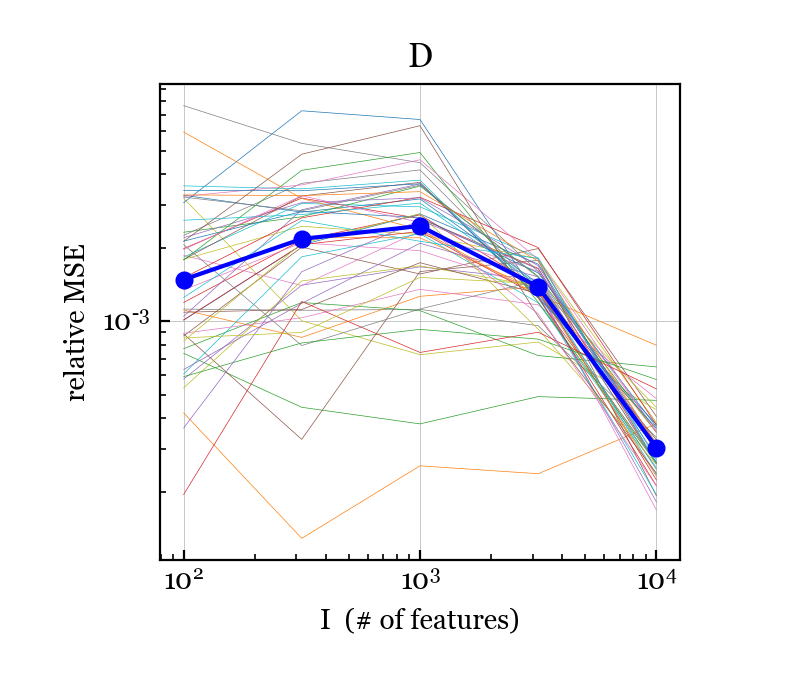}\\
  \includegraphics[trim=0.6cm 0 1.5cm 0, clip, width=0.325\textwidth]{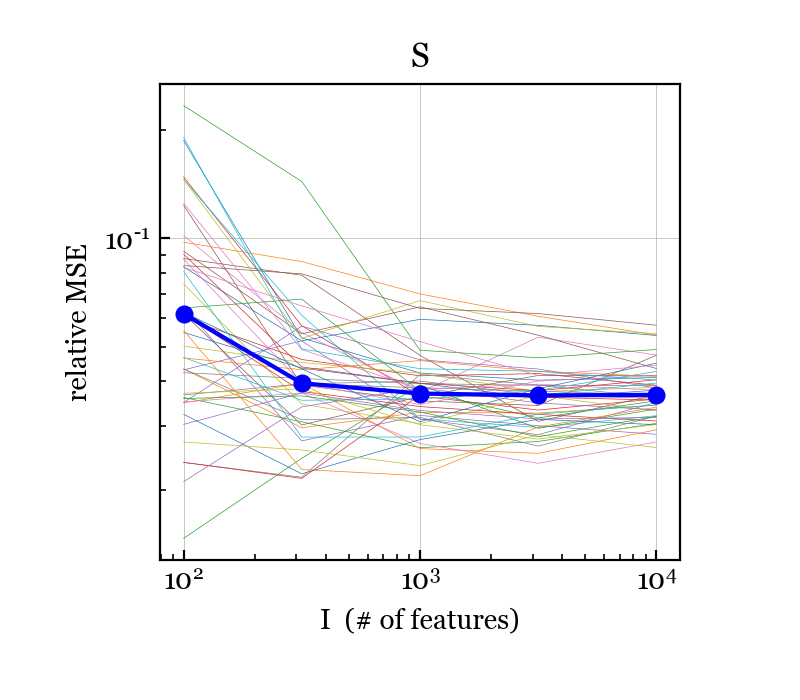}
  \includegraphics[trim=0.6cm 0 1.5cm 0, clip, width=0.325\textwidth]{figures/sim/results-J=100-K=4-L=2-M=3-NegativeBinomial-Normal-Normal/err_T.png}
  \includegraphics[trim=0.6cm 0 1.5cm 0, clip, width=0.325\textwidth]{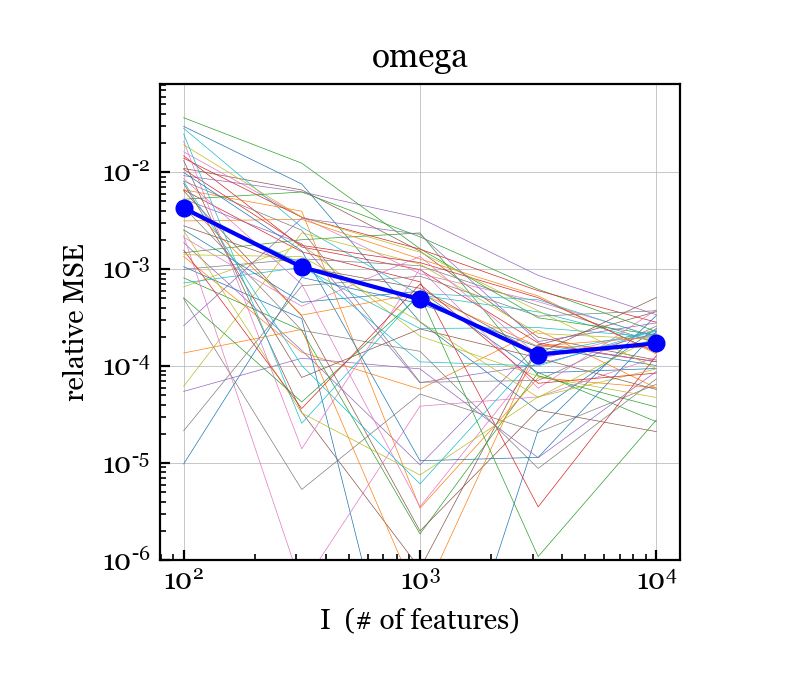}
  \caption{Relative mean-squared error between estimated and true parameter values. Each plot contains $50$ thin lines, one for each run, along with the median over the $50$ runs (thick blue line).}
  \label{figure:consistency-all}
\end{figure}

\iffalse
Since the number of parameters grows with the data matrix size, it is not obvious how GBM estimates will behave in terms of consistency and efficiency.

the number of rows $I$ is large and the number of columns $J$ tends to be smaller, so 

In this section, we empirically assess whether the parameter estimates are converging to the true parameters as the data matrix size grows (consistency),
and how quickly they are converging as the data grows (statistical efficiency).
Since the model is symmetric with respect to rows and columns, identical results would be obtained with increasing $J$ and fixed $I$.

It is less clear whether to expect consistency in $D$ and $\omega$ as $I\to\infty$ since, 
even though the amount of data informing them is growing, they closely depend on the estimates of $U$ and $S$.

we generate data matrices with increasing $I$, holding $J$ fixed.
\fi

\subsubsection*{Accuracy of standard errors -- Details on Section~\ref{section:coverage}}
\label{section:coverage-details}

To estimate the actual coverage at every target coverage level from 0\% to 100\%, we use the fact that the actual coverage of a 
$100\times (1-\alpha) \%$ interval for some parameter $\theta$ can be written as
$$ \Pr(|\hat{\theta} - \theta_0| < z_{\alpha/2}\,\hat{\mathrm{se}}) = \Pr\Big(1 - 2\big(1 - \Phi(|\hat{\theta} - \theta_0|/\hat{\mathrm{se}})\big) < 1-\alpha\Big)$$
where $z_{\alpha/2} = \Phi^{-1}(1-\alpha/2)$ and $\Phi(x)$ is the $\N(0,1)$ CDF.
Thus, since $1-\alpha$ is the target coverage, the curve of actual coverage versus target coverage 
is simply the CDF of the random variable $1 - 2\big(1 - \Phi(|\hat{\theta} - \theta_0|/\hat{\mathrm{se}})\big)$.
The plots in Figure~\ref{figure:coverage} are empirical CDFs of this random variable, aggregating across all entries of each parameter matrix/vector.

\begin{figure}
  \centering
  \includegraphics[trim=0.3cm  0 5.5cm 0, clip, height=1.6in]{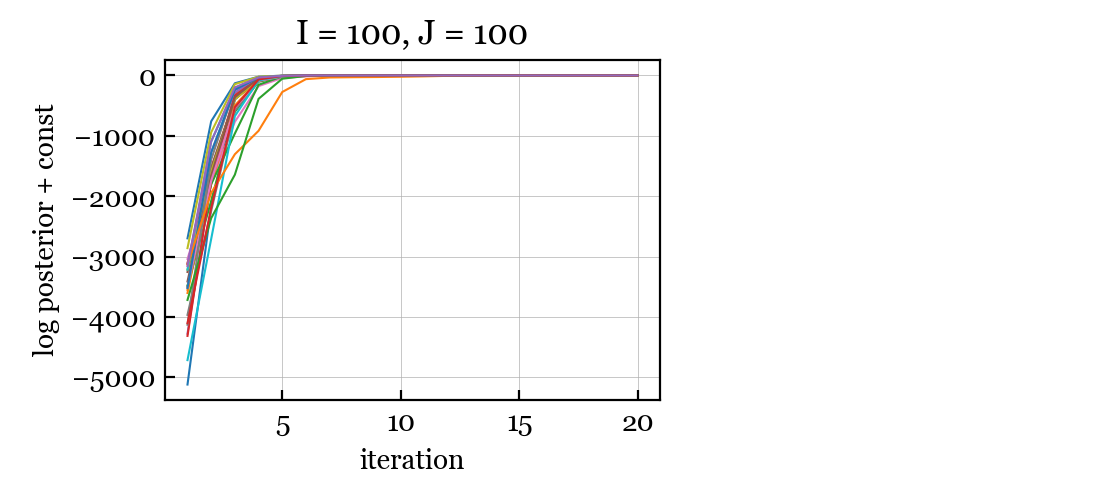}
  \includegraphics[trim=0.75cm 0 5.5cm 0, clip, height=1.6in]{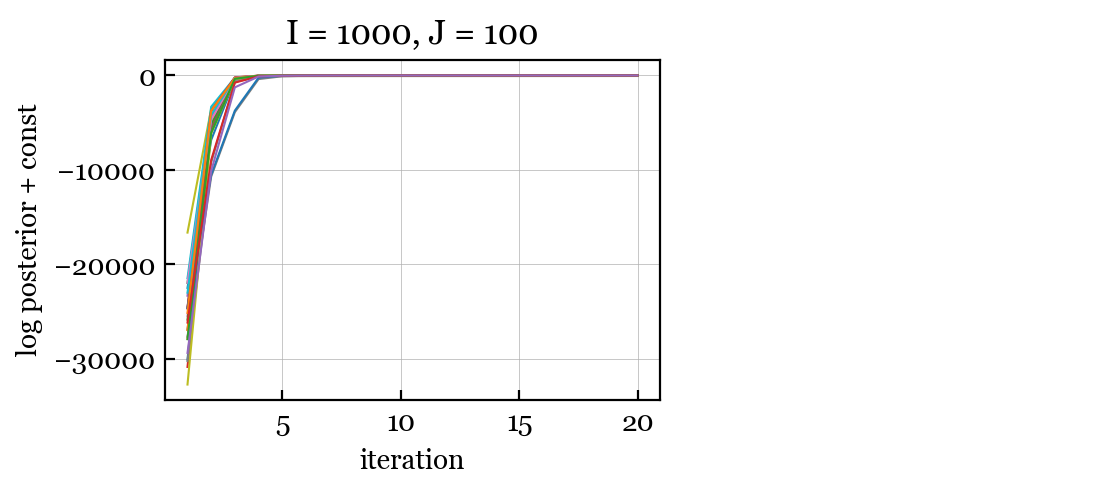}
  \includegraphics[trim=0.5cm  0 5.5cm 0, clip, height=1.6in]{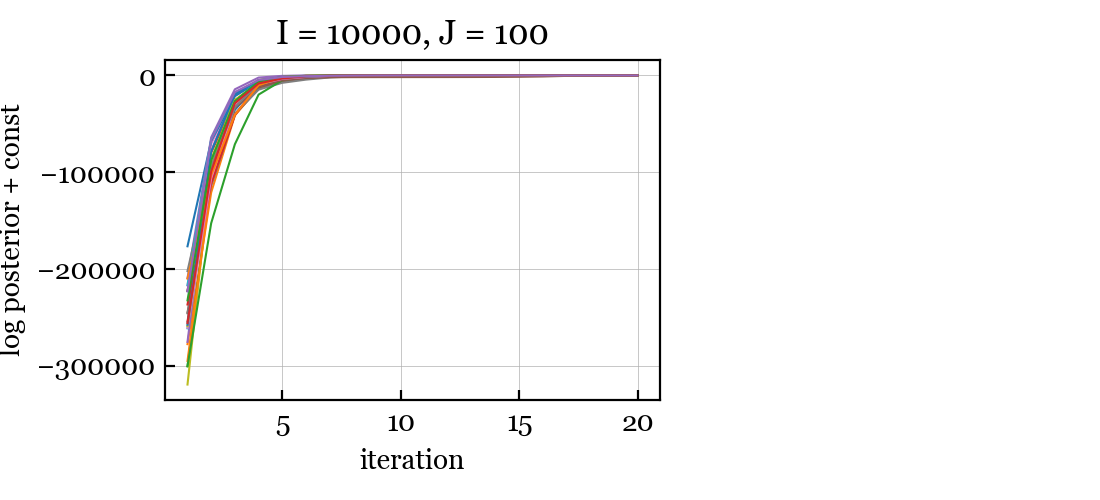}
  \caption{Log posterior density (plus constant) versus iteration for the estimation algorithm.}
  \label{figure:algorithm-convergence}
\end{figure}

\subsubsection*{Algorithm convergence -- Details on Section~\ref{section:convergence}}
Figure~\ref{figure:algorithm-convergence} shows the results of the algorithm convergence experiments described in Section~\ref{section:convergence}.
For each $I$, the plot shows 25 curves of the log-likelihood+log-prior (plus a constant) versus iteration number,
one for each of the 25 simulation runs.
% Since the log posterior density is only known up to an additive constant,
For visual interpretability, we add a constant to each curve such that the final value after iteration $50$ is equal to zero.
Based on our experiments, the estimation algorithm converges rapidly.

\subsubsection*{Robustness to the outcome distribution -- Details on Section~\ref{section:robustness}}
Figures~\ref{figure:LogNormalPoisson-outcomes}, \ref{figure:Poisson-outcomes}, and \ref{figure:Geometric-outcomes}
show the results of the robustness experiments in Section~\ref{section:robustness}. % for $A$, $B$, $U$, and $V$.
% We do not assess $S$, $T$, and $\omega$ since the Poisson and Geometric outcomes have no dispersion parameters.
With LNP outcomes (Figure~\ref{figure:LogNormalPoisson-outcomes}), the results are very similar to when the true outcome distribution is actually NB (Figures~\ref{figure:consistency-all} and \ref{figure:coverage}).
With Poisson outcomes (Figure~\ref{figure:Poisson-outcomes}), the estimation accuracy is even better than with NB outcomes, presumably 
because there is less variability and the Poisson distribution is a limiting case of NB when the dispersion goes to zero.
The coverage with Poisson outcomes is essentially the same as NB, slightly better for some parameters and slightly worse for others.
With Geometric outcomes (Figure~\ref{figure:Geometric-outcomes}), the same parameters appear to be consistently estimated, however, the estimation accuracy in terms of relative MSE 
is worse by roughly a factor of 10.
Compared to NB outcomes, the coverage with Geometric outcomes is similar for some parameters and slightly worse for others.
Overall, it appears that the estimation and inference algorithms are quite robust to misspecification of the outcome distribution.

\begin{figure}
  \centering
  \includegraphics[trim=0.6cm 0 1.5cm 0, clip, height=0.2\textheight]{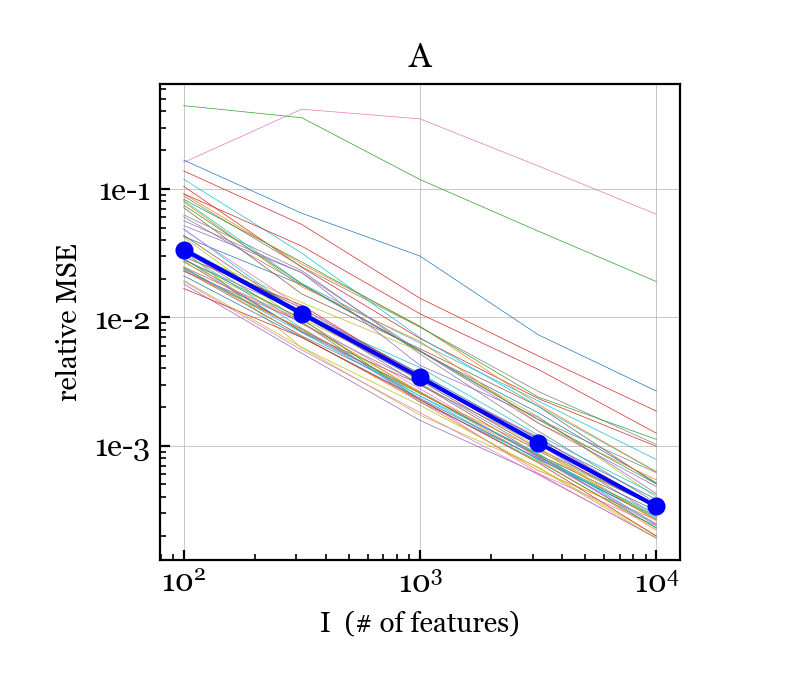}
  \includegraphics[trim=1.1cm 0 1.5cm 0, clip, height=0.2\textheight]{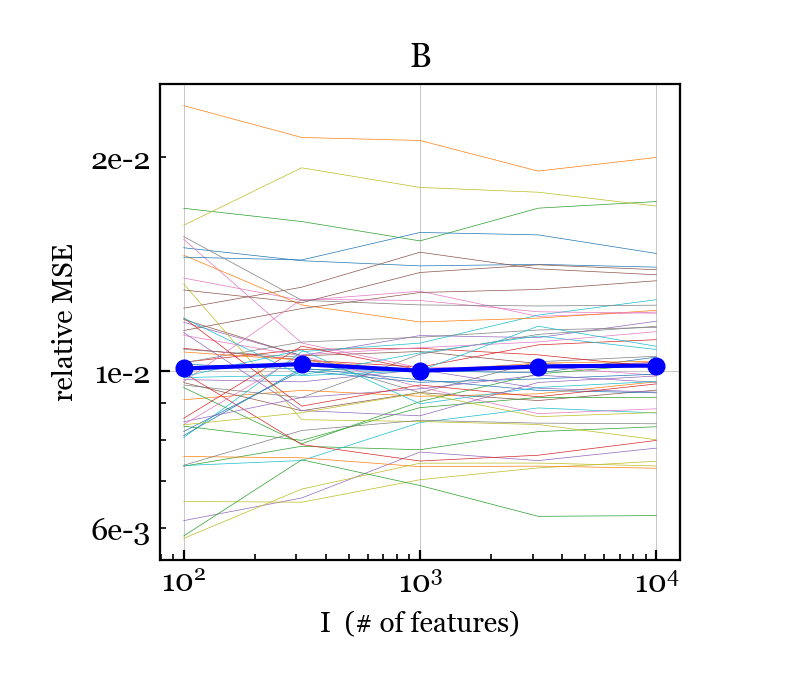}
  \includegraphics[trim=1.1cm 0 1.5cm 0, clip, height=0.2\textheight]{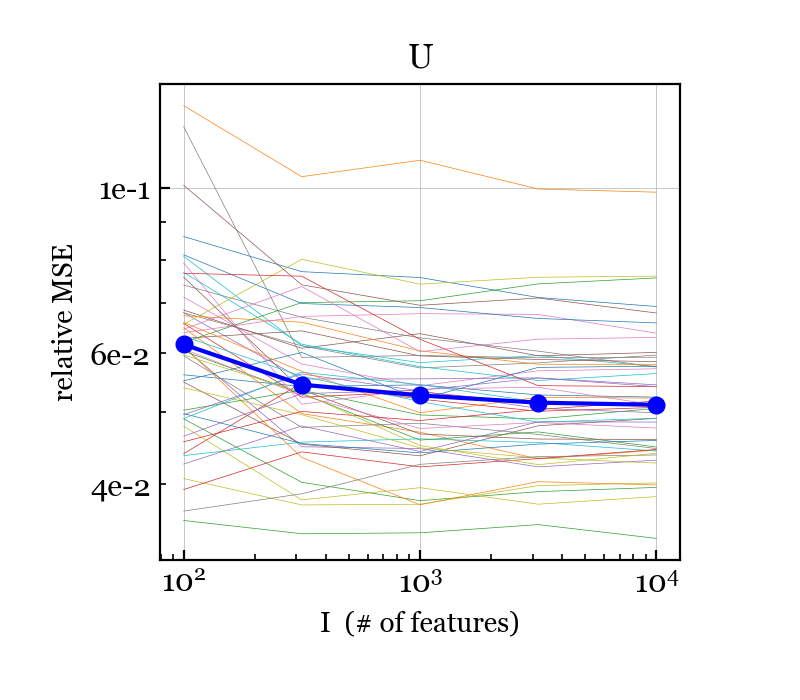}
  \includegraphics[trim=1.1cm 0 1.3cm 0, clip, height=0.2\textheight]{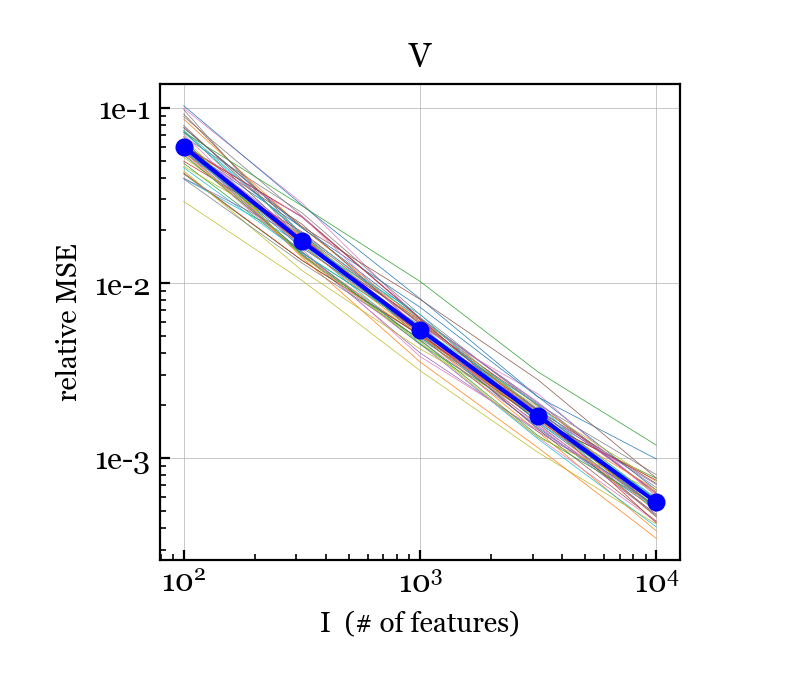}\\
  \includegraphics[trim=0.6cm 0 1.3cm 0, clip, height=0.2\textheight]{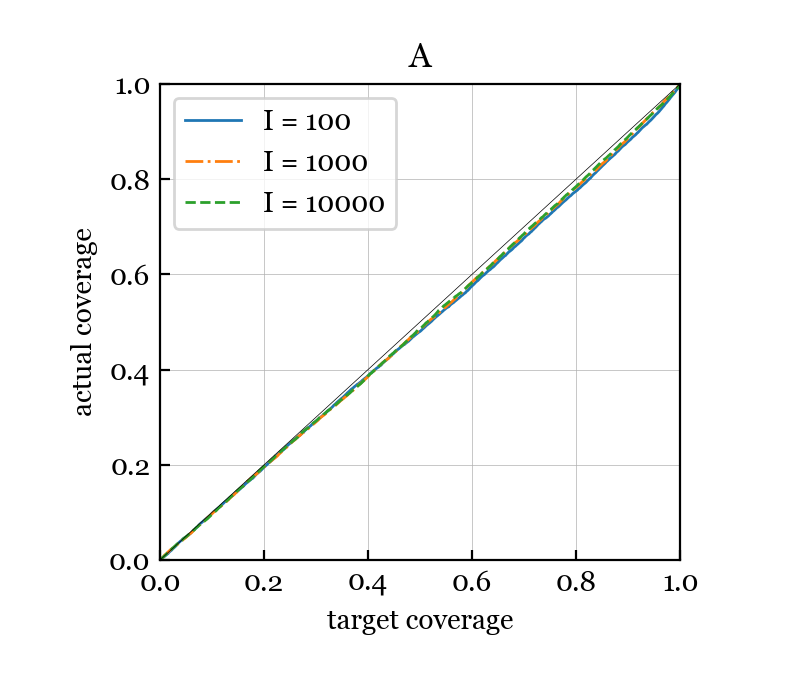}
  \includegraphics[trim=1.3cm 0 1.3cm 0, clip, height=0.2\textheight]{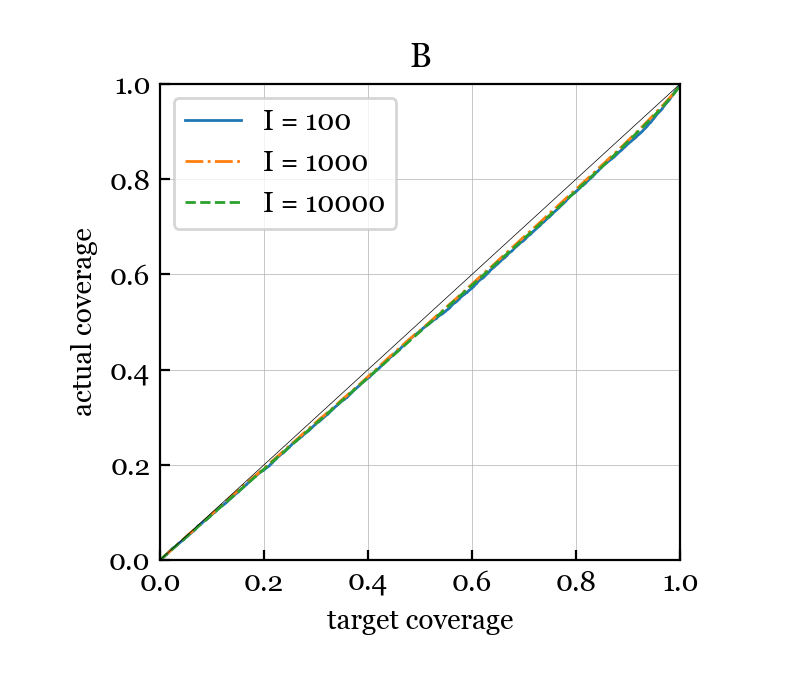}
  \includegraphics[trim=1.3cm 0 1.3cm 0, clip, height=0.2\textheight]{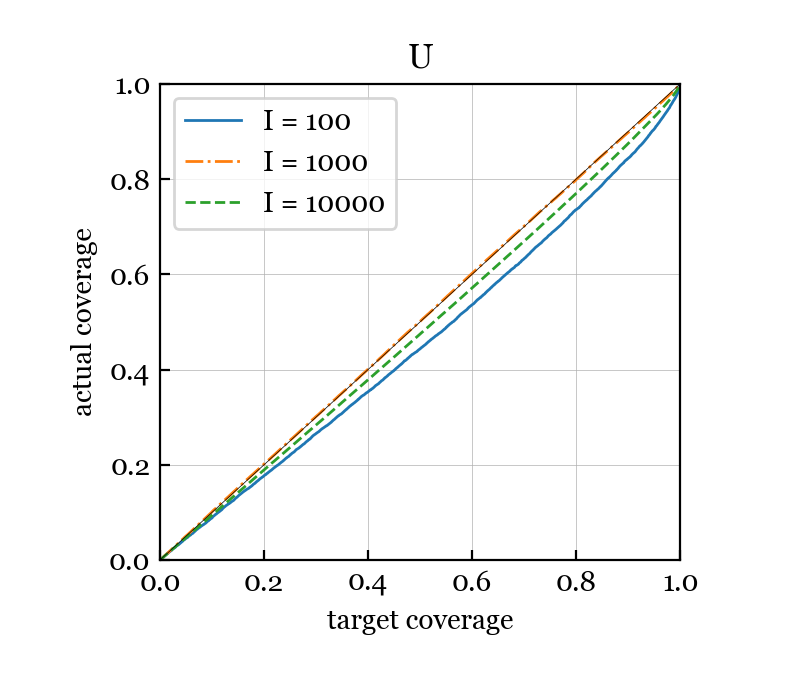}
  \includegraphics[trim=1.3cm 0 1.3cm 0, clip, height=0.2\textheight]{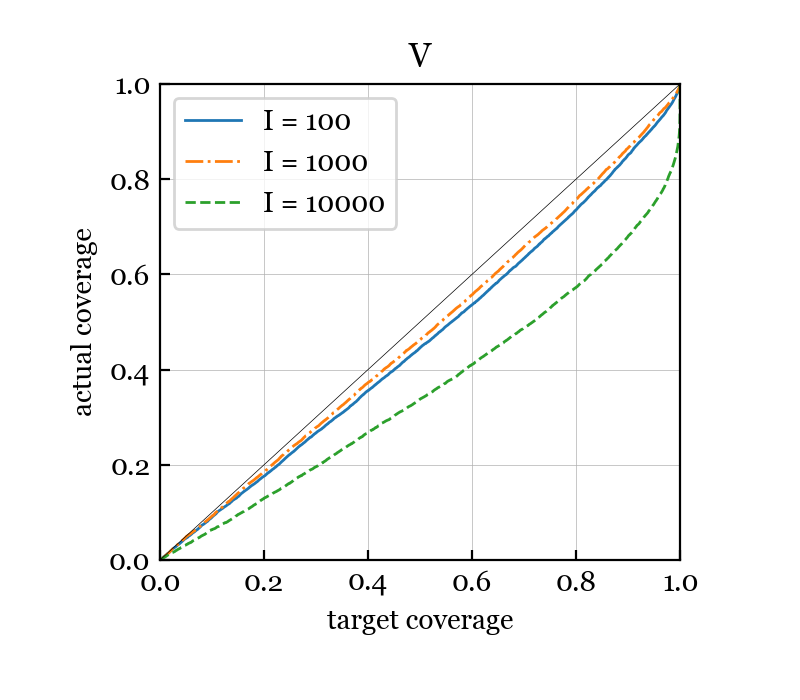}
  \caption{Robustness to outcome: $K = 4$, $L = 2$, $M = 3$, \tt{LNP/Normal/Normal}.}
  \label{figure:LogNormalPoisson-outcomes}
\end{figure}

\begin{figure}
  \centering
  \includegraphics[trim=0.6cm 0 1.5cm 0, clip, height=0.2\textheight]{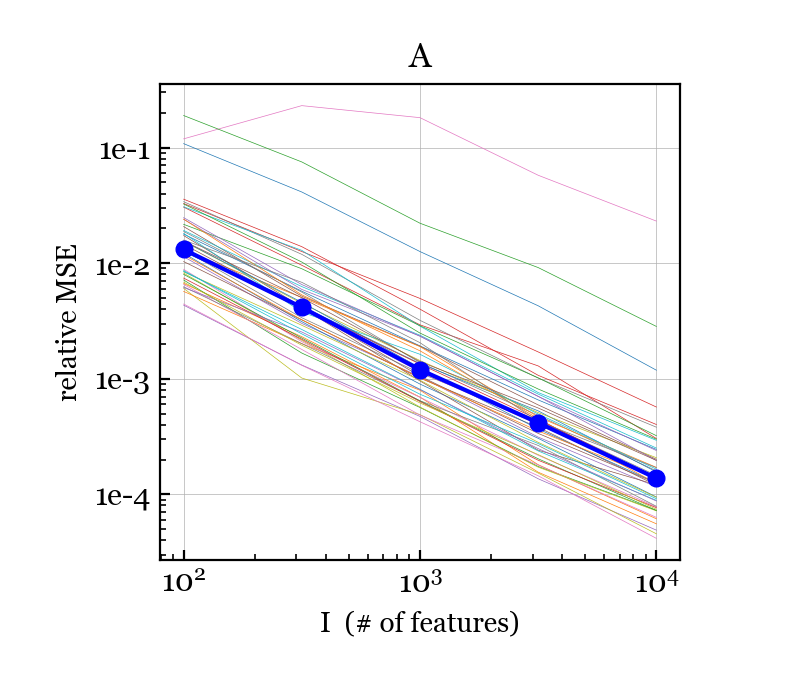}
  \includegraphics[trim=1.1cm 0 1.5cm 0, clip, height=0.2\textheight]{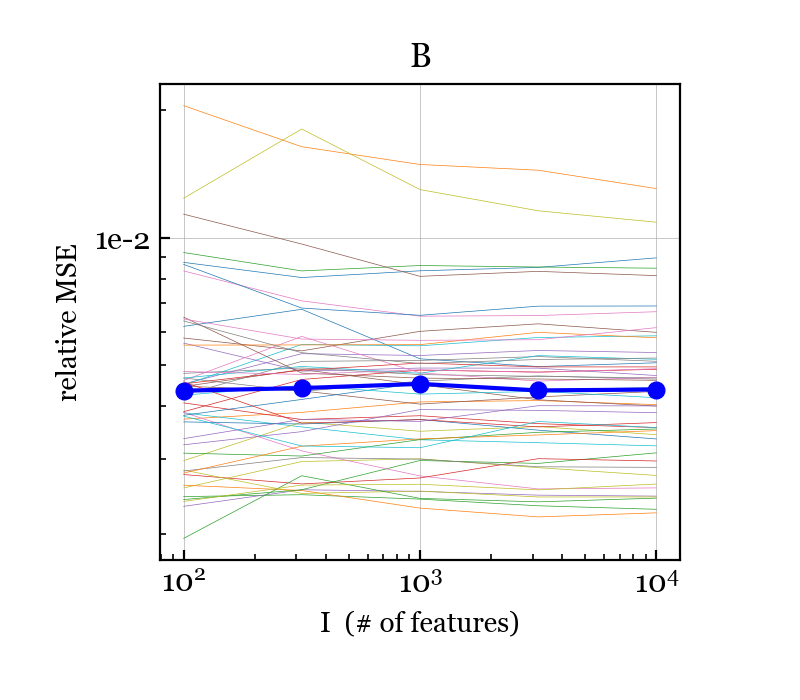}
  \includegraphics[trim=1.1cm 0 1.5cm 0, clip, height=0.2\textheight]{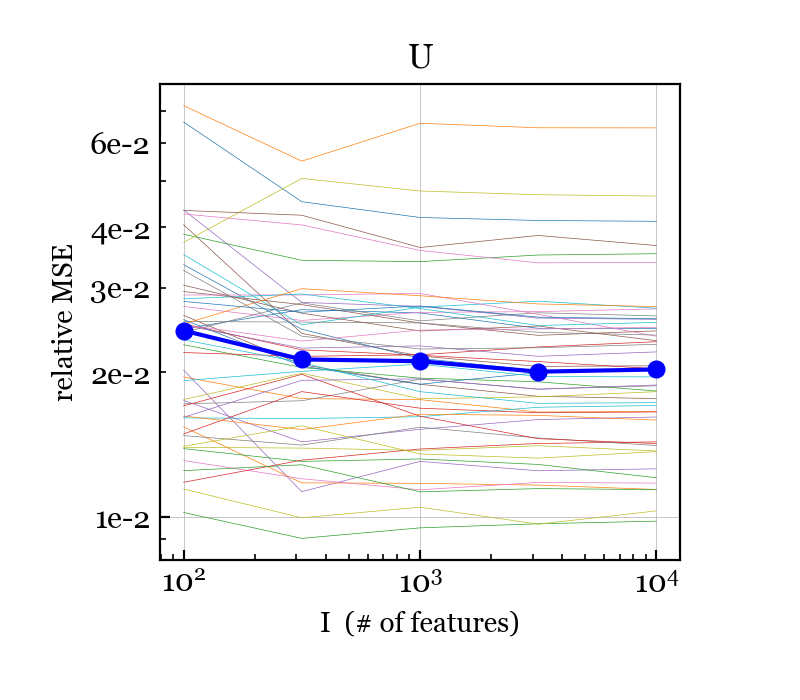}
  \includegraphics[trim=1.1cm 0 1.3cm 0, clip, height=0.2\textheight]{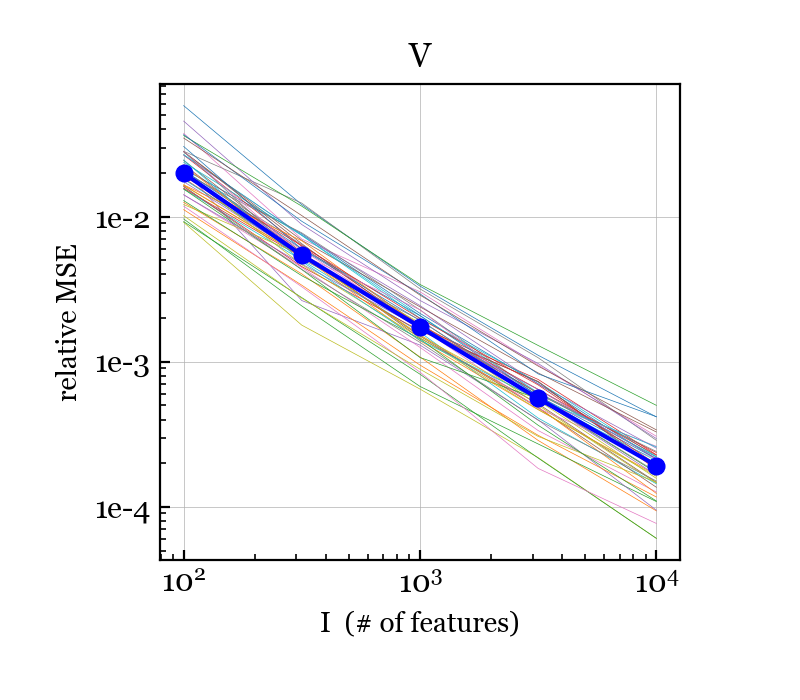}\\
  \includegraphics[trim=0.6cm 0 1.3cm 0, clip, height=0.2\textheight]{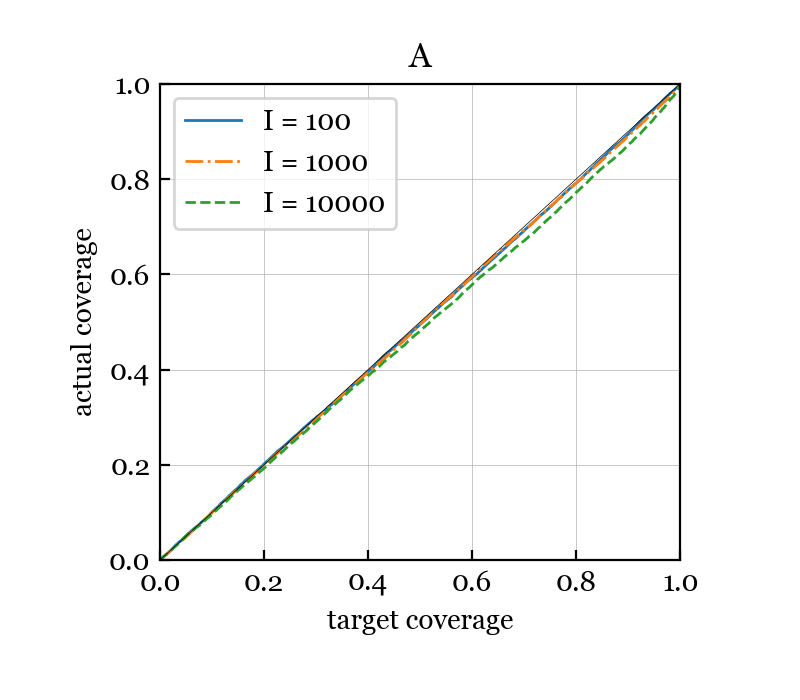}
  \includegraphics[trim=1.3cm 0 1.3cm 0, clip, height=0.2\textheight]{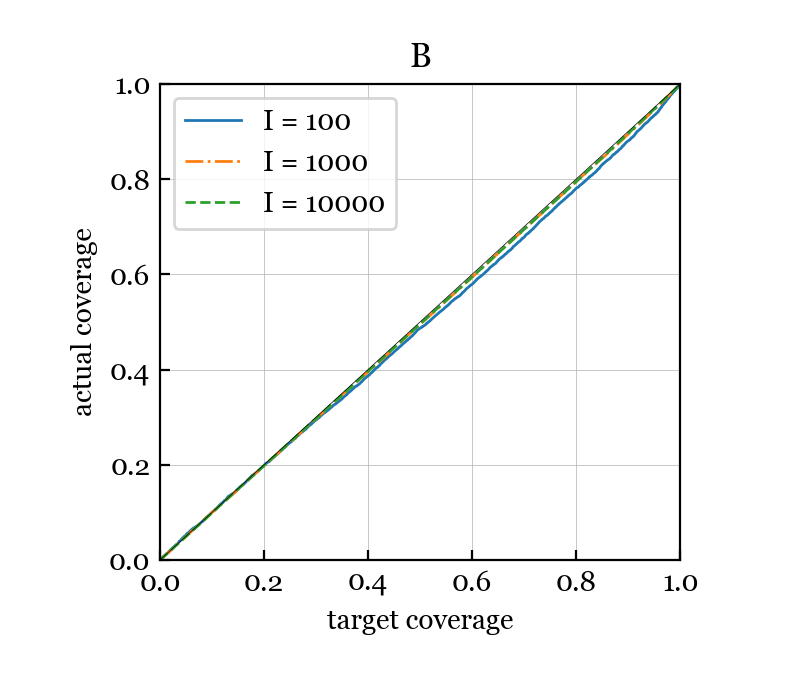}
  \includegraphics[trim=1.3cm 0 1.3cm 0, clip, height=0.2\textheight]{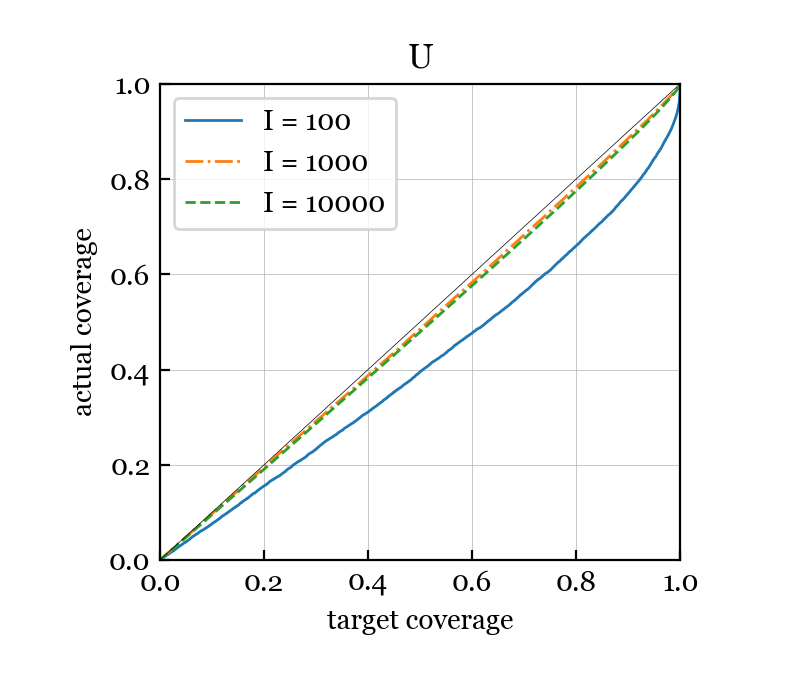}
  \includegraphics[trim=1.3cm 0 1.3cm 0, clip, height=0.2\textheight]{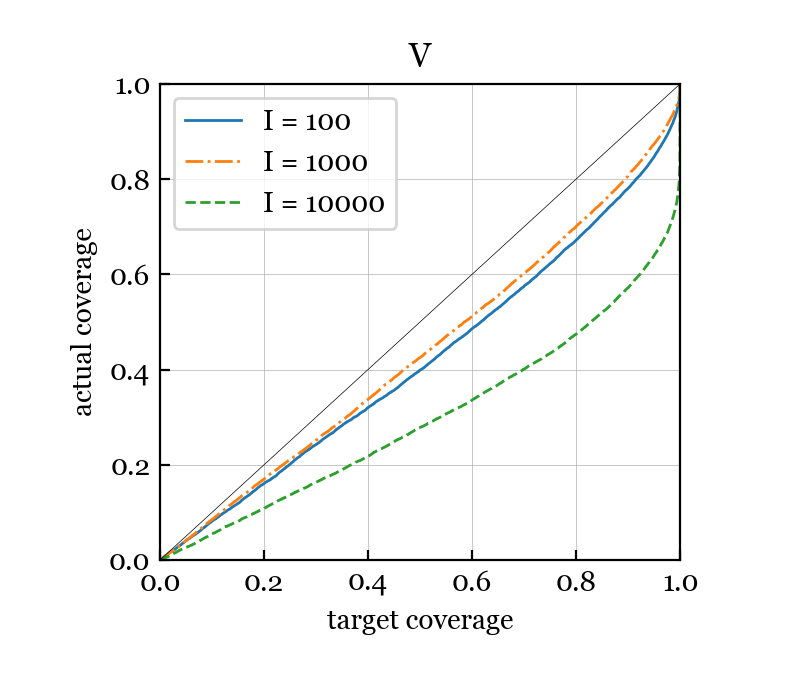}
  \caption{Robustness to outcome: $K = 4$, $L = 2$, $M = 3$, \tt{Poisson/Normal/Normal}.}
  \label{figure:Poisson-outcomes}
\end{figure}

\begin{figure}
  \centering
  \includegraphics[trim=0.6cm 0 1.5cm 0, clip, height=0.2\textheight]{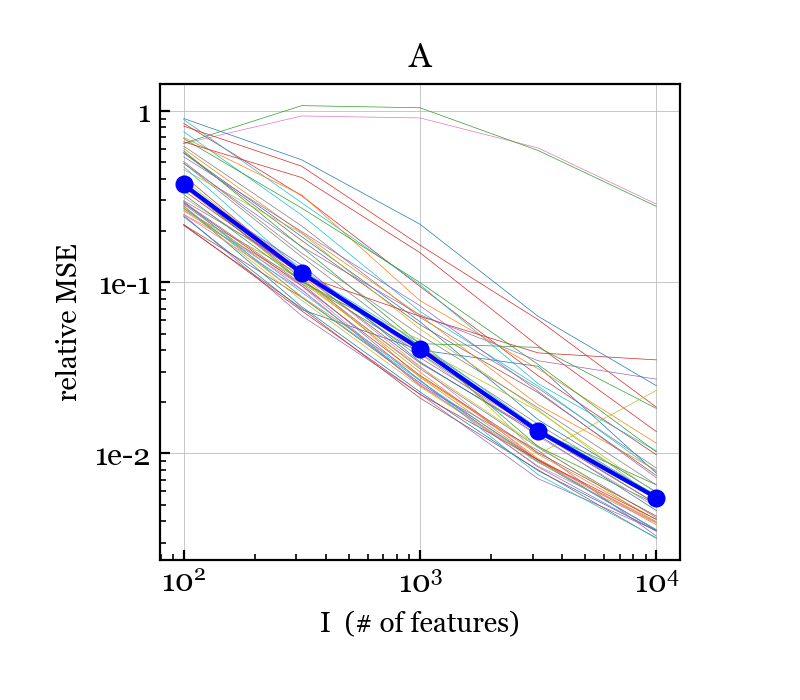}
  \includegraphics[trim=1.1cm 0 1.5cm 0, clip, height=0.2\textheight]{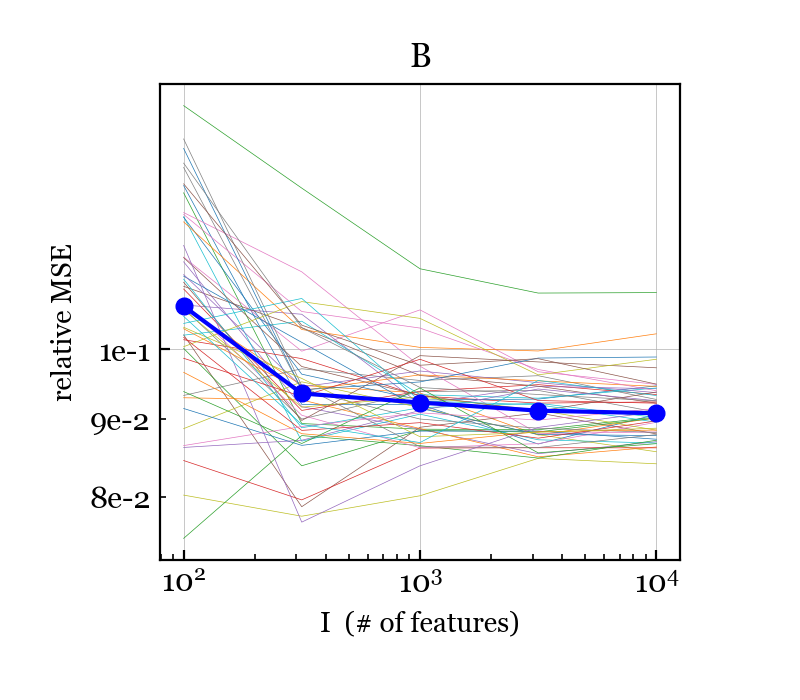}
  \includegraphics[trim=1.1cm 0 1.5cm 0, clip, height=0.2\textheight]{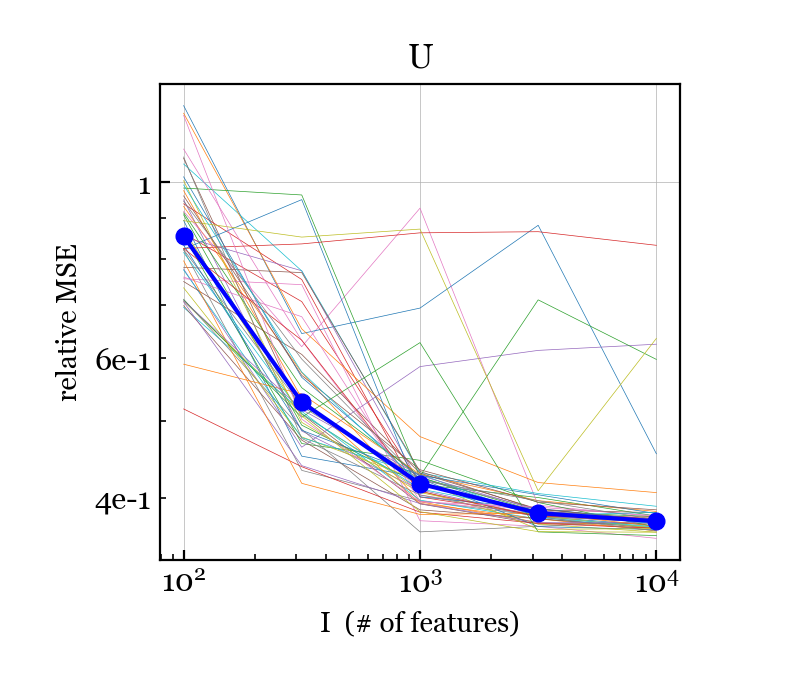}
  \includegraphics[trim=1.1cm 0 1.3cm 0, clip, height=0.2\textheight]{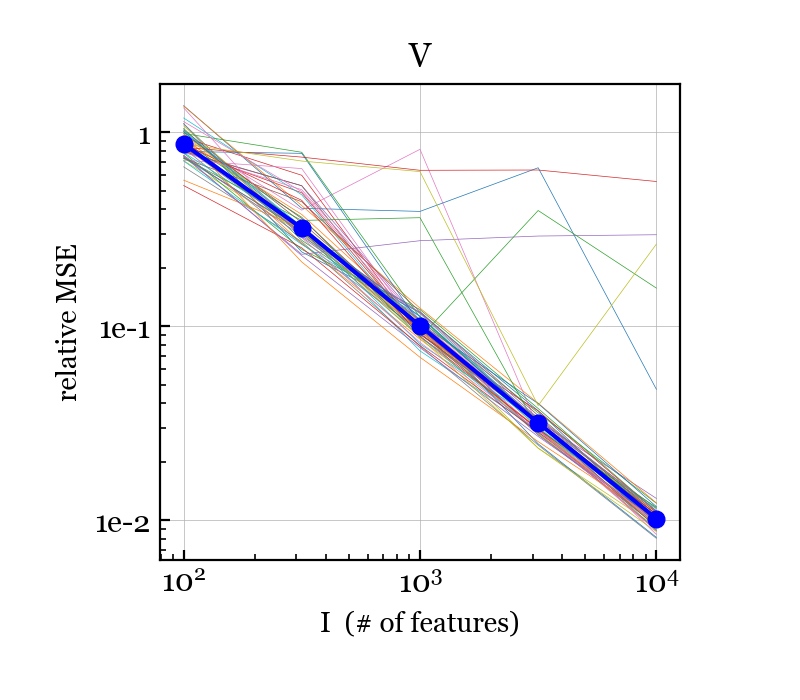}\\
  \includegraphics[trim=0.6cm 0 1.3cm 0, clip, height=0.2\textheight]{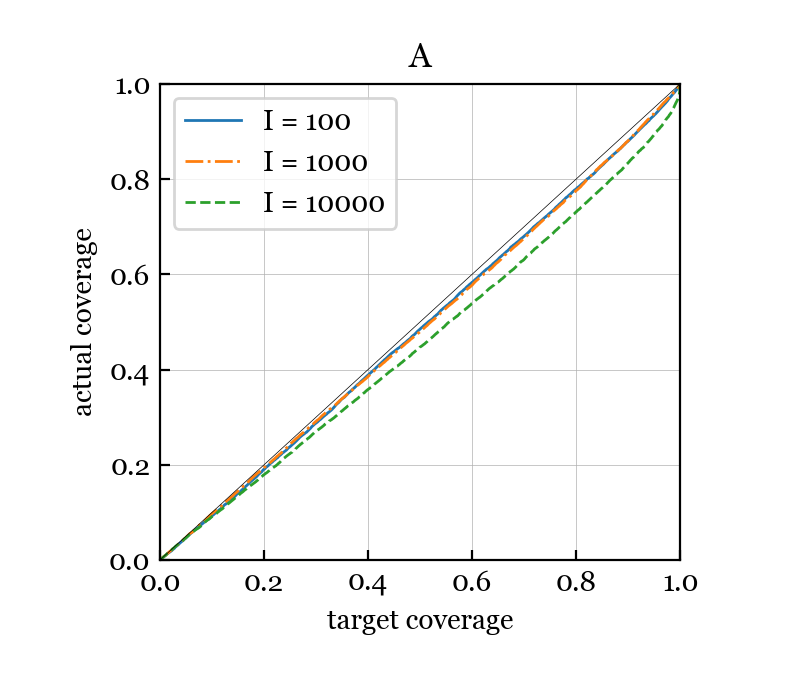}
  \includegraphics[trim=1.3cm 0 1.3cm 0, clip, height=0.2\textheight]{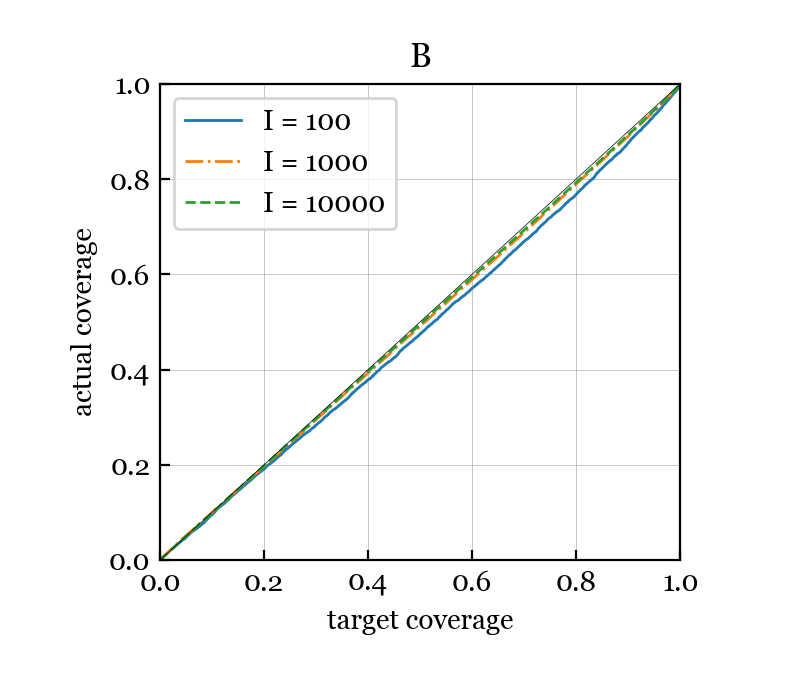}
  \includegraphics[trim=1.3cm 0 1.3cm 0, clip, height=0.2\textheight]{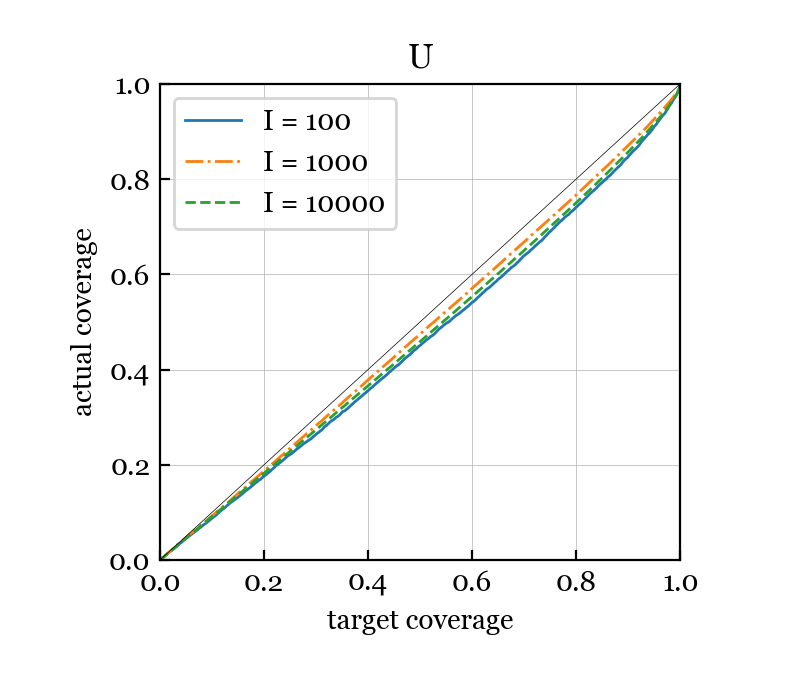}
  \includegraphics[trim=1.3cm 0 1.3cm 0, clip, height=0.2\textheight]{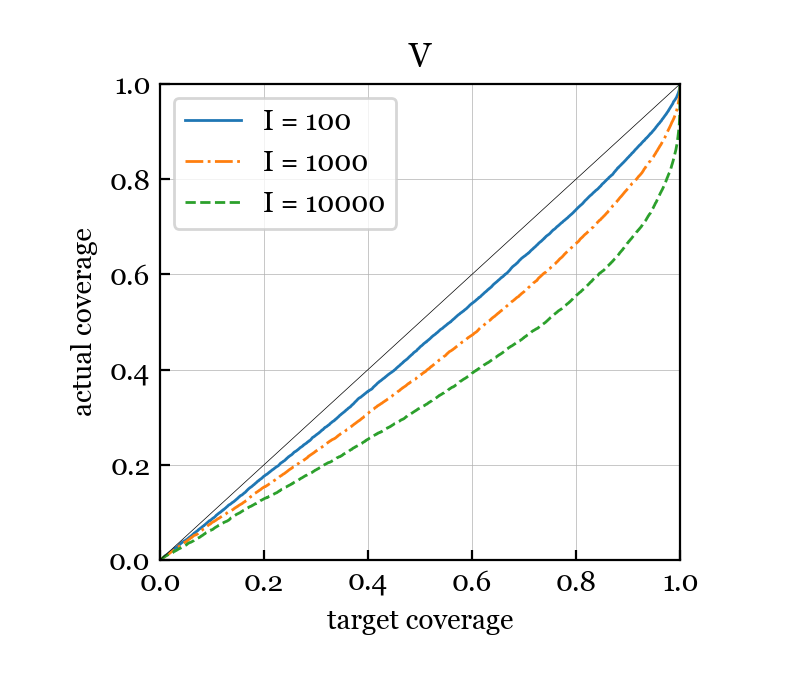}
  \caption{Robustness to outcome: $K = 4$, $L = 2$, $M = 3$, \tt{Geometric/Normal/Normal}.}
  \label{figure:Geometric-outcomes}
\end{figure}

\subsubsection*{Convergence to global optimum}
Based on the consistency experiments reported in Figure~\ref{figure:consistency-all}, it seems unlikely that the algorithm is getting stuck in a suboptimal local mode.
To more directly assess whether the algorithm is converging to a global optimum, we compare the estimates obtained when
(a) initializing the estimation algorithm to the true parameter values versus
(b) initializing using our proposed approach.
Table~\ref{table:initialization-error} shows that the difference between these estimates is negligible.
In detail, we generate $50$ data matrices using the \texttt{NB/Normal/Normal} scheme with $I=1000$, $J=100$, $K=4$, $L=2$, and $M=3$.
For each data matrix, we run the estimation algorithm with initialization approaches (a) and (b) and compute the relative MSE between the two resulting estimates.
In Table~\ref{table:initialization-error}, for each parameter we report the largest observed value of the relative MSE over these $50$ simulation runs.
The maximum relative MSE is extremely small in each case, meaning that initializing at the true parameter values yields nearly identical estimates 
as initializing with our proposed approach.
This suggests that our estimation algorithm is not getting stuck in a suboptimal local mode.

\begingroup
\renewcommand*{\arraystretch}{1.2}
\begin{table} %[!ht]
\centering
% \addtolength{\tabrowsep}{-1pt}
\addtolength{\tabcolsep}{-1.5pt}
\footnotesize
\caption{Initialization error.
% Initializing at the true parameter values is 
% We compare estimates using proposed initialization versus initializing at the true parameter values.
% Each entry is the maximum relative MSE over 50 runs with $I=1000$ and $J=100$.}%
Maximum relative MSE between estimates when initializing at the true values versus the proposed initialization approach, 
over 50 runs with $I=1000$ and $J=100$.}%
\label{table:initialization-error}%
\begin{tabular}{|c|c|c|c|c|c|c|c|c|}%
\hline
$A$ & $B$ & $C$ & $D$ & $U$ & $V$ & $S$ & $T$ & $\omega$ \\
\hline
$2\times 10^{-7}$ & $9\times 10^{-7}$ & $7\times 10^{-9}$ & $1\times 10^{-8}$ & $4\times 10^{-6}$ & $3\times 10^{-7}$ & $3\times 10^{-7}$ & $4\times 10^{-8}$ & $2\times 10^{-9}$ \\
\hline
\end{tabular}
\end{table}
\endgroup

% A: 1.6661698393849036e-7
% B: 9.024450234869325e-7
% C: 6.9519477627992285e-9
% D: 1.4053458978537885e-8
% U: 4.2702185562537806e-6
% V: 2.830042087684861e-7
% S: 3.299170205326515e-7
% T: 4.02358586841068e-8
% omega: 2.4635511977025706e-9

\subsubsection*{Choice of identifiability constraint on log-dispersions}

The choice of identifiability constraint on the log-dispersions $S$ and $T$ has a significant effect on estimation accuracy.
Perhaps the most obvious choice would be sum-to-zero constraints: $\sum_i s_i = 0$ and $\sum_j t_j = 0$.
However, it turns out that this leads to poor performance in terms of estimation accuracy.
The reason is that low log-dispersions (for instance, around $t_j \approx -2$ or less) are difficult to estimate accurately,
and noisy estimates of a few low values of $t_j$ can have an undue effect on $\sum_j t_j$, causing significant error in the rest of the $t_j$'s.
Figure~\ref{figure:log-dispersion-constraint} (left panel) illustrates the issue in a simulated example using the \tt{NB/Normal/Normal} scheme with $I = 10000$, $J = 100$, $K = 4$, $L = 2$, and $M = 3$.  Here, the sum-to-zero constraints are used on both the true parameters and the estimates.

\begin{figure}
  \centering
  \includegraphics[trim=0 0.7cm 0 0, clip, width=0.4\textwidth]{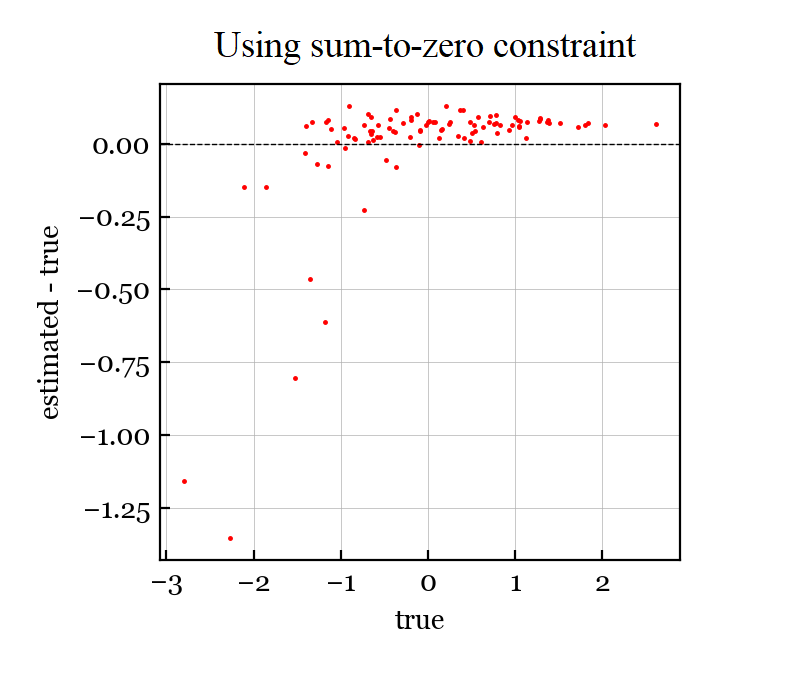}
  \includegraphics[trim=0 0.7cm 0 0, clip, width=0.4\textwidth]{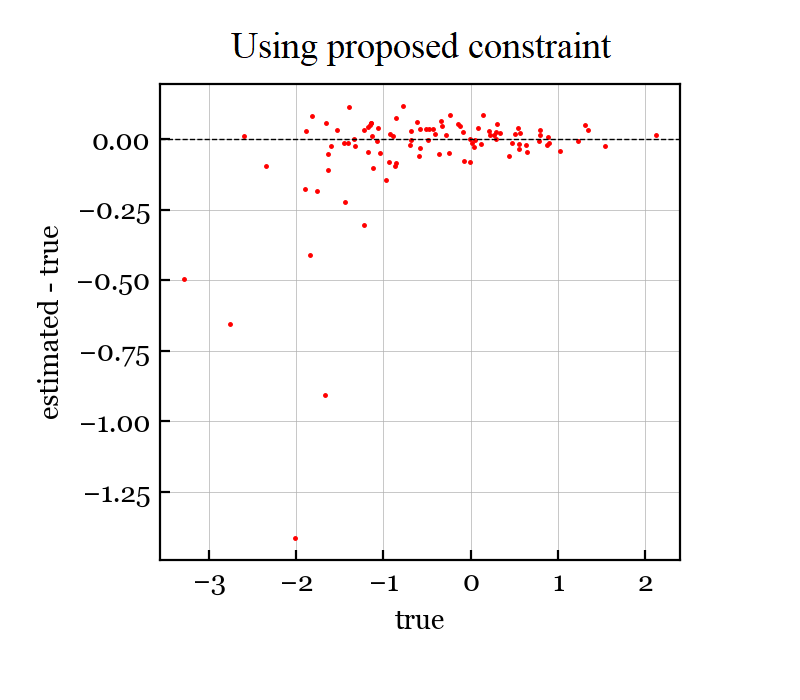}
  \caption{Scatterplots of $\hat{t}_j - t_{0 j}$ versus $t_{0 j}$ for a typical simulated data matrix.}
  \label{figure:log-dispersion-constraint}
\end{figure}

In our proposed algorithm, we instead constrain $\frac{1}{I}\sum_i e^{s_i} = 1$ and $\frac{1}{J}\sum_j e^{t_j} = 1$, which has the effect of downweighting the low values (which are harder to estimate) and upweighting the large values (which are easier to estimate).
As illustrated in Figure~\ref{figure:log-dispersion-constraint} (right panel), this effectively resolves the issue.
Here, we run the same algorithm on the same data matrix, but instead use our proposed constraints.
We can see that the undesirable shift is effectively removed.

\subsubsection*{Special case of no latent factors}
If the number of latent factors is zero, that is, $M = 0$, then the model is no longer ``bilinear'' since the $U D V^\T$ term is no longer present.
The resulting model with $g(\E(\bm{Y})) = X A^\T + B Z^\T + X C Z^\T$ can be viewed as a standard GLM via the vectorization approach discussed in Section~\ref{section:challenges}, however, computation using standard GLM methods becomes intractable when $I$ or $J$ is large.
Thus, our methods provide a computationally tractable approach to estimation and inference in this quite general class of GLMs for matrix data.
In Figure~\ref{figure:K=4-L=2-M=0}, we present a subset of simulation results using the \tt{NB/Normal/Normal} simulation scheme with $J = 100$, $K = 4$, $L = 2$, and $M = 0$.
As expected, the estimates for $A$ and $C$ appear to be consistent and rapidly converging to the true values.
Further, the coverage is nearly perfect for $A$ and $B$, and possibly slightly conservative for $C$.

\begin{figure}
  \centering
  \includegraphics[trim=0.6cm 0 1.5cm 0, clip, height=0.2\textheight]{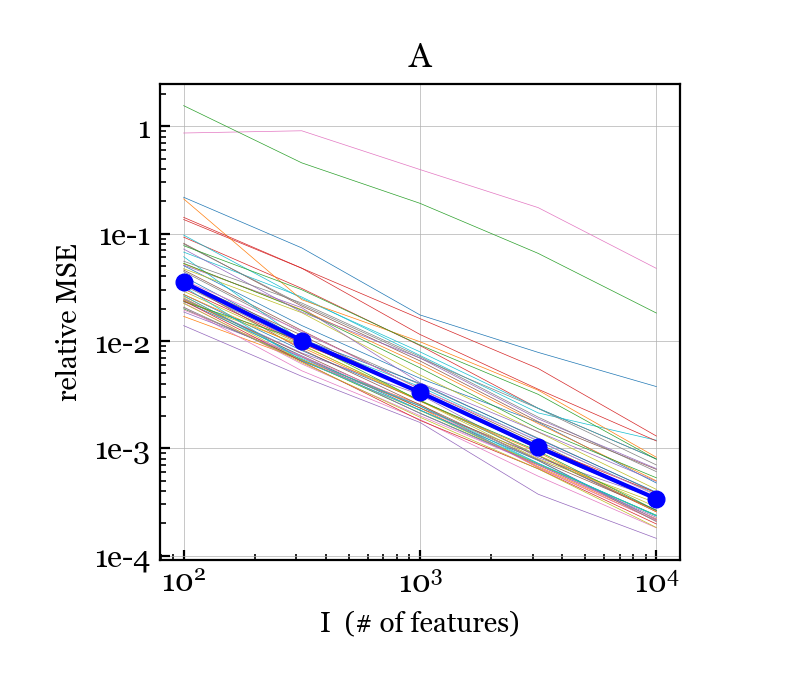}
  \includegraphics[trim=1.1cm 0 1.5cm 0, clip, height=0.2\textheight]{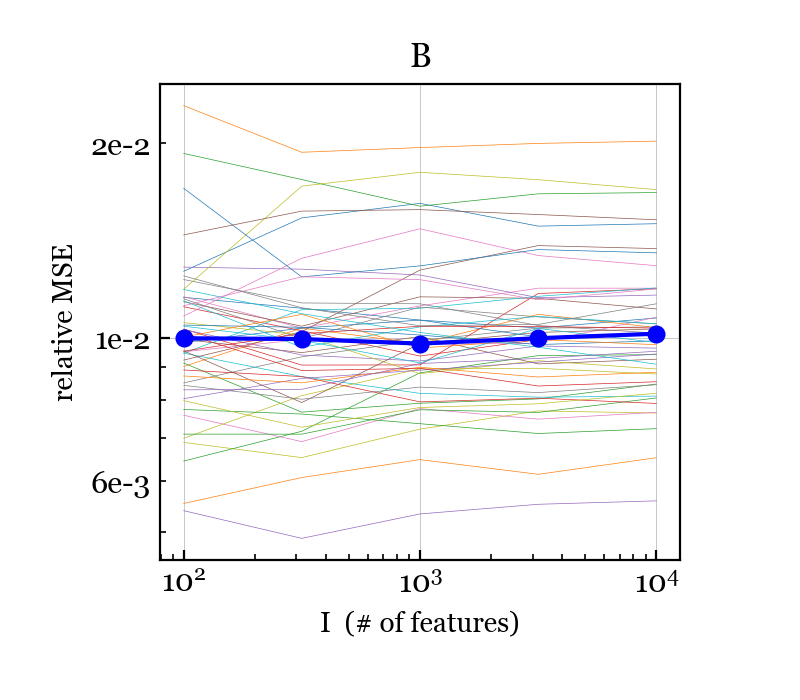}
  \includegraphics[trim=1.1cm 0 1.3cm 0, clip, height=0.2\textheight]{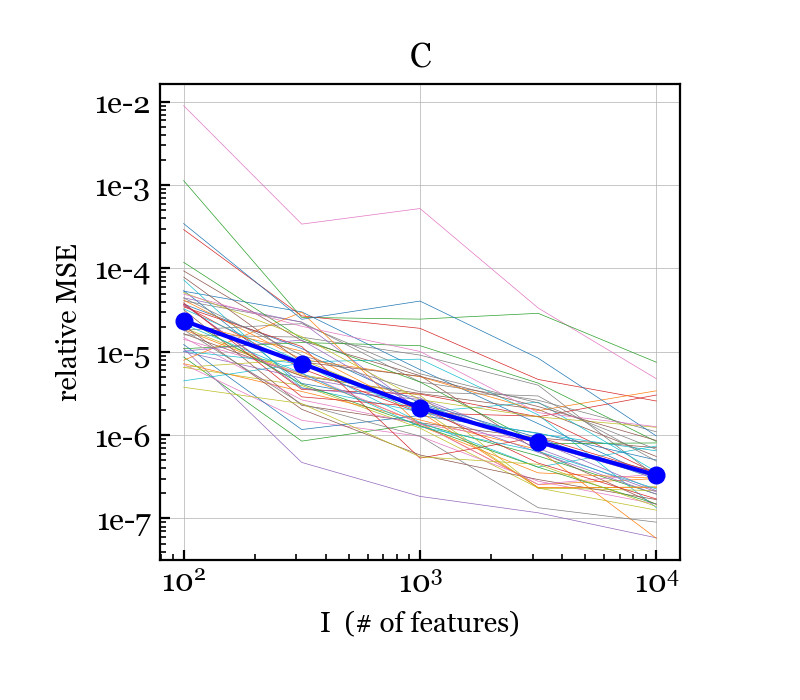}\\
  \includegraphics[trim=0.6cm 0 1.3cm 0, clip, height=0.2\textheight]{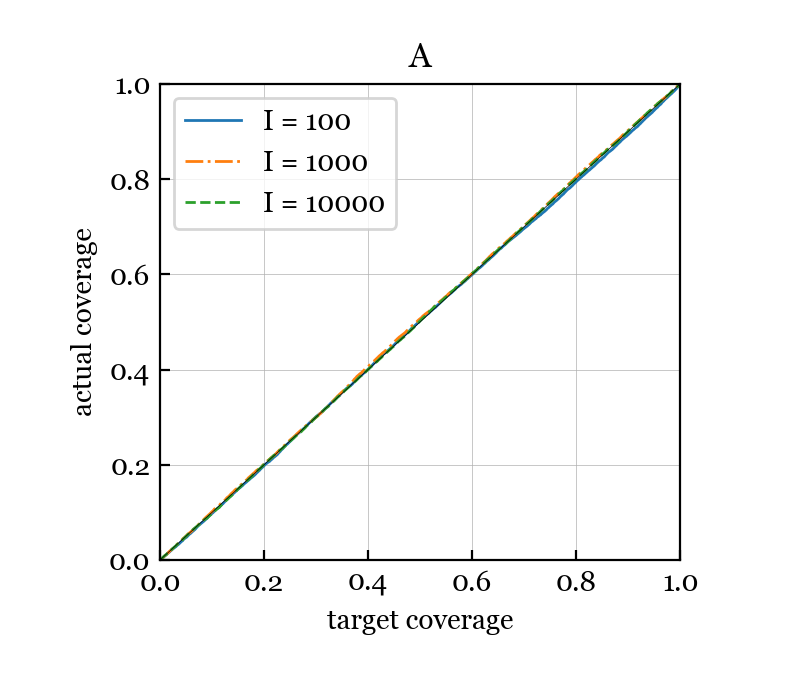}
  \includegraphics[trim=1.3cm 0 1.3cm 0, clip, height=0.2\textheight]{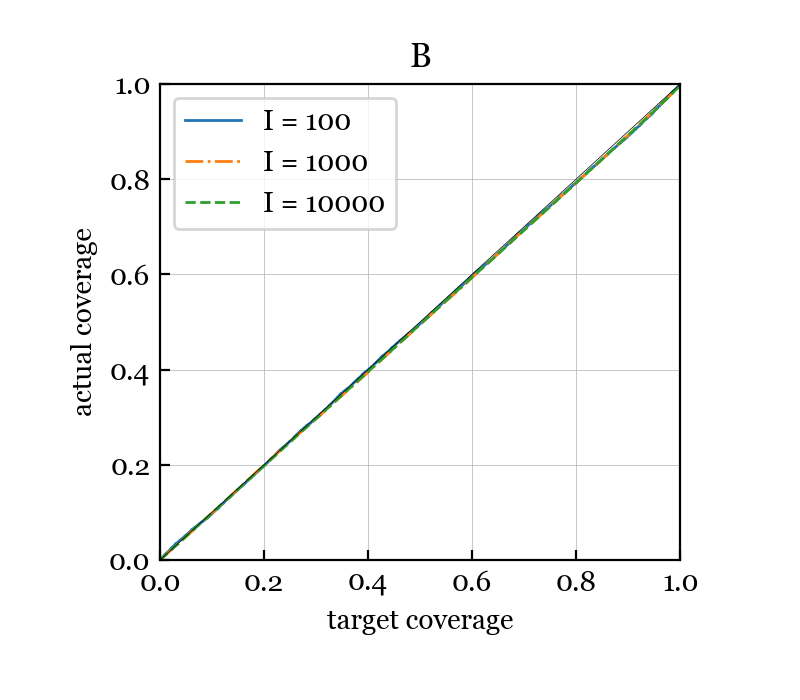}
  \includegraphics[trim=1.3cm 0 1.3cm 0, clip, height=0.2\textheight]{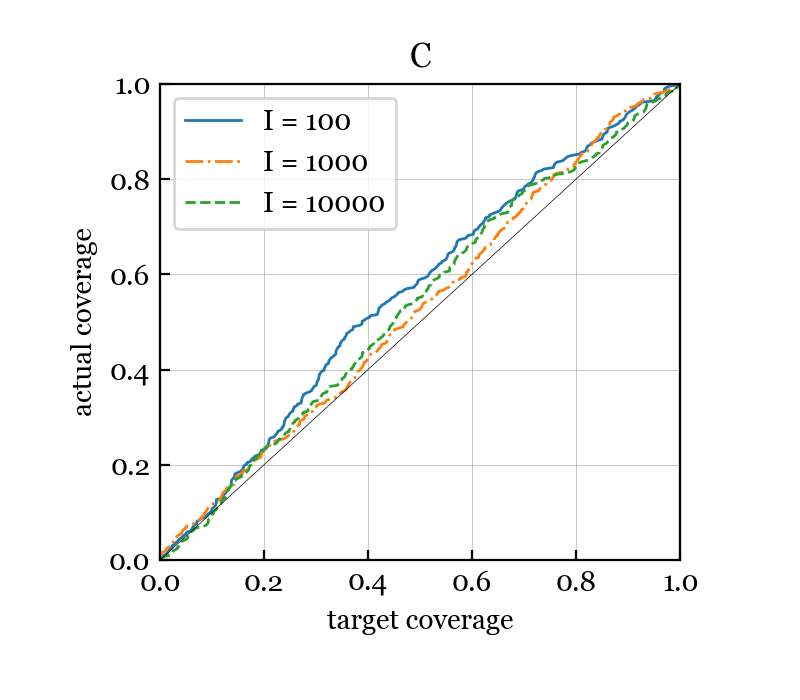}
  \caption{Special case of no latent factors: $K = 4$, $L = 2$, $M = 0$, \tt{NB/Normal/Normal}.}
  \label{figure:K=4-L=2-M=0}
\end{figure}

\begin{figure}
  \centering
  \includegraphics[trim=0.6cm 0 1.5cm 0, clip, height=0.2\textheight]{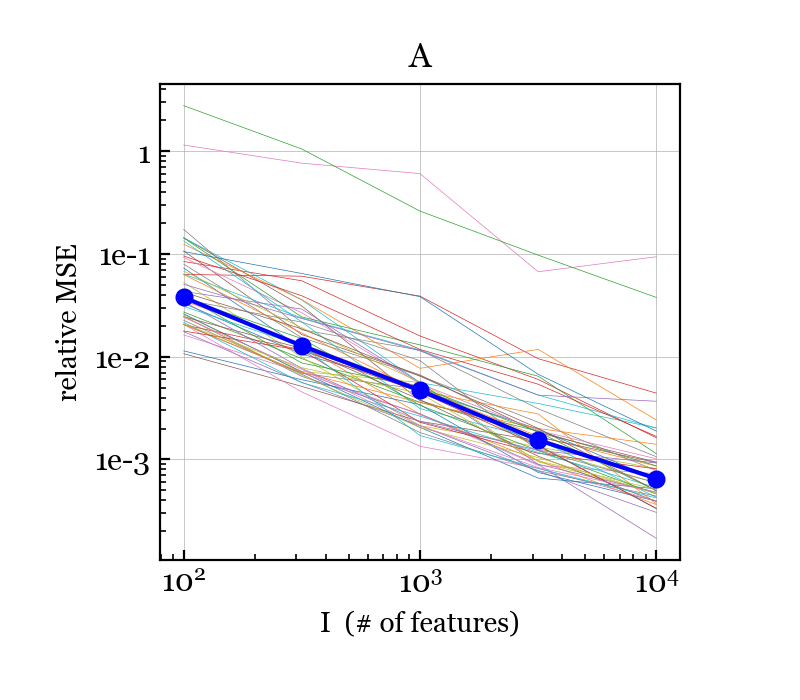}
  \includegraphics[trim=1.1cm 0 1.5cm 0, clip, height=0.2\textheight]{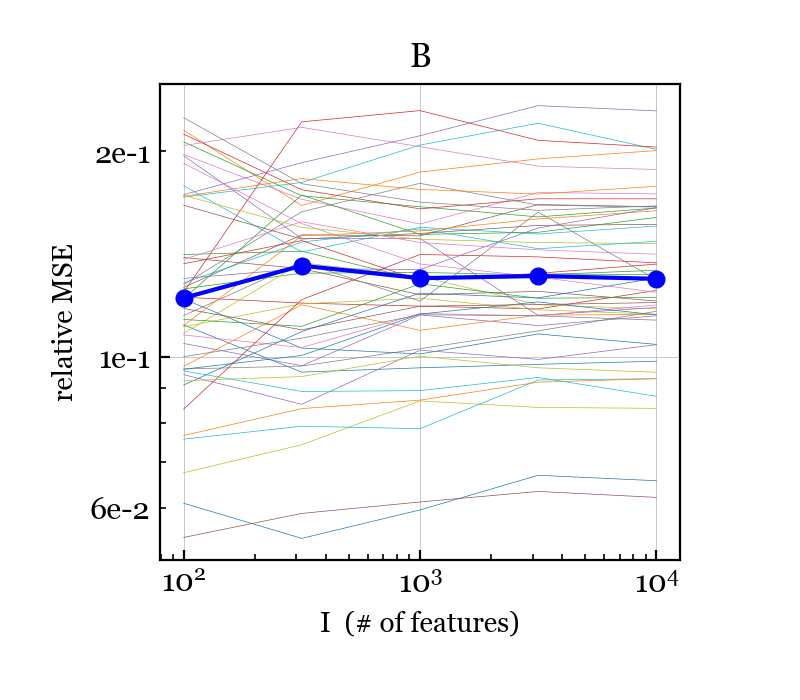}
  \includegraphics[trim=1.1cm 0 1.3cm 0, clip, height=0.2\textheight]{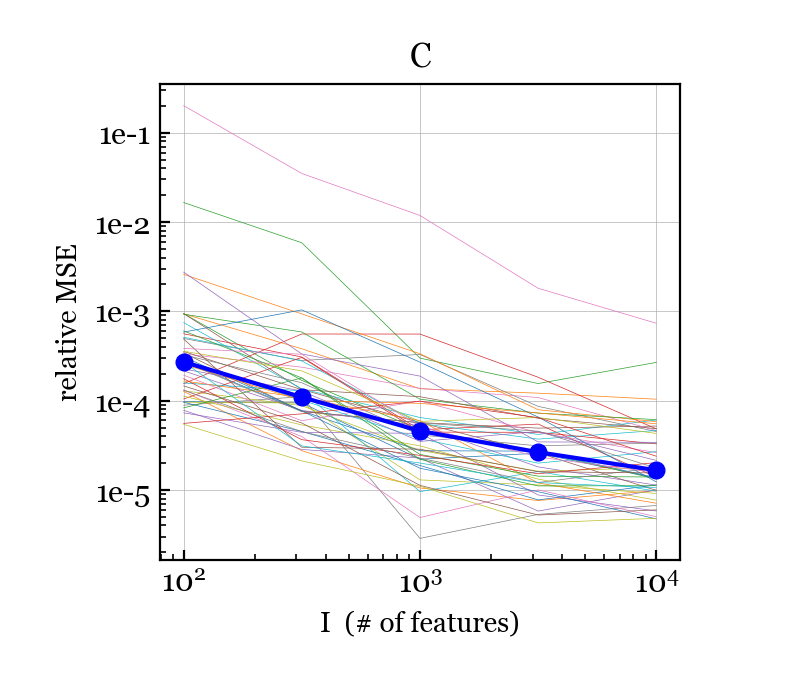}\\
  \includegraphics[trim=0.6cm 0 1.3cm 0, clip, height=0.2\textheight]{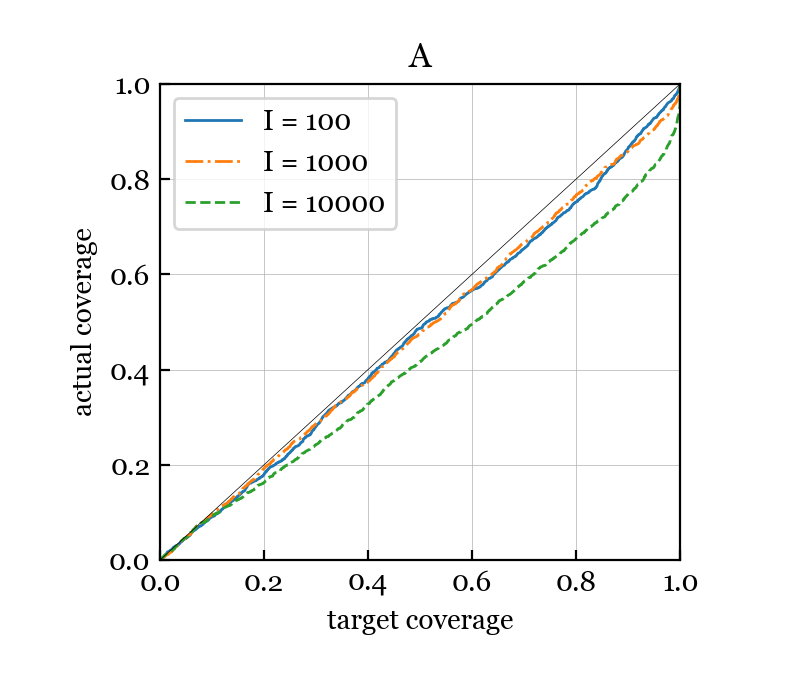}
  \includegraphics[trim=1.3cm 0 1.3cm 0, clip, height=0.2\textheight]{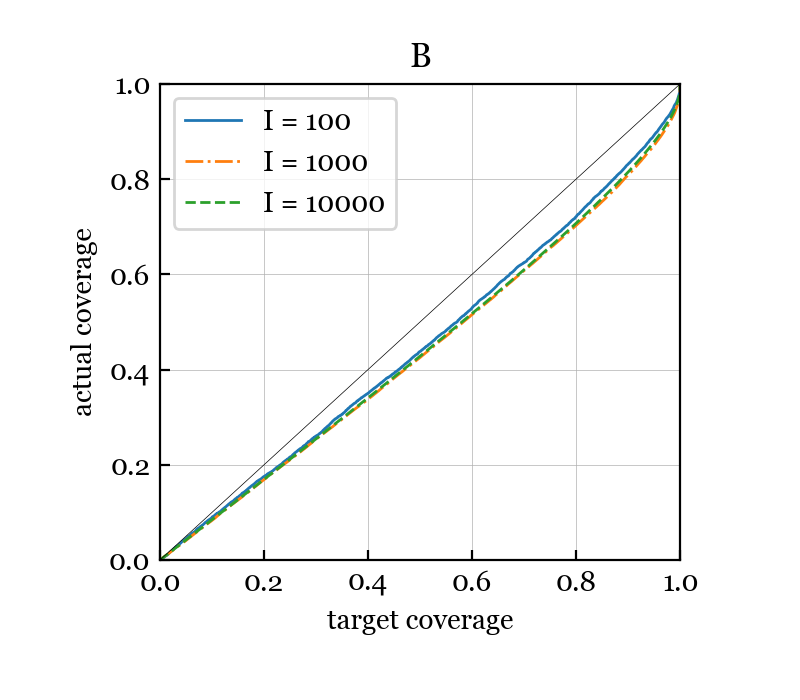}
  \includegraphics[trim=1.3cm 0 1.3cm 0, clip, height=0.2\textheight]{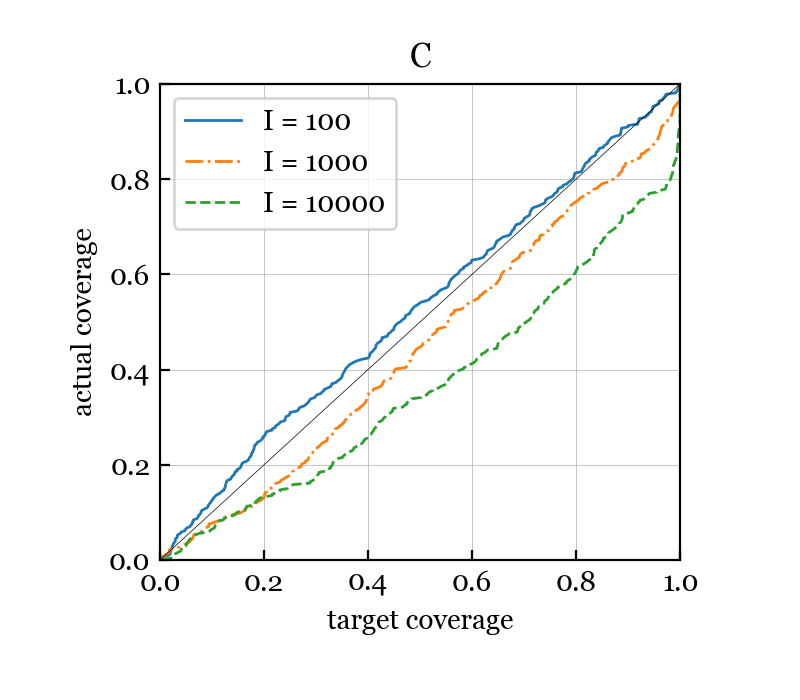}
  \caption{Effect of small sample size: $J = 10$, $K = 4$, $L = 2$, $M = 0$, \tt{NB/Normal/Normal}.}
  \label{figure:J=10-K=4-L=2-M=0}
\end{figure}

\begin{figure}
  \centering
  \includegraphics[trim=0.6cm 0 1.5cm 0, clip, height=0.2\textheight]{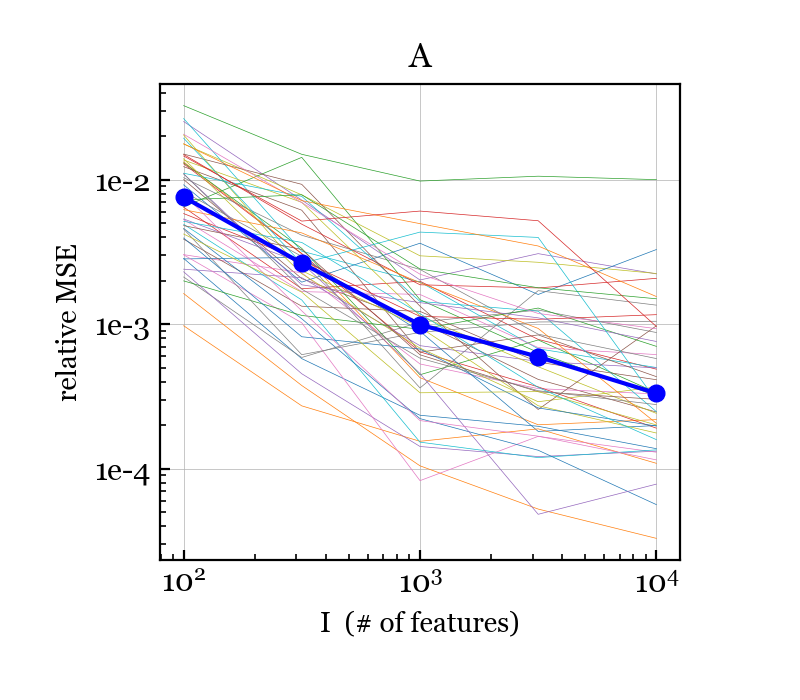}
  \includegraphics[trim=1.1cm 0 1.5cm 0, clip, height=0.2\textheight]{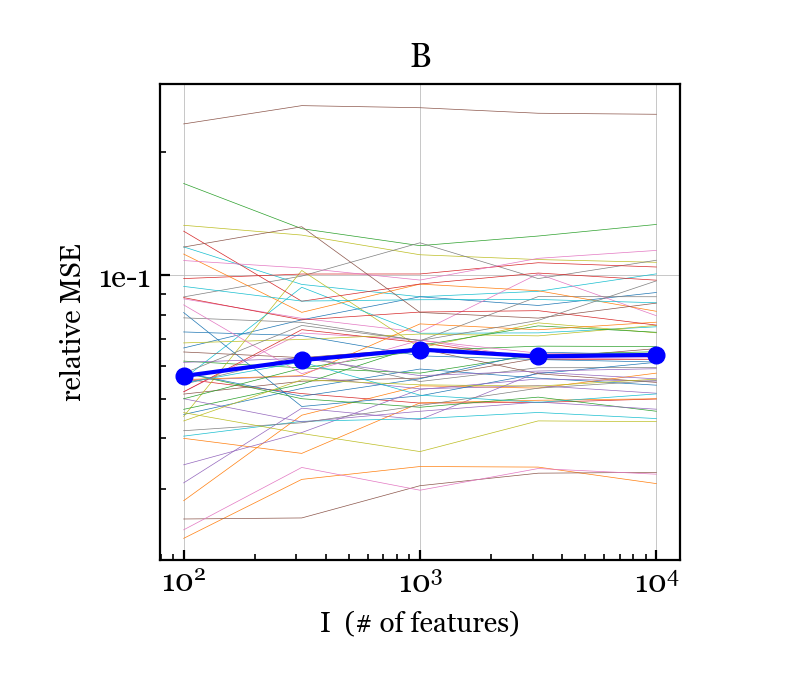}
  \includegraphics[trim=1.1cm 0 1.5cm 0, clip, height=0.2\textheight]{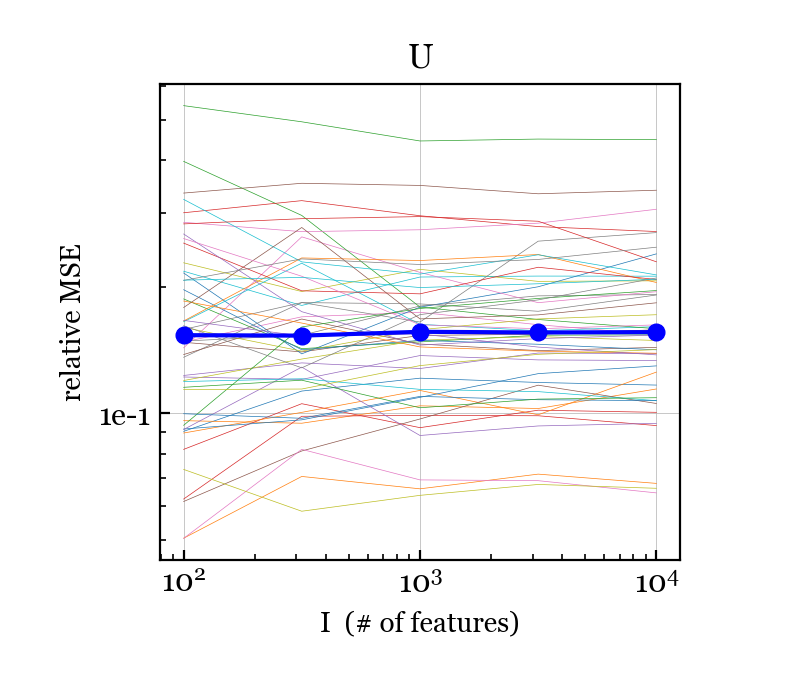}
  \includegraphics[trim=1.1cm 0 1.3cm 0, clip, height=0.2\textheight]{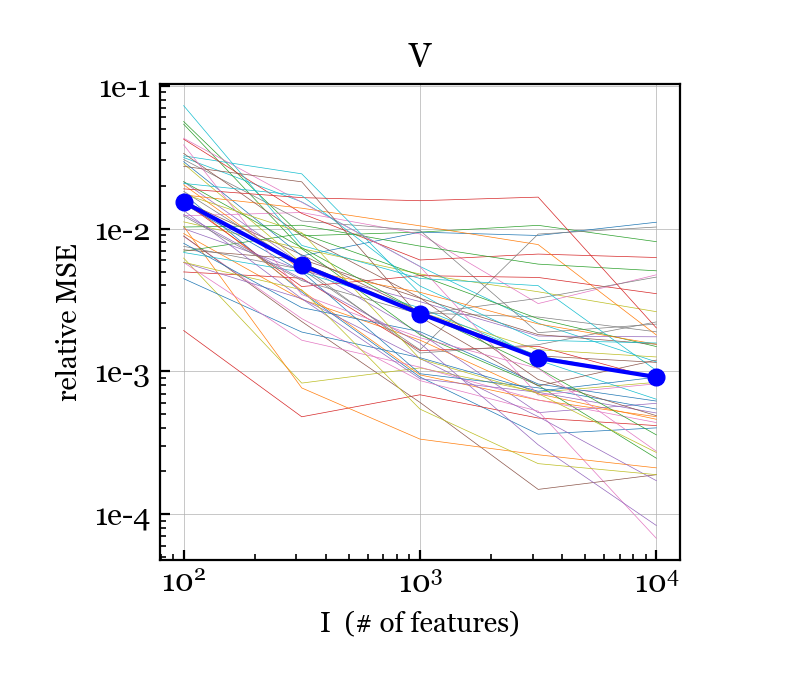}\\
  \includegraphics[trim=0.6cm 0 1.3cm 0, clip, height=0.2\textheight]{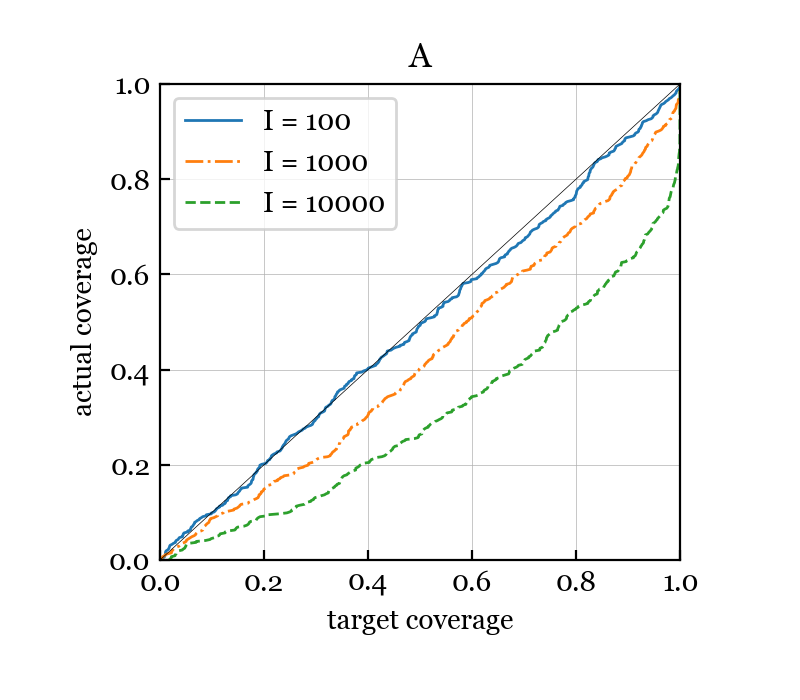}
  \includegraphics[trim=1.3cm 0 1.3cm 0, clip, height=0.2\textheight]{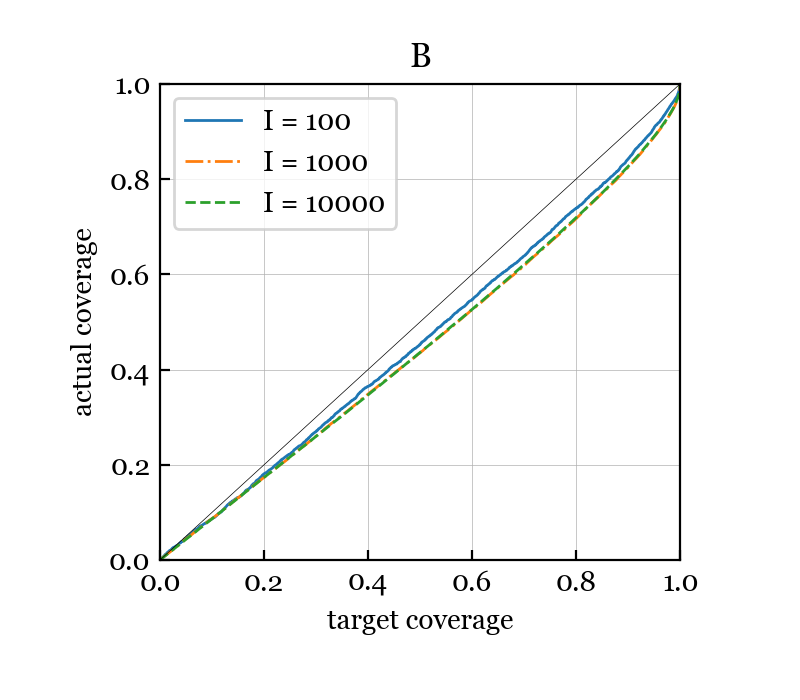}
  \includegraphics[trim=1.3cm 0 1.3cm 0, clip, height=0.2\textheight]{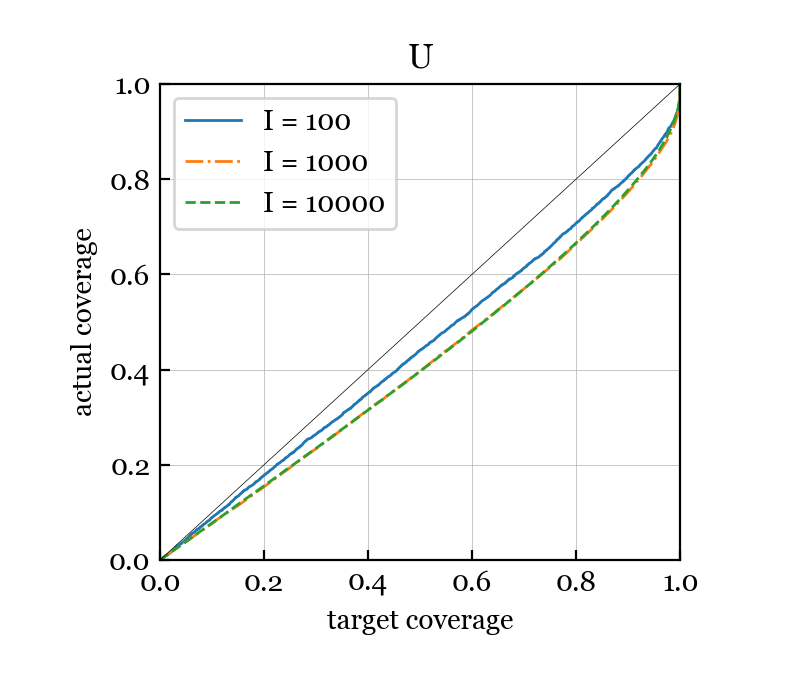}
  \includegraphics[trim=1.3cm 0 1.3cm 0, clip, height=0.2\textheight]{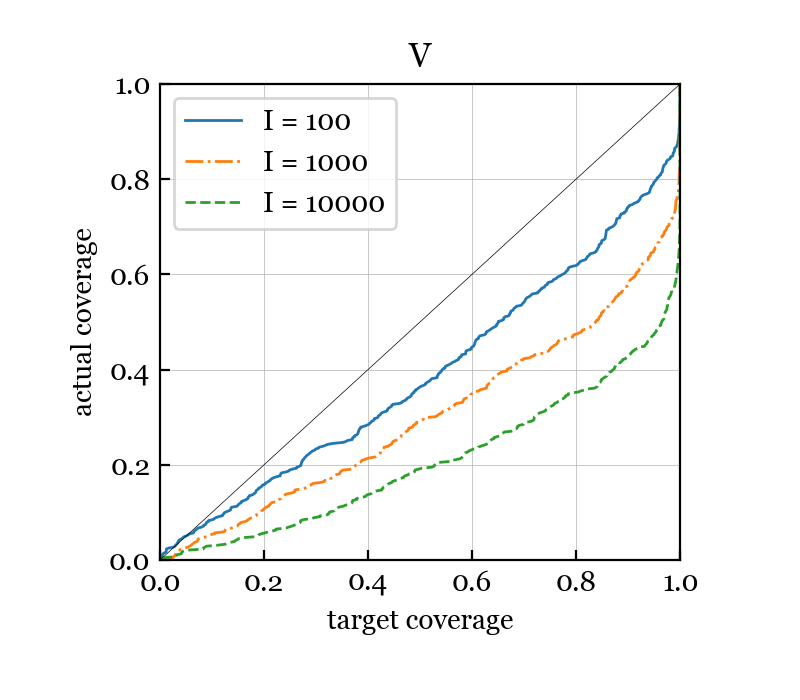}
  \caption{Effect of small sample size: $J = 10$, $K = 1$, $L = 1$, $M = 1$, \tt{NB/Normal/Normal}.}
  \label{figure:J=10-K=1-L=1-M=1}
\end{figure}

\subsubsection*{Effect of small sample size}
In applications such as gene expression analysis, sometimes a very small number of samples are available.
To explore this regime, we run simulations using the \tt{NB/Normal/Normal} scheme with $J = 10$ and
(a) $K = 4$, $L = 2$, $M = 0$, and
(b) $K = 1$, $L = 1$, $M = 1$.
As Figures~\ref{figure:J=10-K=4-L=2-M=0} and \ref{figure:J=10-K=1-L=1-M=1} show, consistency for the same parameters appears to hold, although in some cases the rate is less clear than when $J = 100$ (Figures~\ref{figure:K=4-L=2-M=0} and \ref{figure:K=1-L=1-M=1}).
Meanwhile, the coverage performance is noticably worse when $J = 10$ than when $J = 100$.

\subsubsection*{Effect of varying latent dimension}

The ``bilinear'' latent factorization term $U D V^\T$ is more challenging in terms of estimation and inference
compared to the linear terms $X A^\T$, $B Z^\T$, and $X C Z^\T$.
To assess the effect of the latent dimension $M$ on consistency, efficiency, and accuracy of standard errors,
we run the simulations
% we rerun the simulation experiments in Sections~\ref{section:consistency} and \ref{section:coverage}
using: (a) $K = 1$, $L = 1$, $M = 1$, and (b) $K = 1$, $L = 1$, and $M = 5$.
This is similar to the GAMMI model \citep{van1995multiplicative}, but also includes $S$, $T$, and $\omega$.
We see in Figure~\ref{figure:K=1-L=1-M=1} that when $M=1$,
the estimates of $A$ and $V$ are converging to the true values as expected,
and the coverage for $A$, $B$, $U$ and $V$ is excellent, except that the coverage for $V$ degrades slightly when $I = 10000$.
When the latent dimension is increased to $M = 5$, Figure~\ref{figure:K=1-L=1-M=5} shows that the performance is similar,
except that the coverage for $V$ degrades more severely when $I = 10000$.

Based on these and other experiments, we find that as long as the entries of $D$ are not too small and are sufficiently well-spaced, the performance remains good.
The magnitude of each entry of $D$ controls the strength of the signal for the corresponding latent factor;
thus, a small entry in $D$ makes it difficult to estimate the corresponding columns of $U$ and $V$.
Meanwhile, if two entries of $D$ are similar in magnitude, then there is near non-identifiability of $U$ and $V$ with respect to those two latent factors
due to rotational invariance, leading to a loss of performance in terms of statistical efficiency and coverage.

\begin{figure}
  \centering
  \includegraphics[trim=0.6cm 0 1.5cm 0, clip, height=0.2\textheight]{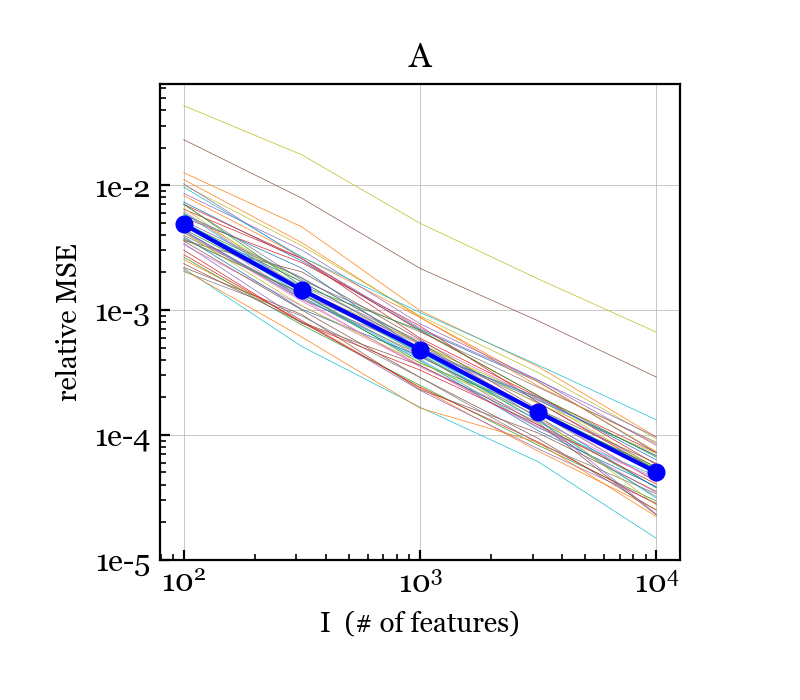}
  \includegraphics[trim=1.1cm 0 1.5cm 0, clip, height=0.2\textheight]{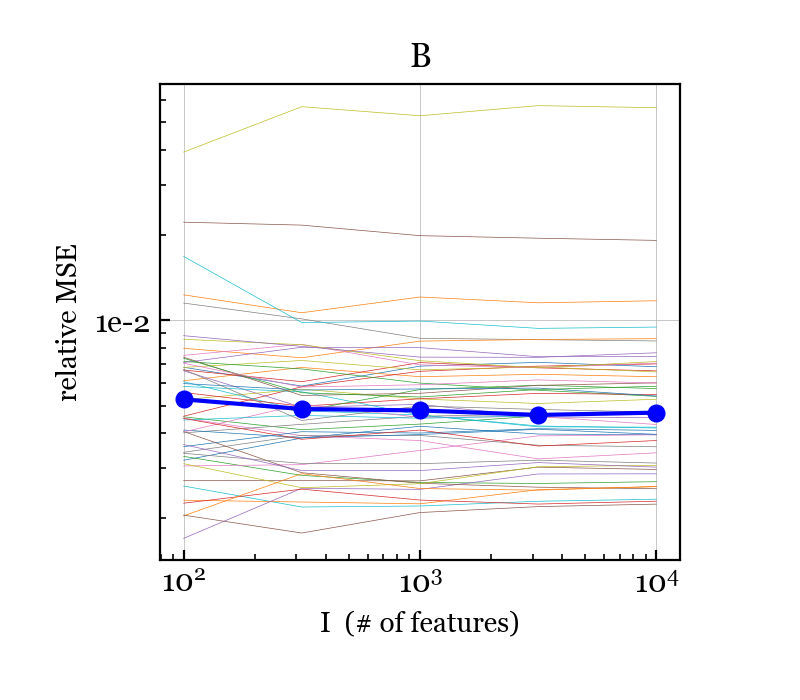}
  \includegraphics[trim=1.1cm 0 1.5cm 0, clip, height=0.2\textheight]{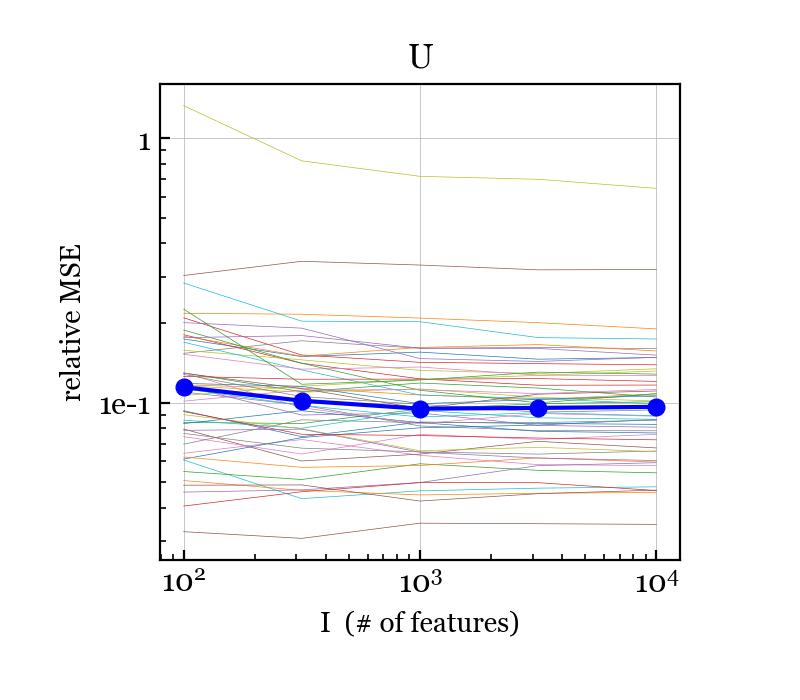}
  \includegraphics[trim=1.1cm 0 1.3cm 0, clip, height=0.2\textheight]{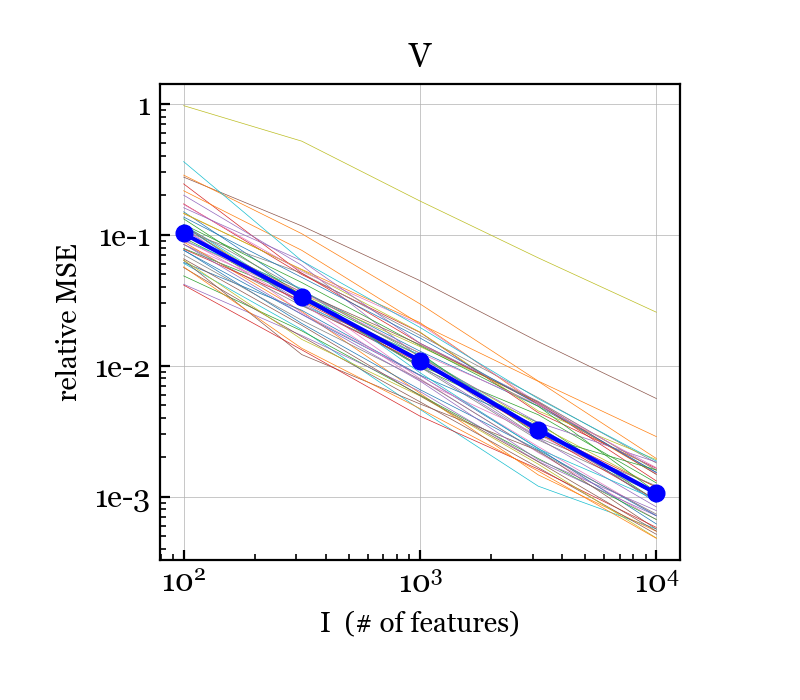}\\
  \includegraphics[trim=0.6cm 0 1.3cm 0, clip, height=0.2\textheight]{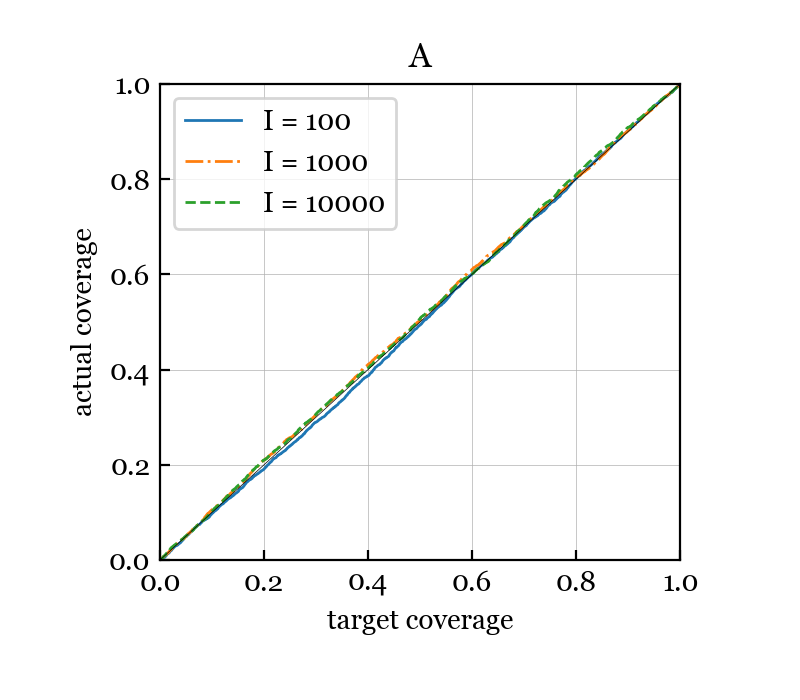}
  \includegraphics[trim=1.3cm 0 1.3cm 0, clip, height=0.2\textheight]{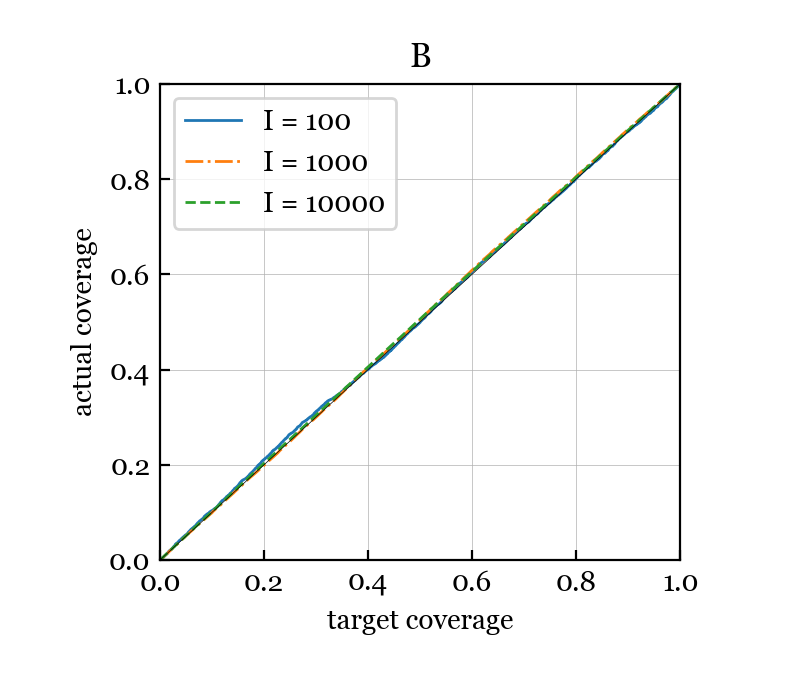}
  \includegraphics[trim=1.3cm 0 1.3cm 0, clip, height=0.2\textheight]{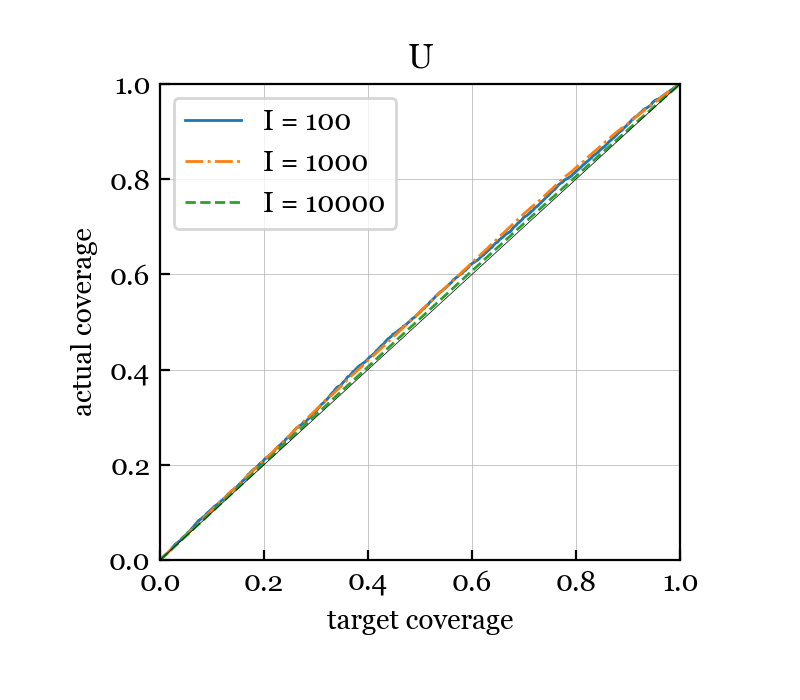}
  \includegraphics[trim=1.3cm 0 1.3cm 0, clip, height=0.2\textheight]{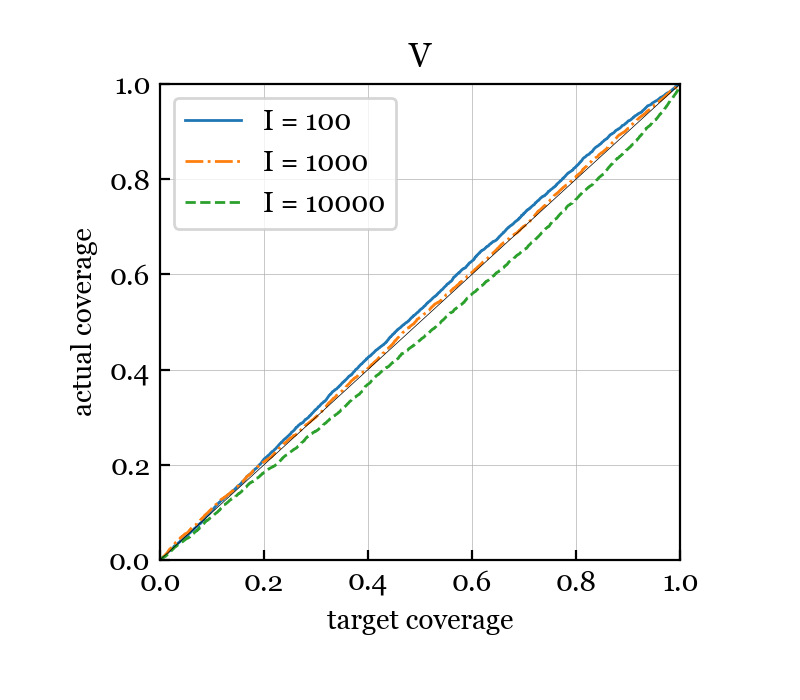}
  \caption{Effect of varying latent dimension: $K = 1$, $L = 1$, $M = 1$, \tt{NB/Normal/Normal}.}
  \label{figure:K=1-L=1-M=1}
\end{figure}

\begin{figure}
  \centering
  \includegraphics[trim=0.6cm 0 1.5cm 0, clip, height=0.2\textheight]{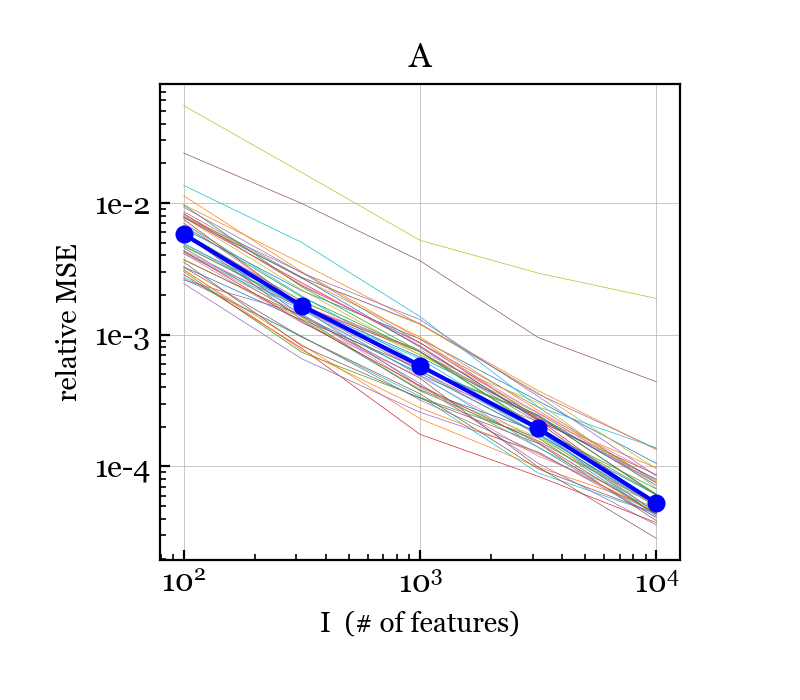}
  \includegraphics[trim=1.1cm 0 1.5cm 0, clip, height=0.2\textheight]{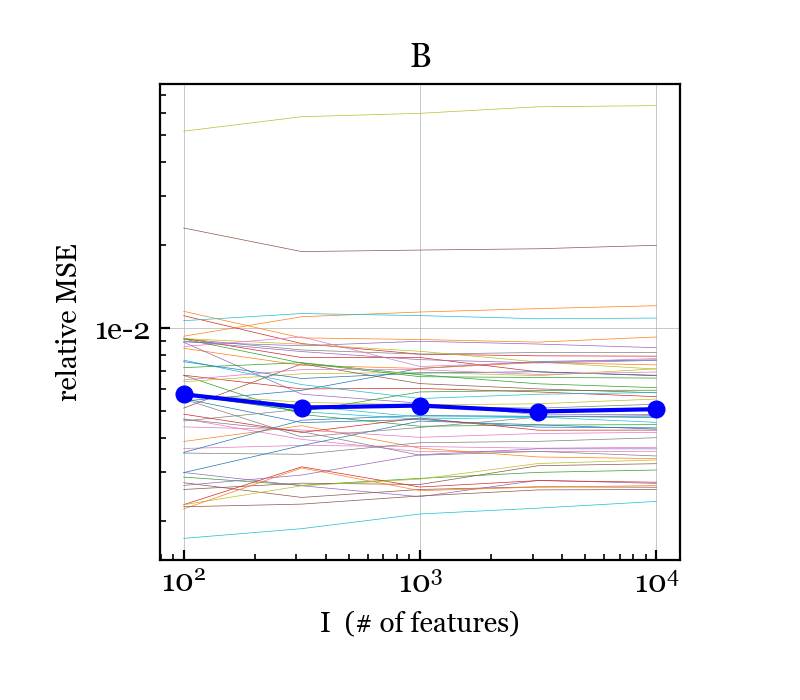}
  \includegraphics[trim=1.1cm 0 1.5cm 0, clip, height=0.2\textheight]{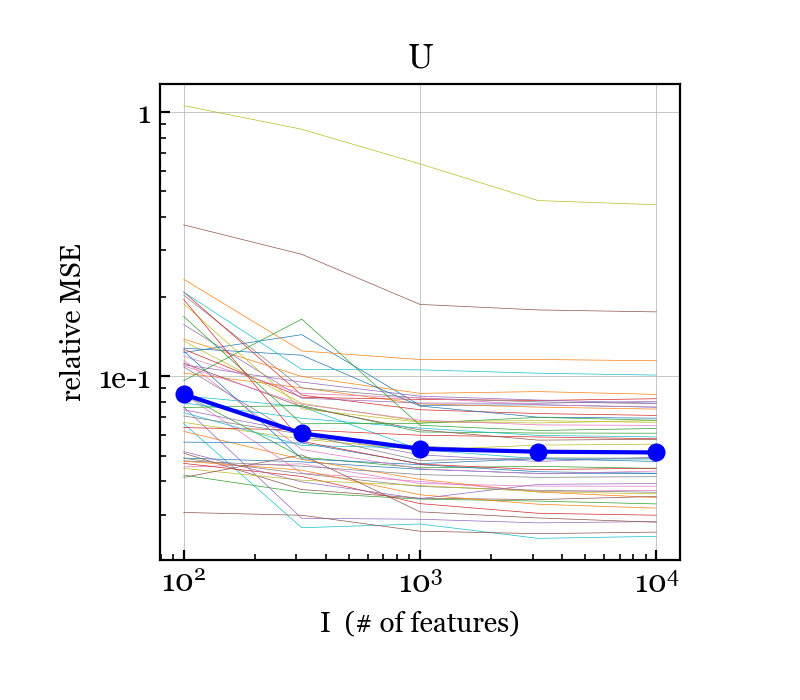}
  \includegraphics[trim=1.1cm 0 1.3cm 0, clip, height=0.2\textheight]{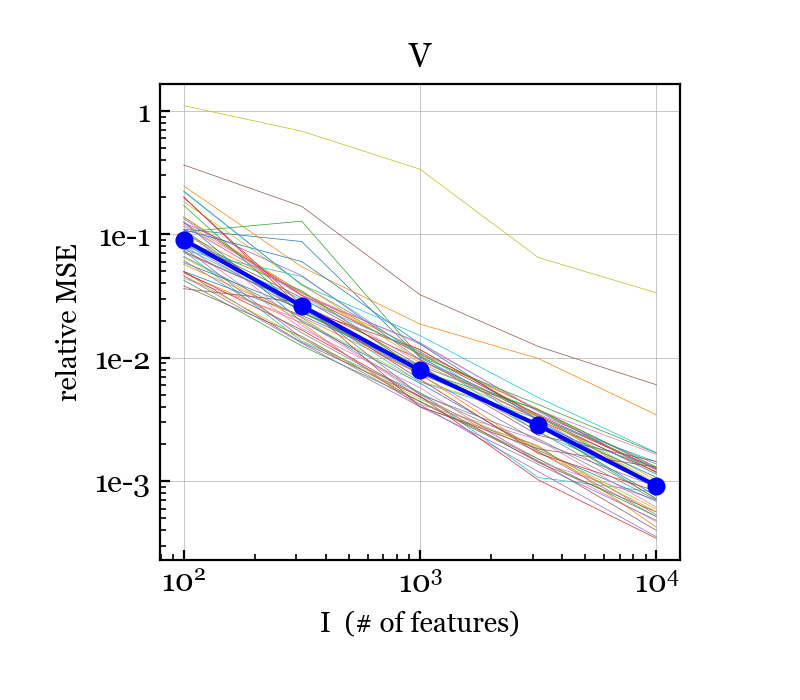}
  \includegraphics[trim=0.6cm 0 1.3cm 0, clip, height=0.2\textheight]{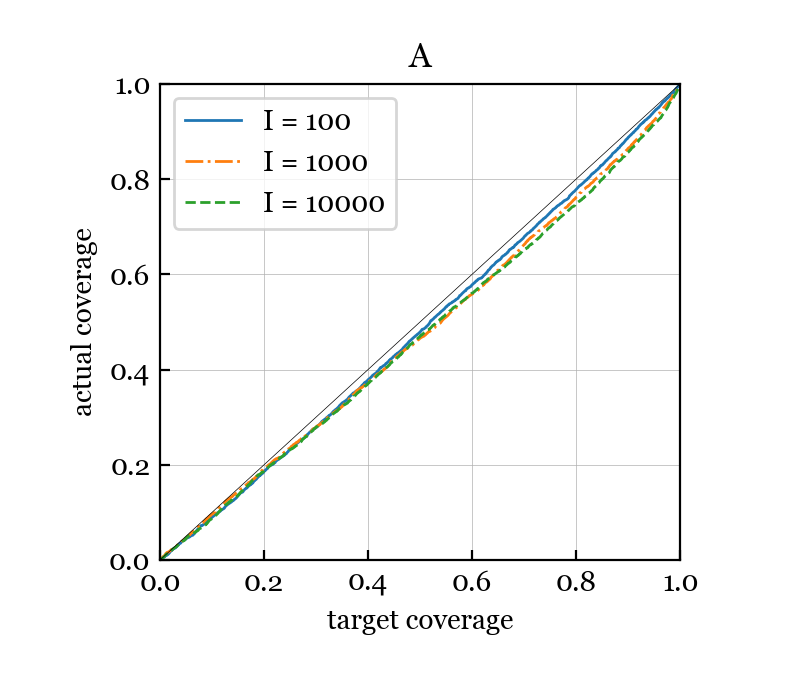}
  \includegraphics[trim=1.3cm 0 1.3cm 0, clip, height=0.2\textheight]{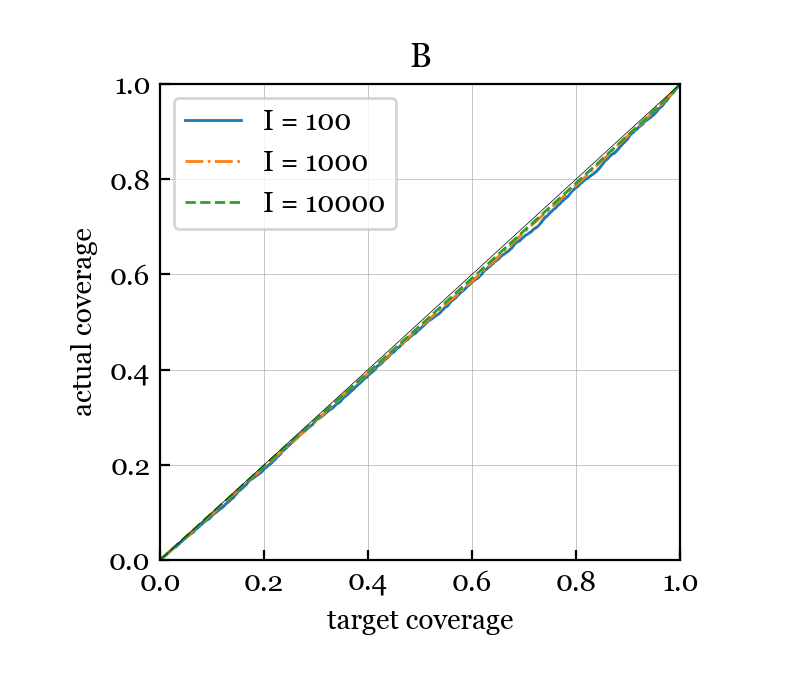}
  \includegraphics[trim=1.3cm 0 1.3cm 0, clip, height=0.2\textheight]{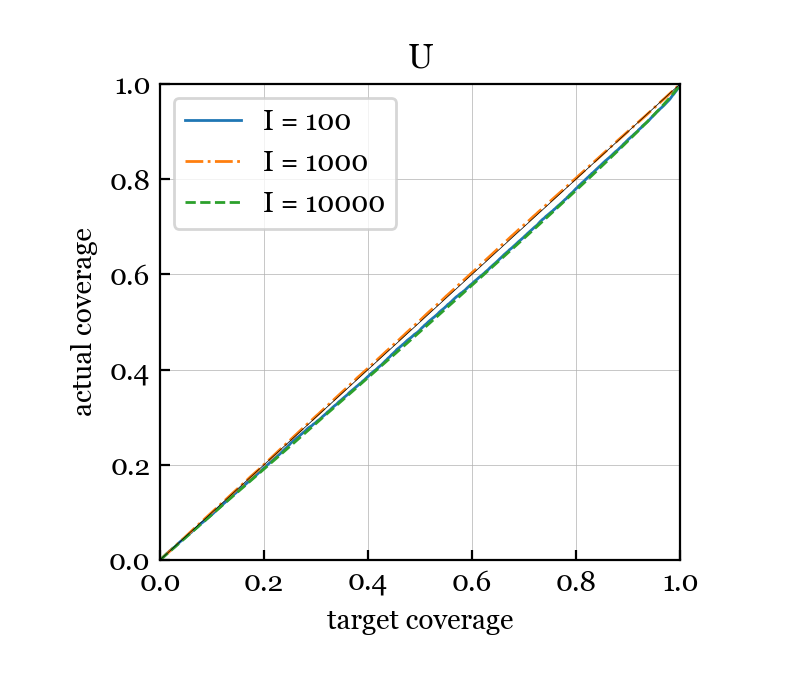}
  \includegraphics[trim=1.3cm 0 1.3cm 0, clip, height=0.2\textheight]{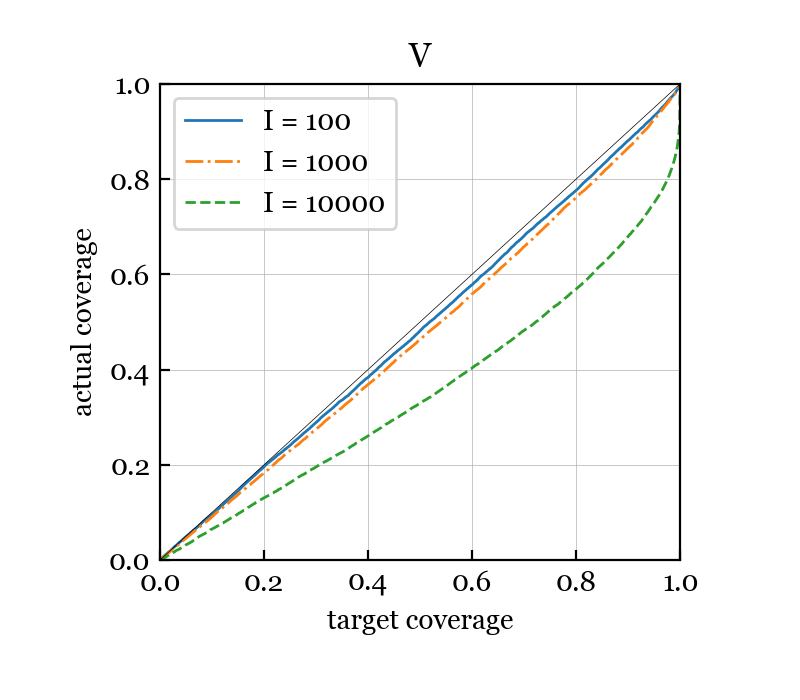}
  \caption{Effect of varying latent dimension: $K = 1$, $L = 1$, $M = 5$, \tt{NB/Normal/Normal}.}
  \label{figure:K=1-L=1-M=5}
\end{figure}

\subsubsection*{Varying the distribution of covariates and parameters}
As usual in regression, since the model is defined conditionally on the covariates, we would expect the results to be robust to the distribution of the covariates.
To verify that this is indeed the case, we run the simulations with the following simulation schemes as defined above:
(a) \tt{NB/Bernoulli/Normal} (Bernoulli-distributed covariates) and
(b) \tt{NB/Gamma/Normal} (Gamma-distributed covariates),
with $J = 100$, $K = 4$, $L = 2$, and $M = 3$;
see Figures~\ref{figure:Bernoulli-covariates} and \ref{figure:Gamma-covariates}.
Similarly, although we have a prior on the parameters, we would not expect the results to be highly sensitive to the distribution of the true parameters.
To check this, we also run the simulations with the \tt{NB/Normal/Gamma} scheme (Gamma-distributed true parameters)
with $J = 100$, $K = 4$, $L = 2$, and $M = 3$; see Figure~\ref{figure:Gamma-parameters}.
The results are nearly identical to the case of normally distributed covariates and true parameters.

\begin{figure}
  \centering
  \includegraphics[trim=0.6cm 0 1.5cm 0, clip, height=0.2\textheight]{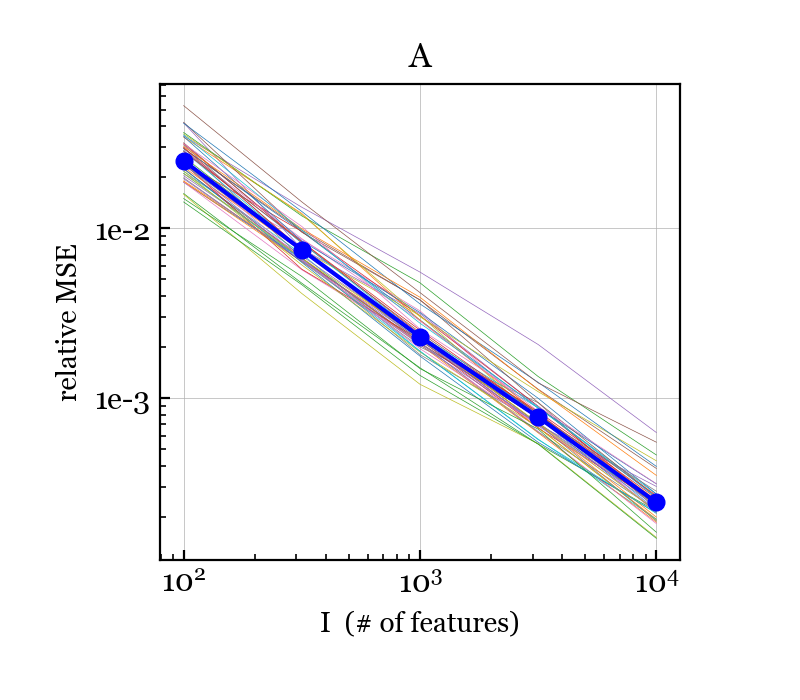}
  \includegraphics[trim=1.1cm 0 1.5cm 0, clip, height=0.2\textheight]{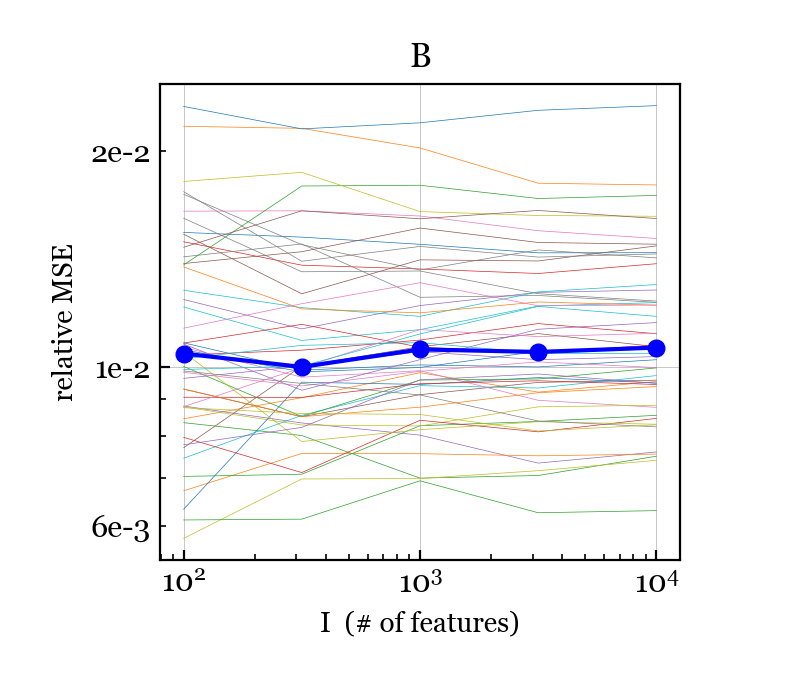}
  \includegraphics[trim=1.1cm 0 1.5cm 0, clip, height=0.2\textheight]{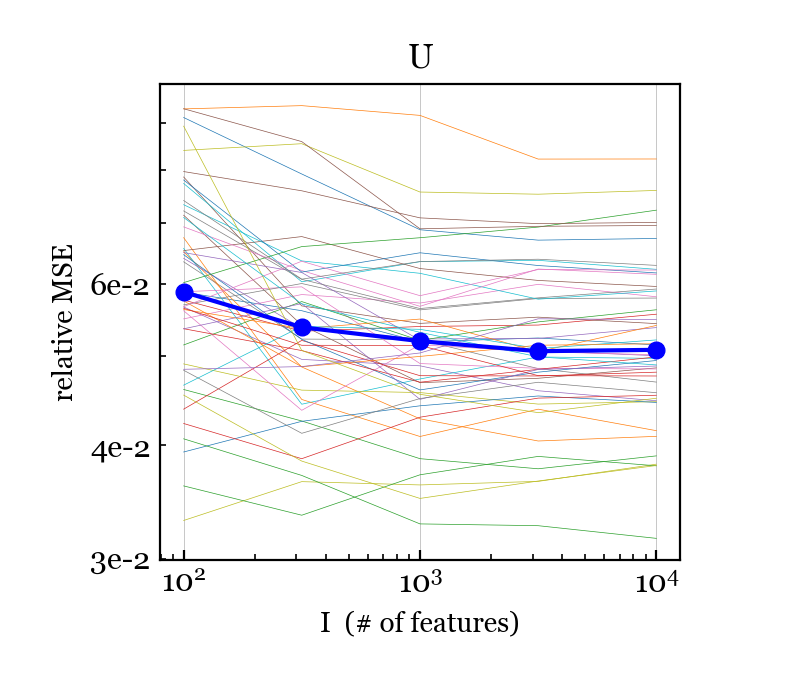}
  \includegraphics[trim=1.1cm 0 1.3cm 0, clip, height=0.2\textheight]{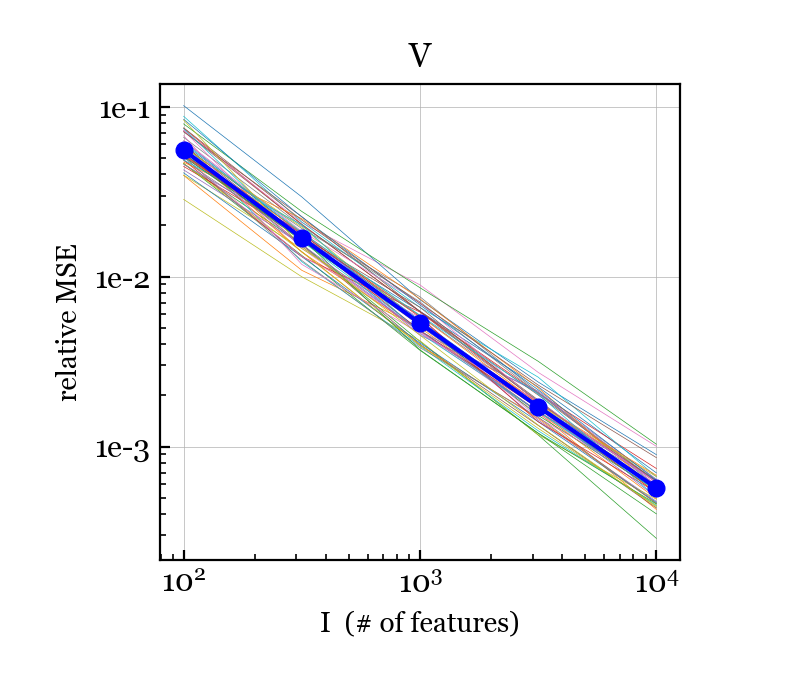}\\
  \includegraphics[trim=0.6cm 0 1.3cm 0, clip, height=0.2\textheight]{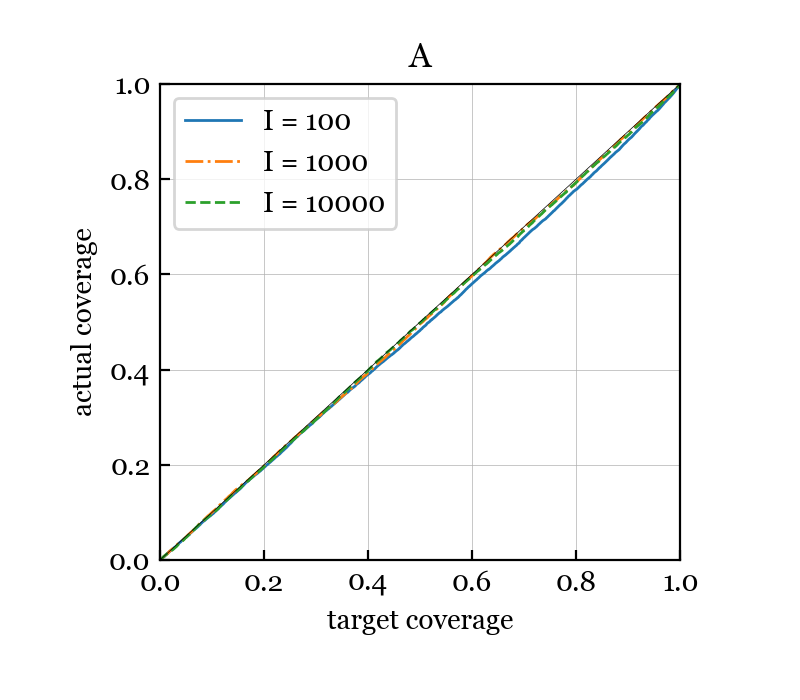}
  \includegraphics[trim=1.3cm 0 1.3cm 0, clip, height=0.2\textheight]{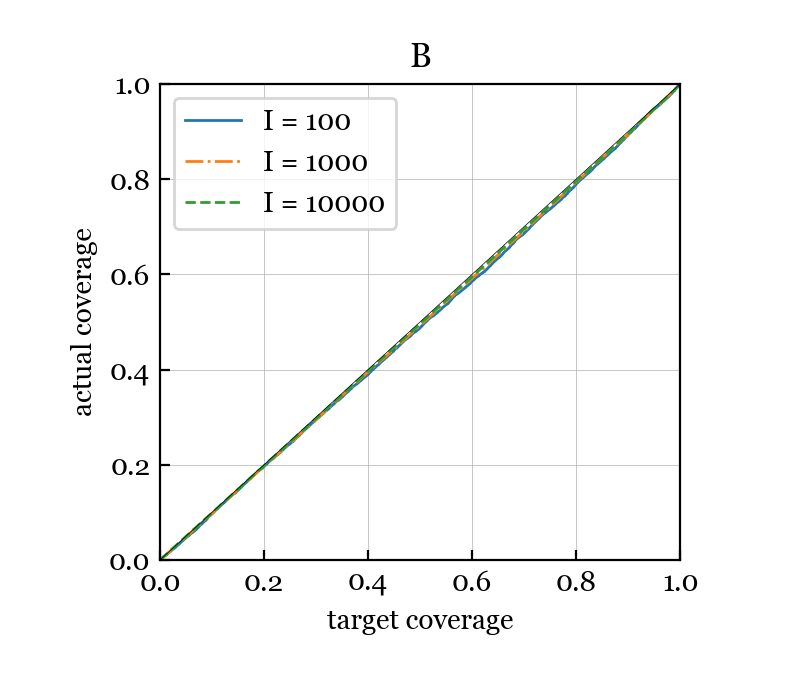}
  \includegraphics[trim=1.3cm 0 1.3cm 0, clip, height=0.2\textheight]{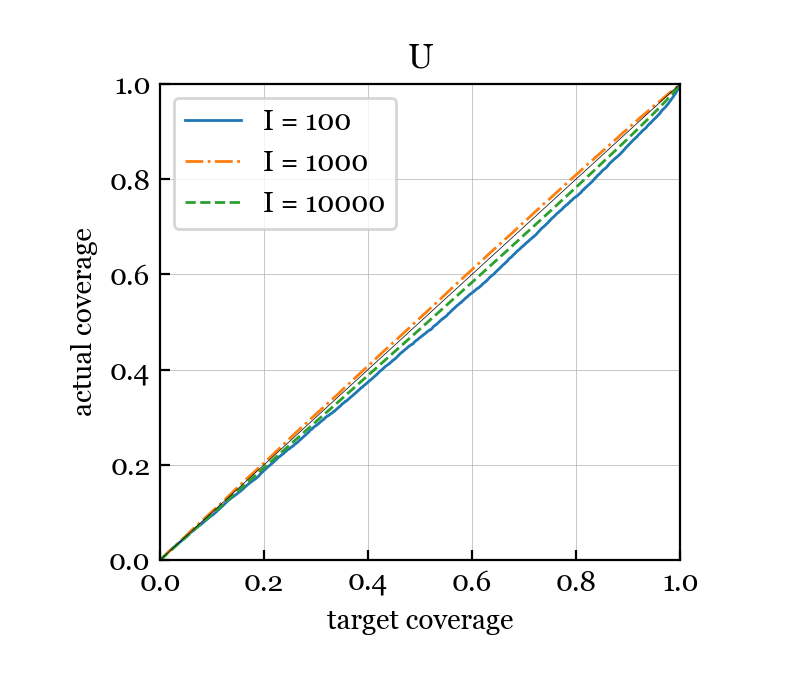}
  \includegraphics[trim=1.3cm 0 1.3cm 0, clip, height=0.2\textheight]{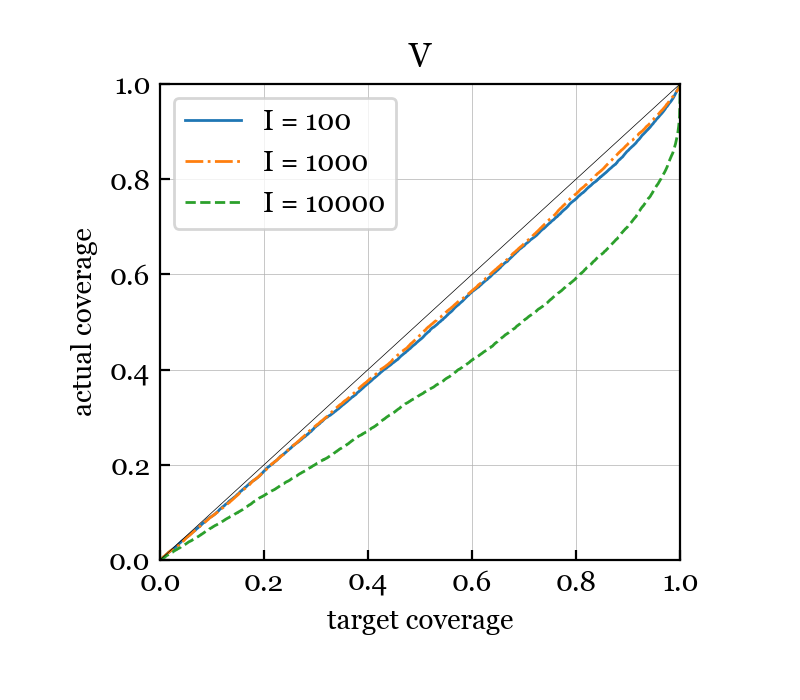}
  \caption{Varying the covariate distribution: $K = 4$, $L = 2$, $M = 3$, \tt{NB/Bernoulli/Normal}.}
  \label{figure:Bernoulli-covariates}
\end{figure}

\begin{figure}
  \centering
  \includegraphics[trim=0.6cm 0 1.5cm 0, clip, height=0.2\textheight]{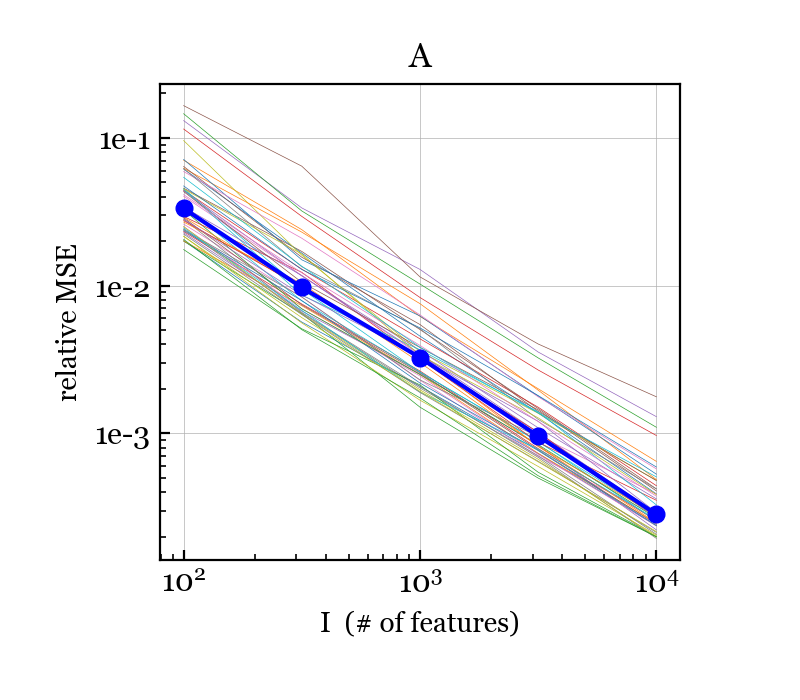}
  \includegraphics[trim=1.1cm 0 1.5cm 0, clip, height=0.2\textheight]{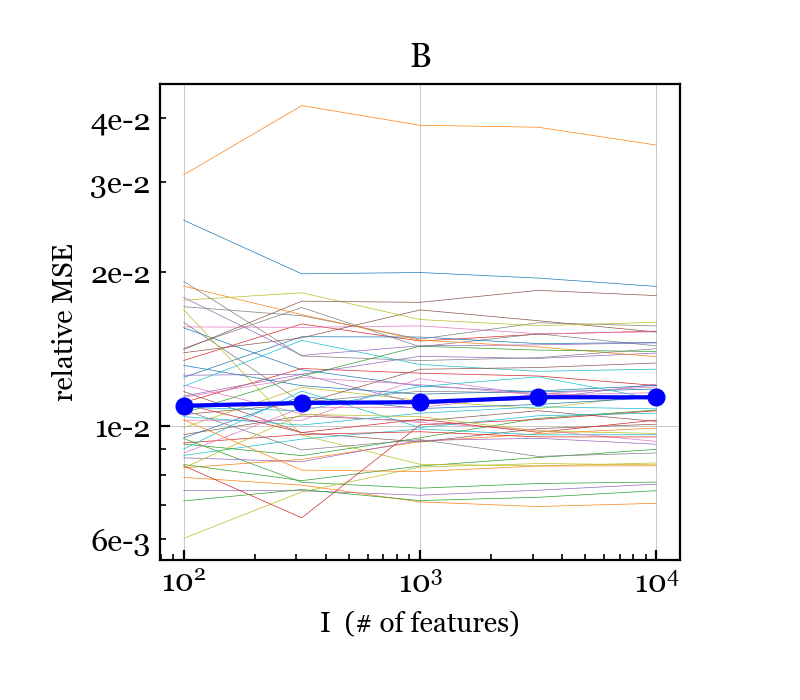}
  \includegraphics[trim=1.1cm 0 1.5cm 0, clip, height=0.2\textheight]{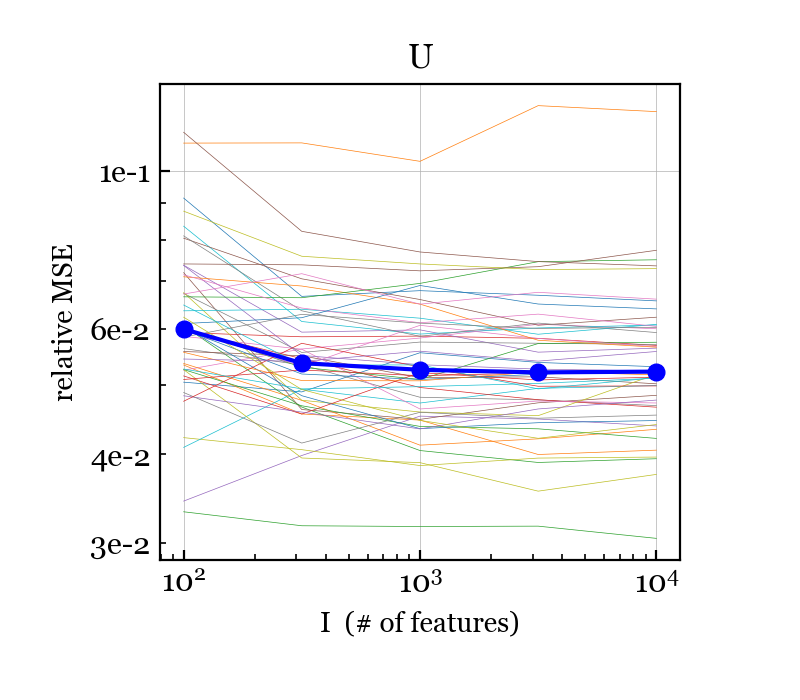}
  \includegraphics[trim=1.1cm 0 1.3cm 0, clip, height=0.2\textheight]{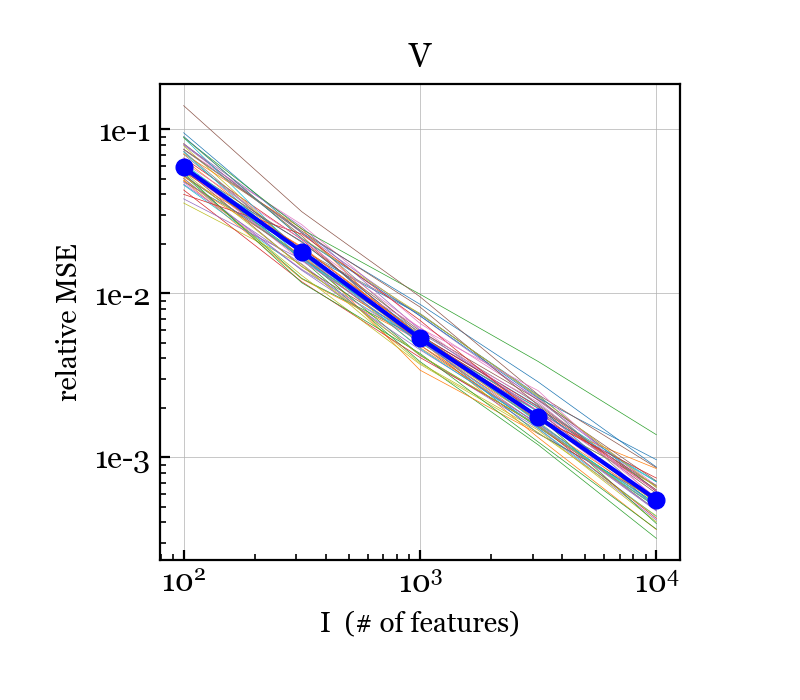}\\
  \includegraphics[trim=0.6cm 0 1.3cm 0, clip, height=0.2\textheight]{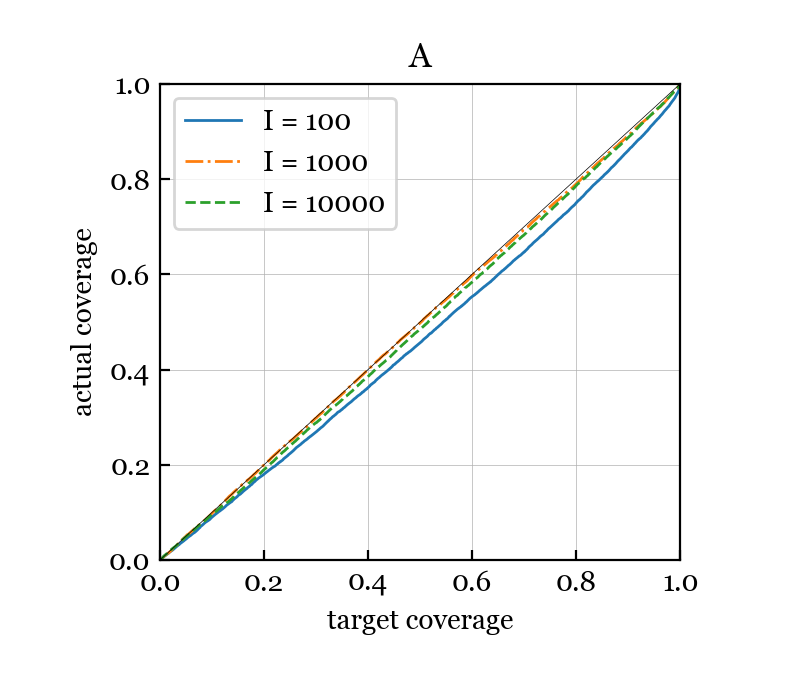}
  \includegraphics[trim=1.3cm 0 1.3cm 0, clip, height=0.2\textheight]{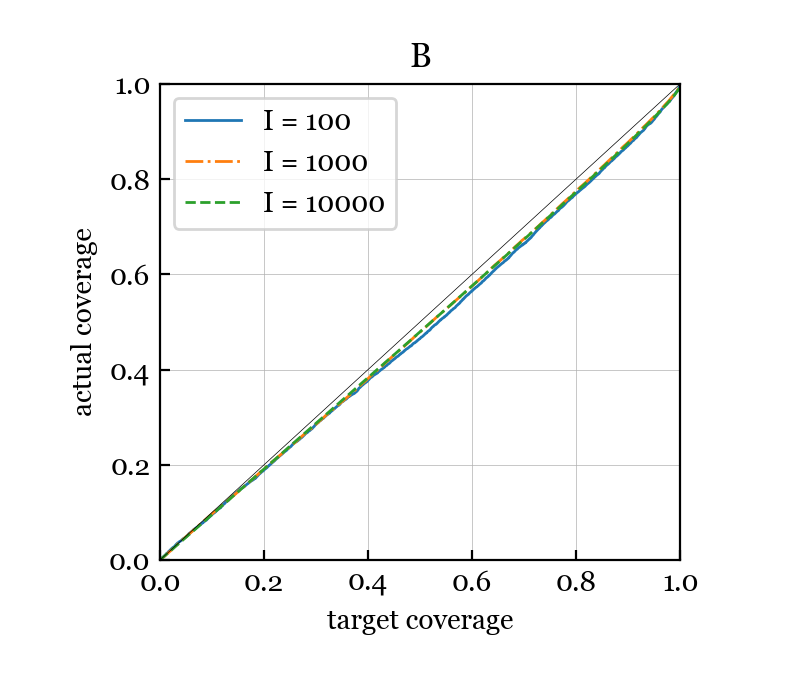}
  \includegraphics[trim=1.3cm 0 1.3cm 0, clip, height=0.2\textheight]{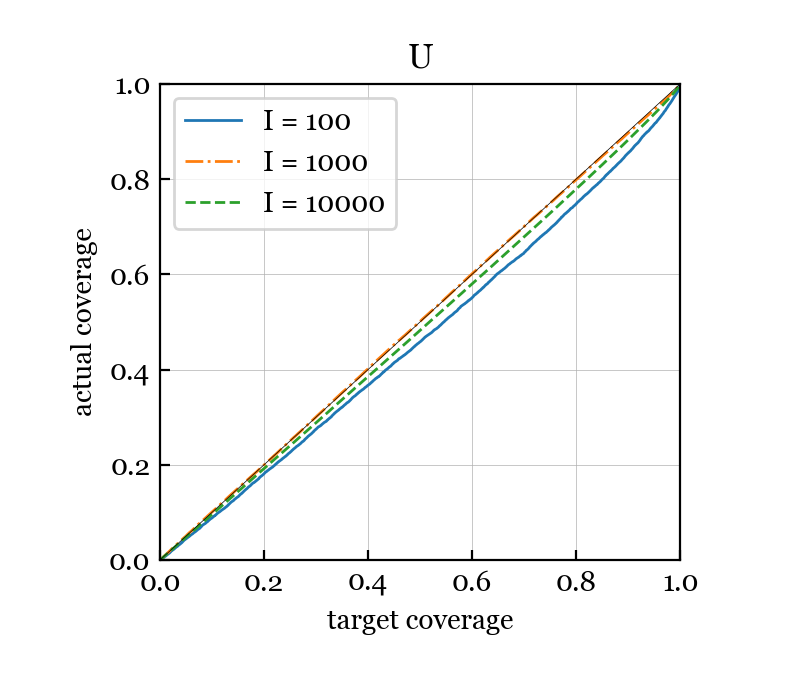}
  \includegraphics[trim=1.3cm 0 1.3cm 0, clip, height=0.2\textheight]{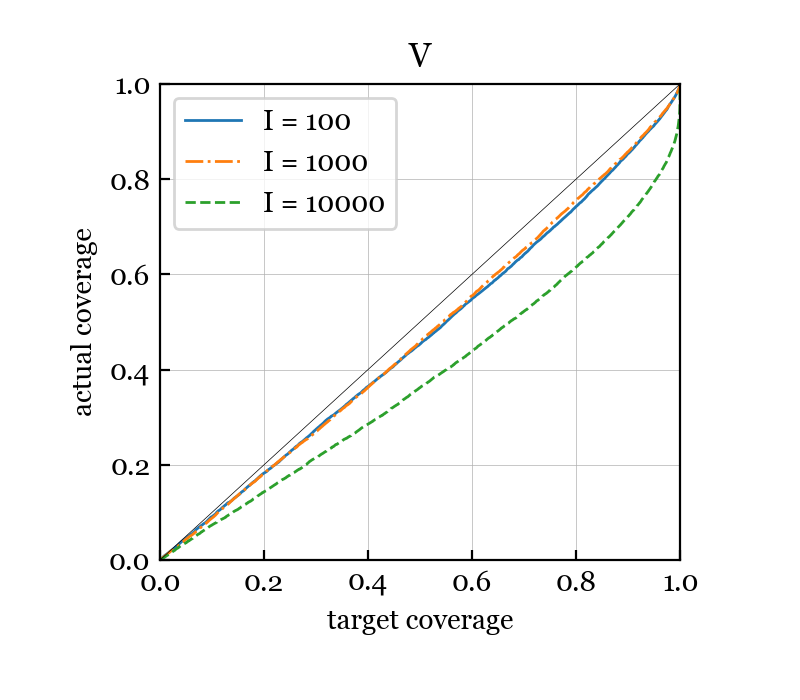}
  \caption{Varying the covariate distribution: $K = 4$, $L = 2$, $M = 3$, \tt{NB/Gamma/Normal}.}
  \label{figure:Gamma-covariates}
\end{figure}

\begin{figure}
  \centering
  \includegraphics[trim=0.6cm 0 1.5cm 0, clip, height=0.2\textheight]{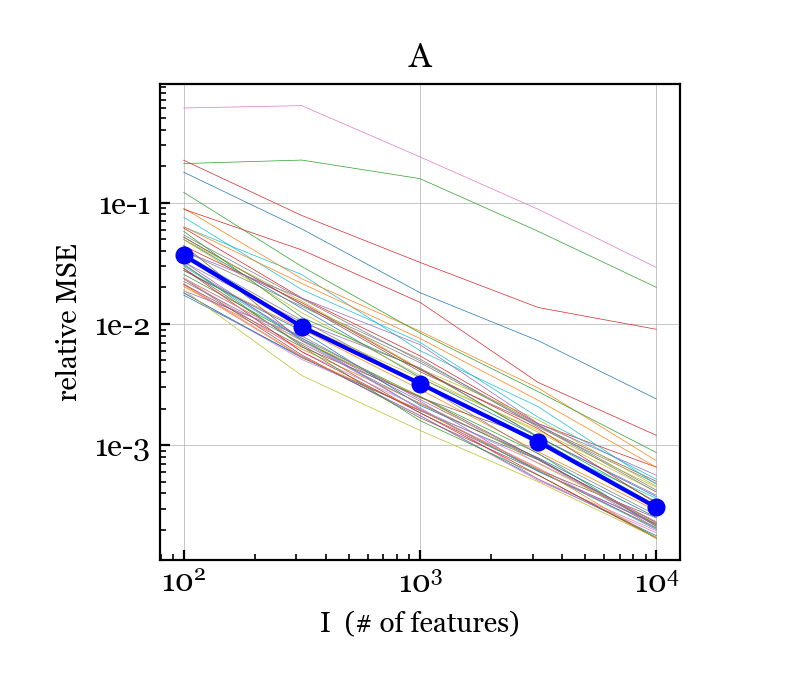}
  \includegraphics[trim=1.1cm 0 1.5cm 0, clip, height=0.2\textheight]{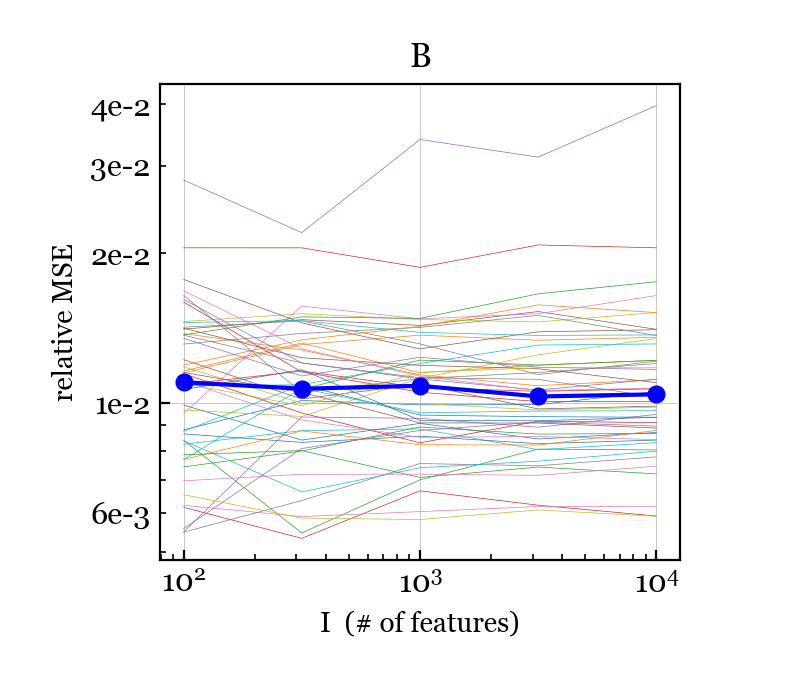}
  \includegraphics[trim=1.1cm 0 1.5cm 0, clip, height=0.2\textheight]{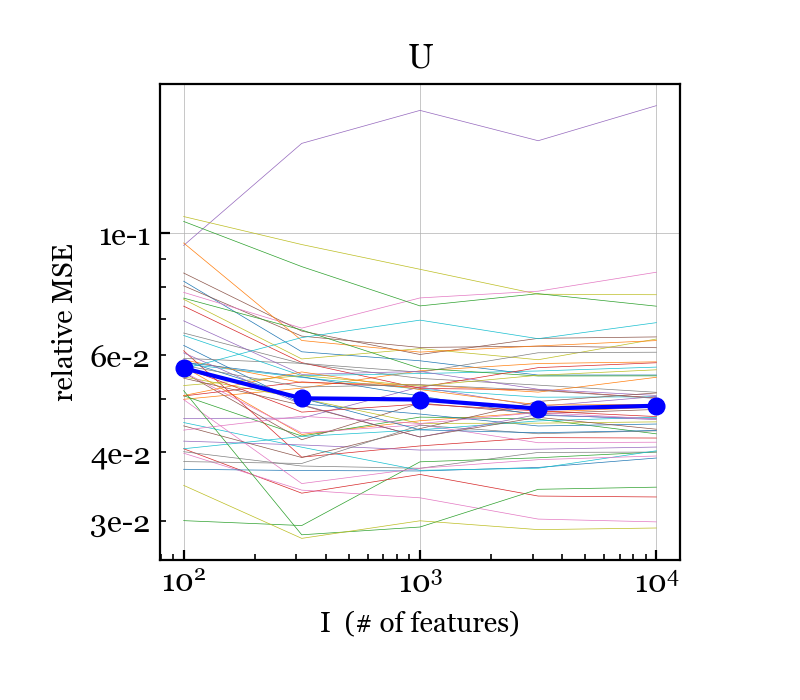}
  \includegraphics[trim=1.1cm 0 1.3cm 0, clip, height=0.2\textheight]{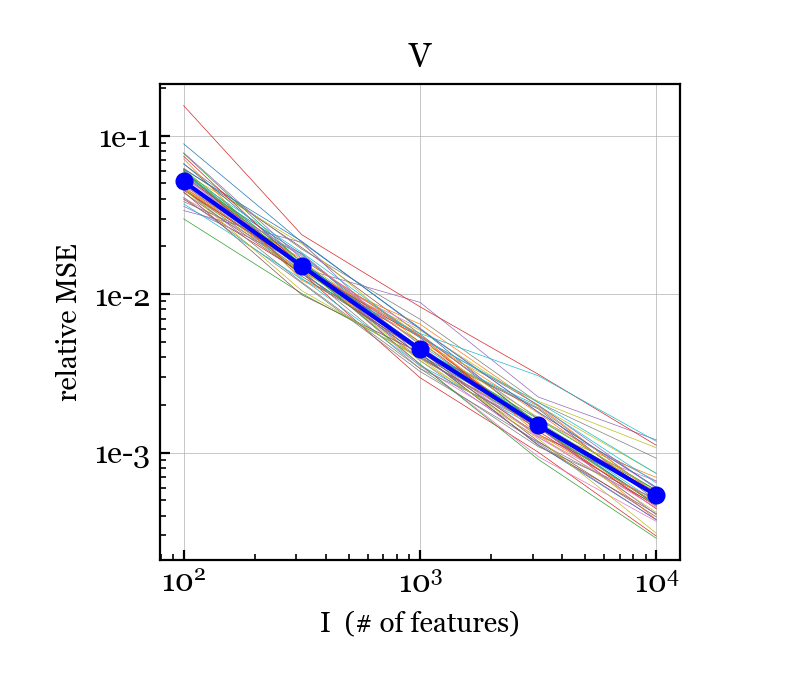}\\
  \includegraphics[trim=0.6cm 0 1.3cm 0, clip, height=0.2\textheight]{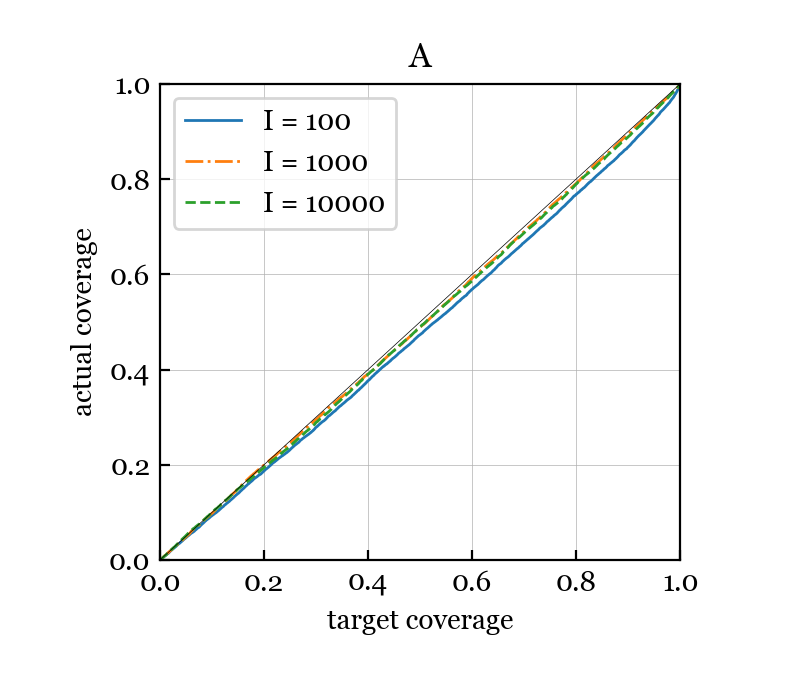}
  \includegraphics[trim=1.3cm 0 1.3cm 0, clip, height=0.2\textheight]{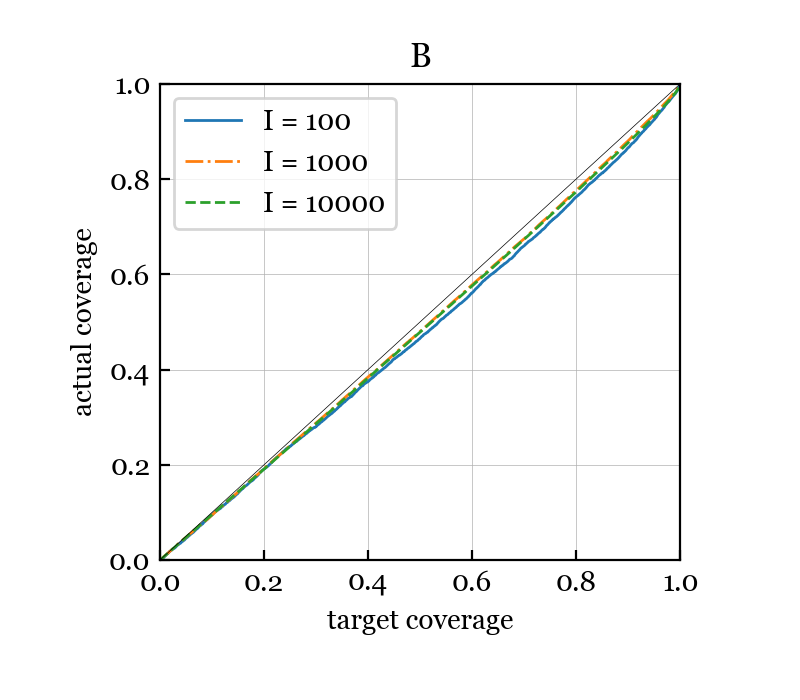}
  \includegraphics[trim=1.3cm 0 1.3cm 0, clip, height=0.2\textheight]{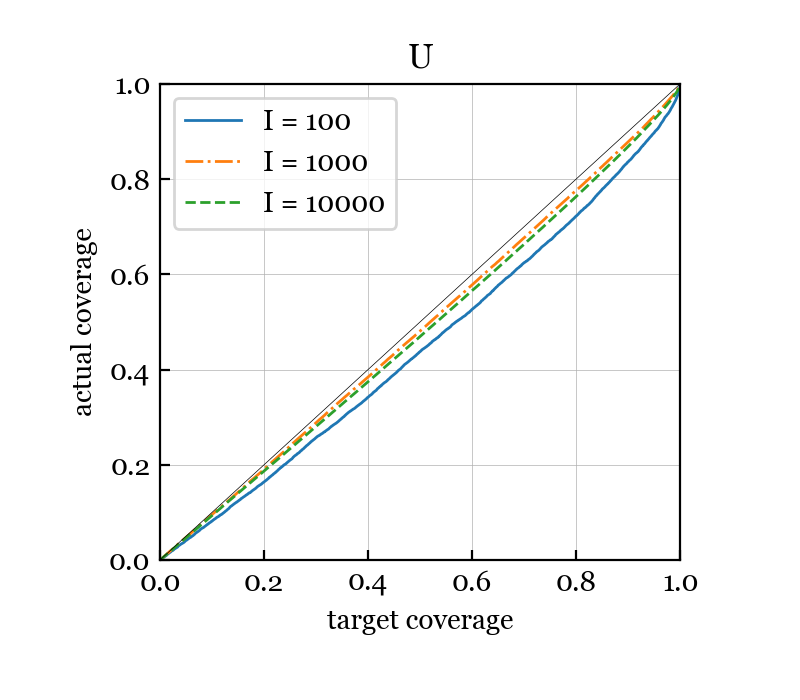}
  \includegraphics[trim=1.3cm 0 1.3cm 0, clip, height=0.2\textheight]{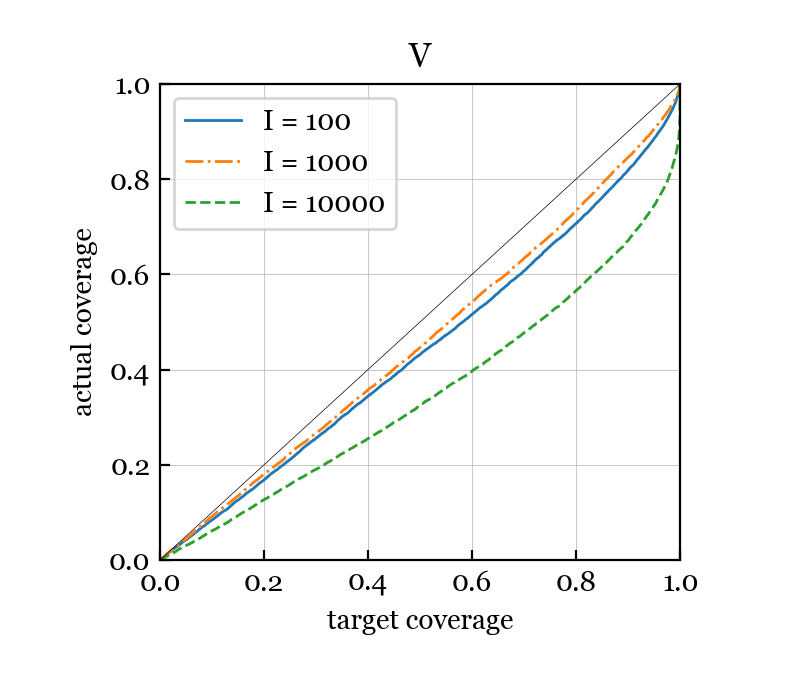}
  \caption{Varying the parameter distribution: $K = 4$, $L = 2$, $M = 3$, \tt{NB/Normal/Gamma}.}
  \label{figure:Gamma-parameters}
\end{figure}

\section{Additional application results and details}

% \todo{Consider re-running the VST plot using plotPCA with \texttt{ntop} equal to all genes, rather than the default which is 500.}

\begin{figure}
  \centering
  \includegraphics[trim=1.8cm 1cm 1.8cm 1cm, clip, height=0.45\textheight]{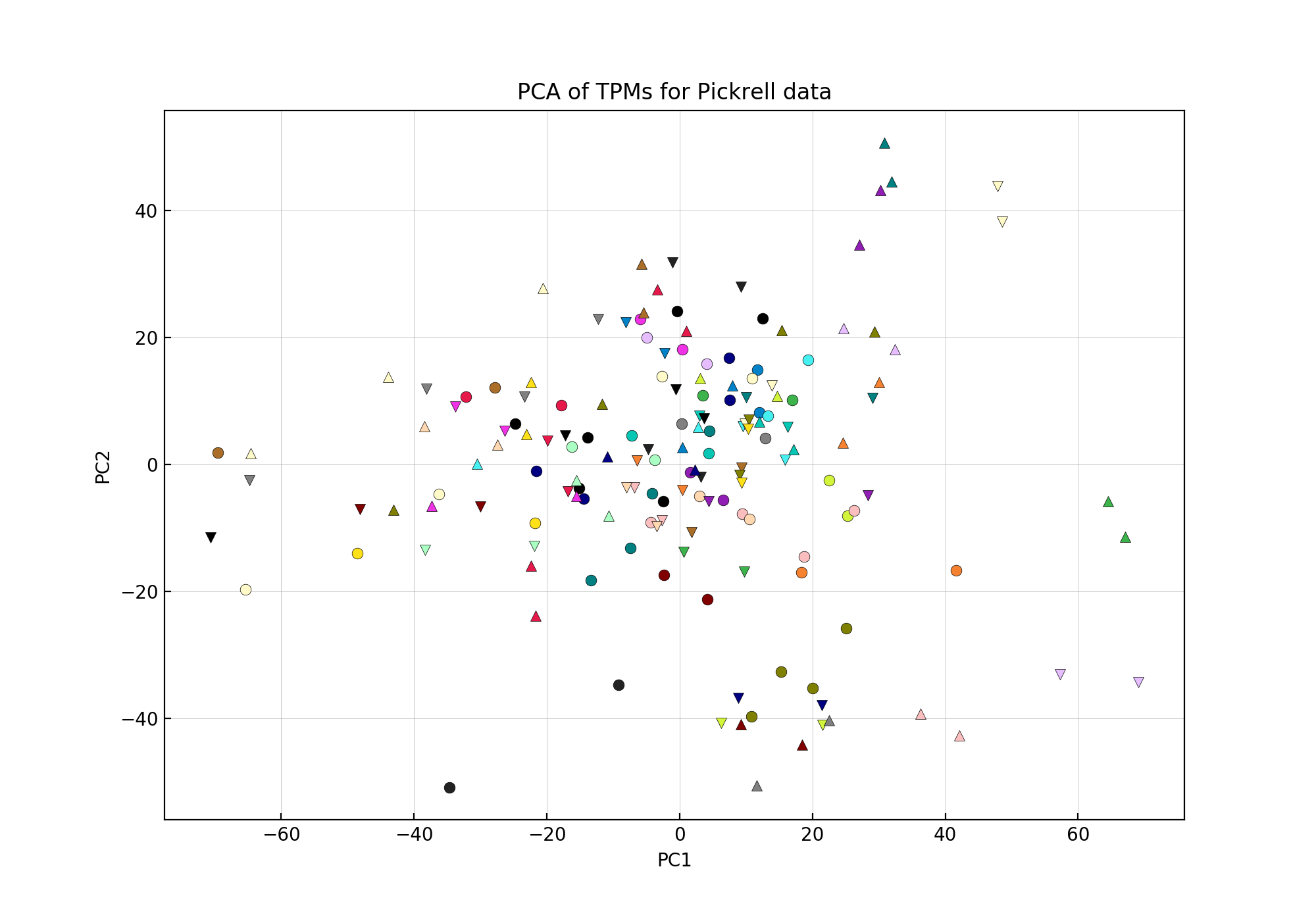}\\
  \includegraphics[trim=1.8cm 1cm 1.8cm 1cm, clip, height=0.45\textheight]{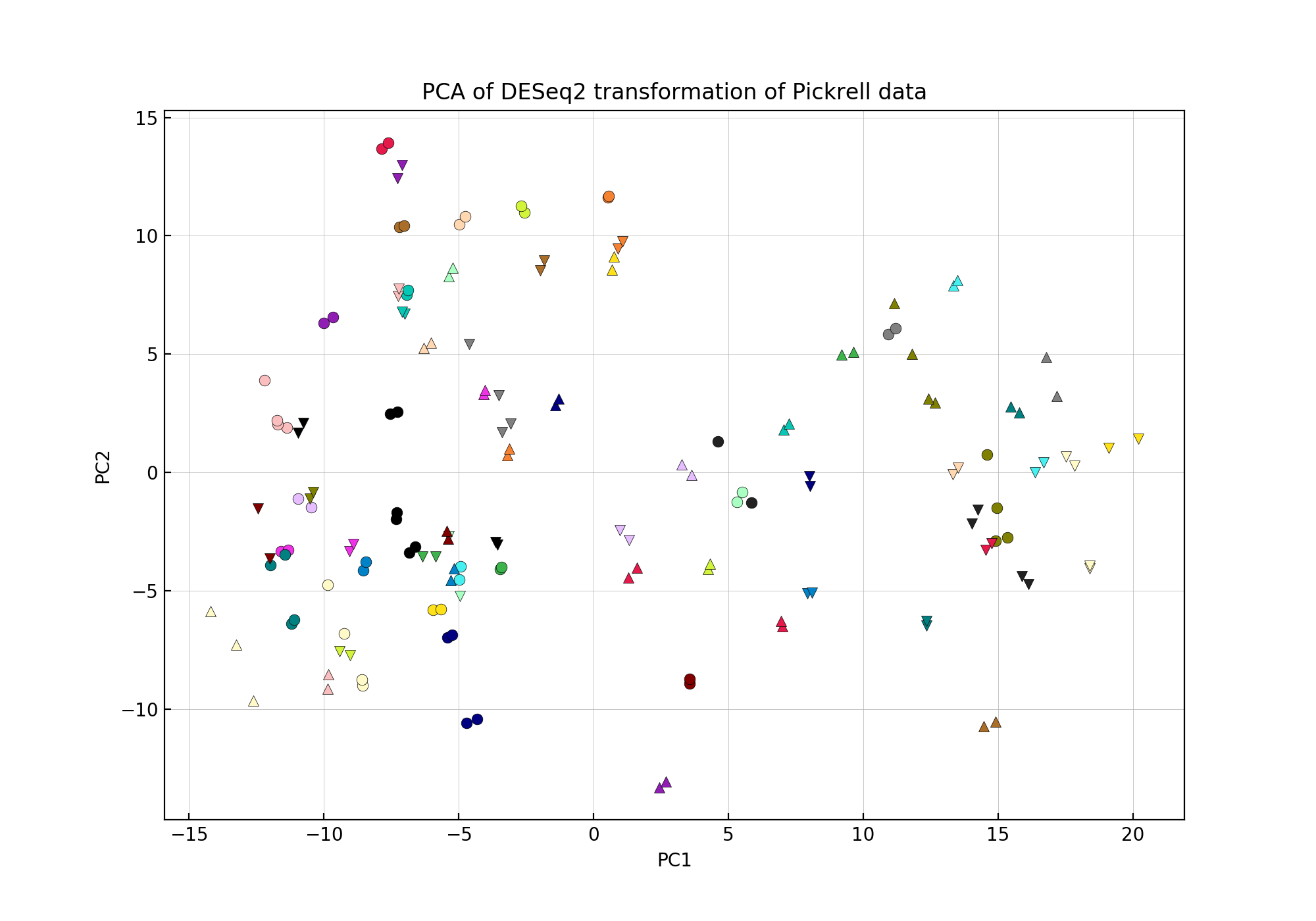}
  \caption{PCA of the Pickrell data using: (top) log-transformed TPMs and (bottom) the VST method in the DESeq2 software package, including the GC bias adjustment from CQN. Compare with the GBM approach in Figure~\ref{figure:pickrell-latent-factors}.}
  \label{figure:pickrell-pca}
\end{figure}

\begin{figure}
  \centering
  \includegraphics[trim=0.5cm 0cm 0cm 0cm, clip, height=0.5\textheight]{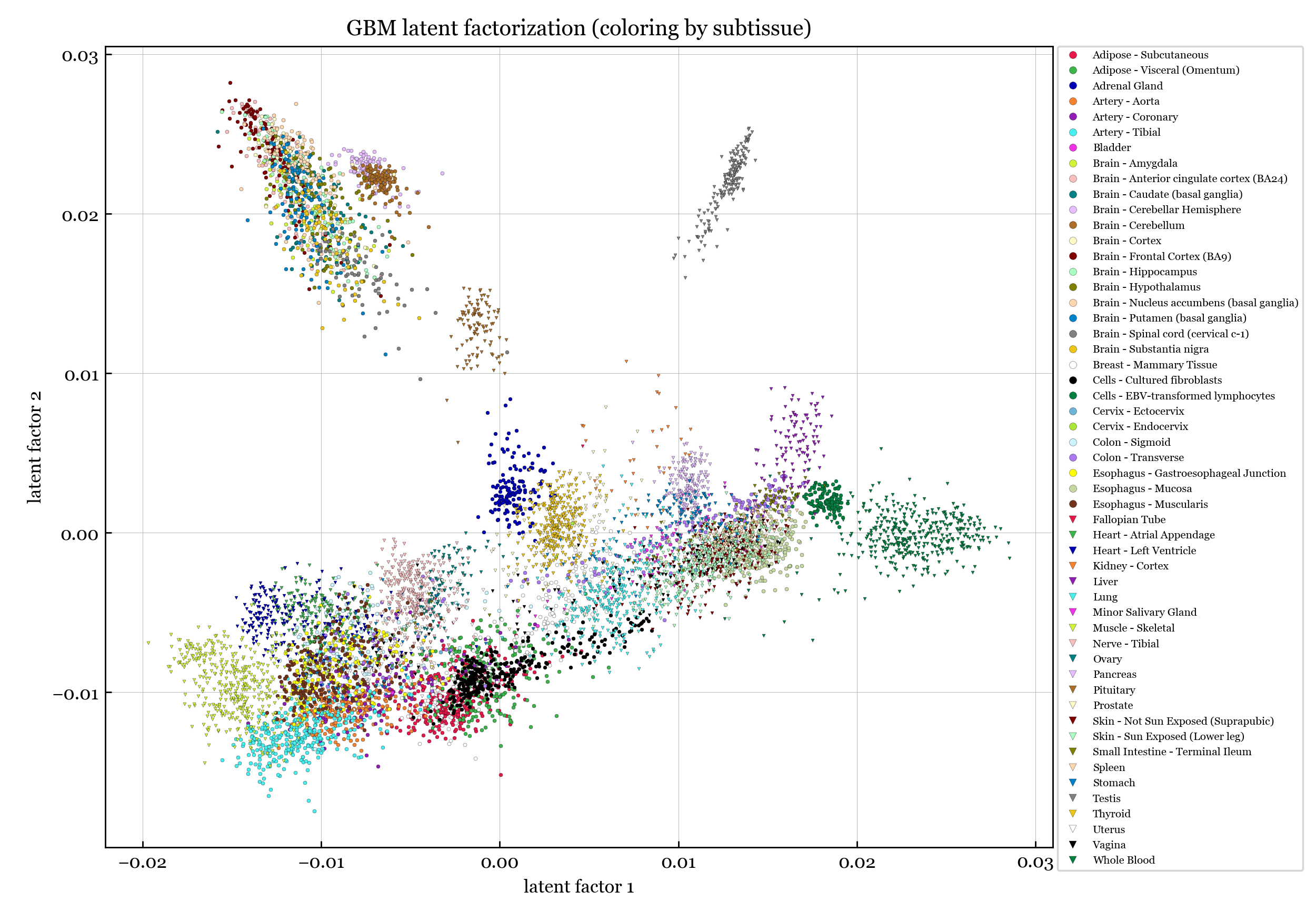}
  \caption{Visualization of GTEx data using NB-GBM latent factors, adjusting for covariates. Each dot represents one of the 8{,}551 samples, and the color indicates the subtissue type.}
  \label{figure:gtex-smtsd}
\end{figure}

\begin{figure}
  \centering
  \includegraphics[trim=0.5cm 0cm 0.5cm 0cm, clip, height=0.5\textheight]{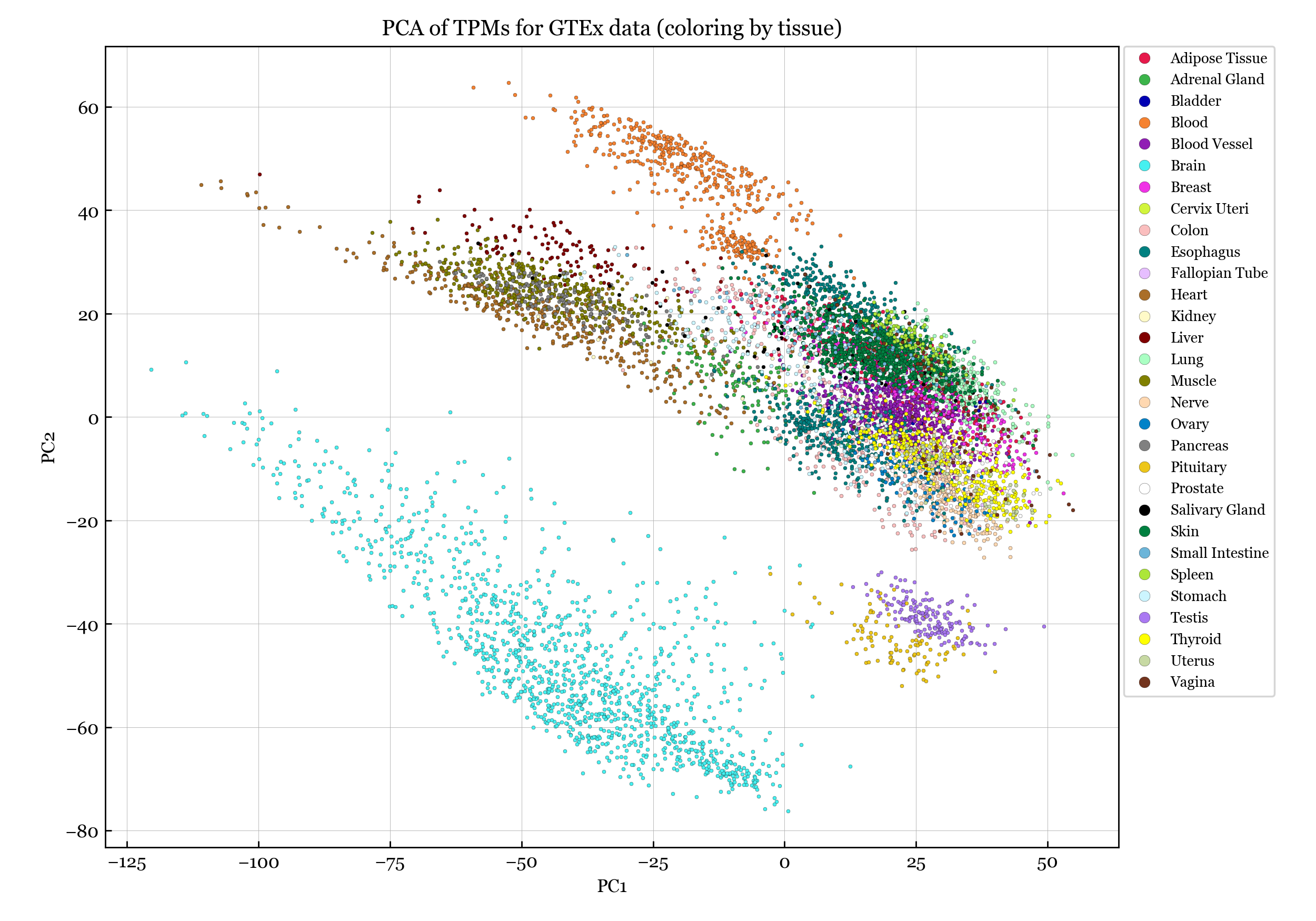}
  \caption{PCA of the GTEx data using log-transformed TPMs, specifically, $\log(\texttt{TPM}_{i j}+1)$. Each dot represents one of the 8{,}551 samples, and the color indicates the tissue type.}
  \label{figure:gtex-pca-tpms}
\end{figure}

\iffalse
\begin{figure}
  \centering
  \includegraphics[trim=0.5cm 0cm 0.5cm 0cm, clip, height=0.5\textheight]{figures/gtex/pca-deseq2-smts-gc=true.png}
  \caption{PCA of the GTEx data using the VST method in the DESeq2 software package, including the GC bias adjustment from CQN. Each dot represents one of the 8{,}551 samples, and the color indicates the tissue type.}
  \label{figure:gtex-pca-tpms}
\end{figure}
\fi

\subsection{Gene expression application -- Details on Section~\ref{section:gene-expression}}

Here we provide additional results and details on the gene expression applications in Section~\ref{section:gene-expression}.
For the Pickrell data,
Figure~\ref{figure:pickrell-pca} shows the PCA plots based on (a) log-transformed TPMs, specifically, $\log(\texttt{TPM}_{i j}+1)$,
and (b) the variance stabilizing transform (VST) method in the DESeq2 software package, using the GC adjustment from CQN.
For the GTEx data, Figure~\ref{figure:gtex-smtsd} shows the latent factors ($v_{j 2}$ versus $v_{j 1}$) as in 
Figure~\ref{figure:gtex-smts}, but coloring the points according to tissue subtype instead tissue type.
We see that the samples tend to fall into clusters according to tissue subtype, further resolving the clustering in Figure~\ref{figure:gtex-smts}.
Figure~\ref{figure:gtex-pca-tpms} shows a PCA plot of the log TPMs of the GTEx data,
which is not nearly as clear as Figure~\ref{figure:gtex-smts} in terms of tissue type clustering.
For the analysis of aging-related genes in the Heart-LV subtissue using the GTEx data (Section~\ref{section:gtex}),
Tables~\ref{table:go-bp-terms} and \ref{table:go-cc-terms} show the top 20 enriched GO terms in the Biological Process and Cellular Component categories.

\begingroup
\renewcommand*{\arraystretch}{1.2}
\begin{table}  %[b]
\scriptsize
\centering
\caption{Top GO terms (Biological Process) for age-related expression in Heart-LV.}%
\begin{tabular}{|c|l|c|c|c|}%
\hline
GO term ID  &  Description  &  Count  &  p-value  &  Benjamini  \\
\hline
GO:0098609  &  cell-cell adhesion  &  48  &  $5.1\text{e-}12$  &  $1.5\text{e-}08$  \\
GO:0006418  &  tRNA aminoacylation for protein translation  &  16  &  $1.4\text{e-}09$  &  $2.0\text{e-}06$  \\
GO:0006099  &  tricarboxylic acid cycle  &  12  &  $3.7\text{e-}07$  &  $3.6\text{e-}04$  \\
GO:1904871  &  positive regulation of protein localization to Cajal body  &  7  &  $1.1\text{e-}06$  &  $6.1\text{e-}04$  \\
GO:1904851  &  positive regulation of establishment of protein localization to telomere  &  7  &  $1.1\text{e-}06$  &  $6.1\text{e-}04$  \\
GO:0006607  &  NLS-bearing protein import into nucleus  &  10  &  $1.3\text{e-}06$  &  $6.2\text{e-}04$  \\
GO:0006914  &  autophagy  &  22  &  $1.8\text{e-}05$  &  $7.6\text{e-}03$  \\
GO:0016192  &  vesicle-mediated transport  &  24  &  $2.6\text{e-}05$  &  $8.3\text{e-}03$  \\
GO:0006511  &  ubiquitin-dependent protein catabolic process  &  24  &  $2.6\text{e-}05$  &  $8.3\text{e-}03$  \\
GO:0006888  &  ER to Golgi vesicle-mediated transport  &  24  &  $3.5\text{e-}05$  &  $1.0\text{e-}02$  \\
GO:0006886  &  intracellular protein transport  &  31  &  $4.3\text{e-}05$  &  $1.1\text{e-}02$  \\
GO:1904874  &  positive regulation of telomerase RNA localization to Cajal body  &  7  &  $8.3\text{e-}05$  &  $2.0\text{e-}02$  \\
GO:0006090  &  pyruvate metabolic process  &  8  &  $9.6\text{e-}05$  &  $2.1\text{e-}02$  \\
GO:0070125  &  mitochondrial translational elongation  &  16  &  $1.1\text{e-}04$  &  $2.2\text{e-}02$  \\
GO:0006446  &  regulation of translational initiation  &  10  &  $1.5\text{e-}04$  &  $2.8\text{e-}02$  \\
GO:0043039  &  tRNA aminoacylation  &  5  &  $1.6\text{e-}04$  &  $3.0\text{e-}02$  \\
GO:0018107  &  peptidyl-threonine phosphorylation  &  10  &  $2.9\text{e-}04$  &  $4.9\text{e-}02$  \\
GO:0000462  &  maturation of SSU-rRNA from tricistronic rRNA transcript  &  9  &  $3.3\text{e-}04$  &  $5.4\text{e-}02$  \\
GO:0006610  &  ribosomal protein import into nucleus  &  5  &  $3.7\text{e-}04$  &  $5.6\text{e-}02$  \\
GO:0016236  &  macroautophagy  &  14  &  $4.0\text{e-}04$  &  $5.9\text{e-}02$  \\
\hline
\end{tabular}
\label{table:go-bp-terms}
\end{table}
\endgroup

\begingroup
\renewcommand*{\arraystretch}{1.2}
\begin{table}  %[b]
\scriptsize
\centering
\caption{Top GO terms (Cellular Component) for age-related expression in Heart-LV.}%
\begin{tabular}{|c|l|c|c|c|}%
\hline
GO term ID  &  Description  &  Count  &  p-value  &  Benjamini  \\
\hline
GO:0016020  &  membrane  &  220  &  $9.8\text{e-}21$  &  $3.7\text{e-}18$  \\
GO:0005739  &  mitochondrion  &  157  &  $1.2\text{e-}20$  &  $3.7\text{e-}18$  \\
GO:0070062  &  extracellular exosome  &  242  &  $4.3\text{e-}16$  &  $9.1\text{e-}14$  \\
GO:0005829  &  cytosol  &  282  &  $1.0\text{e-}15$  &  $1.6\text{e-}13$  \\
GO:0005913  &  cell-cell adherens junction  &  57  &  $9.5\text{e-}15$  &  $1.2\text{e-}12$  \\
GO:0005737  &  cytoplasm  &  380  &  $2.3\text{e-}13$  &  $2.4\text{e-}11$  \\
GO:0043209  &  myelin sheath  &  36  &  $4.7\text{e-}13$  &  $4.2\text{e-}11$  \\
GO:0005759  &  mitochondrial matrix  &  47  &  $5.7\text{e-}09$  &  $4.5\text{e-}07$  \\
GO:0005654  &  nucleoplasm  &  217  &  $1.1\text{e-}08$  &  $7.8\text{e-}07$  \\
GO:0000502  &  proteasome complex  &  18  &  $1.4\text{e-}08$  &  $8.0\text{e-}07$  \\
GO:0005743  &  mitochondrial inner membrane  &  56  &  $1.4\text{e-}08$  &  $8.0\text{e-}07$  \\
GO:0042645  &  mitochondrial nucleoid  &  14  &  $3.5\text{e-}07$  &  $1.8\text{e-}05$  \\
GO:0014704  &  intercalated disc  &  14  &  $8.5\text{e-}07$  &  $4.2\text{e-}05$  \\
GO:0005832  &  chaperonin-containing T-complex  &  7  &  $2.5\text{e-}06$  &  $1.1\text{e-}04$  \\
GO:0005643  &  nuclear pore  &  16  &  $5.2\text{e-}06$  &  $2.2\text{e-}04$  \\
GO:0043231  &  intracellular membrane-bounded organelle  &  55  &  $2.7\text{e-}05$  &  $1.1\text{e-}03$  \\
GO:0002199  &  zona pellucida receptor complex  &  6  &  $2.9\text{e-}05$  &  $1.1\text{e-}03$  \\
GO:0043034  &  costamere  &  8  &  $5.4\text{e-}05$  &  $1.9\text{e-}03$  \\
GO:0043234  &  protein complex  &  42  &  $7.8\text{e-}05$  &  $2.6\text{e-}03$  \\
GO:0045254  &  pyruvate dehydrogenase complex  &  5  &  $1.5\text{e-}04$  &  $4.6\text{e-}03$  \\
\hline
\end{tabular}
\label{table:go-cc-terms}
\end{table}
\endgroup

\subsection{Cancer genomics application -- Details on Section~\ref{section:cancer}}
\label{section:cancer-details}

\textbf{Data acquisition and preprocessing.}
We downloaded BAM files for the 326 CCLE whole-exome samples from the Genomic Data Commons (GDC) Legacy Archive of the National Cancer Institute
(\url{https://portal.gdc.cancer.gov/legacy-archive/}), using the GDC Data Transfer Tool v1.6.0.
Using the PreprocessIntervals tool from the Genome Analysis Toolkit (GATK) v4.1.8.1 (\url{https://github.com/broadinstitute/gatk/}) running on Java v1.8,
we preprocessed the CCLE exome target region interval list to pad the intervals by 250 base pairs on either side (options: \texttt{padding 250}, \texttt{bin-length 0}, \texttt{interval-merging-rule OVERLAPPING\_ONLY}).
% (``Whole exome Agilent 1.1 RefSeq plus 3 boosters'')
% We ran PreprocessIntervals using the b37 human reference genome (\url{http://www.broadinstitute.org/ftp/pub/seq/references/Homo_sapiens_assembly19.fasta}) since the BAM files were created using this reference.
Then, to convert each BAM file to a vector of counts, we counted the number of reads in each target region using the CollectReadCounts tool from GATK
(options: \texttt{interval-merging-rule OVERLAPPING\_ONLY}).
For analysis, we included all target regions in chromosomes 1--22 that have nonzero median across the 326 samples.

\textbf{De-segmenting the training samples.}
We randomly selected half of the samples to use as a training set, and the rest were used as a test set to evaluate performance.
% \todo{Include train/test indicator in a sample info file along with code.}
To be able to treat the training samples as ``pseudo-normal'' (non-cancer) samples, we de-segment them as follows.
% We computed $\rho_{i j}^\textrm{basic}$ as defined in the main text, and 
We first compute a rough estimate of copy ratio, defined as $\rho_{i j} := \tilde{Y}_{i j} / (\alpha_i \beta_j)$ 
where $\alpha_i = \frac{1}{J}\sum_{j=1}^J \tilde{Y}_{i j}$, $\beta_j = \frac{1}{I}\sum_{i=1}^I \tilde{Y}_{i j} / \alpha_i$,
and $\tilde{Y}_{i j} = Y_{i j} + 1/8$; here, $1/8$ is a pseudocount that avoids issues when taking logs.
For each sample $j$, we then run a standard binary segmentation algorithm \citep[Eqn 2]{killick2012optimal} on $\log \rho_{i j}$ to detect changepoints.
% For binary segmentation, we use cost function $\mathcal{C}(x_{1:n}) = -\frac{1}{n}(\sum_{i=1}^n x_i)^2 / \sigma_j^2$
For binary segmentation, we use cost function $\mathcal{C}(x_{1:n}) = -(\frac{1}{\sqrt{n}}\sum_{i=1}^n x_i/\sigma_j)^2$
and penalty $\beta = 1000$
where $\sigma_j^2$ is the sample variance of $(\log \rho_{i j} - \log \rho_{i+1, j})/\sqrt{2}$.
Define $o_{i j}$ to be the average of $\log \rho_{i j}$ over the segment containing region $i$.

We then compute the de-segmented counts $Y_{i j}^\textrm{deseg} := \mathrm{round}\big(\alpha_i \beta_j \exp(\log(\rho_{i j}) - o_{i j})\big)$.
% and treat these as the input training data.
The idea is that this adjusts out the departures from copy neutral (that is, from normal diploid) as inferred by the segmentation algorithm,
to create a panel of pseudo-normals.

% Find the split point that minimizes:
% $$ Q = -\frac{1}{|A_1|} \sum_{i \in A_1} y_i^2 - \frac{1}{|A_2|} \sum_{i \in A_2} y_i^2 $$
% where $A_1$ and $A_2$ are the indices in the two halves.
% Split if 
% $$ Q + \sigma^2\gamma < -\frac{1}{|A|}\sum_{i \in A} y_i^2 $$
% where $\sigma^2$ is the variance and $A$ is the set of indices in the current segment (before splitting).

\textbf{Running the GBM to estimate copy ratios.}
% We run the GBM estimation algorithm on the training data using $o_{i j}$ as a fixed offset to $\mu_{i j}$, 
% which has the effect of adjusting out the segment means.  
% \todo{Modify detailed algorithm listing to include a fixed offset matrix $O$.  It is very minor, see my code.}
We run the GBM estimation algorithm on the de-segmented training data using $M = 5$ latent factors.
To avoid overfitting of the latent factors on the training data, we fix the sample-specific log-dispersion $T$ after
running the initialization procedure -- that is, we do not update $T$ at each iteration of the algorithm.
We use the defaults for all other algorithm settings.
For covariates, we construct $X$ to include $\log(\texttt{length}_i)$, $\texttt{gc}_i$, and $(\texttt{gc}_i - \overline{\texttt{gc}})^2$,
and we use no sample covariates.

% For the test data, we use no offsets and we update $T$ at each iteration as usual.
On the test data, we update $T$ at each iteration as usual.
We construct $X$ to include the same covariates as before, 
along with 5 additional covariates equal to the columns of the $U$ matrix that was estimated on the training data.
We use no sample covariates, and we set $M = 0$ on the test data.
We define the GBM copy ratio estimates as the exponentiated residuals $\rho_{i j}^\textrm{GBM} = \tilde{Y}_{i j}/\hat{\mu}_{i j}$ where $\tilde{Y}_{i j} = Y_{i j} + 1/8$; also see Section~\ref{section:residuals}.

\textbf{Running GATK to estimate copy ratios.}
We run the CreateReadCountPanelOfNormals GATK tool on the de-segmented training data, and to enable comparison with the GBM, 
we set the options to use 5 principal components and include all regions 
(\texttt{number-of-eigensamples 5}, \texttt{minimum-interval-median-percentile 0}, \texttt{maximum-zeros-in-interval-percentage 100}).
To estimate copy ratios for the test data, we run the DenoiseReadCounts GATK tool using the pon file from CreateReadCountPanelOfNormals.
On the test data, we use the original counts (not de-segmented).

\textbf{Performance metrics.}
Suppose $x,w\in\R^n$ and $k\in\{0,2,4,\ldots\}$, where $x_1,\ldots,x_n$ represent noisy measurements of a signal of interest, $w_i$ is a weight for point $x_i$, 
and $k$ is a smoothing bandwidth.
We define the local relative standard error as 
$$ \mathrm{LRSE}(x,w,k) = \sqrt{\frac{1}{n}\sum_{i=1}^n (w_i / \bar{w}_i)^2 (x_i - \bar{x}_i)^2} $$
where $\bar{w}_i$ is the moving average of $w$ using bandwidth $k$,
and $\bar{x}_i$ is the weighted moving average of $x$ using bandwidth $k$ and weights $w$.
More precisely, $\bar{w}_i = \frac{1}{|A_i|} \sum_{j \in A_i} w_j$ and 
\begin{align} \label{equation:weighted-moving-average}
\bar{x}_i = \frac{\sum_{j \in A_i} w_j x_j}{\sum_{j \in A_i} w_j}
\end{align}
where $A_i = \{\max(1,\,i - k/2),\ldots,\min(n,\,i+k/2)\}$.
The idea is that if one is trying to estimate the mean of the signal, then a natural approach would be to use a weighted moving average,
and the LRSE approximates the standard deviation (times $\sqrt{k}$) of this estimator.

We define the weighted median absolute difference as 
$$ \mathrm{WMAD}(x,w,k) = \mathrm{median} \big\{ k\,|\bar{x}_{i+1} - \bar{x}_i| \,:\, i = 1,\ldots,n-1\big\} $$
where $\bar{x}_i$ is the same as above.
The WMAD is similar to the median absolute deviation metric that is frequently used in this application,
but it allows one to account for weights.

To assess copy ratio estimation performance for sample $j$, we take $x_i$ to be the log copy ratio estimate for region $i$ 
and we use a bandwidth of $k = 100$ for both metrics.
We take $w_i = W_{i j}$ for the GBM (following Section~\ref{section:residuals}), and for GATK we take $w_i = 1$ since GATK does not provide weights/precisions.

\textbf{Estimated copy ratios for all test samples.}
Figures~\ref{figure:ccle-heatmap-gatk} and \ref{figure:ccle-heatmap-gbm} show heatmaps of the GATK and GBM copy ratio estimates on the $\log_2$ scale for all 163 samples in the test set.
For visualization, these heatmaps are smoothed using a moving weighted average 
as in Equation~\ref{equation:weighted-moving-average} with $k = 3$.
The GBM estimates are visibly less noisy and appear to infer copy neutral regions (white portions of the heatmap) more accurately than GATK.

\begin{figure}
  \centering
  \includegraphics[trim=1.9cm 2cm 0cm 0cm, clip, width=1\textwidth]{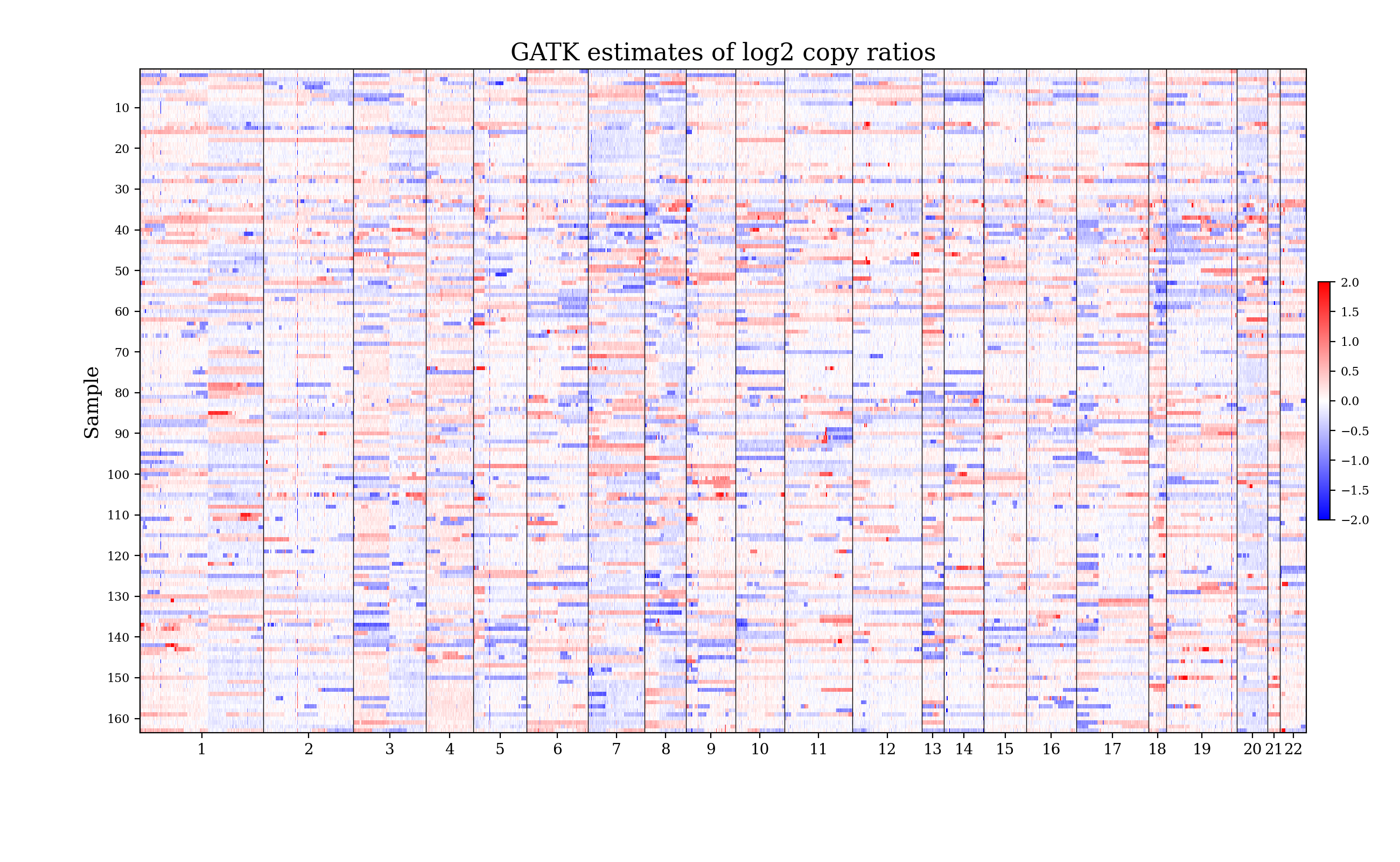}
  \caption{GATK copy ratio estimates (on the log2 scale) for all 163 test samples from the CCLE whole-exome sequencing dataset.
  The x-axis represents genomic position. Red and blue indicate copy gains and losses, respectively.}
  \label{figure:ccle-heatmap-gatk}
\end{figure}

% \newpage

\begin{figure}
  \centering
  \includegraphics[trim=1.9cm 2cm 0cm 0cm, clip, width=1\textwidth]{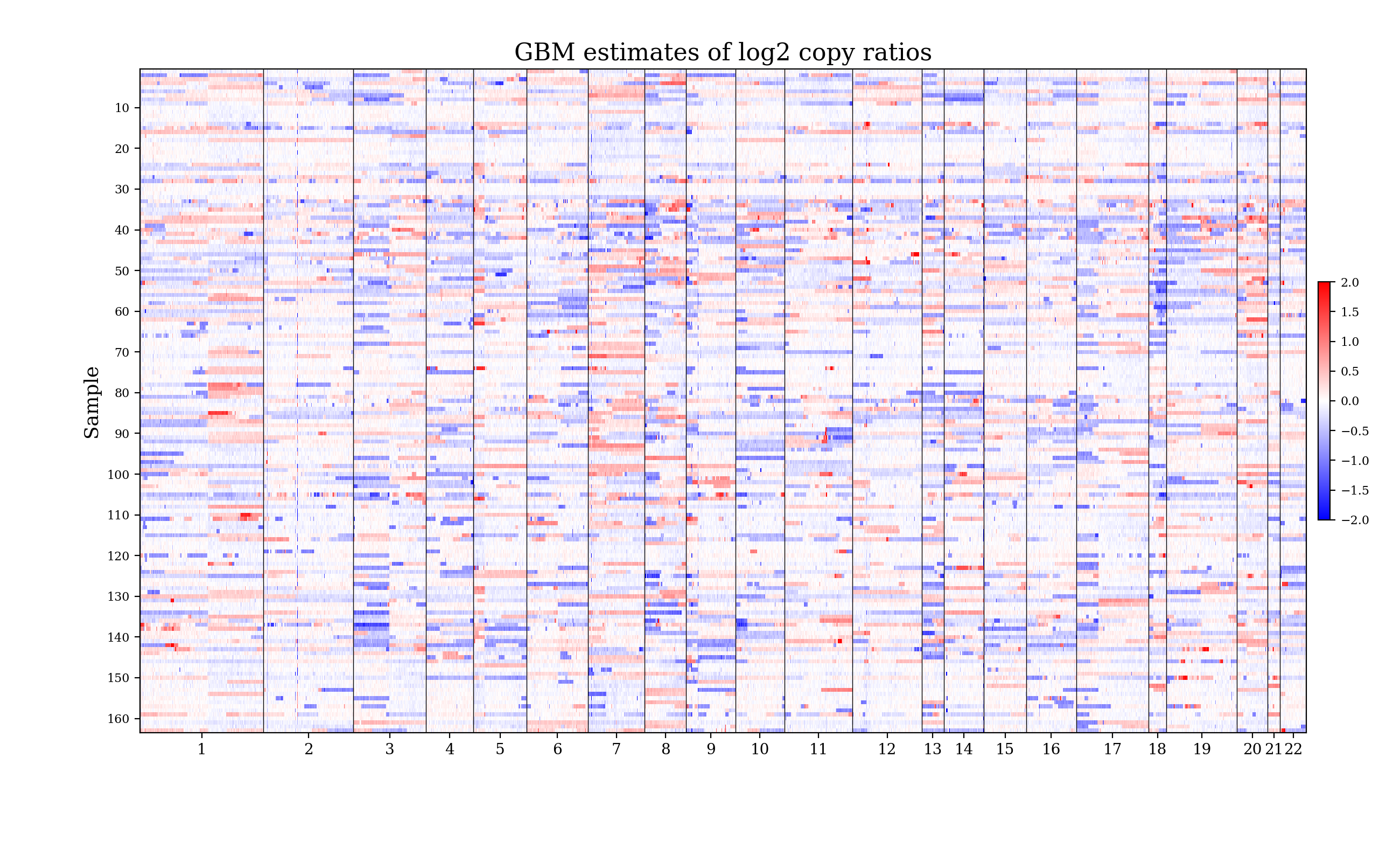}
  \caption{GBM copy ratio estimates (on the log2 scale) for all 163 test samples from the CCLE whole-exome sequencing dataset.
  The x-axis represents genomic position. Red and blue indicate copy gains and losses, respectively.}
  \label{figure:ccle-heatmap-gbm}
\end{figure}

\section{Exponential dispersion families}
\label{section:edfs}

For completeness, we state the basic results for discrete EDFs that we use in this paper.
See \citet{jorgensen1987exponential} or \citet{agresti2015foundations} for reference on this material.
Suppose
\begin{align}
\label{equation:edf-supp}
f(y\mid\theta,r) = \exp(\theta y - r\kappa(\theta)) h(y,r) 
\end{align}
is a p.m.f.\ on $y\in\Z$ for $\theta\in\Theta$ and $r\in\mathcal{R}$,
where $\Theta\subseteq\R$ and $\mathcal{R}\subseteq (0,\infty)$ are convex open sets.
If $Y\sim f(y\mid \theta,r)$, then
\begin{align}
\label{equation:edf-moments}
\E(Y\mid\theta,r) = r\kappa'(\theta) \text{~~ and ~~} \Var(Y\mid\theta,r) = r\kappa''(\theta).
\end{align}
These properties are straightforward to derive from the fact that 
$\log\sum_{y\in\Z} \exp(\theta y) h(y,r) = r\kappa(\theta)$.
For $y,\theta$, and $r$ such that $h(y,r) > 0$, let
\begin{align*}
\L &= \L(y,\theta,r) := \log f(y\mid \theta,r) = \theta y - r\kappa(\theta) + \log h(y,r), \\
\mu &= \mu(\theta,r) := \E(Y\mid\theta,r) = r\kappa'(\theta), \text{ and}\\
\sigma^2 &= \sigma^2(\theta,r) := \Var(Y\mid\theta,r) = r\kappa''(\theta).
\end{align*}
Assume $r$ is not functionally dependent on $\theta$.
Then we have 
\begin{align*}
\frac{\partial\L}{\partial\theta} = y - r\kappa'(\theta) = y - \mu \text{~~ and ~~}
\frac{\partial^2\L}{\partial\theta^2} = -r\kappa''(\theta) = -\sigma^2.
\end{align*}
Assume $\kappa'$ is invertible; this holds as long as $\Var(Y\mid \theta,r) > 0$ for all $\theta$ and $r$.
Then since $\mu = r\kappa'(\theta)$, 
we have $\theta = \kappa'^{-1}(\mu/r)$,
and it follows that $\partial\theta/\partial\mu = 1/(r\kappa''(\theta)) = 1/\sigma^2$
by the inverse function theorem.
Similarly, defining $\eta := g(\mu)$, where $g(\cdot)$ is a smooth function such that $g'$ is positive, we have $\partial\mu / \partial\eta = 1/g'(\mu)$. 
Therefore, by the chain rule,
$$
\frac{\partial\L}{\partial\eta} = \frac{\partial\L}{\partial\theta} \frac{\partial\theta}{\partial\mu} \frac{\partial\mu}{\partial\eta}
= (y - r\kappa'(\theta))\frac{1}{r\kappa''(\theta)}\frac{1}{g'(\mu)}
= \frac{y - \mu}{\sigma^2 g'(\mu)},
$$
and
$$\frac{\partial^2\L}{\partial\eta^2} = 
\Big(-\frac{\partial\mu}{\partial\eta}\Big)\frac{1}{\sigma^2 g'(\mu)} + (y-\mu)\frac{\partial}{\partial\eta}\Big(\frac{1}{\sigma^2 g'(\mu)}\Big).$$
Thus, if $Y\sim f(y\mid\theta,r)$ and $\L = \log f(Y\mid\theta,r)$, then
\begin{align}
\label{equation:edf-eta-derivatives}
\frac{\partial\L}{\partial\eta} = \frac{Y - \mu}{\sigma^2 g'(\mu)} \text{ ~~~and~~~ }
\E\Big(-\frac{\partial^2\L}{\partial\eta^2}\Big) = \frac{1}{\sigma^2 g'(\mu)^2}.
\end{align}
If $\eta$ depends on some parameters $\alpha$ and $\beta$ (and $r$ does not), then
$$ \frac{\partial\L}{\partial\alpha} = \frac{\partial\L}{\partial\eta} \frac{\partial\eta}{\partial\alpha}
\text{ ~~~and~~~ }
\frac{\partial^2\L}{\partial\alpha\partial\beta} = \frac{\partial^2\L}{\partial\eta^2} \frac{\partial\eta}{\partial\alpha}\frac{\partial\eta}{\partial\beta}
+ \frac{\partial\L}{\partial\eta} \frac{\partial^2\eta}{\partial\alpha\partial\beta}. $$
Therefore, by Equation~\ref{equation:edf-eta-derivatives},
\begin{align}
\frac{\partial\L}{\partial\alpha} &= \frac{Y - \mu}{\sigma^2 g'(\mu)} \,\frac{\partial\eta}{\partial\alpha} \label{equation:edf-dalpha} \\
\E\Big(-\frac{\partial^2\L}{\partial\alpha\partial\beta}\Big) &= \frac{1}{\sigma^2 g'(\mu)^2} \,\frac{\partial\eta}{\partial\alpha}\frac{\partial\eta}{\partial\beta} \label{equation:edf-ddalpha} 
\end{align}
since $\partial^2\eta/\partial\alpha\partial\beta$ does not depend on $Y$ and $\E(\partial\L/\partial\eta) = 0$.

\section{Gradient and Fisher information for EDF-GBMs}
\label{section:gbm-derivatives}

% In this section, we derive the gradient and Fisher information matrix for a discrete EDF-GBM.
In Section~\ref{section:gbm-component-derivatives}, we derive the gradient and Fisher information matrix with respect to each of 
$A$, $B$, $C$, $D$, $U$, and $V$ individually, for a discrete EDF-GBM.
In Section~\ref{section:nb-gbm-derivatives}, we specialize these formulas to the case of NB-GBMs, 
and we derive the gradient and observed Fisher information for $S$ and $T$ in this case.
For completeness, in Section~\ref{section:gbm-cross-info} we also provide the cross-terms of the Fisher information between all pairs of components,
although not all of these are needed for our approach.
The formulas contain factors reminiscent of the gradient and Fisher information for a standard GLM,
which take a standard form based on the link function and the mean-variance relationship \citep{agresti2015foundations}.

For our estimation algorithm, we only require the Fisher information matrix with respect to each parameter matrix/vector individually, rather than jointly.
Meanwhile, for inference, we also use the constraint-augmented Fisher information matrix for $U$ and $V$ jointly; see Section~\ref{section:uv-inference}.
To enable comparison with the standard approach of using the full Fisher information matrix, 
in Section~\ref{section:full-info-constraint-blocks} we provide the constraint-augmented Fisher information matrix for all parameters jointly.
% note that full Fisher information matrix is not invertible, and thus, to perform inference using the standard approach the full Fisher information matrix must be augmented to account for the identifiability constraints.

Consider a GBM with discrete EDF outcomes,
that is, suppose $Y_{i j} \sim f(y\mid \theta_{i j}, r_{i j})$ 
% such that $\E(Y_{i j} \mid \theta_{i j},r_{i j})$ satisfies Equation~\ref{equation:model-univariate}.
where $\theta_{i j} := \kappa'^{-1}(\mu_{i j}/r_{i j})$, $\mu_{i j} := g^{-1}(\eta_{i j})$,
and $\eta = \eta(A,B,C,D,U,V)\in\R^{I \times J}$ as in Equation~\ref{equation:eta}.
This makes $\mu_{i j} = \E(Y_{i j} \mid \theta_{i j},r_{i j})$; also, define $\sigma_{i j}^2 := \Var(Y_{i j} \mid \theta_{i j},r_{i j})$.
Let
\begin{align}
\label{equation:log-likelihood}
\L := \sum_{i=1}^I \sum_{j=1}^J \L_{i j} = \sum_{i=1}^I \sum_{j=1}^J \log f(Y_{i j} \mid \theta_{i j},r_{i j})
\end{align}
denote the overall log-likelihood, where
$\L_{i j} := \log f(Y_{i j}\mid \theta_{i j}, r_{i j})$.
By Equations~\ref{equation:edf-dalpha} and \ref{equation:edf-ddalpha},
for any two univariate entries of $A$, $B$, $C$, $D$, $U$, and $V$, say $\alpha$ and $\beta$, we have
\begin{align}
\label{equation:edf-gbm-derivatives}
\frac{\partial\L_{i j}}{\partial\alpha} = e_{i j} \,\frac{\partial\eta_{i j}}{\partial\alpha} \text{ ~~~and~~~ }
\E\Big(-\frac{\partial^2\L_{i j}}{\partial\alpha\partial\beta}\Big) = w_{i j} \,\frac{\partial\eta_{i j}}{\partial\alpha}\frac{\partial\eta_{i j}}{\partial\beta}
\end{align}
where we define the matrices $E\in\R^{I\times J}$ and $W\in\R^{I\times J}$ such that
\begin{align}
\label{equation:edf-gbm-ew}
e_{i j} := \frac{\partial\L_{i j}}{\partial\eta_{i j}} = \frac{Y_{i j} - \mu_{i j}}{\sigma_{i j}^2 g'(\mu_{i j})} \text{ ~~~and~~~ } 
w_{i j} := \E\Big(-\frac{\partial^2\L_{i j}}{\partial\eta_{i j}^2}\Big) = \frac{1}{\sigma_{i j}^2 g'(\mu_{i j})^2}.
\end{align}
The partial derivatives $\partial\eta_{i j}/\partial\alpha$ with respect to each entry of $A$, $B$, $C$, $D$, $U$, and $V$ are:
\begin{align}
\label{equation:eta-derivatives}
\frac{\partial\eta_{i j}}{\partial a_{j' k}} &= x_{i k}\,\I(j=j') \qquad &  \frac{\partial\eta_{i j}}{\partial v_{j' m}} &= u_{i m} d_{m m}\,\I(j=j') \notag\\
\frac{\partial\eta_{i j}}{\partial b_{i' \ell}} &= z_{j \ell}\,\I(i=i')   & \frac{\partial\eta_{i j}}{\partial u_{i' m}} &= v_{j m} d_{m m}\,\I(i=i') \\
\frac{\partial\eta_{i j}}{\partial c_{k \ell}} &= x_{i k} z_{j \ell}   & \frac{\partial\eta_{i j}}{\partial d_{m m}} &= u_{i m} v_{j m}. \notag
\end{align}

% \todo{Maybe try univariate updates for EDF-GBM?  Maybe later... Even if I end up preferring the univariate algorithm, we need the gradients and Hessians anyways for inference --- so may as well type them up.}

\subsection{Gradient and Fisher information for each component}
\label{section:gbm-component-derivatives}

By Equations~\ref{equation:edf-gbm-derivatives} and \ref{equation:eta-derivatives}, we have
\begin{align}
\label{equation:edf-gbm-gradients-univariate}
\frac{\partial\L}{\partial a_{j k}} &= \sum_{i=1}^I x_{i k} e_{i j} \qquad & 
\frac{\partial\L}{\partial v_{j m}} &= \sum_{i=1}^I u_{i m} d_{m m} e_{i j} \notag\\
\frac{\partial\L}{\partial b_{i \ell}} &= \sum_{j=1}^J z_{j \ell} e_{i j}  & 
\frac{\partial\L}{\partial u_{i m}} &= \sum_{j=1}^J v_{j m} d_{m m} e_{i j} \\
\frac{\partial\L}{\partial c_{k \ell}} &= \sum_{i=1}^I \sum_{j=1}^J x_{i k} e_{i j} z_{j \ell} & 
\frac{\partial\L}{\partial d_{m m}} &= \sum_{i=1}^I \sum_{j=1}^J u_{i m} e_{i j} v_{j m}. \notag
\end{align}

To express these equations in matrix notation, we vectorize the parameter matrices as in Section~\ref{section:inference-outline}: $\vec{a} := \mathrm{vec}(A^\T)$, $\vec{b} := \mathrm{vec}(B^\T)$, $\vec{c} := \mathrm{vec}(C)$, $\vec{d} := \diag(D)$, $\vec{u} := \mathrm{vec}(U^\T)$, and $\vec{v} := \mathrm{vec}(V^\T)$.
By Equation~\ref{equation:edf-gbm-gradients-univariate}, the gradient of $\L$ with respect to each vectorized component is then:
\begin{align}
\label{equation:edf-gbm-gradients}
\nabla_{\!\vec{a}}\,\L &= \mathrm{vec}(X^\T E) \qquad &  \nabla_{\!\vec{v}}\,\L &= \mathrm{vec}((U D)^\T E) \notag\\
\nabla_{\!\vec{b}}\,\L &= \mathrm{vec}(Z^\T E^\T)  &   \nabla_{\!\vec{u}}\,\L &= \mathrm{vec}((V D)^\T E^\T) \\
\nabla_{\!\vec{c}}\,\L &= \mathrm{vec}(X^\T E Z) & \nabla_{\!\vec{d}}\,\L &= \diag(U^\T E V). \notag
\end{align}

For each component of the model, the entries of the Fisher information matrices are as follows, 
by Equations~\ref{equation:edf-gbm-derivatives} and \ref{equation:eta-derivatives}:
\begin{align}
\label{equation:edf-gbm-fisher-univariate}
\E\Big(-\frac{\partial^2\L}{\partial a_{j k}\partial a_{j' k'}}\Big) &= \I(j=j')\sum_{i=1}^I w_{i j} x_{i k} x_{i k'} \\
\E\Big(-\frac{\partial^2\L}{\partial b_{i \ell}\partial b_{i' \ell'}}\Big) &= \I(i=i')\sum_{j=1}^J w_{i j} z_{j \ell} z_{j \ell'} \notag\\ 
\E\Big(-\frac{\partial^2\L}{\partial c_{k \ell}\partial c_{k' \ell'}}\Big) &= \sum_{i=1}^I \sum_{j=1}^J w_{i j} x_{i k} x_{i k'} z_{j \ell} z_{j \ell'} \notag\\
\E\Big(-\frac{\partial^2\L}{\partial v_{j m}\partial v_{j' m'}}\Big) &= \I(j=j')\sum_{i=1}^I w_{i j} u_{i m} d_{m m} u_{i m'} d_{m' m'} \notag\\
\E\Big(-\frac{\partial^2\L}{\partial u_{i m}\partial u_{i' m'}}\Big) &= \I(i=i')\sum_{j=1}^J w_{i j} v_{j m} d_{m m} v_{j m'} d_{m' m'} \notag\\
\E\Big(-\frac{\partial^2\L}{\partial d_{m m}\partial d_{m' m'}}\Big) &= \sum_{i=1}^I \sum_{j=1}^J w_{i j} u_{i m} v_{j m} u_{i m'} v_{j m'}. \notag
\end{align}

To express these equations in matrix form, we introduce the following notation.
We write $\Diag(Q_1,\ldots,Q_n)$ to denote the block diagonal matrix with blocks $Q_1,\ldots,Q_n$,
$$\Diag(Q_1,\ldots,Q_n) := 
\renewcommand{\arraystretch}{0.65}
\begin{bmatrix} Q_1 & 0 & \cdots & 0\\ 0 & Q_2 &  & 0 \\ \vdots & & \ddots & \vdots \\ 0 & 0 & \cdots & Q_n \end{bmatrix}, $$ 
and we write $\mathrm{block}(Q_{i j} : i,j\in\{1,\ldots,n\})$ to denote the block matrix with block $i,j$ equal to $Q_{i j}$.
For a matrix $Q\in\R^{m\times n}$, we write $Q_{i *}$ and $Q_{* j}$ to denote the diagonal matrices constructed from the $i$th row and $j$th column, respectively, that is, $Q_{i *} := \Diag(q_{i 1},\ldots,q_{i n})$ and $Q_{* j} := \Diag(q_{1 j},\ldots,q_{m j})$.
Then, by Equation~\ref{equation:edf-gbm-fisher-univariate}, the Fisher information matrices for each component of the model are:
\begin{align}
\label{equation:edf-gbm-fisher}
\E(-\nabla_{\!\vec{a}}^2\,\L) &= \Diag(X^\T W_{* 1} X,\, \ldots,\, X^\T W_{* J} X) \\
\E(-\nabla_{\!\vec{b}}^2\,\L) &= \Diag(Z^\T W_{1 *} Z,\, \ldots,\, Z^\T W_{I *} Z) \notag\\
\E(-\nabla_{\!\vec{c}}^2\,\L) &= \mathrm{block}\Big({\textstyle\sum_{j=1}^J} z_{j\ell}z_{j\ell'}(X^\T W_{* j} X) : \ell,\ell'\in\{1,\ldots,L\}\Big) \notag\\
\E(-\nabla_{\!\vec{v}}^2\,\L) &= \Diag((U D)^\T W_{* 1} (U D),\, \ldots,\, (U D)^\T W_{* J} (U D)) \notag\\
\E(-\nabla_{\!\vec{u}}^2\,\L) &= \Diag((V D)^\T W_{1 *} (V D),\, \ldots,\, (V D)^\T W_{I *} (V D))  \notag\\
\E(-\nabla_{\!\vec{d}}^2\,\L) &= \sum_{j=1}^J (U V_{j *})^\T W_{* j} (U V_{j *}). \notag
\end{align}

\subsection{NB-GBM with log link: Gradient and Fisher information}
\label{section:nb-gbm-derivatives}

The negative binomial distribution has probability mass function
$$\NegBin(y \mid \mu,r) = \frac{\Gamma(y + r)}{\Gamma(y+1)\Gamma(r)} \Big(\frac{\mu}{\mu + r}\Big)^y \Big(\frac{r}{\mu + r}\Big)^r $$
for $y\in\{0,1,2,\ldots\}$, given $\mu>0$ and $r>0$.
This is a discrete EDF of the form in Equation~\ref{equation:edf-supp} with 
$\theta = \log(\mu/(\mu + r))$ and $\kappa(\theta) = -\log(1 - \exp(\theta))$.
Observe that
\begin{align*}
\kappa'(\theta) &= \frac{e^\theta}{1 - e^\theta} = \frac{\mu}{r} \\
\kappa''(\theta) &= \frac{e^\theta}{(1 - e^\theta)^2} = \frac{\mu + \mu^2/r}{r}.
\end{align*}
Thus, letting $Y\sim \NegBin(\mu,r)$, we have $\E(Y) = \mu$ and $\Var(Y) = \mu + \mu^2/r$
by Equation~\ref{equation:edf-moments}.
For an NB-GBM with log link $g(\mu) = \log(\mu)$, the gradients and Fisher information matrices for $A$, $B$, $C$, $D$, $U$, and $V$ are given 
by Equations~\ref{equation:edf-gbm-gradients} and \ref{equation:edf-gbm-fisher} where, by Equation~\ref{equation:edf-gbm-ew},
\begin{align}
\label{equation:nb-gbm-ew}
\mu_{i j} = \exp(\eta_{i j}), ~~~~
w_{i j} = \frac{r_{i j} \mu_{i j}}{r_{i j} + \mu_{i j}},
\text{ ~~~and~~~ } 
e_{i j} = (Y_{i j} - \mu_{i j})\, \frac{w_{i j}}{\mu_{i j}}.
\end{align}

Next, we derive the gradient and observed information matrix for the log-dispersions.
First, consider a single entry. Letting $\psi(x)$ denote the digamma function, we have
\begin{align}
\label{equation:inv-disp-derivs-overflow}
\begin{split}
\frac{\partial}{\partial r} \log \NegBin(Y \mid  \mu,r) &= \psi(Y+r)-\psi(r) + \log(r) + 1 - \log(\mu+r) - \frac{Y+r}{\mu+r} \\
\frac{\partial^2}{\partial r^2} \log \NegBin(Y \mid  \mu,r) &= \psi'(Y+r)-\psi'(r) + \frac{1}{r} - \frac{1}{\mu+r} - \frac{\mu-Y}{(\mu+r)^2}.
\end{split}
\end{align}
We work with the observed information rather than the expected information since $\E(\psi'(r+Y))$ does not seem to have a simple expression.

It turns out that Equation~\ref{equation:inv-disp-derivs-overflow} tends to lead to arithmetic overflow/underflow in extreme cases.
Although these extreme cases occur only occasionally for individual entries,
large GBMs have so many entries that these failures occur persistently and must be addressed.
To avoid arithmetic overflow/underflow, we rewrite Equation~\ref{equation:inv-disp-derivs-overflow} as follows:
\begin{align}
\label{equation:inv-disp-derivs}
\begin{split}
\frac{\partial}{\partial r} \log \NegBin(Y \mid  \mu,r) &= \psi_\Delta(Y,r) - \mathrm{log1p}(\mu/r) - \frac{Y - \mu}{r + \mu} + \mathrm{err}(Y,r) \\
\frac{\partial^2}{\partial r^2} \log \NegBin(Y \mid  \mu,r) &= \psi'_\Delta(Y,r) + \frac{Y + \mu^2/r}{(r + \mu)^2} + \mathrm{err}'(Y,r)
\end{split}
\end{align}
where
\begin{align}
\label{equation:special-funcs}
\mathrm{log1p}(x) &= \branch{x + \log((1+x)/e^x)}{x<1}{\log(1+x)}{x\geq 1}\notag \\
\psi_\Delta(y,r) &= \branch{\psi(y+r) - \psi(r)}{r < 10^8}{\mathrm{log1p}(y/r)}{r \geq 10^8} \\
\psi'_\Delta(y,r) &= \branch{\psi'(y+r) - \psi'(r)}{r < 10^8}{-(y/r)/(y+r)}{r \geq 10^8} \notag
\end{align}
and the error terms $\mathrm{err}(y,r)$ and $\mathrm{err}'(y,r)$ are typically exceedingly small and can be safely ignored.
Mathematically, $\mathrm{log1p}(x) = \log(1+x)$,
however, numerically, the expression in Equation~\ref{equation:special-funcs} computes this value in a way that helps avoid arithmetic overflow and underflow.
Similarly, $\psi_\Delta(y,r)$ and $\psi'_\Delta(y,r)$ compute $\psi(y+r) - \psi(r)$ and $\psi'(y+r) - \psi'(r)$, respectively, 
to very high accuracy while avoiding overflow/underflow; the errors in these approximations are $\mathrm{err}(y,r)$ and $\mathrm{err}'(y,r)$, respectively.
To derive Equation~\ref{equation:inv-disp-derivs}, we group terms of similar magnitude and we use 
the asymptotics of $\psi(x)$ and $\psi'(x)$.

% where $\mathrm{log1p}(x)$ is a function that computes $\log(1+x)$ without arithmetic underflow when $x$ is very small,
% for instance, $\mathrm{log1p}(x) = x + \log((1+x)/e^x)$.
% the identities 
% $$ \psi(x+n) - \psi(x) = \sum_{k=0}^{n-1} \frac{1}{x+k} \text{~~~and~~~} \psi'(x+n) - \psi'(x) = -\sum_{k=0}^{n-1} \frac{1}{(x+k)^2} $$
% for $x>0$ and $n\in\{0,1,2,\ldots\}$,
% which are easily derived from

Now, we derive the gradient and observed information for $S$ and $T$ in the NB-GBM.
Recall that we parametrize $r_{i j} = \exp(-s_i - t_j - \omega)$
and we work in terms of the vector of feature-specific log-dispersion offsets $S\in\R^I$,
the vector of sample-specific log-dispersion offsets $T\in\R^J$,
and the overall log-dispersion $\omega$,
subject to the identifiability constraints $\frac{1}{I}\sum_i e^{s_i} = 1$ and $\frac{1}{J}\sum_j e^{t_j} = 1$.
Let $\L_{i j} = \log \NegBin(Y_{i j}\mid\mu_{i j},r_{i j})$ and $\L = \sum_{i=1}^I \sum_{j=1}^J \L_{i j}$ as before.
The derivatives with respect to $s_i$ and $t_j$ are then
\begin{align}
\label{equation:log-disp-derivs}
\frac{\partial\L}{\partial s_i} &= \sum_{j=1}^J \frac{\partial\L_{i j}}{\partial r_{i j}} (-r_{i j})  \qquad\qquad & 
\frac{\partial^2\L}{\partial s_i^2} &= \sum_{j=1}^J \Big(\frac{\partial^2\L_{i j}}{\partial r_{i j}^2} r_{i j}^2 + \frac{\partial\L_{i j}}{\partial r_{i j}} r_{i j}\Big) \\
\frac{\partial\L}{\partial t_j} &= \sum_{i=1}^I \frac{\partial\L_{i j}}{\partial r_{i j}} (-r_{i j})  \qquad\qquad &
\frac{\partial^2\L}{\partial t_j^2} &= \sum_{i=1}^I \Big(\frac{\partial^2\L_{i j}}{\partial r_{i j}^2} r_{i j}^2 + \frac{\partial\L_{i j}}{\partial r_{i j}} r_{i j}\Big)\notag
\end{align}
since 
$$
\frac{\partial r_{i j}}{\partial s_i} = \frac{\partial r_{i j}}{\partial t_j} = -r_{i j} \qquad\text{ and }\qquad
\frac{\partial^2 r_{i j}}{\partial s_i^2} = \frac{\partial^2 r_{i j}}{\partial t_j^2} = r_{i j}.
$$
By Equation~\ref{equation:inv-disp-derivs}, 
\begin{align}
\label{equation:loglik-derivs-wrt-r}
\begin{split}
\frac{\partial\L_{i j}}{\partial r_{i j}} &= \psi_\Delta(Y_{i j},r_{i j}) - \mathrm{log1p}(\mu_{i j}/r_{i j}) - \frac{Y_{i j}-\mu_{i j}}{r_{i j}+\mu_{i j}} + \mathrm{err}(Y_{i j},r_{i j}) \\
\frac{\partial^2\L_{i j}}{\partial r_{i j}^2} &= \psi'_\Delta(Y_{i j},r_{i j}) + \frac{Y_{i j}+\mu_{i j}^2/r_{i j}}{(r_{i j}+\mu_{i j})^2} + \mathrm{err}'(Y_{i j},r_{i j}) 
\end{split}
\end{align}
where the error terms are negligible.
Note that $\omega$ is implicitly optimized in the optimization-projection steps for $S$ and $T$ due to the likelihood-preserving projections,
making it unnecessary to optimize with respect to $\omega$ directly.

In practice, we do not use the off-diagonal terms of the observed information matrix for the log-dispersion parameters,
since they seem to have a very small effect on the results.
However, for completeness, we note here that 
$$ \frac{\partial^2\L}{\partial s_i \partial t_j} = \frac{\partial^2\L_{i j}}{\partial r_{i j}^2} r_{i j}^2 + \frac{\partial\L_{i j}}{\partial r_{i j}} r_{i j} $$
for all $i,j$, and 
$$ \frac{\partial^2\L}{\partial s_i \partial s_{i'}} = 0 \qquad\qquad \frac{\partial^2\L}{\partial t_j \partial t_{j'}} = 0 $$
for all $i\neq i'$ and $j\neq j'$.

% For a GBM with link function $g(\cdot)$, we reparametrize the EDF outcome distribution in terms of $\eta$ and $s$,
% where $\eta = g(\E(Y)) = g(r \kappa'(\theta))$ and $s = -\log(r)$; this is justified since $g$ and $\kappa'$ are invertible functions.
% Consider a GBM with discrete EDF outcomes $Y_{i j} \sim p(y\mid \eta_{i j}, s_{i j})$,
% where $\eta_{i j} = g(\mu_{i j})$ is a function of $A$, $B$, $C$, $D$, $U$, and $V$ as defined by Equation~\ref{equation:model-univariate}.
% ...
% where $\mu_{i j} = g^{-1}(\eta_{i j})$ and 
% $\sigma_{i j}^2 = r_{i j}\kappa''(\theta_{i j}) = e^{-s_{i j}} \kappa''(\kappa'^{-1}(\mu_{i j} e^{s_{i j}}))$
% in terms of $\eta_{i j}$ and $s_{i j}$.

\subsection{Fisher information between components of the model}
\label{section:gbm-cross-info}

In this section, we provide formulas for the off-block-diagonal entries of the full Fisher information matrix for a discrete EDF-GBM, that is, 
the entries that involve more than one of $A$, $B$, $C$, $D$, $U$, $V$, $S$, $T$, and $\omega$.
With the exception of the entries involving $U$ and $V$, we do not use these formulas in our proposed algorithms,
however, we provide the formulas here to enable comparison with the standard inference approach based on the full Fisher information matrix.
% \todo{Is there a general result for EDF-GBMs that the cross-info between $r$ and other params is zero? If so, state that, otherwise, mention the result for NB-GBMs.}
As in Equation~\ref{equation:edf-gbm-ew}, we define
\begin{align*}
e_{i j} = \frac{\partial\L_{i j}}{\partial\eta_{i j}} \text{~~~~and~~~~}
w_{i j} = \E\Big(-\frac{\partial^2\L_{i j}}{\partial\eta_{i j}^2}\Big).
\end{align*}
For any univariate entry of $A$, $B$, $C$, $D$, $U$, or $V$, say $\beta$, we have 
\begin{align}
\label{equation:deriv-e}
\E\Big(-\frac{\partial e_{i j}}{\partial \beta}\Big) = \E\Big(-\frac{\partial^2 \L_{i j}}{\partial \eta_{i j}^2}\Big) \frac{\partial\eta_{i j}}{\partial\beta} = 
w_{i j}  \frac{\partial\eta_{i j}}{\partial\beta}.
\end{align}
Using Equation~\ref{equation:deriv-e} along with Equations~\ref{equation:eta-derivatives} and \ref{equation:edf-gbm-gradients-univariate}, the cross terms among $A$, $B$, and $C$ are:
\begin{align}
\E\Big(-\frac{\partial^2\L}{\partial a_{j k} \partial b_{i \ell}}\Big) &= w_{i j} x_{i k} z_{j \ell} \\
\E\Big(-\frac{\partial^2\L}{\partial a_{j k} \partial c_{k' \ell}}\Big) &= \sum_{i=1}^I w_{i j} x_{i k} x_{i k'} z_{j \ell} \\
\E\Big(-\frac{\partial^2\L}{\partial b_{i \ell} \partial c_{k \ell'}}\Big) &= \sum_{j=1}^J w_{i j} x_{i k} z_{j \ell} z_{j \ell'}.
\end{align}
Similarly, the cross terms among $U$, $V$, and $D$ are:
\begin{align}
\label{equation:duv-cross-info}
\E\Big(-\frac{\partial^2\L}{\partial u_{i m} \partial v_{j m'}}\Big) &= w_{i j} u_{i m'} d_{m' m'} d_{m m} v_{j m} \\
\E\Big(-\frac{\partial^2\L}{\partial u_{i m} \partial d_{m' m'}}\Big) &= \sum_{j=1}^J w_{i j} u_{i m'} v_{j m'} v_{j m} d_{m m}  \\
\E\Big(-\frac{\partial^2\L}{\partial v_{j m} \partial d_{m' m'}}\Big) &=  \sum_{i=1}^I w_{i j} u_{i m} d_{m m} u_{i m'} v_{j m'}.
\end{align}
The cross terms between $A$, $B$, $C$ and $U$, $V$, and $D$ are:
\begin{align}
\E\Big(-\frac{\partial^2\L}{\partial a_{j k} \partial u_{i m}}\Big) &= w_{i j} x_{i k} v_{j m} d_{m m} \\
\E\Big(-\frac{\partial^2\L}{\partial a_{j k} \partial v_{j' m}}\Big) &= \I(j=j')\sum_{i=1}^I w_{i j} x_{i k} u_{i m} d_{m m} \\
\E\Big(-\frac{\partial^2\L}{\partial a_{j k} \partial d_{m m}}\Big) &= \sum_{i=1}^I w_{i j} x_{i k} u_{i m} v_{j m}\\
\E\Big(-\frac{\partial^2\L}{\partial b_{i \ell} \partial u_{i' m}}\Big) &= \I(i=i')\sum_{j=1}^J w_{i j} z_{j \ell} v_{j m} d_{m m} \\
\E\Big(-\frac{\partial^2\L}{\partial b_{i \ell} \partial v_{j m}}\Big) &= w_{i j} z_{j \ell} u_{i m} d_{m m} \\
\E\Big(-\frac{\partial^2\L}{\partial b_{i \ell} \partial d_{m m}}\Big) &= \sum_{j=1}^J w_{i j} z_{j \ell} u_{i m} v_{j m} \\
\E\Big(-\frac{\partial^2\L}{\partial c_{k \ell} \partial u_{i m}}\Big) &= \sum_{j=1}^J w_{i j} x_{i k} z_{j \ell} v_{j m} d_{m m}  \\
\E\Big(-\frac{\partial^2\L}{\partial c_{k \ell} \partial v_{j m}}\Big) &= \sum_{i=1}^I w_{i j} x_{i k} z_{j \ell} u_{i m} d_{m m} \\
\E\Big(-\frac{\partial^2\L}{\partial c_{k \ell} \partial d_{m m}}\Big) &= \sum_{i=1}^I \sum_{j=1}^J w_{i j} x_{i k} z_{j \ell} u_{i m} v_{j m}.
\end{align}

All of the cross terms between $(S,T,\omega)$ and $(A,B,C,D,U,V)$ are zero.
% Furthermore, it turns out that all remaining entries of the (expected) Fisher information matrix that involve $S$ and $T$ are zero.
To see this, first note that by differentiating Equation~\ref{equation:inv-disp-derivs-overflow} with respect to $\mu$,
$$ \frac{\partial^2}{\partial\mu\partial r} \log \NegBin(Y \mid  \mu,r) = \frac{Y - \mu}{(r + \mu)^2}. $$
Hence, for any entry of $A$, $B$, $C$, $D$, $U$, or $V$, say $\beta$, we have $\E(-\partial^2\L/\partial\beta\partial s_i) = 0$ since
$$ \E\Big(-\frac{\partial^2\L_{i j}}{\partial\beta\partial s_i}\Big) 
= \E\Big(-\frac{\partial^2\L_{i j}}{\partial\mu_{i j}\partial r_{i j}} \frac{\partial \mu_{i j}}{\partial\beta}\frac{\partial r_{i j}}{\partial s_i}\Big)
= -\E\Big(\frac{Y_{i j} - \mu_{i j}}{(r_{i j} + \mu_{i j})^2}\Big) \frac{\partial \mu_{i j}}{\partial\beta}\frac{\partial r_{i j}}{\partial s_i} = 0. $$
Similarly, $\E(-\partial^2\L/\partial\beta\partial t_j) = 0$ and $\E(-\partial^2\L/\partial\beta\partial \omega) = 0$.

\section{Constraint-augmented Fisher information}
\label{section:constraint-augmented-fisher}

During estimation, we enforce constraints that ensure identifiability of the parameters, 
however, these constraints are not accounted for in the Fisher information matrix.
Consequently, it turns out that to appropriately quantify uncertainty in the parameters, 
it is necessary to augment the Fisher information matrix to account for the identifiability constraints;
see Section~\ref{section:augmented-fisher-technique}.

In our proposed algorithm, we use this constraint-augmentation technique only for $U$ and $V$ jointly (see Section~\ref{section:uv-inference}),
and for the remaining components we use our delta propagation technique, which does not involve special handling of the constraints.
We also compare our proposed method to the standard approach of inverting the full constraint-augmented Fisher information matrix.
In Section~\ref{section:augmented-fisher-technique}, we review the constraint-augmentation technique in general, and in Section~\ref{section:full-info-constraint-blocks}, we provide formulas for the full constraint-augmented Fisher information for GBMs.

\subsection{Constraint-augmentation technique}
\label{section:augmented-fisher-technique}

In this section, we review the constraint-augmentation technique of \citet{aitchison1958maximum} and \citet{silvey1959lagrangian}
in the setting of a finite-dimensional parametric model and we extend it to our setting; also see \citet[Section 4.7]{silvey1975statistical}.
Consider an i.i.d.\ model with parameter $\theta\in\R^d$, and suppose we want to perform estimation and inference
subject to the constraint that $g(\theta) = 0$, where $g(\theta) = (g_1(\theta),\ldots,g_k(\theta))^\T\in\R^k$ is a differentiable function.
Suppose $\hat\theta$ is the maximum likelihood estimator subject to $g(\theta) = 0$;
this is referred to as a ``restricted maximum likelihood estimator'' \citep{silvey1975statistical}.

Let $J_\theta \in \R^{k\times d}$ be the Jacobian matrix of $g$, that is, $J_{\theta, i j} = \partial g_i / \partial \theta_j$.
Let $\mathcal{I}_\theta \in \R^{d\times d}$ be the Fisher information matrix for $\theta$, that is, $\mathcal{I}_\theta = \E(-\nabla_\theta^2 \mathcal{L})$ 
where $\mathcal{L}$ is the log-likelihood.
% $\mathcal{I}_{\theta, i j} = \E(-\partial^2 \mathcal{L} / \partial\theta_i
Suppose $\theta_0$ is the true value of the parameter.
\citet{aitchison1958maximum} show that under regularity conditions, when $\mathcal{I}_{\theta_0}$ is invertible,
the covariance matrix of $\hat\theta$ is approximately equal to the 
leading $d\times d$ submatrix of 
{\footnotesize $\begin{bmatrix} \mathcal{I}_{\theta_0} & J_{\theta_0}^\T \\ J_{\theta_0} & 0 \end{bmatrix}^{-1}$}.
% evaluated at $\theta = \theta_0$.
However, in GBMs, we need to consider situations in which $\mathcal{I}_{\theta_0}$ is not invertible since the model is overparametrized unless the identifiability constraints are imposed.
When $\mathcal{I}_{\theta_0}$ is not invertible, \citet{silvey1959lagrangian} extends the technique by showing that 
the covariance matrix of $\hat\theta$ is approximately equal to the leading $d\times d$ submatrix of 
{\footnotesize $\begin{bmatrix} \mathcal{I}_{\theta_0} + J_{\theta_0}^\T J_{\theta_0} & J_{\theta_0}^\T \\ J_{\theta_0} & 0 \end{bmatrix}^{-1}$}.

Since we employ \textit{maximum a posteriori} estimates, we modify the technique to use
the regularized Fisher information matrix $F_\theta = \E(-\nabla_\theta^2 (\mathcal{L} + \log \pi))$ in place of $\mathcal{I}_\theta$.
For our choice of prior, $F_{\theta_0}$ is invertible even when $\mathcal{I}_{\theta_0}$ is not invertible, and consequently it turns out that 
the leading $d\times d$ submatrices of 
{\footnotesize $\begin{bmatrix} F_{\theta_0} & J_{\theta_0}^\T \\ J_{\theta_0} & 0 \end{bmatrix}^{-1}$}
and
{\footnotesize $\begin{bmatrix} F_{\theta_0} + J_{\theta_0}^\T J_{\theta_0} & J_{\theta_0}^\T \\ J_{\theta_0} & 0 \end{bmatrix}^{-1}$}
coincide;
see Proposition~\ref{proposition:augmented-fisher}.
Thus, when using $F_\theta$ instead of $\mathcal{I}_\theta$, it is not necessary to employ the extended version provided by \citet{silvey1959lagrangian};
this provides a big advantage, computationally, when we apply the method to perform inference for $(U,V)$ in GBMs.
Although \citet{aitchison1958maximum} and \citet{silvey1959lagrangian} justify the technique in the i.i.d.\ setting,
empirically we find that it works well in our non-i.i.d.\ setting also.

\begin{proposition}
\label{proposition:augmented-fisher}
Let $F\in\R^{d\times d}$ and $J\in\R^{k\times d}$.  If $F$ and $J F^{-1} J^\T$ are invertible, then the leading $d\times d$ submatrices of 
{\footnotesize $\begin{bmatrix} F & J^\T \\ J & 0 \end{bmatrix}^{-1}$} and 
{\footnotesize $\begin{bmatrix} F + J^\T J  & J^\T \\ J & 0 \end{bmatrix}^{-1}$} are equal.
\end{proposition}
\begin{proof}
Let $A := F + J^\T J$.
By the formula for inverting a $2\times 2$ block matrix,
\begin{align}
\label{equation:block-matrix-inversion}
\begin{bmatrix} A & J^\T \\ J & 0 \end{bmatrix}^{-1} = \begin{bmatrix} A^{-1} - A^{-1} J^\T (J A^{-1} J^\T)^{-1} J A^{-1} & ~\ast \\ \ast & ~\ast \end{bmatrix},
\end{align}
and the same formula applies for {\footnotesize $\begin{bmatrix} F & J^\T \\ J & 0 \end{bmatrix}^{-1}$}, but with $F$ in place of $A$.
Here, $\ast$ denotes unneeded entries.
By the Woodbury matrix inversion formula, 
\begin{align}
\label{equation:woodbury}
A^{-1} = F^{-1} - F^{-1} J^\T (\Id + J F^{-1} J^\T)^{-1} J F^{-1}.
\end{align}
Defining $C := J F^{-1} J^\T$, this implies
$J A^{-1} J^\T = C - C (\Id + C)^{-1} C = C(\Id + C)^{-1}$, hence
\begin{align}
\label{equation:sandwich-inverse}
(J A^{-1} J^\T)^{-1} = (\Id + C) C^{-1}.
\end{align}
Using Equation~\ref{equation:woodbury} again, we see that 
$J A^{-1} = J F^{-1} - C (\Id + C)^{-1} J F^{-1}$, and thus
\begin{align}
\label{equation:sandwich-inverse-plus}
(J A^{-1} J^\T)^{-1} J A^{-1} = (\Id + C) C^{-1} J F^{-1} - J F^{-1} = C^{-1} J F^{-1}.
\end{align}
Likewise, $A^{-1} J^\T = F^{-1} J^\T - F^{-1} J^\T (\Id + C)^{-1} C$, so combining this with Equations~\ref{equation:sandwich-inverse-plus}
and \ref{equation:woodbury} and canceling, we have
$$ A^{-1} - A^{-1} J^\T (J A^{-1} J^\T)^{-1} J A^{-1} = F^{-1} - F^{-1} J^\T (J F^{-1} J^\T)^{-1} J F^{-1}. $$
The result follows by Equation~\ref{equation:block-matrix-inversion}.
\end{proof}

\subsection{Constraint-augmented Fisher information for $(U,V)$}
\label{section:uv-inference}

We have found that it is important to quantify uncertainty in $U$ and $V$ jointly and account for the identifiability constraints.
% due to the strong interdependency between them and the substantial uncertainty in these parameters.
% using the constraint-augmented Fisher information matrix (Section~\ref{section:augmented-fisher-technique}).
In this section, we derive a computationally efficient method for computing the diagonal of 
the inverse of the constraint-augmented Fisher information matrix (Section~\ref{section:augmented-fisher-technique})
to obtain approximate standard errors for $U$ and $V$.

Let $J_u$ and $J_v$ denote the constraint Jacobian matrices for $U$ and $V$; see Section~\ref{section:full-info-constraint-blocks}.
The regularized, constraint-augmented Fisher information matrix for $(U,V)$ is then
$$ \tilde{F}_{(u,v)} := \begin{bmatrix} F_u & F_{u v} & J_u^\T & 0 \\ F_{u v}^\T & F_v & 0 & J_v^\T \\ J_u & 0 & 0 & 0 \\ 0 & J_v & 0 & 0 \end{bmatrix} $$
where $F_u = \E(-\nabla_{\!\vec{u}}^2\,\L) + \lambda_u \Id$, $F_v = \E(-\nabla_{\!\vec{v}}^2\,\L) + \lambda_v \Id$, and $F_{u v} = \E(-\nabla_{\!\vec{u}}\nabla_{\!\vec{v}}^\T\,\L)$; formulas for each of these expectations are given in Equations~\ref{equation:edf-gbm-fisher} and \ref{equation:duv-cross-info}.

We need the first $I M + J M$ entries of $\diag(\tilde{F}_{(u,v)}^{-1})$ 
in order to obtain approximate standard errors for the entries of $U$ and $V$,
however, naively performing this matrix inversion is computationally intractable when $I$ (or $J$) is large.
To compute this efficiently, we structure the calculation as follows.
First, let $P$ be the permutation matrix such that 
$$ P\tilde{F}_{(u,v)}P^\T = \begin{bmatrix} F_u & J_u^\T & F_{u v} & 0 \\ J_u & 0 & 0 & 0 \\ F_{u v}^\T & 0 & F_v & J_v^\T \\ 0 & 0 & J_v & 0 \end{bmatrix}
= \begin{bmatrix} \mathsf{A} & \mathsf{B} \,\\ \mathsf{B}^\T & \mathsf{C} \,\end{bmatrix} $$
where we define
$$ 
\mathsf{A} = \begin{bmatrix} F_u & J_u^\T \,\\ J_u & 0 \,\end{bmatrix}, \qquad
\mathsf{B} = \begin{bmatrix} F_{u v} & 0 \,\\ 0 & 0 \,\end{bmatrix}, \qquad
\mathsf{C} = \begin{bmatrix} F_v & J_v^\T \,\\ J_v & 0 \,\end{bmatrix}.
$$
(We use sans serif font for these block matrices, such as $\mathsf{A}$, to distinguish them from parameter matrices such as $A$.)
Since $\diag((P\tilde{F}_{(u,v)}P^\T)^{-1}) = \diag(P\tilde{F}_{(u,v)}^{-1} P^\T) = P \diag(\tilde{F}_{(u,v)}^{-1})$, we can 
compute the diagonal of the inverse of $P\tilde{F}_{(u,v)}P^\T$ and then permute back to get $\diag(\tilde{F}_{(u,v)}^{-1})$.
By the formula for inversion of a $2\times 2$ block matrix,
$$ (P\tilde{F}_{(u,v)}P^\T)^{-1} = \begin{bmatrix} \mathsf{A} & \mathsf{B} \,\\ \mathsf{B}^\T & \mathsf{C} \,\end{bmatrix}^{-1} 
= \begin{bmatrix} \mathsf{A}^{-1} + \mathsf{A}^{-1} \mathsf{B} (\mathsf{C} - \mathsf{B}^\T \mathsf{A}^{-1} \mathsf{B})^{-1} \mathsf{B}^\T \mathsf{A}^{-1} & \ast \\ \ast & (\mathsf{C} - \mathsf{B}^\T \mathsf{A}^{-1} \mathsf{B})^{-1} \end{bmatrix} $$
where $\ast$ denotes entries that are not needed.  Similarly, by the same formula,
$$ \mathsf{A}^{-1} = \begin{bmatrix} F_u & J_u^\T \,\\ J_u & 0 \,\end{bmatrix}^{-1} = \begin{bmatrix} \mathsf{D} & \mathsf{E} \,\\ \mathsf{E}^\T & \ast \,\end{bmatrix} $$
where $\mathsf{D} = F_u^{-1} - F_u^{-1} J_u^\T (J_u F_u^{-1} J_u^\T)^{-1} J_u F_u^{-1}$ and $\mathsf{E} = F_u^{-1} J_u^\T (J_u F_u^{-1} J_u^\T)^{-1}$.
Thus, 
$$ (\mathsf{C} - \mathsf{B}^\T \mathsf{A}^{-1} \mathsf{B})^{-1} = \begin{bmatrix} F_v - F_{u v}^\T \mathsf{D} F_{u v} & J_v^\T \,\\ J_v & 0 \,\end{bmatrix}^{-1}
= \begin{bmatrix} \,\mathsf{G} & \,\ast \,\\ \ast & \ast \,\end{bmatrix}, $$
where $\mathsf{G}$ is defined to be the leading $J M \times J M$ block of this matrix.
Putting these pieces together justifies defining the approximate variance of each entry of $\vec{v} = \mathrm{vec}(V^\T)$ as
$$\widehat{\mathrm{var}}(\vec{v}) := \diag(G)$$
since these are the entries of $\diag((P\tilde{F}_{(u,v)}P^\T)^{-1})$ corresponding to $\vec{v}$.
To approximate the variance of the entries of $\vec{u} = \mathrm{vec}(U^\T)$, observe that
\begin{align*}
\mathsf{A}^{-1} \mathsf{B} (\mathsf{C} - \mathsf{B}^\T \mathsf{A}^{-1} \mathsf{B})^{-1} \mathsf{B}^\T \mathsf{A}^{-1} &=
\begin{bmatrix} \mathsf{D} & \mathsf{E} \,\\ \mathsf{E}^\T & \ast \,\end{bmatrix} 
\begin{bmatrix} F_{u v} & 0 \,\\ 0 & 0 \,\end{bmatrix}
\begin{bmatrix} \mathsf{G} & \ast \,\\ \ast & \ast \,\end{bmatrix}
\begin{bmatrix} F_{u v}^\T & 0 \,\\ 0 & 0 \,\end{bmatrix}
\begin{bmatrix} \mathsf{D} & \mathsf{E} \,\\ \mathsf{E}^\T & \ast \,\end{bmatrix} \\
& = 
\begin{bmatrix} \mathsf{D} F_{u v} \mathsf{G} F_{u v}^\T \mathsf{D} & \ast \,\\ \ast & \ast \,\end{bmatrix}.
\end{align*}
Therefore, we define
$$\widehat{\mathrm{var}}(\vec{u}) := \diag(\mathsf{D} + \mathsf{D} F_{u v} \mathsf{G} F_{u v}^\T \mathsf{D})$$
since these are the entries of $\diag((P\tilde{F}_{(u,v)}P^\T)^{-1})$ corresponding to $\vec{u}$.

\subsection{Full constraint-augmented Fisher information for GBMs}
\label{section:full-info-constraint-blocks}

To facilitate comparison with the classical approach, we derive the constraint-augmented Fisher information matrix for all of the components $(A,B,C,D,U,V)$ jointly, even though our approach only requires the matrix for $(U,V)$.
It is not necessary to include the log-dispersion parameters $(S,T,\omega)$ since the Fisher information (as well as the constraint Jacobian) between these parameters and $(A,B,C,D,U,V)$ is zero (see Section~\ref{section:gbm-cross-info}); thus, in the classical approach, inference for $(S,T,\omega)$ and $(A,B,C,D,U,V)$ can be performed independently since the constraint-augmented Fisher information decomposes.

For each of $A$, $B$, $C$, $D$, $U$, and $V$, we vectorize both the parameter matrix and the corresponding constraints, as follows.
For $A$, the constraint $Z^\T A = 0$ can be written as $g(A) = 0$ where $g(A) = \mathrm{vec}(A^\T Z)\in\R^{K L}$.
The constraint Jacobian for $\vec{a} = \mathrm{vec}(A^\T)$ is then $J_a := Z^\T \otimes \Id_K$, where $\otimes$ denotes the Kronecker product and $\Id_K$ is the $K\times K$ identity matrix.
Likewise, vectorizing the constraint on $B$ as $\mathrm{vec}(B^\T X) = 0$, the constraint Jacobian for $\vec{b} = \mathrm{vec}(B^\T)$ is $J_b := X^\T \otimes \Id_L$.

For uncertainty quantification, the key constraints on $U$ and $V$ are $X^\T U = 0$, $Z^\T V = 0$, $U^\T U = \Id$, and $V^\T V = \Id$. 
Vectorizing, the constraints on $U$ can be written as $\mathrm{vec}(U^\T X) = 0$ and $\mathrm{vec}(U^\T U - \Id) = 0$.
% $f(U) = \mathrm{vec}(U^\T X) \in \R^{M K}$ and $g(U) = \mathrm{vec}(U^\T U - \Id) \in \R^{M^2}$.
Thus, the constraint Jacobian for $\vec{u}$ is 
\begin{align*}
% \label{equation:u-constraint-jacobian}
J_u := \begin{bmatrix} X_{1 \bc} \otimes \Id_M  & \cdots & X_{I \bc} \otimes \Id_M   \\ (U_{1 \bc} \otimes \Id_M) + (\Id_M \otimes U_{1 \bc}) & \cdots & (U_{I \bc} \otimes \Id_M) + (\Id_M \otimes U_{I \bc}) \end{bmatrix} 
\in \R^{(M K + M^2)\times I M}.
\end{align*}
Here, for a matrix $Q\in\R^{m\times n}$, we write $Q_{i \bc}$ and $Q_{\bc j}$ to denote the $i$th row and $j$th column as column vectors, respectively, that is, $Q_{i \bc} := (q_{i 1},\ldots,q_{i n})^\T \in \R^n$ and $Q_{\bc j} := (q_{1 j},\ldots,q_{m j})^\T \in \R^m$.
The constraint Jacobian for $\vec{v}$, namely $J_v$, is computed the same way as $J_u$ but with $V$, $Z$, $J$, and $L$ in place of $U$, $X$, $I$, and $K$, respectively.
There are no constraints on $C$, and the remaining constraints on $D$, $U$, and $V$ do not reduce the dimensionality of the parameter space.
%  and can be ignored for the purposes of uncertainty quantification.
Thus, altogether, the constraint Jacobian for $(A,B,C,D,U,V)$ is
$$ \mathcal{J} := \begin{bmatrix} J_a & 0 & 0 & 0 & 0 \\ 0 & J_b & 0 & 0 & 0 \\ 0 & 0 & 0 & J_u & 0 \\ 0 & 0 & 0 & 0 & J_v \end{bmatrix} $$
where the column of zero blocks corresponding to $(C,D)$ has width $K L + M$ and height $J K + I L + I M + J M$.

Let $\mathcal{I} = \E(-\nabla^2 \L)$ denote the full Fisher information matrix for $(A,B,C,D,U,V)$.
Formulas for all of the entries of $\mathcal{I}$ are given in Equation~\ref{equation:edf-gbm-fisher} and in Section~\ref{section:gbm-cross-info}.
The full constraint-augmented Fisher information matrix for all of the parameters is
$$ \tilde{\mathcal{I}} := \begin{bmatrix} \mathcal{I} & \mathcal{J}^\T \\ \mathcal{J} & 0 \end{bmatrix}. $$
The classical approach is then to define approximate standard errors as the square roots of the entries of $\diag(\tilde{\mathcal{I}}^{-1})$ corresponding to each component (Section~\ref{section:augmented-fisher-technique} and \citealp{silvey1975statistical}); for instance, the first $J K$ entries of $\diag(\tilde{\mathcal{I}}^{-1})$ are the variances of the entries of $\mathrm{vec}(A^\T)$.

% Taking the gradient of each entry of $f$ with respect to $\vec{u}$, we have
% These Jacobians are given by the following formulas:
% $D_{u} f = X^\T \otimes \Id_M$ where $\otimes$ denotes the Kronecker product and $\Id_M$ is the $M\times M$ identity matrix, and 
% $$ D_{u} g = \Big[(U_{1 \bc} \otimes \Id_M) + (\Id_M \otimes U_{1 \bc}) \quad \cdots \quad (U_{1 \bc} \otimes \Id_M) + (\Id_M \otimes U_{1 \bc})\Big]. $$
% and column $(i-1)I + m$ of $D_u g$ is
% $$(D_{u} g)_{\bc,(i-1)M + m} = \mathrm{vec}(U_{i \bc} \delta_m^\T + \delta_m U_{i \bc}^\T)$$
% where $\delta_m = (0,\ldots,0,1,0,\ldots,0)^\T\in\R^M$, with the $1$ in entry $m$.

\section{Step-by-step estimation algorithm}
\label{section:estimation-details}

Given the inputs and preprocessing as described in Section~\ref{section:estimation}, the algorithm is as follows.

\subsection{Initialization procedure}
\label{section:initialization}

We initialize the GBM estimation algorithm by (a) solving for values of $A$, $B$, and $C$ 
to minimize the sum-of-squares of the GBM residuals $\varepsilon_{i j}$,
(b) randomly initializing $D$, $U$, and $V$, and in the NB-GBM case, (c) iteratively updating $S$, $T$, and $\omega$ for a few iterations.
This approach has the advantage of being simple, fast, and effective; see below for details.

It is somewhat tricky to initialize the algorithm well due to a chicken-and-egg problem.
The issue is that having decent estimates of $S$ and $T$ is important to avoid overfitting to outliers,
but we need a reasonable estimate of the mean matrix in order to estimate $S$ and $T$.
Our solution to this problem is to exclude $D$, $U$, and $V$ from the initial fitting of $A$, $B$, and $C$ in step (a) above.
This prevents the latent factors from overfitting to outlier samples or outlier features, thus helping avoid getting stuck at a suboptimal point.
In detail, we initialize as follows.

\begin{enumerate}[label=\textup{(\arabic*)},noitemsep,topsep=0pt]
\item Compute $\check{Y}_{i j} = g(Y_{i j} + \epsilon)$ where $\epsilon = 1/8$, and set $\check{Y} = (\check{Y}_{i j})\in\R^{I\times J}$.
\item $C \gets X^\ps \check{Y} (Z^\ps)^\T$
\item $A \gets (X^\ps \check{Y} - C Z^\T)^\T$
\item $B \gets \check{Y} (Z^\ps)^\T - X C$
\item Sample $q_{i j}\sim\N(0, 10^{-16})$ i.i.d.\ for all $i,j$, and set $Q = (q_{i j})\in\R^{I\times J}$.
\item Compute $U$, $D$, and $V$ such that $U D V^\T = Q$ is the compact SVD of rank $M$. %, where $D$ is a diagonal matrix with $d_{1 1} \geq \cdots \geq d_{M M}$.
\item Initialize $S = 0$, $T = 0$, and $\omega = 0$, and then run the updates to $S$ and $T$ as defined in Section~\ref{section:detailed-updates} for 4 iterations, using the current values of $A$, $B$, $C$, $D$, $U$, and $V$. Note that this also modifies $\omega$.
\end{enumerate}

\subsection{Updates to each component of the model}
\label{section:detailed-updates}

In this section, we provide step-by-step algorithms for updating each component of the model
using the optimization-projection approach.
The optimization part of each update is based on the bounded regularized Fisher scoring step in Equation~\ref{equation:generic-fisher-update},
using the formulas for the gradients and Fisher information matrices derived in Section~\ref{section:gbm-derivatives}.
The projection part of each update is based on Theorem~\ref{theorem:projections}.
In the updates to $G = U D$ and $H = V D$, we use the priors on $G$ and $H$ induced by the priors on $U$ and $V$, given $D$;
see Section~\ref{section:priors}.

For a matrix $Q\in\R^{m\times n}$, we denote 
$Q_{i \bc} := (q_{i 1},\ldots,q_{i n})^\T \in \R^n$,
$Q_{\bc j} := (q_{1 j},\ldots,q_{m j})^\T \in \R^m$,
$Q_{i *} = \Diag(Q_{i \bc})$,
$Q_{* j} = \Diag(Q_{\bc j})$,
$\mathrm{vec}(Q)$ is the column-wise vectorization of $Q$, and
$\mathrm{block}(Q_{i j} : i,j\in\{1,\ldots,n\})$ is the block matrix with blocks $Q_{i j}$.
When multiplying by a diagonal matrix such as $Q_{i *}$ or $Q_{* j}$, we do not allocate the full diagonal matrix and perform matrix multiplication. Instead, it is much more efficient to simply multiply each row (when left-multiplying) or column (when right-multiplying) by the corresponding diagonal entry.

\subsubsection*{Computing $\eta$, $\mu$, $W$, and $E$}
    \begin{enumerate}[label=\textup{(\arabic*)},noitemsep,topsep=0pt]
    \item $\eta\gets X A^\T + B Z^\T + X C Z^\T + U D V^\T$
    \item For $i = 1,\ldots,I$ and $j = 1,\ldots,J$,
        \begin{enumerate}[label=\textup{(\alph*)},noitemsep,topsep=0pt]
        \item $\mu_{i j} \gets \exp(\eta_{i j})$
        \item $w_{i j} \gets r_{i j} \mu_{i j} / (r_{i j} + \mu_{i j})$
        \item $e_{i j} \gets (Y_{i j} - \mu_{i j}) w_{i j} / \mu_{i j}$
        \end{enumerate}
    \end{enumerate}

\subsubsection*{Updating $A$}
    \begin{enumerate}[label=\textup{(\arabic*)},noitemsep,topsep=0pt]
    \item Recompute $W$ and $E$ using the current parameter estimates.
    \item For $j = 1,\ldots,J$
        \begin{enumerate}[label=\textup{(\alph*)},noitemsep,topsep=0pt]
        \item $\xi \gets (X^\T W_{* j} X + \lambda_a \Id)^{-1} (X^\T E_{\bc j} - \lambda_a A_{j \bc})$ ~~~~(compute Fisher scoring step)
        \item $A_{j\bc} \gets A_{j\bc} + \xi \min\{1,\, \rho\sqrt{K} / \|\xi\|\}$ ~~~~(apply modified step to $A_{j\bc}$)
        \end{enumerate}
    \item $Q \gets Z^\ps A$ ~~~(efficiently structure computation of projection)
    \item $A \gets A - Z Q$ ~~~(enforce $Z^\T A = 0$ by projecting onto nullspace of $Z^\T$)
    \item $C \gets C + Q^\T$ ~~~(compensate to preserve likelihood)
    \end{enumerate}
    
\subsubsection*{Updating $B$}
    \begin{enumerate}[label=\textup{(\arabic*)},noitemsep,topsep=0pt]
    \item Recompute $W$ and $E$ using the current parameter estimates.
    \item For $i = 1,\ldots,I$
        \begin{enumerate}[label=\textup{(\alph*)},noitemsep,topsep=0pt]
        \item $\xi \gets (Z^\T W_{i *} Z + \lambda_b \Id)^{-1} (Z^\T E_{i \bc} - \lambda_b B_{i \bc})$ ~~~~(compute Fisher scoring step)
        \item $B_{i\bc} \gets B_{i\bc} + \xi \min\{1,\, \rho\sqrt{L} / \|\xi\|\}$ ~~~~(apply modified step to $B_{i\bc}$)
        \end{enumerate}
    \item $Q \gets X^\ps B$ ~~~(efficiently structure computation of projection)
    \item $B \gets B - X Q$  ~~~(enforce $X^\T B = 0$ by projecting onto nullspace of $X^\T$)
    \item $C \gets C + Q$ ~~~(compensate to preserve likelihood)
    \end{enumerate}
    
\subsubsection*{Updating $C$}
    \begin{enumerate}[label=\textup{(\arabic*)},noitemsep,topsep=0pt]
    \item Recompute $W$ and $E$ using the current parameter estimates.
    \item $F \gets \mathrm{block}\Big({\textstyle\sum_{j=1}^J} z_{j\ell}z_{j\ell'}(X^\T W_{* j} X) : \ell,\ell'\in\{1,\ldots,L\}\Big)$ ~~~~(compute Fisher info)
    \item $\xi \gets (F + \lambda_c \Id)^{-1} (\mathrm{vec}(X^\T E Z) - \lambda_c \mathrm{vec}(C))$ ~~~~(compute Fisher scoring step)
    \item $\mathrm{vec}(C) \gets \mathrm{vec}(C) + \xi \min\{1,\, \rho\sqrt{K L} / \|\xi\|\}$ ~~~(apply modified step to $C$)
    \end{enumerate}
    
\subsubsection*{Updating $D$}
    \begin{enumerate}[label=\textup{(\arabic*)},noitemsep,topsep=0pt]
    \item Recompute $W$ and $E$ using the current parameter estimates.
    \item $F \gets \sum_{j=1}^J (U V_{j *})^\T W_{* j} (U V_{j *})$ ~~~~(compute Fisher information)
    \item $\xi \gets (F + \lambda_d \Id)^{-1} (\diag(U^\T E V) - \lambda_d\, \mathrm{diag}(D))$ ~~~~(compute Fisher scoring step)
    \item $\mathrm{diag}(D) \gets \mathrm{diag}(D) + \xi \min\{1,\, \rho\sqrt{M} / \|\xi\|\}$ ~~~(apply modified step to $D$)
    % \item Find a permutation $\sigma$ such that $d_{\sigma_1 \sigma_1} \geq \cdots \geq d_{\sigma_M \sigma_M}$. ~~~~(sort diagonal entries)
    % \item $d_{m m} \gets d_{\sigma_m \sigma_m}$ for $m=1,\ldots,M$ ~~~~(enforce constraint that $d_{1,1} \geq \cdots \geq d_{m,m}$)
    % \item $U_{\bc\, \sigma_m} \gets U_{\bc\, \sigma_m}$ and $V_{\bc\,\sigma_m} \gets V_{\bc\, \sigma_m}$ for $m=1,\ldots,M$ ~~~~(permute columns of $U$ and $V$)
    \end{enumerate}
    
\subsubsection*{Updating $G = U D$}
    \begin{enumerate}[label=\textup{(\arabic*)},noitemsep,topsep=0pt]
    \item Recompute $W$ and $E$ using the current parameter estimates.
    \item $G \gets U D$
    \item $\Lambda \gets (D^2 / \lambda_u)^{-1}$ ~~~~(precision matrix for prior on each row of $G$)
    \item For $i = 1,\ldots,I$
        \begin{enumerate}[label=\textup{(\alph*)},noitemsep,topsep=0pt]
        \item $\xi \gets (V^\T W_{i *} V + \Lambda)^{-1} (V^\T E_{i \bc} - \Lambda G_{i \bc})$ ~~~~(compute Fisher scoring step)
        \item $G_{i\bc} \gets G_{i\bc} + \xi \min\{1,\, \rho\sqrt{M} / \|\xi\|\}$ ~~~~(apply modified step to $G_{i\bc}$)
        \end{enumerate}
    \item $Q \gets X^\ps G$ ~~~(efficiently structure computation of projection)
    \item $G \gets G - X Q$ ~~~(enforce $X^\T G = 0$ by projecting onto nullspace of $X^\T$)
    \item $A \gets A + V Q^\T$  ~~~(compensate to preserve likelihood)
    \item $Q \gets Z^\ps A$ ~~~(efficiently structure computation of projection)
    \item $A \gets A - Z Q$ ~~~(enforce $Z^\T A = 0$ by projecting onto nullspace of $Z^\T$)
    \item $C \gets C + Q^\T$ ~~~(compensate to preserve likelihood)
    \item Run compact SVD of rank $M$ on $G V^\T$, yielding $U$, $D$, $V$ such that $U D V^\T = G V^\T$.
    \end{enumerate}
    
\subsubsection*{Updating $H = V D$}
    \begin{enumerate}[label=\textup{(\arabic*)},noitemsep,topsep=0pt]
    \item Recompute $W$ and $E$ using the current parameter estimates.
    \item $H \gets V D$
    \item $\Lambda \gets (D^2 / \lambda_v)^{-1}$ ~~~~(precision matrix for prior on each row of $H$)
    \item For $j = 1,\ldots,J$
        \begin{enumerate}[label=\textup{(\alph*)},noitemsep,topsep=0pt]
        \item $\xi \gets (U^\T W_{* j} U + \Lambda)^{-1} (U^\T E_{\bc j} - \Lambda H_{j \bc})$ ~~~~(compute Fisher scoring step)
        \item $H_{j\bc} \gets H_{j\bc} + \xi \min\{1,\, \rho\sqrt{M} / \|\xi\|\}$ ~~~~(apply modified step to $H_{j\bc}$)
        \end{enumerate}
    \item $Q \gets Z^\ps H$ ~~~(efficiently structure computation of projection)
    \item $H \gets H - Z Q$ ~~~(enforce $Z^\T H = 0$ by projecting onto nullspace of $Z^\T$)
    \item $B \gets B + U Q^\T$  ~~~(compensate to preserve likelihood)
    \item $Q \gets X^\ps B$ ~~~(efficiently structure computation of projection)
    \item $B \gets B - X Q$ ~~~(enforce $X^\T B = 0$ by projecting onto nullspace of $X^\T$)
    \item $C \gets C + Q$ ~~~(compensate to preserve likelihood)
    \item Run compact SVD of rank $M$ on $U H^\T$, yielding $U$, $D$, $V$ such that $U D V^\T = U H^\T$.
    \end{enumerate}
    
\subsubsection*{Updating $S$ in an NB-GBM}
    For the updates to $S$ and $T$, we employ adaptive maximum step sizes $\rho_{s i}$ and $\rho_{t j}$ for $s_i$ and $t_j$, respectively.
    This helps prevent occasional lack of convergence due to oscillating estimates.
    At the start of the algorithm, we initialize $\rho_{s i} \gets \rho$ and $\rho_{t j} \gets \rho$.
    Define $\mathrm{log1p}(x)$, $\psi_\Delta(y,r)$, and $\psi'_\Delta(y,r)$ as in Equation~\ref{equation:special-funcs}.
    Note that we do not explicitly update $\omega$, since $\omega$ is implicitly updated in the projection part of the updates to $S$ and $T$.  
    \begin{enumerate}[label=\textup{(\arabic*)},noitemsep,topsep=0pt]
    % \item $\eta \gets X A^\T + B Z^\T + X C Z^\T + U D V^\T$
    \item Compute $\mu$ using the current parameter estimates.
    \item For $i = 1,\ldots,I$ and $j = 1,\ldots,J$, ~~~~(differentiate each term in the log-likelihood)
        \begin{enumerate}[label=\textup{(\alph*)},noitemsep,topsep=0pt]
        % \item $\mu_{i j} \gets \exp(\eta_{i j})$
        \item $\delta_{i j} \gets -r_{i j} \Big(\psi_\Delta(Y_{i j},r_{i j}) - \mathrm{log1p}(\mu_{i j}/r_{i j}) - (Y_{i j}-\mu_{i j})/(r_{i j}+\mu_{i j})\Big)$
        \item $\delta'_{i j} \gets -\delta_{i j} + r_{i j}^2 \psi'_\Delta(Y_{i j},r_{i j}) + (Y_{i j} + \mu_{i j}^2 / r_{i j})/(1 + \mu_{i j}/r_{i j})^2$
        \end{enumerate}
    \item\label{item:updating-s-loop} For $i = 1,\ldots,I$
        \begin{enumerate}[label=\textup{(\alph*)},noitemsep,topsep=0pt]
        \item $g \gets -\lambda_s (s_i - m_s) + \sum_{j=1}^J \delta_{i j}$ ~~~~(derivative of log-posterior with respect to $s_i$)
        \item $h \gets -\lambda_s + \sum_{j=1}^J \delta'_{i j}$ ~~~~(second derivative of log-posterior with respect to $s_i$)
        \item If $h < 0$ then $\xi \gets -g/h$, otherwise, $\xi \gets g$. (Newton if valid, otherwise gradient)
        \item $s_i \gets s_i + \xi \min\{1,\, \rho_{s i} / |\xi|\}$ ~~~(apply modified optimization step to $s_i$)
        \item If $|\xi| > \rho_{s i}$ then $\rho_{s i} \gets \rho_{s i}/2$, otherwise, $\rho_{s i} \gets \rho$. ~~~(adapt maximum step size)
        \end{enumerate}
    \item\label{item:updating-s-shift} $c \gets \log(\frac{1}{I} \sum_{i=1}^I e^{s_i})$
    \item\label{item:updating-s-project} $S \gets S - c$ ~~~(enforce constraint by projecting)
    \item $\omega \gets \omega + c$ ~~~(compensate to preserve likelihood)
    \item $r_{i j} \gets \exp(-s_i-t_j-\omega)$ for $i=1,\ldots,I$ and $j=1,\ldots,J$ ~~~(update inverse dispersions)
    \end{enumerate}
    
\subsubsection*{Updating $T$ in an NB-GBM}
Steps (1)-(2) and (7) are the same as in the update to $S$.  Steps (3)-(6) become:
    \begin{enumerate}[label=\textup{(\arabic*)},noitemsep,topsep=0pt]
    \item[(3)] For $j = 1,\ldots,J$
        \begin{enumerate}[label=\textup{(\alph*)},noitemsep,topsep=0pt]
        \item $g \gets -\lambda_t (t_j - m_t) + \sum_{i=1}^I \delta_{i j}$ ~~~~(derivative of log-posterior with respect to $t_j$)
        \item $h \gets -\lambda_t + \sum_{i=1}^I \delta'_{i j}$ ~~~~(second derivative of log-posterior with respect to $t_j$)
        \item If $h < 0$ then $\xi \gets -g/h$, otherwise, $\xi \gets g$. (Newton if valid, otherwise gradient)
        \item $t_j \gets t_j + \xi \min\{1,\, \rho_{t j} / |\xi|\}$ ~~~(apply modified optimization step to $s_i$)
        \item If $|\xi| > \rho_{t j}$ then $\rho_{t j} \gets \rho_{t j}/2$, otherwise, $\rho_{t j} \gets \rho$. ~~~(adapt maximum step size)
        \end{enumerate}
    \item[(4)] $c \gets \log(\frac{1}{J} \sum_{j=1}^J e^{t_j})$
    \item[(5)] $T \gets T - c$ ~~~(enforce constraint by projecting)
    \item[(6)] $\omega \gets \omega + c$ ~~~(compensate to preserve likelihood)
    \end{enumerate}

\subsubsection*{Bias correction for $S$ and $T$ in an NB-GBM}
Empirically, when the true values of $S$ and $T$ are low, the maximum likelihood estimates tend to exhibit a downward bias.
Occasionally, this leads to massive underestimation of some of the log-dispersion values.
This issue is mitigated somewhat by using a prior to shrink the estimates toward zero, however, 
it seems difficult to tune the prior to appropriately balance the bias.
Thus, we employ the following simple bias correction procedure, applied after the final iteration of the estimation algorithm.
Choose lower bounds $s_*$ and $t_*$ on $s_i$ and $t_j$, respectively; we use $s_* = t_* = -4$ as defaults.
    \begin{enumerate}[label=\textup{(\arabic*)},noitemsep,topsep=0pt]
    \item $s_i \gets s_* + \log(\exp(s_i - s_*) + 1)$ for $i = 1,\ldots,I$ ~~~~(apply bias correction to $S$)
    \item $c \gets \log(\frac{1}{I} \sum_{i=1}^I e^{s_i})$
    \item $S \gets S - c$ ~~~(enforce constraint by projecting)
    \item $\omega \gets \omega + c$ ~~~(compensate to preserve likelihood)
    \end{enumerate}
The same procedure is applied to $T$, with $t_*$ in place of $s_*$.
We find that this improves the accuracy of the log-dispersion estimates when the true values are at the low end.

\subsection{Remarks on the estimation algorithm}
We continue iterating until either
(a) the relative change in log-likelihood+log-prior (Equations~\ref{equation:log-likelihood} and \ref{equation:log-prior}) from one iteration to the next is less than the convergence tolerance $\tau$
or (b) the maximum number of iterations has been reached.
% Since the normalizing constant of the log-posterior cancels, we only need to compute the change in log-likelihood plus log-prior (Equations~\ref{equation:log-likelihood} and \ref{equation:log-prior}).
% In the NB-GBM case, for numerical stability we compute the negative binomial log-density as follows:
% $\log \NegBin(y\mid \mu,r) = \gamma_\Delta(y,r) - \gamma_\Delta(y,1) - y \mathrm{log1p}(r/\mu) - r \mathrm{log1p}(\mu/r)$
% where $\gamma_\Delta(y,r) = \sum_{k = 0}^{y-1} \log(r + k)$. % for $y\in\{0,1,2,\ldots\}$, $\mu > 0$, and $r > 0$.

In the updates to $G = U D$ and $H = V D$, we use the compact SVD to enforce the constraints on $D$, $U$, and $V$.
Fast computation of the compact SVD can be done using procedures for the truncated SVD, which allows one to specify the rank (that is, the number of latent factors $M$).
Procedures for the truncated SVD are available in many programming languages.
It is not necessary to enforce the ordering and sign constraints on $D$, $U$, and $V$ (Conditions~\ref{condition:diagonal} and \ref{condition:sign}) during the iterative updates since both the likelihood and the prior are invariant to the order and sign of the latent factors.

Note that, due to the symmetry of the model, the updates for $A$ and $B$ are similar enough that a single function can be used to compute both of them,
with an option to handle the transpose for $C$.
Likewise, a single function can be used to compute both the $U D$ and $V D$ updates, with an option to handle transposes appropriately.
% \todo{Could other projections be used instead, e.g., by weighting $X$ and $Z$? Something like in Choulakian (1996)?  If so, mention this.}

\section{Step-by-step inference algorithm}
\label{section:inference-details}

\textbf{Notation.}
For matrices $A$ and $B$, we use
$A \otimes B$ to denote the Kronecker product and $A \odot B$ for the element-wise product.
Normally, we write $A B$ for matrix multiplication, but for improved clarity we use $A \times B$ to denote matrix multiplication when multi-letter variables such as $\tt{invFc}$ are involved.
We write $\mathrm{hcat}(A_1,\ldots,A_n)$ for the horizontal concatenation of matrices $A_1,\ldots,A_n$, 
that is, $\mathrm{hcat}(A_1,\ldots,A_n) := [A_1 \; \cdots \; A_n]$.
Likewise, $\mathrm{vcat}$ denotes vertical concatenation.
We define $\mathrm{block}(j,K) := ((j-1) K + 1, (j-1) K + 2, \ldots, j K)$.
For a matrix $A\in\R^{m\times n}$, $\mathrm{colsums}(A)$ denotes the vector of column sums, that is, 
$\mathrm{colsums}(A) = (\sum_i a_{i 1},\ldots,\sum_i a_{i n})^\T \in \R^n$.
Likewise, $\mathrm{rowsums}$ denotes the row sums.
For a vector $x\in\R^{m n}$, we define $\mathrm{reshape}(x,m,n)$ to be the matrix $A\in\R^{m\times n}$ such that $x = \mathrm{vec}(A)$.
For a vector $x\in\R^n$, we write $\Diag(x)$ to denote the $n\times n$ diagonal matrix with $x$ on the diagonal.
For a matrix $A\in\R^{m\times n}$ and vectors $x\in\R^m$, $y\in\R^n$, with $m\neq n$, we extend the $\odot$ operator as follows:
$A \odot x = x \odot A := \Diag(x) A$ and $A \odot y = y \odot A := A\, \Diag(y)$.
We write $\mathrm{sqrt}(\cdot)$ to denote the element-wise square root.
We use $\Id_n$ to denote the $n\times n$ identity matrix.
    
\bigskip
\noindent\textbf{Preprocessing.}
    \begin{enumerate}[label=\textup{(\arabic*)},noitemsep,topsep=0pt]
    \item Compute the inverse dispersions $r_{i j} \gets \exp(-s_i-t_j-\omega)$ for all $i,j$.
    \item Compute $\mu$, $W$, and $E$ as in Section~\ref{section:estimation-details}.
    \item Compute $\tt{dWM} \in \R^{I\times J}$ where $\tt{dWM}_{i j} \gets \mu_{i j} r_{i j}^2 / (r_{i j} + \mu_{i j})^2$.
    \item Compute $\tt{dEM} \in \R^{I\times J}$ where $\tt{dEM}_{i j} \gets -\mu_{i j} r_{i j} (r_{i j} + Y_{i j}) / (r_{i j} + \mu_{i j})^2$.
    \item $\tt{gradA} \gets E^\T X$
    \item $\tt{gradB} \gets E Z$
    \item $\tt{gradC} \gets X^\T E Z$
    \item Compute $\delta_{i j}$ and $\delta'_{i j}$ for all $i,j$ using the formula from the $S$ update during estimation.
    \item $\tt{gradS}_i \gets -\lambda_s (s_i - m_s) + \sum_{j = 1}^J \delta_{i j}$ for $i = 1,\ldots,I$.
    \item $\tt{gradT}_j \gets -\lambda_t (t_j - m_t) + \sum_{i = 1}^I \delta_{i j}$ for $j = 1,\ldots,J$.
    \end{enumerate}
\bigskip
\noindent\textbf{Compute conditional uncertainty for each component.}
    \begin{enumerate}[label=\textup{(\arabic*)},noitemsep,topsep=0pt]
    \item $\tt{invFa}_j \gets (X^\T W_{* j} X + \lambda_a \Id)^{-1}$ for $j = 1,\ldots,J$.
    \item $\tt{invFb}_i \gets (Z^\T W_{i *} Z + \lambda_b \Id)^{-1}$ for $i = 1,\ldots,I$.
    % \item $\tt{Fc} \gets \mathrm{block}\Big({\textstyle\sum_{j=1}^J} z_{j\ell}z_{j\ell'}(X^\T W_{* j} X) : \ell,\ell'\in\{1,\ldots,L\}\Big)$.
    \item $\tt{invFc} \gets (\E(-\nabla_{\!\vec{c}}^2\,\L) + \lambda_c \Id)^{-1}$ where $\E(-\nabla_{\!\vec{c}}^2\,\L)$ is given in Equation~\ref{equation:edf-gbm-fisher}.
    \item $\tt{invFu}_i \gets ((V D)^\T W_{i *} (V D) + \lambda_u \Id)^{-1}$ for $i = 1,\ldots,I$.
    \item $\tt{invFv}_j \gets ((U D)^\T W_{* j} (U D) + \lambda_v \Id)^{-1}$ for $j = 1,\ldots,J$.
    \item $\tt{invFs}_i \gets 1/(\lambda_s - \sum_{j=1}^J \delta'_{i j})$ for all $i = 1,\ldots,I$.
    \item $\tt{invFt}_j \gets 1/(\lambda_t - \sum_{i=1}^I \delta'_{i j})$ for all $j = 1,\ldots,J$.
    \item $\tt{invFs} \gets (\tt{invFs}_1,\ldots,\tt{invFs}_I)^\T$
    \item $\tt{invFt} \gets (\tt{invFt}_1,\ldots,\tt{invFt}_J)^\T$
    \end{enumerate}
\bigskip
\noindent\textbf{Compute constraint Jacobians for $U$ and $V$.}
    \begin{enumerate}[label=\textup{(\arabic*)},noitemsep,topsep=0pt]
    % \item $\tt{Duf}_i \gets X_{i \bc} \otimes \Id_M$ for $i = 1,\ldots,I$.
    % \item $\tt{Dug}_i \gets (U_{i \bc} \otimes \Id_M) + (\Id_M \otimes U_{i \bc}) $ for $i = 1,\ldots,I$.
    % \item $\tt{Ju}_i \gets \mathrm{vcat}(\tt{Duf}_i, \,\tt{Dug}_i)$ for $i = 1,\ldots,I$.
    \item $\tt{Ju}_i \gets \mathrm{vcat}(X_{i \bc} \otimes \Id_M, \; (U_{i \bc} \otimes \Id_M) + (\Id_M \otimes U_{i \bc}))$ for $i = 1,\ldots,I$.
    \item $\tt{Ju} \gets \mathrm{hcat}(\tt{Ju}_1, \ldots, \tt{Ju}_I)$
    
    % \item $\tt{Dvf}_j \gets Z_{j \bc} \otimes \Id_M$ for $j = 1,\ldots,J$.
    % \item $\tt{Dvg}_j \gets (V_{j \bc} \otimes \Id_M) + (\Id_M \otimes V_{j \bc}) $ for $j = 1,\ldots,J$.
    % \item $\tt{Jv}_j \gets \mathrm{vcat}(\tt{Dvf}_j, \,\tt{Dvg}_j)$ for $j = 1,\ldots,J$.
    \item $\tt{Jv}_j \gets \mathrm{vcat}(Z_{j \bc} \otimes \Id_M, \; (V_{j \bc} \otimes \Id_M) + (\Id_M \otimes V_{j \bc}))$ for $j = 1,\ldots,J$.
    \item $\tt{Jv} \gets \mathrm{hcat}(\tt{Jv}_1, \ldots, \tt{Jv}_J)$
    
    % \item $\tt{Dvf} \gets Z^\T \otimes \Id_M$
    % \item $\tt{Dvg} \gets (V^\T \otimes \Id_M) + [(\Id_M \otimes V_{1 \bc}),\ldots,(\Id_M \otimes V_{J \bc})]$
    % \item $\tt{Jv} \gets [\tt{Dvf}^\T \;\, \tt{Dvg}^\T]^\T$
    \end{enumerate}
\bigskip
\noindent\textbf{Compute joint uncertainty in $(U,V)$ accounting for constraints.}
    \begin{enumerate}[label=\textup{(\arabic*)},noitemsep,topsep=0pt]
    % \item $\tt{Fuv} \gets \mathrm{block}\big(w_{i j} (D V_{j\bc}) (D U_{i\bc})^\T \;:\; i\in\{1,\ldots,I\},\, j\in\{1,\ldots,J\}\big)$
    \item $\tt{Fuv}_i \gets \mathrm{hcat}(w_{i 1} (D V_{1\bc}) (D U_{i\bc})^\T, \ldots, w_{i J} (D V_{J\bc}) (D U_{i\bc})^\T)$ for $i = 1,\ldots,I$
    \item $\tt{Fuv} \gets \mathrm{vcat}(\tt{Fuv}_1, \ldots, \tt{Fuv}_I)$
    \item $\tt{FJ} \gets \mathrm{vcat}(\tt{invFu}_1 \times \tt{Ju}_1^\T, \ldots, \tt{invFu}_I \times \tt{Ju}_I^\T)$
    \item $\tt{FuvFJ} \gets \tt{Fuv}^\T \times \tt{FJ}$
    \item $\tt{invJFJ} \gets (\tt{Ju} \times \tt{FJ})^{-1}$
    \item $\tt{FFuv} \gets \mathrm{vcat}(\tt{invFu}_1 \times \tt{Fuv}_1, \ldots, \tt{invFu}_I \times \tt{Fuv}_I)$
    \item $\tt{FuvFFuv} \gets \tt{Fuv}^\T \times \tt{FFuv}$
    \item $\tt{Fv} \gets$ block diagonal matrix with $j$th block equal to $(U D)^\T W_{* j} (U D) + \lambda_v \Id$.
    \item $\tt{A} \gets \tt{Fv} - \tt{FuvFFuv} + \tt{FuvFJ} \times \tt{invJFJ} \times \tt{FuvFJ}^\T$
    % \item $\tt{G} \gets$ leading $J M\times J M$ block of $\displaystyle \begin{bmatrix}\tt{A} & \tt{Jv}^\T \,\\ \tt{Jv} & 0 \,\end{bmatrix}^{-1}$.
    \item $\tt{B} \gets \displaystyle \begin{bmatrix}\tt{A} & \tt{Jv}^\T \,\\ \tt{Jv} & 0 \,\end{bmatrix}^{-1}$
    \item $\tt{C} \gets \tt{B}[1 \! : \! J M, \, 1 \! : \! J M]$ (that is, $\tt{C}$ is the leading $J M\times J M$ block of $\tt{B}$)
    \item $\tt{FuvD} \gets \tt{FFuv}^\T - \tt{FuvFJ}\times \tt{invJFJ} \times \tt{FJ}^\T$
    \item $\tt{d} \gets \mathrm{vcat}(\diag(\tt{invFu}_1), \ldots, \diag(\tt{invFu}_I))$
    \item $\tt{f} \gets \mathrm{colsums}(\tt{FJ}^\T \odot (\tt{invJFJ} \times \tt{FJ}^\T))$
    \item $\tt{g} \gets \mathrm{colsums}(\tt{FuvD} \odot (\tt{C} \times \tt{FuvD}))$
    \item $\tt{varU} \gets \tt{d} - \tt{f} + \tt{g}$
    \item $\tt{varV} \gets \diag(\tt{C})$
    \end{enumerate}
\bigskip
\noindent\textbf{Propagate uncertainty from $U$ to $A$.}
    \begin{enumerate}[label=\textup{(\arabic*)},noitemsep,topsep=0pt]
    % \item Initialize $\tt{dA} \in \R^{J K \times I M}$ to all zeros.
    \item For $j = 1,\ldots,J$,
        \begin{enumerate}[label=\textup{(\alph*)},noitemsep,topsep=0pt]
        \item $\tt{Q}_j \gets (-\tt{invFa}_j \times (X \odot \tt{dWM}_{\bc j})^\T) \odot (X \times \tt{invFa}_j \times \tt{gradA}_{j\bc})^\T + \tt{invFa}_j \times (X \odot \tt{dEM}_{\bc j})^\T$
        % \item For $m = 1,\ldots,M$,
            % \begin{enumerate}
            % \item $S \gets (m, m + M, m + 2 M, \ldots, m + (I-1) M)$
            % \item[] $\tt{dA}[\mathrm{block}(j,K),S] \gets \tt{Q}_j V_{j m} D_{m m}$
            % \end{enumerate}
        \end{enumerate}
    \item $\tt{dA} \gets \mathrm{vcat}(\tt{Q}_1 \otimes (D V_{1 \bc})^\T, \ldots, \tt{Q}_J \otimes (D V_{J \bc})^\T)$
    \item $\tt{varAfromU} \gets \mathrm{colsums}(\tt{dA}^\T \odot (\tt{varU} \odot \tt{dA}^\T))$
    \end{enumerate}
\bigskip
\noindent\textbf{Propagate uncertainty from $V$ to $A$.}
    \begin{enumerate}[label=\textup{(\arabic*)},noitemsep,topsep=0pt]
    \item For $j = 1,\ldots,J$,
    \begin{enumerate}[label=\textup{(\alph*)},noitemsep,topsep=0pt]
        \item Initialize $\tt{dA} \in \R^{K \times M}$ to all zeros.
        \item For $m = 1,\ldots,M$,
            \begin{enumerate}[label=\textup{(\roman*)},noitemsep,topsep=0pt]
            % \item $\tt{dW} \gets $
            \item $\tt{XdE} \gets X^\T (\tt{dEM}_{\bc j} \odot (D_m U_{\bc m}))$
            \item $\tt{XdWX} \gets (X^\T ((\tt{dWM}_{\bc j} \odot (D_m U_{\bc m})) \odot X))$
            \item $\tt{dA}_{\bc m} \gets (-\tt{invFa}_j \times \tt{XdWX}) \times (\tt{invFa}_j \times \tt{gradA}_{j\bc}) + \tt{invFa}_j \times \tt{XdE}$
            \end{enumerate}
        \item $\tt{varAfromV}_j \gets \mathrm{colsums}(\tt{dA}^\T \odot (\tt{varV}[\mathrm{block}(j,M)] \odot \tt{dA}^\T))$
        \end{enumerate}
    \item $\tt{varAfromV} \gets \mathrm{vcat}(\tt{varAfromV}_1,\ldots,\tt{varAfromV}_J)$
    \end{enumerate}
\bigskip
\noindent\textbf{Propagate uncertainty from $U$ and $V$ to $B$.}
    \begin{enumerate}[label=\textup{(\arabic*)},noitemsep,topsep=0pt]
    \item Computing $\tt{varBfromU}$ is identical to calculating $\tt{varAfromV}$, but with 
    $Z$, $V$, $\tt{varU}$, $\tt{invFb}$, $\tt{gradB}$, $\tt{dWM}^\T$, and $\tt{dEM}^\T$ in place of 
    $X$, $U$, $\tt{varV}$, $\tt{invFa}$, $\tt{gradA}$, $\tt{dWM}$, and $\tt{dEM}$, respectively.
    \item Computing $\tt{varBfromV}$ is identical to calculating $\tt{varAfromU}$, but with 
    $Z$, $U$, $\tt{varV}$, $\tt{invFb}$, $\tt{gradB}$, $\tt{dWM}^\T$, and $\tt{dEM}^\T$ in place of 
    $X$, $V$, $\tt{varU}$, $\tt{invFa}$, $\tt{gradA}$, $\tt{dWM}$, and $\tt{dEM}$, respectively.
    \end{enumerate}
\bigskip
\noindent\textbf{Propagate uncertainty from $A$ to $C$.}
    \begin{enumerate}[label=\textup{(\arabic*)},noitemsep,topsep=0pt]
    \item Initialize $\tt{dC} \in \R^{K L \times J K}$ to all zeros.
    \item For $j = 1,\ldots,J$ and $k = 1,\ldots,K$,
        \begin{enumerate}[label=\textup{(\alph*)},noitemsep,topsep=0pt]
        \item $\tt{dF} \gets (Z_{j\bc} Z_{j\bc}^\T) \otimes (X^\T ((\tt{dWM}_{\bc j} \odot X_{\bc k}) \odot X))$
        \item $\tt{dgradC} \gets (X^\T \times (\tt{dEM}_{\bc j} \odot X_{\bc k})) Z_{j \bc}^\T$
        \item $\tt{dC}[\bc,(j-1)K + k] \gets \tt{invFc}\times (-\tt{dF} \times (\tt{invFc}\times \mathrm{vec}(\tt{gradC})) + \mathrm{vec}(\tt{dgradC}))$
        \end{enumerate}
    \item $\tt{invFdC}_j \gets \tt{invFa}_j \times \tt{dC}[\bc,\mathrm{block}(j,K)]^\T$ for $j = 1,\ldots,J$
    \item $\tt{invFdC} \gets \mathrm{vcat}(\tt{invFdC}_1,\ldots,\tt{invFdC}_J)$
    \item $\tt{varCfromA} \gets \mathrm{colsums}(\tt{dC}^\T \odot \tt{invFdC})$
    \end{enumerate}
\bigskip
\noindent\textbf{Propagate uncertainty from $B$ to $C$.}
    \begin{enumerate}[label=\textup{(\arabic*)},noitemsep,topsep=0pt]
    \item Initialize $\tt{dC} \in \R^{K L \times I L}$ to all zeros.
    \item For $i = 1,\ldots,I$ and $\ell = 1,\ldots,L$,
        \begin{enumerate}[label=\textup{(\alph*)},noitemsep,topsep=0pt]
        \item $\tt{dF} \gets (Z^\T ((\tt{dWM}_{i \bc} \odot Z_{\bc \ell}) \odot Z)) \otimes (X_{i\bc} X_{i\bc}^\T)$
        \item $\tt{dgradC} \gets X_{i \bc} ((\tt{dEM}_{i \bc} \odot Z_{\bc \ell})^\T \times Z)$
        \item $\tt{dC}[\bc,(i-1)L + \ell] \gets \tt{invFc}\times (-\tt{dF} \times (\tt{invFc}\times \mathrm{vec}(\tt{gradC})) + \mathrm{vec}(\tt{dgradC}))$
        \end{enumerate}
    \item $\tt{invFdC}_i \gets \tt{invFb}_i \times \tt{dC}[\bc,\mathrm{block}(i,L)]^\T$ for $i = 1,\ldots,I$
    \item $\tt{invFdC} \gets \mathrm{vcat}(\tt{invFdC}_1,\ldots,\tt{invFdC}_I)$
    \item $\tt{varCfromB} \gets \mathrm{colsums}(\tt{dC}^\T \odot \tt{invFdC})$
    \end{enumerate}
\bigskip
\noindent\textbf{Compute approximate variances for $A$, $B$, and $C$.}
    \begin{enumerate}[label=\textup{(\arabic*)},noitemsep,topsep=0pt]
    \item $\tt{varA} \gets \mathrm{vcat}(\diag(\tt{invFa}_1), \ldots, \diag(\tt{invFa}_J)) + \tt{varAfromU} + \tt{varAfromV}$
    \item $\tt{varB} \gets \mathrm{vcat}(\diag(\tt{invFb}_1), \ldots, \diag(\tt{invFb}_I)) + \tt{varBfromU} + \tt{varBfromV}$
    \item $\tt{varC} \gets \diag(\tt{invFc}) + \tt{varCfromA} + \tt{varCfromB}$
    \end{enumerate}
\bigskip
\noindent\textbf{Propagate uncertainty from $(A,B,U,V)$ to $S$.}\\
First, we describe how to compute $\tt{varSfromU}$ and $\tt{varSfromV}$.
    \begin{enumerate}[label=\textup{(\arabic*)},noitemsep,topsep=0pt]
    \item Compute $Q\in\R^{I\times J}$ where $q_{i j} \gets -w_{i j} e_{i j} / r_{i j}$.
    \item Compute $P\in\R^{I\times J}$ where $p_{i j} \gets 2 w_{i j} q_{i j} / \mu_{i j}$.
    \item $\tt{dgradS} \gets Q V D$
    \item $\tt{dF} \gets \tt{dgradS} - P V D$
    \item $\tt{dS} \gets (-\tt{invFs} \odot \tt{dF} \odot \tt{invFs} \odot \tt{gradS}) + (\tt{invFs}\odot \tt{dgradS})$
    \item $\tt{varSfromU} \gets \mathrm{rowsums}(\tt{dS} \odot \mathrm{reshape}(\tt{varU},M,I)^\T \odot \tt{dS})$
    \item For $i = 1,\ldots,I$,
        \begin{enumerate}[label=\textup{(\alph*)},noitemsep,topsep=0pt]
        \item $\tt{dgradS} \gets Q_{i\bc} (D U_{i \bc})^\T$
        \item $\tt{dF} \gets \tt{dgradS} - P_{i\bc} (D U_{i \bc})^\T$
        \item $\tt{dS} \gets -\tt{invFs}_i \cdot \tt{dF} \cdot \tt{invFs}_i \cdot \tt{gradS}_i + \tt{invFs}_i \cdot \tt{dgradS}$
        \item $\tt{varSfromV}_i \gets \tt{varV}^\T \times \mathrm{vec}((\tt{dS} \odot \tt{dS})^\T)$
        \end{enumerate}
    \item $\tt{varSfromV} \gets (\tt{varSfromV}_1,\ldots,\tt{varSfromV}_I)^\T$
    \end{enumerate}
Next, $\tt{varSfromB}$ and $\tt{varSfromA}$ are computed in exactly the same way as $\tt{varSfromU}$ and $\tt{varSfromV}$, respectively,
but with $X$, $Z$, $\Id$, $\tt{varB}$, and $\tt{varA}$ in place of $U$, $V$, $D$, $\tt{varU}$, and $\tt{varV}$, respectively.

\bigskip
\noindent\textbf{Propagate uncertainty from $(A,B,U,V)$ to $T$.}
    \begin{enumerate}[label=\textup{(\arabic*)},noitemsep,topsep=0pt]
    \item We compute $\tt{varTfromV}$ and $\tt{varTfromU}$ in exactly the same way as $\tt{varSfromU}$ and $\tt{varSfromV}$, respectively, but with
    $Y^\T$, $\mu^\T$, $W^\T$, $E^\T$, $r^\T$, $V$, $U$, $\tt{gradT}$, $\tt{invFt}$, $\tt{varV}$, and $\tt{varU}$ in place of
    $Y$, $\mu$, $W$, $E$, $r$, $U$, $V$, $\tt{gradS}$, $\tt{invFs}$, $\tt{varU}$, and $\tt{varV}$, respectively.
    \item We compute $\tt{varTfromA}$ and $\tt{varTfromB}$ in exactly the same way as $\tt{varSfromU}$ and $\tt{varSfromV}$, respectively, but with
    $Y^\T$, $\mu^\T$, $W^\T$, $E^\T$, $r^\T$, $Z$, $X$, $\Id$, $\tt{gradT}$, $\tt{invFt}$, $\tt{varA}$, and $\tt{varB}$ in place of
    $Y$, $\mu$, $W$, $E$, $r$, $U$, $V$, $D$, $\tt{gradS}$, $\tt{invFs}$, $\tt{varU}$, and $\tt{varV}$, respectively.
    \end{enumerate}
    
\bigskip
\noindent\textbf{Compute approximate standard errors.}
    \begin{enumerate}[label=\textup{(\arabic*)},noitemsep,topsep=0pt]
    \item $\hat{\mathrm{se}}_A \gets \mathrm{reshape}(\mathrm{sqrt}(\tt{varA}), K,J)^\T$
    \item $\hat{\mathrm{se}}_B \gets \mathrm{reshape}(\mathrm{sqrt}(\tt{varB}), L,I)^\T$
    \item $\hat{\mathrm{se}}_C \gets \mathrm{reshape}(\mathrm{sqrt}(\tt{varC}), K,L)$
    \item $\hat{\mathrm{se}}_U \gets \mathrm{reshape}(\mathrm{sqrt}(\tt{varU}), M,I)^\T$
    \item $\hat{\mathrm{se}}_V \gets \mathrm{reshape}(\mathrm{sqrt}(\tt{varV}), M,J)^\T$
    \item $\hat{\mathrm{se}}_S \gets \mathrm{sqrt}(\tt{invFs} + \tt{varSfromA} + \tt{varSfromB} + \tt{varSfromU} + \tt{varSfromV})$
    \item $\hat{\mathrm{se}}_T \gets \mathrm{sqrt}(\tt{invFt} + \tt{varTfromA} + \tt{varTfromB} + \tt{varTfromU} + \tt{varTfromV})$
    \end{enumerate}
We do not attempt to provide standard errors for $D$, since it seems difficult to estimate $D$ without significant bias.
Note that here, we reshape the vectorized standard errors to matrices having the same dimensions as the corresponding components,
for instance, $\hat{\mathrm{se}}_A$ has the same dimensions as $A$, namely $J\times K$.

\section{Priors for regularization}
\label{section:priors}

We place independent normal priors on all the entries of $A$, $B$, $C$, $D$, $U$, and $V$, 
and in the NB-GBM, on the entries of $S$ and $T$ as well.
Specifically, the prior is $\pi(A,B,C,D,U,V,S,T) = \pi_a(A)\pi_b(B)\pi_c(C)\pi_d(D)\pi_u(U)\pi_v(V)\pi_s(S)\pi_t(T)$ where
\begin{align}
\label{equation:priors}
\pi_a(A) &= {\textstyle\prod_{j,k}\;}\N(a_{j k}\mid 0,\lambda_a^{-1})       \qquad & \pi_u(U) &= {\textstyle\prod_{i,m}\;}\N(u_{i m}\mid 0,\lambda_u^{-1}) \notag\\
\pi_b(B) &= {\textstyle\prod_{i,l}\;}\N(b_{i \ell}\mid 0,\lambda_b^{-1})    \qquad & \pi_v(V) &= {\textstyle\prod_{j,m}\;}\N(v_{j m}\mid 0,\lambda_v^{-1}) \\
\pi_c(C) &= {\textstyle\prod_{k,\ell}\;}\N(c_{k \ell}\mid 0,\lambda_c^{-1}) \qquad & \pi_s(S) &= {\textstyle\prod_{i}\;}\N(s_{i}\mid m_s,\lambda_s^{-1})\notag\\
\pi_d(D) &= {\textstyle\prod_{m}\;}\N(d_{m m}\mid 0,\lambda_d^{-1})         \qquad & \pi_t(T) &= {\textstyle\prod_{j}\;}\N(t_{j}\mid m_t,\lambda_t^{-1}).\notag
\end{align}
Thus, the log-prior is
\begin{align}
\label{equation:log-prior}
\begin{split}
\log\pi &= \mathrm{const}
- \tfrac{1}{2}\lambda_a \sum_{j,k} a_{j k}^2
- \tfrac{1}{2}\lambda_b \sum_{i,\ell} b_{i \ell}^2
- \tfrac{1}{2}\lambda_c \sum_{k,\ell} c_{k \ell}^2
- \tfrac{1}{2}\lambda_d \sum_{m} d_{m m}^2 \\
&~~~~ - \tfrac{1}{2}\lambda_u \sum_{i,m} u_{i m}^2
- \tfrac{1}{2}\lambda_v \sum_{j,m} v_{j m}^2
- \tfrac{1}{2}\lambda_s \sum_{i} (s_i - m_s)^2
- \tfrac{1}{2}\lambda_t \sum_{j} (t_j - m_t)^2.
\end{split}
\end{align}
For the prior parameters, we use the following default settings:
$\lambda_a = \lambda_b = \lambda_c = \lambda_d = \lambda_u = \lambda_v = 1$,
$m_s = m_t = 0$, and $\lambda_s = \lambda_t = 1$.
These defaults are fairly generally applicable 
since they are acting on coefficients that are on the same scale in terms of units, 
due to the fact that we standardize the covariates to have zero mean and unit variance,
that is, $\frac{1}{I}\sum_{i=1}^I x_{i k} = 0$ and $\frac{1}{I}\sum_{i=1}^I x_{i k}^2 = 1$ for all $k\geq 2$
and $\frac{1}{J}\sum_{j=1}^J z_{j \ell} = 0$ and $\frac{1}{J}\sum_{j=1}^J z_{j \ell}^2 = 1$ for all $\ell\geq 2$.
However, specific applications may call for departures from these defaults.

\iffalse
In the optimization algorithm, we use the gradients and Fisher information matrices based on the log-likelihood plus the log-prior.
To account for the log-prior term, it is straightforward to check that, for example, 
% $$ \nabla_{\!\vec{a}} \log\pi = -\lambda_a \vec{a} \text{~~~~and~~~~} -\nabla_{\!\vec{a}}^2 \log\pi = \lambda_a $$
$$ \frac{\partial}{\partial a_{j k}} \log\pi = -\lambda_a a_{j k} \text{~~~~~~~} -\frac{\partial^2}{\partial a_{j k}^2} \log\pi = \lambda_a $$
with corresponding expressions for the other components.
Also note that for any two distinct entries of $A$, $B$, $C$, $D$, $U$, $V$, and $S$, say $\alpha$ and $\beta$,
we have $-(\partial^2 / \partial\alpha\partial\beta) \log\pi = 0$; in other words, the Hessian of $\log\pi$ is a diagonal matrix.
\fi

For the updates to $G = U D$ and $H = V D$ in the GBM estimation algorithm (Section~\ref{section:estimation-details}),
we use the priors on $G$ and $H$ induced by the priors on $U$ and $V$, given $D$.
First consider $G$.
For any fixed $D$, the induced prior on $g_{i m} = u_{i m} d_{m m}$ is $\pi(g_{i m}) = \N(g_{i m} \mid 0, d_{m m}^2 / \lambda_u)$.
Thus, given $D$, the prior on each row of $G$ is $G_{i\bc} \sim \N(0,\Lambda^{-1})$ where $\Lambda = (D^2 / \lambda_u)^{-1}$.
The gradient and Hessian of the log-prior on $G_{i\bc}$ are therefore $-\Lambda G_{i \bc}$ and $-\Lambda$, respectively.
Hence, with this prior, the regularized Fisher scoring approach (Equation~\ref{equation:generic-fisher-update})
yields the $G$ update formulas used in the algorithm (Section~\ref{section:estimation-details}).
The $H$ update is similar, except that the induced prior is $H_{j\bc} \sim \N(0,\Lambda^{-1})$ where $\Lambda = (D^2 / \lambda_v)^{-1}$.
It seems reasonable to hold $D$ fixed when computing the induced priors on $G$ and $H$, rather than integrating it out, 
since $D$ tends to be more accurately estimated than $U$ or $V$ (in terms of relative MSE), presumably due to the fact that $D$ has only $M$ nonzero entries, 
each of which is informed by all of the data; see Figure~\ref{figure:consistency-all} for an empirical example.

% $\nabla_{\! G_{i\bc}} \log \N(G_{i\bc} \mid 0, \Lambda^{-1}) = -\Lambda G_{i \bc}$ and 

\section{Proofs}
\label{section:proofs}

\subsection{Identifiability and interpretation}

\begin{proof}[Proof of Theorem~\ref{theorem:identifiability}]
Left-multiplying both sides of Equation~\ref{equation:identifiability} by $X^\T$, we have 
\begin{equation}
\label{equation:identifiability-1}
X^\T X A_1^\T + X^\T X C_1 Z^\T = X^\T X A_2^\T + X^\T X C_2 Z^\T
\end{equation}
by Condition~\ref{condition:orthogonal}.  Since $X^\T X$ is invertible, this implies 
\begin{equation}
\label{equation:identifiability-2}
A_1^\T + C_1 Z^\T = A_2^\T + C_2 Z^\T.
\end{equation}
Right-multiplying Equation~\ref{equation:identifiability-2} by $Z$, we have $C_1 Z^\T Z = C_2 Z^\T Z$ by Condition~\ref{condition:orthogonal}.
Since $Z^\T Z$ is invertible, this implies that $C_1 = C_2$.
Plugging $C_1 = C_2$ into Equation~\ref{equation:identifiability-2} yields $A_1 = A_2$.
Plugging $A_1 = A_2$ and $C_1 = C_2$ into Equation~\ref{equation:identifiability}, we have
\begin{equation}
\label{equation:identifiability-3}
B_1 Z^\T + U_1 D_1 V_1^\T = B_2 Z^\T + U_2 D_2 V_2^\T.
\end{equation}
Right-multiplying Equation~\ref{equation:identifiability-3} by $Z$, using Condition~\ref{condition:orthogonal}, and using the fact that $Z^\T Z$ is invertible, we obtain $B_1 = B_2$.
This implies that $U_1 D_1 V_1^\T = U_2 D_2 V_2^\T$.
By the uniqueness properties of the singular value decomposition, Conditions~\ref{condition:orthonormal} and \ref{condition:diagonal}
imply that $D_1 = D_2$, $U_1 = U_2 \mathsf{S}$, and $V_1^\T = \mathsf{S} V_2^\T$ for a diagonal matrix $\mathsf{S}$ of the form $\mathsf{S} = \Diag(\pm 1,\ldots,\pm 1)$ \citep{blum2020foundations}.
By Condition~\ref{condition:sign}, $\mathsf{S} = \Id$.  Therefore, $U_1 = U_2$ and $V_1 = V_2$.
This proves that $A$, $B$, $C$, $D$, $U$, and $V$ are uniquely determined by $\E(\bm{Y})$ for any given $X$, $Z$, $M$.
% Therefore, in particular, they are uniquely determined by the distribution of $\bm{Y}$.
\end{proof}

\begin{proof}[Proof of Theorem~\ref{theorem:interpretation}]
First, $\sum_{j = 1}^J a_{j k} = 0$ follows from the fact that $z_{j 1} = 1$ for all $j$ by Condition~\ref{condition:interpretation}(a)
and $Z^\T A = 0$ by Condition~\ref{condition:orthogonal}.
Likewise, $\sum_{i = 1}^I b_{i \ell} = 0$ follows from $x_{i 1} = 1$ and $X^\T B = 0$.
In the same way, $\sum_{i = 1}^I u_{i m} = 0$ and $\sum_{j = 1}^J v_{j m} = 0$
since $x_{i 1} = 1$, $z_{j 1} = 1$, $X^\T U = 0$, and $Z^\T V = 0$.

When Condition~\ref{condition:interpretation}(a) holds, we can rearrange Equation~\ref{equation:model-univariate} as
\begin{align}
\label{equation:interpretation}
g(\mu_{i j}) &= 
c_{1 1} + a_{j 1} + b_{i 1} + \sum_{k = 2}^K (c_{k 1} + a_{j k}) x_{i k} + \sum_{\ell = 2}^L (c_{1 \ell} + b_{i \ell}) z_{j \ell} \notag \\
&~~~~ + \sum_{k = 2}^K \sum_{\ell = 2}^L c_{k \ell} x_{i k} z_{j \ell} + \sum_{m = 1}^M u_{i m} d_{m m} v_{j m}. 
\end{align}
Averaging Equation~\ref{equation:interpretation} over all $i$, and using these sum-to-zero properties
(specifically, using that $\sum_{i = 1}^I x_{i k} = 0$ for $k\geq 2$, $\sum_{i = 1}^I b_{i \ell} = 0$, and $\sum_{i = 1}^I u_{i m} = 0$),
\begin{align}
\label{equation:interpretation-feature-mean}
\frac{1}{I}\sum_{i = 1}^I g(\mu_{i j}) &= c_{1 1} + a_{j 1} + ({\textstyle \frac{1}{I}\sum_i}\, b_{i 1}) + \sum_{k = 2}^K (c_{k 1} + a_{j k}) ({\textstyle \frac{1}{I}\sum_i}\, x_{i k}) + \sum_{\ell = 2}^L \big(c_{1 \ell} + ({\textstyle \frac{1}{I}\sum_i}\, b_{i \ell})\big) z_{j \ell}  \notag\\
&~~~~ + \sum_{k = 2}^K \sum_{\ell = 2}^L c_{k \ell} ({\textstyle \frac{1}{I}\sum_i}\, x_{i k}) z_{j \ell} + \sum_{m = 1}^M ({\textstyle \frac{1}{I}\sum_i}\, u_{i m}) d_{m m} v_{j m} \notag\\
&= c_{1 1} + a_{j 1} + \sum_{\ell = 2}^L c_{1 \ell} z_{j \ell}.
\end{align}
In the same way, averaging Equation~\ref{equation:interpretation} over all $j$ 
(and using that $\sum_{j = 1}^J z_{j \ell} = 0$ for $\ell\geq 2$, $\sum_{j = 1}^J a_{j k} = 0$, and $\sum_{j = 1}^J v_{j m} = 0$),
we have
$$ \frac{1}{J}\sum_{j = 1}^J g(\mu_{i j}) = c_{1 1} + b_{i 1} + \sum_{k = 2}^K c_{k 1} x_{i k}. $$
Finally, averaging Equation~\ref{equation:interpretation-feature-mean} over all $j$, we have
$\frac{1}{I J}\sum_{i = 1}^I \sum_{j = 1}^J g(\mu_{i j}) = c_{1 1}$.
\end{proof}

\begin{proof}[Proof of Theorem~\ref{theorem:sum-of-squares}]
For any $Q\in\R^{m\times n}$, we have $\mathrm{SS}(Q) = \mathrm{tr}(Q^\T Q)$, where $\mathrm{tr}(\cdot)$ denotes the trace.
Define $Q = X A^\T + B Z^\T + X C Z^\T + U D V^\T$.  By using $X^\T B = 0$, $Z^\T A = 0$, $X^\T U = 0$, and $Z^\T V = 0$, we have that
\begin{align*} Q^\T Q &= (A X^\T + Z B^\T + Z C^\T X^\T + V D U^\T) (X A^\T + B Z^\T + X C Z^\T + U D V^\T) \\
&= A X^\T X A + A X^\T X C Z^\T + Z B^\T B Z^\T + Z B^\T U D V^\T + Z C^\T X^\T X A^\T \\
&~~~ + Z C^\T X^\T X C Z^\T + V D U^\T B Z^\T + V D U^\T U D V^\T.
\end{align*}
By the cyclic property of the trace, 
\begin{align*}
\mathrm{tr}(A X^\T X C Z^\T) &= \mathrm{tr}(X C Z^\T A X^\T) = 0, \\
\mathrm{tr}(Z B^\T U D V^\T) &= \mathrm{tr}(B^\T U D V^\T Z) = 0, \\
\mathrm{tr}(Z C^\T X^\T X A^\T) &= \mathrm{tr}(X A^\T Z C^\T X^\T) = 0, \\
\mathrm{tr}(V D U^\T B Z^\T) &= \mathrm{tr}(B Z^\T V D U^\T) = 0.
\end{align*}
Therefore, by the linearity of the trace,
\begin{align*}
\mathrm{SS}(Q) &= \mathrm{tr}(Q^\T Q) = \mathrm{tr}(A X^\T X A^\T) + \mathrm{tr}(Z B^\T B Z^\T) + \mathrm{tr}(Z C^\T X^\T X C Z^\T) + \mathrm{tr}(V D U^\T U D V^\T)  \\
&= \mathrm{SS}(X A^\T) + \mathrm{SS}(B Z^\T) + \mathrm{SS}(X C Z^\T) + \mathrm{SS}(U D V^\T). 
\end{align*}
\end{proof}

\subsection{Likelihood-preserving projections}

\begin{proof}[Proof of Theorem~\ref{theorem:projections}]
(1.) For the projection of $\tilde{A}$, plugging in the definitions of $A_1$ and $C_1$, we have
$$ X A_1^\T + X C_1 Z^\T = X \tilde{A}^\T - X(Z^\ps \tilde{A})^\T Z^\T + X C Z^\T + X (Z^\ps \tilde{A})^\T Z^\T = X \tilde{A}^\T + X C Z^\T, $$
and therefore, $\eta(A_1,B,C_1,D,U,V) = \eta(\tilde{A},B,C,D,U,V)$.
To see that Condition~\ref{condition:identifiability} is satisfied, first note that 
$Z^\T (\Id - Z Z^\ps) = Z^\T - Z^\T Z (Z^\T Z)^{-1} Z^\T = 0$, and therefore 
$$ Z^\T A_1 = Z^\T (\tilde{A} - Z Z^\ps \tilde{A}) = Z^\T (\Id - Z Z^\ps) \tilde{A} = 0. $$
(2.) Similarly, for the projection of $\tilde{B}$, we have $B_1 Z^\T + X C_1 Z^\T = \tilde{B} Z^\T + X C Z^\T$ and $X^\T B_1 = 0$.
(3.) For the projection of $\tilde{G}$, we have
\begin{align*}
X A_1^\T + X C_1 Z^\T + U_1 D_1 V_1^\T &= X A_0^\T - X(Z^\ps A_0)^\T Z^\T + X C Z^\T + X(Z^\ps A_0)^\T Z^\T + G_0 V^\T \\
&= X A^\T + X(X^\ps \tilde{G}) V^\T + X C Z^\T + \tilde{G} V^\T - X(X^\ps \tilde{G})V^\T \\
&= X A^\T + X C Z^\T + \tilde{G} \Id V^\T,
\end{align*}
and thus, $\eta(A_1,B,C_1,D_1,U_1,V_1) = \eta(A,B,C,\Id,\tilde{G},V)$. 
To check that Condition~\ref{condition:identifiability} is satisfied, first observe that 
$Z^\T A_1 = Z^\T (\Id - Z Z^\ps) A_0 = 0$ and $X^\T G_0 = X^\T (\Id - X X^\ps) \tilde{G} = 0$.
Hence,
$$ 0 \stackrel{\text{(a)}}{=} X^\T G_0 V^\T V_1 D_1^{-1} \stackrel{\text{(b)}}{=} X^\T U_1 D_1 V_1^\T V_1 D_1^{-1} \stackrel{\text{(c)}}{=} X^\T U_1 $$
where we have used (a) $X^\T G_0 = 0$, (b) $G_0 V^\T = U_1 D_1 V_1^\T$, 
and (c)  $V_1^\T V_1 = \Id$ and $D_1 D_1^{-1} = \Id$. Likewise, since $V^\T Z = 0$ by assumption,
$$ 0 = D_1^{-1} U_1^\T G_0 V^\T Z = D_1^{-1} U_1^\T U_1 D_1 V_1^\T Z = V_1^\T Z $$
since $U_1^\T U_1 = \Id$ and $D_1^{-1} D_1 = \Id$.

(4.) For the projection of $\tilde{H}$, in an altogether similar way, we have 
$$ B_1 Z^\T + X C_1 Z^\T + U_1 D_1 V_1^\T = B Z^\T + X C Z^\T + U \Id \tilde{H}^\T. $$
Further, $X^\T B_1 = X^\T (\Id - X X^\ps) B_0 = 0$ and $Z^\T H_0 = 0$, thus
\begin{align*}
0 &= Z^\T H_0 U^\T U_1 D_1^{-1} = Z^\T V_1 D_1 U_1^\T U_1 D_1^{-1} = Z^\T V_1 \\
0 &= X^\T U H_0^\T V_1 D_1^{-1} = X^\T U_1 D_1 V_1^\T V_1 D_1^{-1} = X^\T U_1.
\end{align*}
\end{proof}

Computationally, it is highly advantageous to structure the calculation of the projections in Theorem~\ref{theorem:projections} as follows.
First, one can precompute the pseudoinverses $X^\ps$ and $Z^\ps$ since $X$ and $Z$ are fixed throughout the algorithm.
In the updates to $A$ (or $B$), it is advantageous to first compute $Z^\ps \tilde{A}$ (or $X^\ps \tilde{B}$, respectively)
in order to avoid explicitly computing and storing $X X^\ps \in \R^{I\times I}$ and $Z Z^\ps \in \R^{J\times J}$.
Likewise, in the projection of $\tilde{G}$ (or $\tilde{H}$), first compute $X^\ps \tilde{G}$ and $Z^\ps A_0$ (or $Z^\ps \tilde{H}$ and $X^\ps B_0$, respectively).
We use this approach in the step-by-step algorithm provided in Section~\ref{section:estimation-details}.
% \todo{Could other projections be used instead, e.g., by weighting $X$ and $Z$? Something like in Choulakian (1996)?  If so, mention this.}

The interpretation of the operations in Theorem~\ref{theorem:projections} is as follows.
For $\tilde{A}$, the idea is that $A_1 = (\Id - Z Z^\ps) \tilde{A}$ is an orthogonal projection of the columns of $\tilde{A}$ onto the nullspace of $Z^\T$,
and $C_1$ is a shifted version of $C$ to compensate for the shift from $\tilde{A}$ to $A_1$.  Likewise for $\tilde{B}$, but with $X$ instead of $Z$.
For $\tilde{G}$, the idea is that (a) $G_0$ is a projection of $\tilde{G}$ onto the nullspace of $X^\T$, (b) the SVD enforces the orthonormality, ordering, and sign constraints on $U_1$, $D_1$, and $V_1$ while maintaining $X^\T U_1 = 0$ and $Z^\T V_1 = 0$, (c) $A_0$ compensates for the shift from $\tilde{G}$ to $G_0$,
(d) $A_1$ projects $A_0$ onto the nullspace of $Z^\T$, and (e) $C_1$ compensates for the shift from $A_0$ to $A_1$.
For $\tilde{H}$, the interpretation is similar.

\subsection{Estimation time complexity}
\label{section:estimation-time-derivations}

In this section, we justify the expressions in Section~\ref{section:theory}
giving the time complexity of the estimation algorithm
as a function of $I$, $J$, $K$, $L$, and $M$, assuming $\max\{K^2,L^2,M\} \leq \min\{I,J\}$ (Equation~\ref{equation:cost-dimension-bound}).
The outline of the estimation algorithm is in Section~\ref{section:estimation}, and the step-by-step details are in Section~\ref{section:estimation-details}.
Denote $x \wedge y = \min\{x,y\}$ and $x \vee y = \max\{x,y\}$.
For the updates to each of $A$, $B$, $C$, $D$, $U D$, $V D$, $S$, and $T$, we report the computation time complexity after $\eta$, $\mu$, $W$, and $E$ have been recomputed.

\textbf{Cost of preprocessing and initialization.}
Precomputing the pseudoinverses $X^\ps$ and $Z^\ps$ takes $O(I K^2)$ and $O(J L^2)$ time, respectively;
thus, both are $O(I J)$ by Equation~\ref{equation:cost-dimension-bound}.
Computing $C \gets X^\ps \check{Y} (Z^\ps)^\T$ takes $O(I J (K \wedge L))$ time,
$A \gets (X^\ps \check{Y} - C Z^\T)^\T$ takes $O(I J K)$ time,
and $B \gets \check{Y} (Z^\ps)^\T - X C$ takes $O(I J L)$ time.
Computing $D$, $U$, and $V$ takes $O(I J M)$ time, since
the truncated SVD of rank $M$ for an $I\times J$ matrix can be done in $O(I J M)$ time \citep{halko2011finding}.
Finally, each update to $S$ and $T$ takes $O(I J)$ time (see below).
Thus, overall, preprocessing and initialization takes $O(I J (K \vee L \vee M))$ time.

\textbf{Cost of computing $\eta$, $\mu$, $W$, and $E$.}
Computing $\eta = X A^\T + B Z^\T + X C Z^\T + U D V^\T$ takes $O(I J (K \vee L \vee M))$ time, 
since $X A^\T$, $B Z^\T$, $X C Z^\T$, and $U D V^\T$ take $O(I J K)$, $O(I J L)$, $O(I J (K \wedge L))$, and $O(I J M)$ time, respectively.
Computing $\mu$, $W$, and $E$ takes $O(I J)$ time given $\eta$.
% $\sigma_{i j}^2 = r_{i j} \kappa''(\theta_{i j})$ 
% $\theta_{i j} = \kappa'^{-1}(\mu_{i j}/r_{i j})$

\textbf{Cost of updating $A$.} For each $j$, computing the Fisher scoring step takes $O(I K^2)$ time, so altogether the $J$ steps take $O(I J K^2)$ time.
For the projection, we compute $Q \gets Z^\ps A$ and $Z Q$, which takes $O(J K L)$ time.  
By Equation~\ref{equation:cost-dimension-bound}, we have $L \leq I$, so the cost of computing the projection 
can be absorbed into the cost of the Fisher scoring steps.

\textbf{Cost of updating $B$.} By symmetry, this takes $O(I J L^2)$ time.

\textbf{Cost of updating $C$.} 
Computing the Fisher information matrix $F$ takes $O(I J K^2 + J K^2 L^2)$ time, which is $O(I J K^2)$ 
since $L^2 \leq I$ by Equation~\ref{equation:cost-dimension-bound}.
Inverting $F + \lambda_c \Id$ takes $O((K L)^3) = O(I J K L)$ time (using Equation~\ref{equation:cost-dimension-bound}), 
and computing $X^\T E Z$ takes $O(I J K)$ time,
so the update to $C$ can be done in $O(I J (K^2 \vee L^2)$ time.

\textbf{Cost of updating $D$.}
Computing the Fisher information matrix $F$ takes $O(I J M^2)$ time, inverting $F + \lambda_d \Id$ takes $O(M^3)$ time,
and computing $\diag(U^\T E V)$ takes $O(I J M)$ time.
Thus, the update to $D$ takes $O(I J M^2)$ time.

\textbf{Cost of updating $G = U D$.} By comparison with the $B$ update, the Fisher scoring steps cost $O(I J M^2)$ time.
The projection steps (except for the SVD) take $O(I K M + J K M + J K L)$ time,
and by Equation~\ref{equation:cost-dimension-bound}, this is $O(I J M)$.
% we have $K\leq J$, $K \leq I$, and $K L \leq I$, so this is $O(I J M)$.
Computing the truncated SVD of rank $M$ for an $I\times J$ matrix can be done in $O(I J M)$ time \citep{halko2011finding}.
% \citet{larsen1998lanczos}?
Therefore, the cost of the projection can be absorbed into the Fisher scoring steps.

\textbf{Cost of updating $H = V D$.} By symmetry, this takes $O(I J M^2)$ time.

\textbf{Cost of updating $S$ and $T$.}
Given $\mu$, updating $S$ and $T$ takes $O(I J)$ time,
since computing $\delta$ and $\delta'$ involves a loop over all $i$ and $j$.

\subsection{Inference time complexity}
\label{section:inference-time-derivations}

Here, we justify the expressions in Section~\ref{section:theory}
giving the time complexity of the inference algorithm,
assuming $\max\{K^2,L^2,M\} \leq \min\{I,J\}$ (Equation~\ref{equation:cost-dimension-bound}) and also assuming $I \geq J$.
The outline of the inference algorithm is in Section~\ref{section:inference-outline}, and the step-by-step details are in Section~\ref{section:inference-details}.
Denote $x \wedge y = \min\{x,y\}$ and $x \vee y = \max\{x,y\}$.

\textbf{Cost of preprocessing.}
Computing $\eta$, $\mu$, $W$, and $E$ takes $O(I J (K\vee L \vee M))$ time. 
Computing $\tt{gradA}$, $\tt{gradB}$, and $\tt{gradC}$ takes $O(I J K)$, $O(I J L)$, and $O(I J K)$ time, respectively.
All of the other preprocessing steps take $O(I J)$ time.

\textbf{Cost of computing conditional uncertainty for each component.}
Computing $\tt{invFa}$, $\tt{invFb}$, and $\tt{invFc}$ take $O(I J K^2)$, $O(I J L^2)$, and $O(I J (K^2 \vee L^2))$ time, respectively.
Both $\tt{invFu}$ and $\tt{invFv}$ take $O(I J M^2)$ time, and $\tt{invFs}$ and $\tt{invFt}$ take $O(I J)$ time.
Thus, overall, this part is $O(I J (K^2 \vee L^2 \vee M^2))$.

\textbf{Cost of computing constraint Jacobians for $U$ and $V$.}
Computing $\tt{Ju}$ and $\tt{Jv}$ take $O(I(K M^2 + M^3))$ and $O(J (L M^2 + M^3))$ time, respectively.

\textbf{Cost of computing joint uncertainty in $(U,V)$ accounting for constraints.}
The most expensive steps are computing $\tt{FuvFFuv} \gets \tt{Fuv}^\T \times \tt{FFuv}$,
$$\tt{B} \gets \displaystyle \begin{bmatrix}\tt{A} & \tt{Jv}^\T \,\\ \tt{Jv} & 0 \,\end{bmatrix}^{-1},$$
and $\tt{g} \gets \mathrm{colsums}(\tt{FuvD} \odot (\tt{C} \times \tt{FuvD}))$, which take $O(I J^2 M^3)$, $O(J^3 M^3)$, and $O(I J^2 M^3)$ time, respectively.
Since we assume $I \geq J$, these are all $O(I J^2 M^3)$.
It is tedious but straightforward to check that 
all of the other steps in this part take less than $O(I J^2 M^3)$ time, 
assuming $I \geq J$ and $\max\{K^2,L^2,M\} \leq \min\{I,J\}$ (Equation~\ref{equation:cost-dimension-bound}).

\textbf{Cost of propagating uncertainty from $U$ and $V$ to $A$ and $B$.}
Computing $\tt{varAfromU}$ and $\tt{varAfromV}$ take $O(I J K (K\vee M))$ and $O(I J K^2 M)$ time, respectively; combined, this is $O(I J K^2 M)$.
By symmetry, $\tt{varBfromU}$ and $\tt{varBfromV}$ take $O(I J L^2 M)$ time.

\textbf{Cost of propagating uncertainty from $A$ and $B$ to $C$.}
First, consider computing $\tt{varCfromA}$.
Each step in the loop over $j$ and $k$ takes $O(I K^2)$ time (since $L^2 \leq I$), thus, computing $\tt{dC}$ takes $O(I J K^3)$ time altogether.
Computing $\tt{invFdC}$ takes $O(J K^3 L)$ time, and the last step is $O(J K^2 L)$.
Thus, overall, $\tt{varCfromA}$ takes $O(I J K^3)$ time.
By symmetry, $\tt{varCfromB}$ takes $O(I J L^3)$ time.

\textbf{Cost of propagating uncertainty from $(A,B,U,V)$ to $S$.}
Computing $\tt{varSfromA}$, $\tt{varSfromB}$, $\tt{varSfromU}$, and $\tt{varSfromV}$ take $O(I J K)$, $O(I J L)$, $O(I J M)$, and $O(I J M)$ time, respectively.
Thus, overall this takes $O(I J (K \vee L \vee M))$ time.

\textbf{Cost of propagating uncertainty from $(A,B,U,V)$ to $T$.}
By symmetry, $\tt{varTfromA}$, $\tt{varTfromB}$, $\tt{varTfromU}$, and $\tt{varTfromV}$ takes $O(I J (K \vee L \vee M))$ time overall.

\textbf{Cost of computing approximate standard errors.}
Given the approximate element-wise variances, this is only takes as much time as the dimension of the each of the parameter matrices/vectors; 
namely, $O(J K)$, $O(I L)$, $O(K L)$, $O(I M)$, $O(J M)$, $O(I)$, and $O(J)$ for each of $A$, $B$, $C$, $U$, $V$, $S$, and $T$, respectively.
Using Equation~\ref{equation:cost-dimension-bound} to easily upper bound each of these shows that, overall, this is $O(I J)$.

% \spacingset{1}
% \bibliographystyle{abbrvnatcaplf}
% \small
% \bibliography{refs}
% \normalsize
% \spacingset{1.6}

\end{document}